\documentclass[11pt, notitlepage, reqno]{amsart}
\usepackage[utf8]{inputenc}
\usepackage[T1]{fontenc}
\usepackage[british]{babel}

\usepackage{times}
\usepackage{amsaddr, amsbsy, amscd, amsfonts, amsmath, amssymb, amsthm, dsfont, mathrsfs, mathtools}
\usepackage[pdfusetitle,
		pdfsubject={Paper in Mathematical Physics},%
		pdfkeywords={%
                Many-body Quantum Mechanics, Effective Bosonic Dynamics, Mean-Field, High-Density%
            }]{hyperref}
\usepackage{enumitem}
\usepackage{faktor}
\usepackage{xparse}
\usepackage{soul, xcolor}
\usepackage[bottom, hang, flushmargin]{footmisc}
\usepackage{eqnarray}
\usepackage{scalerel, stackengine}
\usepackage{fancyhdr}
\usepackage{indentfirst}
\usepackage[square,comma,numbers]{natbib}

\usepackage[a4paper, hmargin=2cm, vmargin={1.5cm,2.5cm}, nomarginpar, hcentering, includeheadfoot]{geometry}
\usepackage[toc, page]{appendix}
\let\oldtocsection=\tocsection
\let\oldtocsubsection=\tocsubsection

\renewcommand{\tocsection}[2]{%
    \bfseries\oldtocsection{#1}{#2}%
}
\renewcommand{\tocsubsection}[2]{%
  \hspace{1.5em}%
  \oldtocsubsection{#1}{#2}
}

\makeatletter
\renewcommand{\section}{%
    \@startsection{section}%
    {1}
    {0em}
    {1.5cm \@plus 0.1ex \@minus -0.05ex}
    {0.75cm \@plus 0.2em}
    {\centering \Large \scshape}
}

\renewcommand{\subsection}{%
    \@startsection{subsection}%
    {2}
    {0em}
    {0.75cm \@plus 0.1ex \@minus -0.05ex}
    {0.25cm \@plus 0.2em}
    {\bf\large}
}

\renewcommand{\subsubsection}{%
  \@startsection{subsubsection}%
    {3}
    {0em}
    {0.375cm \@plus 0.1ex \@minus -0.05ex}
    {-0.25cm \@plus 0.2em}
    {\normalfont\normalsize\bfseries}
}

\renewcommand{\paragraph}[1]{%
  \par
  \addvspace{\medskipamount}
  \noindent\textit{#1\@addpunct{.}}\quad\ignorespaces
}

\makeatother

\theoremstyle{definition}
\newtheorem{defn}{Definition}
\newtheorem{assumption}{Assumption}
\theoremstyle{plain}
\newtheorem{note}{Remark}
\newtheorem{theo}{Theorem}[section]
\newtheorem{lemma}[theo]{Lemma}
\newtheorem{prop}[theo]{Proposition}
\newtheorem{cor}[theo]{Corollary}

\numberwithin{equation}{section}
\numberwithin{defn}{section}
\numberwithin{note}{section}

\stackMath

\DeclareMathOperator{\var}{\mathbb{V}ar}
\let\oldker\ker
\let\oldforall\forall
\let\oldexists\exists
\renewcommand\widehat[1]{%
\savestack{\tmpbox}{\stretchto{%
    \scaleto{%
        \scalerel*[\widthof{\ensuremath{#1}}]{\kern.1pt\mathchar"0362\kern.1pt}%
        {\rule{0ex}{\textheight}}
    }{\textheight}%
}{2.4ex}}%
\stackon[-6.9pt]{#1}{\tmpbox}%
}

\NewDocumentCommand \fromToParser {m m m}{%
    \IfValueTF{#3}{%
        \fromTo{#1}{#2}[#3]%
    }{%
        \fromTo{#1}{#2}%
    }%
}
\NewDocumentCommand \fromTo {m m o}{%
    #1\! \IfValueTF{#3}{%
            \to[#3]
        }{%
            \to[\!\quad\!]
        } \!#2%
}
\newcommand{\integrand}[2]{\!#2\; #1}
\NewDocumentCommand \integralLimits {m m}{%
    _{#1}\IfValueT{#2}{^{\,#2}\!}%
}
\NewDocumentCommand \defsign {s}{%
    \IfBooleanTF{#1}{%
        =\vcentcolon%
    }{%
        \vcentcolon =%
    }%
}

\RenewDocumentCommand \newline {o}{%
\hfill\\[\IfValueTF{#1}{#1}{0} pt]
}
\newcommand{\nls}{\mspace{-1.5mu}}
\newcommand{\ls}{\mspace{1.5mu}}
\newcommand{\nquad}{\mspace{-18mu}}
\newcommand{\nqquad}{\mspace{-36mu}}

\providecommand{\hide}[1]{}
\providecommand{\say}[1]{``#1''}
\setstcolor{red}

\ProvideDocumentCommand \comment {o m m}{\hl{#2}\footnote{\IfValueT{#1}{#1:} \textcolor{red}{#3.}}}
\renewcommand{\Im}{\mathrm{Im}}
\renewcommand{\Re}{\mathrm{Re}}
\newcommand{\ie}{\emph{i.e.}}
\newcommand{\eg}{\emph{e.g.}}
\newcommand{\cfr}{\emph{cf.}}
\newcommand{\etal}{\emph{et al.}}

\newcommand{\Id}{\mathds{1}}
\renewcommand{\forall}{\oldforall\,}
\RenewDocumentCommand \exists {s}{%
    \IfBooleanTF{#1}{\oldexists!}{\oldexists} \;%
}
\RenewDocumentCommand \to {o}{%
    \IfValueTF{#1}{%
        \xrightarrow[\,#1\,]{\;}%
    }{%
        \,\rightarrow\,%
    }%
}
\NewDocumentCommand \maps {m >{\SplitArgument{2}{;}} m}{%
    #1\!: \fromToParser #2%
}
\ProvideDocumentCommand \conjugate {s m}{%
    \IfBooleanTF{#1}{%
        \overline{#2}%
    }{%
        \bar{#2}%
    }%
}
\NewDocumentCommand \abs {s m}{%
    \IfBooleanTF{#1}{\left\lvert#2\right\rvert}{\lvert#2\rvert}%
}
\renewcommand*{\vec}[1]{\boldsymbol{#1}}

\NewDocumentCommand \ball {s m m}{%
    \IfBooleanTF{#1}{%
        \mathcal{B}^{\ls  c}_{#2}(#3)%
    }{%
        \mathcal{B}_{#2}(#3)%
    }%
}
\NewDocumentCommand \Char { m o }{
    \Id_{#1}
    \IfValueT{#2}{( #2 )}
}

\NewDocumentCommand \oSmall {s o m}{%
    o\IfValueT{#2}{_{#2}}\IfBooleanTF{#1}{%
        \!\left(#3\right)\!%
    }{(#3)}
}

\NewDocumentCommand \oBig {s o m}{%
    \mathcal{O}\IfValueT{#2}{_{#2}}\IfBooleanTF{#1}{%
        \!\left(#3\right)\!%
    }{(#3)}
}

\NewDocumentCommand \integrate { >{\SplitArgument{1}{;}} o >{\SplitArgument{1}{;}} m}{%
    \int \IfValueT{#1}{\integralLimits #1 \!} \integrand #2%
}

\NewDocumentCommand \hilbert{s}{%
    \IfBooleanTF{#1}{\mathfrak{H}}{\mathscr{H}}%
}
\NewDocumentCommand \fock{o}{%
    \mathcal{F}\IfValueT{#1}{(#1)}%
}
\NewDocumentCommand \fockS{s o}{%
    \mathcal{F}\IfBooleanTF{#1}{_+}{_{\mathrm{s}}}\IfValueT{#2}{(#2)}%
}
\NewDocumentCommand \fockA{s o}{%
    \mathcal{F}\IfBooleanTF{#1}{_-}{_{\mathrm{a}}}\IfValueT{#2}{(#2)}%
}
\NewDocumentCommand \X{s}{%
    \IfBooleanTF{#1}{\mathcal{X}}{\mathfrak{X}}%
}
\NewDocumentCommand \scalar { m m o }{%
    \langle#1,\,#2\rangle\IfValueT{#3}{_{#3}}%
}
\NewDocumentCommand \norm { s m o }{%
    \IfBooleanTF{#1}{\left\lVert#2\right\rVert}{\lVert#2\rVert}\IfValueT{#3}{_{#3}}%
}
\NewDocumentCommand \normConverge {s o m m o}{
    \IfBooleanTF{#1}{%
        \IfValueTF{#5}{%
            \fromTo{%
                \IfValueTF{#2}{%
                    \norm*{#3 - #4}[#2]\!%
                }{%
                    \norm*{#3 - #4}%
                }%
            }{\,0}[#5]
        }{%
            \fromTo{%
                \IfValueTF{#2}{%
                    \norm*{#3 - #4}[#2]\!%
                }{%
                    \norm*{#3 - #4}%
                }%
            }{\,0}%
        }
    }{%
        \IfValueTF{#5}{%
            \fromTo{%
                \IfValueTF{#2}{%
                    \norm{#3 - #4}[#2]\!%
                }{%
                    \norm{#3 - #4}%
                }%
            }{\,0}[#5]
        }{%
            \fromTo{%
                \IfValueTF{#2}{%
                    \norm{#3 - #4}[#2]\!%
                }{%
                    \norm{#3 - #4}%
                }%
            }{\,0}%
        }
    }
}

\NewDocumentCommand \weakConverge {m m o}{%
    #1 \!\IfValueTF{#3}{%
            \xrightharpoonup[#3]{\quad}
        }{%
            \xrightharpoonup{\quad}%
        } \nls  #2
}

\newcommand*{\dom}[1]{\mathfrak{D}(#1)}
\newcommand*{\fdom}[1]{\mathfrak{Q}(#1)}

\renewcommand*{\ker}[1]{\oldker(#1)}
\NewDocumentCommand \tr {s o}{%
    \mathop{\mathrm{Tr}}\IfValueT{#2}{%
        \IfBooleanTF{#1}{%
            \nls \left(#2\right)%
        }{%
        \ls (#2)%
        }%
    }%
}
\newcommand*{\adj}[1]{{#1}^{\ast}}
\newcommand*{\linear}[1]{\mathscr{L}\!\left(#1\right)}
\newcommand*{\bounded}[1]{\mathscr{B}\mspace{-1mu}\left(#1\right)}
\NewDocumentCommand \resolvent {m o}{%
    \mathcal{R}_{#1}\IfValueT{#2}{(#2)}%
}
\NewDocumentCommand \spectrum {o m}{%
    \IfValueTF{#1}{%
        \sigma_{\mathrm{#1}}(#2)%
    }{%
        \sigma(#2)%
    }%
}

\NewDocumentCommand \Lp {s m o} { L^{\nls  #2}\IfBooleanT{#1}{_{\,\mathrm{loc}}}\IfValueT{#3}{(#3)} }
\NewDocumentCommand \LpS {s m o}{
\IfBooleanTF{#1}{\Lp{#2}_{+}}{\Lp{#2}_{\mathrm{s}}}\IfValueT{#3}{(#3)}
}
\NewDocumentCommand \LpA {s m o}{
\IfBooleanTF{#1}{\Lp{#2}_{-}}{\Lp{#2}_{\mathrm{a}}}\IfValueT{#3}{(#3)}
}
\NewDocumentCommand \LpSA {m o}{
\Lp{#1}_{\pm}\IfValueT{#2}{(#2)}
}
\NewDocumentCommand \schatten {m o} {
    \mathfrak{L}^{#1}\IfValueT{#2}{(#2)}
}

\NewDocumentCommand \FT { s m o }{
    \IfBooleanTF{#1}{ \IfValueTF{#3}{ (\mathcal{F} \ls   #2)\nls (#3) }{%
    \mathcal{F} \ls   #2 } }{ \hat{#2}\IfValueT{#3}{(#3)} }
}

\newcommand{\C}{\mathbb{C}}
\newcommand{\R}{\mathbb{R}}

\newcommand{\N}{\mathbb{N}}
\newcommand{\No}{\mathbb{N}_0}
\NewDocumentCommand \Z {o}{%
    \IfValueTF{#1}{%
        \mathbb{Z}_{\,#1}%
    }{%
        \mathbb{Z}%
    }%
}

\allowdisplaybreaks

\title[Effective Dynamics of Weakly Interacting Bosons at High Density]{\texorpdfstring{\large Effective Dynamics for Weakly Interacting Bosons\\ in an Iterated High-Density Thermodynamic Limit}{Effective Dynamics for Weakly Interacting Bosons in an Iterated High-Density Thermodynamic Limit}}

\author[D. Ferretti and K. Koskinen]{\texorpdfstring{Daniele Ferretti*}{Daniele Ferretti} \and \texorpdfstring{Kalle Koskinen\textsuperscript{\textdagger}}{Kalle Koskinen}}

\address{{\small \textsf{daniele.ferretti@math.lmu.de} \hspace{1cm} \textsf{kalle.koskinen@gssi.it}}\\[5pt]
* \emph{Ludwig--Maximilians-Universit{\"a}t (LMU)}, {\small Theresienstra{\ss}e 39 -- 80333, M{\"u}nchen (DE)}\\[-2.5pt]
\emph{\textsuperscript{\textdagger} Gran Sasso Science Institute (GSSI)}, {\small Via Michele Iacobucci 2 -- 67100, L'Aquila (IT)}}

\thanks{\scriptsize * Corresponding author; ORCiD: \href{https://orcid.org/0000-0003-4088-6574}{link to publications}.\newline
The authors warmly thank Prof. Serena Cenatiempo for valuable discussions and helpful comments.
The authors are also grateful for the financial support provided by the European Research Council through the ERC Stg MaTCh, grant agreement no. 101117299.\newline
D. Ferretti acknowledges the support of the GNFM Gruppo Nazionale per la Fisica Matematica - INdAM}

\begin{document}

\begin{abstract}
    We study the time evolution of weakly interacting Bose gases on a three-dimensional torus of arbitrary volume.
    The coupling constant is supposed to be inversely proportional to the density, which is considered to be large and independent of the particle number.
    We take into account a class of initial states exhibiting \emph{quasi-complete} Bose--Einstein condensation.
    For each fixed time in a finite interval, we prove the convergence of the one-particle reduced density matrix towards the projection onto the normalised order parameter describing the condensate -- evolving according to the Hartree equation -- in the iterated limit where the volume (and therefore the particle number) and, subsequently, the density go to infinity.
    The rate of convergence depends only on the density and on the decay of both the expected number of particles and the energy of the initial \emph{quasi-vacuum} state.
\newline[10]
\begin{footnotesize}
\emph{Keywords: Many-Body Hamiltonians, Effective Bosonic Dynamics, Mean-Field Regime, High-Density Limit.}

\noindent \emph{MSC 2020:
    35Q40; 
    46N50; 
    47A99; 
    81Q05; 
    81V70; 
    81V73; 
    82C10.
}
\end{footnotesize}
\end{abstract}

\maketitle

{\vspace{-1cm} \small \tableofcontents}


\section{Introduction}\label{sec:Intro}

The analysis of many-body quantum systems aims to derive an effective description of macroscopic observables from the underlying fundamental microscopic perspective.
At low energies, the elementary constituents typically obey the Schr\"odinger equation; however, the sheer number of degrees of freedom precludes any attempt to compute an explicit solution.
Quantum statistical mechanics offers a rigorous framework both to justify the phenomenological laws arising from the collective behaviour of the particles and to clarify the limits of their validity.

\medskip

\noindent In particular, a significant interest in the thermodynamic properties of Bose gases has grown since the first theoretical predictions concerning the emergence of \emph{Bose--Einstein condensation} at low temperatures~(\cite{Bo1924},~\cite{Ei1924, Ei1925}) -- a phenomenon consisting in the macroscopic occupation of a single quantum state.
More precisely, in three dimensions, there exists a positive critical temperature depending on the density, below which the system undergoes a (second-order) phase transition.
Bose--Einstein condensates were later observed experimentally by the groups of Cornell and Wieman~\cite{CoWi95} and that of Ketterle~\cite{Ke95}.
Since then, the mathematical community has become increasingly active in this area, providing in~\cite{LiSe02} the first mathematical evidence of the existence of such a quantum phase in the ground state of an interacting Hamiltonian in the so-called \emph{Gross--Pitaevskii} regime (for a mathematical exposition about Bose--Einstein condensates in trapped systems, see \eg~\cite{LiSoSeYn05}).\newline
Over the past decades, several strategies and models grounded in the microscopic description of the many-body problem have been devised to extract the relevant degrees of freedom of low-energy bosonic systems.
One of the simplest non-trivial regimes employed to study weakly interacting bosons is the mean-field scaling\footnote{In this regime, the manifestation of Bose--Einstein condensation as a bound state of the Hamiltonian can be proven (see, \eg~\cite{LeNaRo14}).} (for a review on the subject, consult \eg~\cite{FrLe04}, \cite[Chapters 2--4]{BePoSc16}, or \cite{BoLeMiPe24}).
With this approach, one considers the coupling constant of the pairwise interaction to be inversely proportional to the particle number $N$.
This approximation has been extensively studied, since it both facilitates the derivation of explicit estimates and provides a guide to the treatment of more complex situations.\newline
For example, the analysis of the excitation spectrum (see~\citep{Se11, GrSe13, BrSe22} and \cite{BoBrCeSc20}) and of the next-to-leading order term of the ground state energy (\cfr~\cite{LeNaSeSo15, NaSe15, BoPeSe21}) has provided significant insights into Bogoliubov's theory~\cite{Bo1947}, which is essential for tackling tougher, more physically relevant models, such as the Gross--Pitaevskii regime.

\smallskip

\noindent In this paper, we investigate a more challenging scaling than  mean-field, which retains some of its key features.
Here, the coupling constant is set to be inversely proportional to the density of the system $\varrho$, which remains finite as the number of particles grows to infinity.
The motivation behind this choice lies in the dispersion relation of the energy $E_{\mathrm{Bog}}(\vec{p})$ carried by each quasi-particle of momentum $\vec{p}$ and predicted by Bogoliubov's approximation, that is
$$E_{\mathrm{Bog}}(\vec{p})=\sqrt{\frac{\:\abs{\vec{p}}^4\!}{4\ls m^2} +4\pi\hbar^2\ls  \varrho\,a_{\mathrm{s}}\,\frac{\:\abs{\vec{p}}^2\!}{\:m^2\!}\ls }\ls ,$$
where $a_{\mathrm{s}}$ stands for the $\mathrm{s}$-wave \emph{scattering length} associated with the pairwise potential.
This relation reduces to the usual kinetic dispersion in the absence of interaction, whereas it becomes linear at low momenta $\abs{\vec{p}}\nls \ll\nls  \hbar\ls  \sqrt{\varrho\,a_{\mathrm{s}}\ls }\ls $.
Specifically, the proportionality constant $c_s\!=\nls \frac{2\ls \hbar}{m}\ls \sqrt{\pi\ls  \varrho\,a_{\mathrm{s}}\ls }\ls $ for the linear dispersion $E_{\mathrm{Bog}}(\vec{p})\approx c_s \abs{\vec{p}}$ represents the group velocity at which fluctuations propagate through the condensate at zero temperature -- often referred to as the \emph{speed of sound}, by analogy with the behaviour of mechanical waves.\newline
In our regime, the scattering length obeys the \emph{Born expansion}; therefore, its main contribution is of the same order as the $\Lp{1}$-norm of the pairwise potential.
As a consequence, a coupling constant proportional to the inverse of the density is intended to keep the speed of sound fixed.
We stress that this is also the case for the mean-field scaling, where the coupling constant and the volume of the gas are of order $1/N$ and $1$, respectively (and therefore $\varrho\sim N$).
In contrast, we are interested in considering a system enclosed within an arbitrarily large volume with fixed density; in this sense, our model is closer to the thermodynamic setting.
However, one still observes an averaging mechanism typical of the mean-field scaling when the density is large, since the Hartree equation is found to play an important role in the effective dynamics -- meaning that the interaction felt by a single particle can be approximated by the convolution of the pairwise potential with the time-dependent density distribution.\newline
Our focus shall be on the time evolution of the Bose--Einstein condensate after its preparation, which is typically achieved by confining the gas with a proper external field and cooling it to populate low-energy states (in the case of a \emph{complete} condensate, almost\footnote{Technically, a complete Bose--Einstein condensate has $N\nls -\nls n$ of its $N$ particles in the same quantum state, where $n=\oSmall{N}$.} all the particles occupy the same state).
Indeed, upon release from the trap, the condensate is typically expected to remain stable: the complex many-body dynamics can still be approximated by the evolution of a single one-particle wave function -- known as the \emph{order parameter} -- describing the time-dependent condensate.
This collective behaviour is expected to persist until decoherence occurs through interaction with the measuring apparatus or until thermalisation with the environment supplies enough heat to overcome the critical temperature.

\paragraph{Our contribution} For a specific class of initial states exhibiting a (\emph{quasi-complete}) Bose--Einstein condensate, we prove that in the thermodynamic limit at high density, the one-particle reduced density matrix converges in trace norm towards the projection onto the time-dependent order parameter that evolves according to the Hartree equation.

\bigskip

In the remainder of this section, we give a precise formulation of the setting and regime under consideration, briefly review the state of the art, and then outline the main ideas of our proof strategy.

\subsection{The Model}\label{sec:Model}

We consider an isolated system of $N\!\in\nls \N$ non-relativistic, spinless bosons of mass $\frac{1}{2}$ confined on the three-dimensional torus $\Lambda_L\!=\big[\!-\!\tfrac{L}{2},\tfrac{L}{2}\ls \big)^3$, where $L\nls >\nls 0$ denotes the side length.
We focus on the high-density regime characterised by weak interactions in the thermodynamic limit, wherein the coupling constant is inversely proportional to the system density $\varrho\nls >\nls 0$.
Specifically, the number of particles increases proportionally to the volume of the box, so that the density remains $L$-independent.
Consequently, $\varrho$ can be treated as a fixed large parameter once the limit $N,L\to\infty$ is taken.

\medskip

\noindent For \emph{indistinguishable} particles obeying Bose--Einstein statistics, the $N$-particle Hilbert space $\Lp{2}[\Lambda_L^N]$ must be restricted to the subspace symmetric under particle exchange.
More precisely, one defines
$$\LpS{2}[\Lambda_L^N]\nls \vcentcolon=\big\{\psi\!\in\!\Lp{2}[\Lambda_L^N] \,\big|\: \psi(\vec{x}_{\pi(1)},\ldots,\vec{x}_{\pi(N)})\nls =\psi(\vec{x}_1,\ldots,\vec{x}_N),\;\forall \pi\nls \in\nls \mathfrak{S}_N\big\},$$
where $\mathfrak{S}_N$ stands for the group of permutations of $N$ elements.
Then, we define the $N$-body Hamiltonian
\begin{equation}\label{def:Hamiltonian}
    H^N_{\nls \varrho,\ls L}\vcentcolon=-\sum_{i\,=\ls 1}^N\Delta_{\vec{x}_i}\nls +\frac{1}{\varrho}\ls \sum_{i\ls <\ls j}^N V_L(\vec{x}_i\nls -\vec{x}_j),\qquad \text{on } \LpS{2}[\Lambda^N_L].
\end{equation}
Here, the reduced Planck constant $\hbar$ has been set to $1$, and $\varrho^{-1}$ is the coupling constant.
The pairwise interaction $\maps{V_L}{\R^3;\R}$ is defined via \say{periodization} of a given real-valued, spherically symmetric, continuous function $V_\infty\nls \in\nls\Lp{1}[\R^3]$, whose range is taken as the unit length scale.
This function satisfies, for some constants $C,\delta_1, \delta_2>0$, the decay condition
\begin{equation}\label{eq:potentialDecay}
    0\leq V_\infty(\vec{y}) \leq \frac{C}{\;(1+\abs{\vec{y}})^{\ls 3\ls +\ls \delta_1}\nquad}\quad,\qquad \abs{ \FT{V}_\infty(\vec{p})} \leq \frac{C}{\;(1+\abs{\vec{p}})^{\ls 3\ls +\ls \delta_2}\nquad}\quad,\qquad\forall \vec{y},\vec{p}\in\R^3,
\end{equation}
where $\maps{\FT{V}_\infty}{\R^3;\R}$ denotes the Fourier transform\footnote{By the Riemann--Lebesgue lemma, $\FT{V}_\infty\in C^{\ls \lceil\delta_1-1\rceil}\mspace{-2.25mu}(\R^3)$, and $\forall\epsilon\nls >\nls 0\quad\exists K\!\subset\nls \R^3$ compact set such that $\abs{\partial^{\ls \alpha}\FT{V}_\infty(\vec{p})}<\epsilon$, for any $\vec{p}\in \R^3\smallsetminus K$ and three-dimensional multi-index $\alpha$ with $|\alpha|\leq \lceil\delta_1\!-\nls 1\rceil$.} of $V_\infty$
\begin{equation}\label{eq:potentialFT}
    \FT{V}_\infty:\:\vec{p}\,\longmapsto\nls \integrate[\R^3]{e^{-i\ls \vec{p}\ls \cdot\ls \vec{y}}\, V_\infty(\vec{y});\!d\vec{y}}.
\end{equation}
The periodic potential is defined as
\begin{subequations}\label{eqs:pairwisePotentialConstruction}
\begin{equation}\label{def:pairwisePotential}
    V_L:\;\vec{x}\, \longmapsto\, \frac{1}{L^3}\mspace{-7.5mu}\sum_{\vec{p}\ls \in\frac{2\pi}{L}\Z^3} \!\! e^{i\ls \vec{p}\ls \cdot\ls \vec{x}}\:\FT{V}_\infty(\vec{p}),\qquad \vec{x}\in\Lambda_L.
\end{equation}
In this framework, the \emph{Poisson summation formula} holds (\cfr~\cite[Chapter VII, \S 2 - Corollary 2.6]{StWe71})
\begin{equation}\label{eq:Poisson}
    V_L(\vec{x})=\!\sum_{\vec{n}\ls \in\,\Z^3}V_\infty(\vec{x}+\vec{n} L),\qquad \vec{x}\in\Lambda_L,
\end{equation}
\end{subequations}
in the sense that the series on both sides of equation~\eqref{eq:Poisson} converge absolutely and uniformly in $\Lambda_L$ to the same limit.
Note that
\begin{itemize}
    \item identity~\eqref{eq:Poisson} implies that $V_L$ is a non-negative function on the torus;
    \item the uniform convergence in the r.h.s. of equation~\eqref{eq:Poisson} entails the continuity of $V_L$;
    \item $V_L\nls \xrightarrow[\mspace{-3.75mu}L\mspace{-0.75mu}\to\mspace{-0.75mu} \infty\mspace{-3.75mu}]{} \nls V_\infty$ pointwise, since definition~\eqref{def:pairwisePotential} recovers a Riemann sum in the limit;
    \item combining equation~\eqref{eq:Poisson} with the integrability of $V_\infty$ yields $V_L\!\in\nls \Lp{1}[\Lambda_L]$, with
    $$\norm{V_L}[1]\leq \norm{V_\infty}[\Lp{1}[\R^3]]=\vcentcolon\mathfrak{b}=\FT{V}_\infty(\vec{0})\vspace{-0.25cm}.$$
    Moreover, because of the decay condition~\eqref{eq:potentialDecay}
    $$\norm{V_L}[\infty]\leq \norm{V_\infty}[\Lp{\infty}[\R^3]]\!+\frac{1}{L^{3\ls +\ls \delta_1}}\!\!\sum_{\substack{\vec{n}\ls \in\,\Z^3 \ls :\\[1.5pt] \vec{n}\ls \neq\ls  \vec{0}}}\frac{C}{(\abs{\vec{n}}\nls -\nls 1/2)^{\ls 3\ls +\ls \delta_1}}\leq \norm{V_\infty}[\Lp{\infty}[\R^3]]\! + \oBig{L^{-3\ls -\ls \delta_1}},$$
    since $\abs{\vec{x}\nls +\nls \vec{n}L}\geq L\ls (\mspace{0.75mu}\abs{\vec{n}}\nls -\nls 1/2)$ for all $\vec{x}\nls \in\nls \Lambda_L$ and the series converges for any $\delta_1\nls >\nls 0$.
\end{itemize}

\smallskip

\noindent By virtue of equation~\eqref{eq:Poisson}, this construction of the potential $V_L$ is designed to model problems where each particle interacts with all others and with their respective images inside the copies of the box provided by the periodic boundary conditions.

\smallskip

\noindent Due to the boundedness of $V_L$ and the periodic boundary conditions imposed on $\partial\Lambda_L$, the Hamiltonian~\eqref{def:Hamiltonian} is self-adjoint on the domain $H^2\big(\Lambda^N_L\big)\cap \LpS{2}\big(\Lambda^N_L\big)$.

\medskip

\noindent The thermodynamic limit is realised by fixing $\varrho,L\nls >\nls 0$ so that $N\!\in\nls \N$ depends on these two parameters, namely $N=\lceil\varrho L^3\rceil$.
Thus, as $L\nls \to\nls \infty$, $N$ also diverges, but the ratio $N/L^3\longrightarrow \varrho$ remains finite, although we are interested in the high-density regime.
Crucially, the limit $L\!\to\!\infty$ must be taken \emph{before} considering $\varrho$ large.
Here, the physical intuition is that the thermodynamic limit ($L\to\infty$) is the procedure by which the macroscopic description of the system is defined: such a mathematical tool is meant to select the relevant collective degrees of freedom.
In particular, the density $\varrho$ emerges only at this stage as an intrinsic parameter of the system, through the asymptotic relation $N/L^3\to\varrho\ls $; the subsequent limit $\varrho\to\infty$ should therefore be interpreted as a high-density regime of this already established infinite-volume system.
In this sense, the order of limits reflects the fact that one first constructs the thermodynamic model, and only afterwards explores its behaviour in a specific asymptotic regime.\newline
If one instead regards the identity $\varrho= N/L^3$ as exact at finite volume and considers alternative scalings of $\varrho$ and $L$, the resulting regimes could potentially correspond to different physical situations.
Such procedures do not capture the same bulk behaviour and may lead to qualitatively different effective descriptions.
For instance, sending $\varrho\to\infty$ for $L$ fixed corresponds to an infinite-particle system in a finite region that corresponds to the usual mean-field scaling.

\medskip

\noindent Notably, self-interaction contributions -- arising from the force between each particle and its own images in the copies of the box produced by the periodic boundary conditions -- become negligible for large $L$:
$$\tfrac{N}{\varrho}\!\! \sum_{\substack{\vec{n}\ls \in\,\Z^3 \ls :\\[1.5pt] \vec{n}\ls \neq\ls  \vec{0}}}\!\! V_\infty(\vec{n}L)\leq \tfrac{N}{\varrho}\!\!\sum_{\substack{\vec{n}\ls \in\,\Z^3 \ls :\\[1.5pt] \vec{n}\ls \neq\ls  \vec{0}}} \!\frac{C}{(1+|\vec{n}| L)^{\ls 3\ls +\ls \delta_1}\nquad}\quad\leq \tfrac{N}{\varrho L^3\!}\,L^{-\delta_1}\!\!\sum_{\substack{\vec{n}\ls \in\,\Z^3 \ls :\\[1.5pt] \vec{n}\ls \neq\ls  \vec{0}}} \!\nls \frac{C}{\;|\vec{n}|^{3+\delta_1}\!\nls }\:=\oBig{L^{-\delta_1}}.$$
To further clarify the physical interpretation, we emphasise that the box is not intended to model a confining trap generated by an external field; rather, it simply identifies the region of the system under consideration, whose size is ultimately taken to infinity in the thermodynamic limit, thereby reflecting its large scale with respect to the microscopic lengths at which particles typically interact.
In this framework, the use of a torus as the domain of the Hamiltonian is not a particular restriction but a convenient mathematical choice that avoids unnecessary complications due to boundary effects while retaining the relevant physics.\newline
To better understand this point, we stress that by sending the volume to infinity, the problem is endowed with the notion of two asymptotically separated length scales.
On the one hand, a viewpoint adopting a macroscopic perspective works with the unitarily rescaled Hamiltonian on the domain $\Lambda_1$.
In this picture, the system has a volume of order $1$, while the typical range of the pairwise interaction shrinks as $L$ grows -- this is the viewpoint according to which the experiment is described in the laboratory.\newline
On the other hand, from the microscopic perspective, the range of the interaction remains of order $1$, and particles do not effectively distinguish between evolving in a very large finite volume or in an infinite one (which clarifies the simplification given by the limit $L\nls \to\nls \infty$).
Consequently, our analysis focuses on bulk phenomena: the boundaries are pushed far away, and therefore their effects on the microscopic constituents of the system become negligible.\newline
Hence, provided pathological cases are excluded (\eg, domains whose boundary scales too rapidly with respect to volume, or those that do not expand homogeneously enough in the three spatial directions), the precise shape of the box is irrelevant.
Thus, the torus stands as a particularly convenient, yet substantially general, choice among all possible domains.

\medskip

\noindent We are interested in investigating the time evolution of quantum states close to Bose--Einstein condensates.
Accordingly, the Hamiltonian~\eqref{def:Hamiltonian} is to be thought of as the energy operator acting \emph{after} the system has been prepared such that a Bose--Einstein condensate emerges as a bound state at the initial time of our analysis $t\nls =\nls 0$.

\subsubsection*{Comparison with Existing Literature}

To the best of our knowledge, there exists no proof demonstrating that the Hamiltonian $H^{\lceil\varrho L^3\rceil}_{\varrho,\ls L}$ (possibly with the additional presence of a trapping external potential) has a Bose--Einstein condensate as a bound state when $L$ goes to infinity, even in the high-density regime.
The most closely related results in the literature are found in~\cite{GiSe09} and~\cite{DeNa14}.
\begin{itemize}
    \item First, in~\cite{GiSe09} the authors consider a Hamiltonian acting on the torus with an exponentially decaying potential with range $R_0$ and coupling constant $a_0>0$.
    Their model intersects with ours for the specific choices $R_0\nls \sim\nls 1$ and $a_0\nls \sim\nls \varrho^{-1}$; consequently, their results show that the \emph{Lee--Huang--Yang formula}~\cite{LeHuYa57} holds true for the ground state energy of~\eqref{def:Hamiltonian} whenever the scattering length associated with the potential decays proportionally to $\varrho^{-1-\gamma}$, with $\gamma\in\big[0,\frac{4}{63}\big)$.
    In our case, Born's expansion holds true, and therefore $\gamma=0$.
    \item By contrast,~\cite{DeNa14} adopts precisely the same Hamiltonian as ours, and the Bogoliubov approximation for the excitation spectrum is proven to be valid, but in a different regime.
    Indeed, in three-dimensions~\cite[Theorem 1.1]{DeNa14} assumes $L^5\nls \leq\nls \varrho$ for the lower bound and $\max\{L, L^4\}\nls \leq\nls \varrho$ for the upper bound.
    Neither of these requirements is satisfied in our case of interest, where the limit $L\nls \to\nls \infty$ precedes $\varrho\nls \to\nls \infty$.
\end{itemize}
Our goal is to establish an effective evolution of the condensate that is accurate at each fixed time in the thermodynamic limit for sufficiently large $\varrho$.
Remarkably, by sending $L$ to infinity, this approach prescribes, for every fixed $\varrho$, the construction of an infinite-particle system, whose definition can be found, for instance, in~\cite{St69},~\cite{FiFr91},~\cite{BrRo97}, or~\cite{BaBu21}.
However, understanding the effective dynamics of the entire infinite-particle system is far beyond our scope.
We do not pursue this direction further.

\smallskip

\noindent Instead, we are satisfied with an effective finite-volume dynamics for the finite particle system that approximates the many-body evolution at each fixed time when we take the thermodynamic limit on the box (namely, when $L$ approaches infinity), and such an approximation eventually improves as $\varrho$ increases.\newline
A similar result was first obtained for a Fermi gas by Fresta, Porta, and Schlein~\cite{FrPoSc23}, whose work deeply inspired our own (see also~\cite{FrPoSc25}, which investigates the convergence towards the effective dynamics of pseudo-relativistic fermions by testing local averages of observables).
Although a coupling constant equal to the inverse of the density is employed in their case as well, the mathematical setup differs: in the literature, when studying the dynamics, the trapping potential localising the gas in a given region $\Omega\nls \subset\nls \R^3$ is typically present only at the initial time $t\nls =\nls 0$.
Thereafter, the system evolves in the entire space $\R^3$.
In particular, for~\cite{FrPoSc23} and~\cite{PePiSo20}, the coupling constant of the Hamiltonian generating the dynamics is equal to the inverse of the \emph{initial} sample density\footnote{This simplification is justified by the fact that the typical velocities in the setting are of order $1$, and therefore, in principle, one can provide a notion of volume $\Omega(t)$ within which the gas is \say{localised}, whose size is not too far from the initial one (if $t$ is not too large).
As a consequence, $\rho(t)$ is expected to be of the same order as $\rho$.} $\rho=N/|\Omega|$.
In this interpretation, the additional quantity $|\Omega|$ serves to parametrise the family of considered initial data -- it is not meant to be a dynamical variable.\newline
Within this setup, a natural macroscopic limit to be considered consists in taking both $N$ and $|\Omega|$ large.
In comparison with our framework, this approach is characterised by sending $L$ to infinity \emph{before} the double limit $N, |\Omega|\to\infty$, which is possibly coupled to maintain the initial sample density $\rho$ fixed -- as in~\cite{FrPoSc23}, where the convergence towards the effective dynamics is shown for large $\rho$, uniformly in $|\Omega|$.
In other words, they are interested in providing an \emph{infinite-volume} dynamics associated with the finite-particle system that approaches the many-body dynamics when the initial datum undergoes the double joint limit $N, |\Omega|\to\infty$.\newline
By contrast, we couple the double limit $N, L\to\infty$, and we are not concerned with confining our initial datum within a subregion $\Omega\nls \subset\nls  \Lambda_L\ls $, as our focus is on condensation -- a phenomenon involving the localisation of momenta.
Although the two approaches may appear very similar, the order of limits is conceptually significant and might be crucial, in principle, for these kinds of problems.
A key distinctive feature of our setting is indeed the boundedness of the domain within which the gas evolves;
this implies that the density $\varrho$ is an intrinsic property of the system, does not depend on the choice of initial data, and remains constant over time.

\smallskip

\noindent Beyond the differences in the mathematical setting, our approach is also physically distinct from~\cite{PePiSo20}, at least for sufficiently large times.
Indeed, it has been recently shown in~\cite[Corollary 1.3]{NaRaSo26} that, in the whole space $\R^3$, the Bogoliubov dynamics of fluctuations becomes close to the free evolution on time scales of order $N^2$, as a consequence of dispersive estimates (\cite[Theorem 1.1]{NaRaSo26}).
This indicates that correlations progressively weaken at long times, as particles outside the condensate spread over larger and larger regions.\newline
A heuristic interpretation of this behaviour can be given in terms of the speed of sound $c_s\!=\nls \frac{2\ls \hbar}{m}\ls \sqrt{\pi\ls  \rho\,a_{\mathrm{s}}\ls }$.
In the mean-field regime, one has $a_{\mathrm{s}}\sim \frac{1}{N}$ so that $c_s$ remains of order one only as long as the gas is effectively localised in a region of order one.
If the evolution takes place in the whole space, dispersive effects cause the gas to spread over increasing length scales, leading to a decrease of the sample density.
Therefore, the propagation of correlations becomes increasingly suppressed in this regime, and the Bogoliubov dynamics approaches that of free particles, while the condensate itself continues to evolve according to the nonlinear Hartree equation.\newline
In contrast, evolution takes place in a finite region in our setting; therefore, such dispersive mechanisms are not available.
For this reason, there is no analogous tendency for correlations to decay purely due to spatial spreading, and there is no evident obstruction to describing the dynamics of the system, even at long times, within the present framework.
In this sense, our approach is more directly aimed at a macroscopic description of the evolution, rather than at a microscopic tracking of the dynamics in the whole space.\newline
We emphasise, however, that we do not focus on the analysis of the dynamics of particles outside the condensate here, as a quantitative description of their evolution would require a finer level of approximation, beyond the scope of the present work.

\subsubsection*{Our Contribution} Our results are presented in the framework of second quantisation (see Section~\ref{sec:SQ} for definitions and details).
In particular, we shall work with the second-quantised Hamiltonian $\mathcal{H}_{\varrho,\ls L}$ on the symmetric Fock space (defined by~\eqref{def:HamiltonianFock}) corresponding to the $N$-body operator $H^N_{\varrho,\ls L}$.
The family of initial states we are going to deal with shall be generated by the action of the \emph{Weyl operator} $\mathcal{W}(\Psi_{\!\varrho,\ls L})$ (introduced in~\eqref{def:Weyl}), which implements suitable coherence properties on such states that are based on the order parameter $\Psi_{\!\varrho,\ls L}\!\in H^1(\Lambda_L)$.\newline
More precisely, we consider the initial state
$$\varphi^{\ls 0}_{\varrho,\ls L}=\mathcal{W}(\Psi_{\!\varrho,\ls L})\ls  \xi_{\varrho,\ls L},$$
which will be a \emph{quasi-canonical coherent state} (see Definition~\ref{def:Q-CCS}), while the order parameter $\Psi_{\!\varrho,\ls L}$ will be a \emph{quasi-complete Bose--Einstein condensate} for $\varphi^{\ls 0}_{\varrho,\ls L}$ (see Definition~\ref{def:QCBEC}).
Here,
\begin{itemize}
    \item the order parameter $\Psi_{\!\varrho,\ls L}$ is assumed to be such that there exists a macroscopic counterpart $\Psi\nls\in\nls H^1(\Lambda_1)$ in the high-density thermodynamic limit (Assumption~\ref{ass:initialBEC}).
    Moreover, $\Psi_{\!\varrho,\ls L}$ is required to have proper decay conditions on the tails of its macroscopically rescaled Fourier series in the iterated limit (\cfr~Assumptions~\ref{ass:tailCondition} and~\ref{ass:kineticTailCondition}).
    These are sufficient conditions that strengthen the notion of convergence towards the macroscopic order parameter, creating a more comfortable framework in which to prove the well-posedness of the associated time evolution;
    \item $\xi_{\varrho,\ls L}$ is a \emph{quasi-vacuum state} (\cfr~Definition~\ref{def:Q-Omega}) with respect to $\Psi_{\!\varrho,\ls L}$ -- which contains, roughly speaking, few expected particles compared to $\norm{\Psi_{\!\varrho,\ls L}}[2]^2\ls $;
    \item $\varphi^{\ls 0}_{\varrho,\ls L}$ is required to be \emph{energetically quasi-self-consistent} (see Definition~\ref{def:Q-sC}), namely, its associated expected energy is close to the Hartree energy of $\Psi_{\!\varrho,\ls L}\!\in\nls  H^1(\Lambda_L)$.
\end{itemize}

\paragraph{Our main result, Theorem~\ref{th:gamma1Convergence}} Consider the many-body time evolution driven by the Hamiltonian $\mathcal{H}_{\varrho,\ls L}$ and let $\varphi^{\ls t}_{\varrho,\ls L}\!=\nls e^{-i\,\mathcal{H}_{\varrho,\ls L}\ls t}\, \varphi^{\ls 0}_{\varrho,\ls L}$.
Then, for the corresponding \emph{one-particle reduced density matrix} $\gamma^{(1)}_{\varphi^{\ls t}_{\varrho,\ls L}}$ (see~\eqref{def:gamma1Kernel})
$$\exists T>0 \;:\quad\lim_{\varrho\to\infty}\limsup_{L\to\infty}\, \norm*{\gamma^{(1)}_{\varphi^{\ls t}_{\varrho,\ls L}}\!\nls-\,\tfrac{|\Psi^{\ls t}_{\!\varrho,\ls L}\rangle\langle\Psi^{\ls t}_{\!\varrho,\ls L}|}{\vphantom{|^|}\norm{\Psi^{\ls t}_{\!\varrho,\mspace{0.75mu} L}}[2]^2}\,}[\mathrm{Tr}]\!\nls=0,\qquad \forall t\in[0,T],$$
where the wave function $\Psi^{\ls t}_{\!\varrho,\ls L}$ evolves according to the Hartree equation
$$i\partial_t\ls \Psi^{\ls t}_{\!\varrho,\ls L}\nls =-\Delta \Psi^{\ls t}_{\!\varrho,\ls L}\nls  + \tfrac{1}{\varrho}\ls \big(V_L\!\ast\nls \abs{\Psi^{\ls t}_{\!\varrho,\ls L}}^2\big)\ls  \Psi^{\ls t}_{\!\varrho,\ls L},\qquad \text{in }\,\Lambda_L,$$
with the initial datum $\Psi^{\ls 0}_{\!\varrho,\ls L}=\Psi_{\!\varrho,\ls L}$.
Technically speaking, this result could be rephrased by replacing the sequence of normalised order parameters $\frac{\Psi^{\ls t}_{\!\varrho,\ls L}}{\norm{\Psi^{\ls t}_{\!\varrho,\ls L}}[2]\!\nls}$ with $\tfrac{1}{L^{3/2}} \Psi^{\ls t}\big(\tfrac{\vec{\cdot}}{L}\big)$, where $\Psi^{\ls t}:\, \vec{y}\:\longmapsto\, e^{-i\ls\FT{V}_\infty(\vec{0}) \ls\abs{\Psi(\vec{y})}^2\ls t}\,\Psi(\vec{y})$ is the evolution of the macroscopic counterpart of the initial quasi-complete Bose--Einstein condensate (\cfr~Proposition~\ref{th:propagationExistenceBECgeneral}). Morally, $\Psi^{\ls t}_{\!\varrho,\ls L}$ can be regarded to be closer and closer to a stretched and scaled version of a wave function on the unit torus scaled up to $\Lambda_L$, and one can readily observe its features from this exact form.
However, we stress that the effective one-body non-linear Hartree dynamics is meant to approximate the full many-body Schr\"odinger dynamics at finite values of $\varrho$ and $L$: our result states that such an approximation eventually improves in the thermodynamic limit when the density increases.\newline
We provide a more detailed roadmap to the proof of the main theorem in Section~\ref{sec:qcBEC}.

\bigskip

In the following, we briefly review the current state of the art concerning the dynamics of three-dimensional systems modelling weakly interacting non-relativistic bosons in mean-field-related scalings (we refer to~\cite{Na23} for a more detailed review).

\subsection{Background and Related Works}\label{sec:stateOfArt}

The rigorous mathematical study of many-body bosonic dynamics has a long and rich history.
A foundational result for three-dimensional systems was established by Ginibre and Velo~\cite{GiVe79A}, who generalised the earlier one-dimensional work of Hepp~\cite{He74}.
In their paper,  they study a semiclassical limit $\hbar\nls \to\nls  0$, in which the mass of the bosons scales as $m_\hbar=\hbar\, m$ and the coupling constant for the pairwise potential is $\hbar^2$.
Given that the expected number of particles is $\oBig{\hbar^{-1}}$, this scaling is equivalent to the mean-field regime, described by the Hamiltonian $H^N_{N,\ls \infty}$, given by~\eqref{def:Hamiltonian}.
They prove that, as $\hbar$ approaches zero, the particle structure disappears, as the correlation functions in coherent states converge along the evolution to those expected by a classical field $t\longmapsto\varphi^{\ls t}$ (in a suitable Banach space) obeying the Hartree equation
$$i\partial_t\ls \varphi^{\ls t}\nls =-\tfrac{1}{2m}\Delta \varphi^{\ls t}\nls  + \big(V_\infty\nls \ast\nls \abs{\varphi^{\ls t}}^2\big)\ls  \varphi^{\ls t},\qquad \text{in }\,\R^3,$$
for a large class of potentials $V_\infty$.

\subsubsection*{BBGKY hierarchies}
An alternative approach was pursued by Spohn in~\cite{Sp81}.
He considered a model with $\hbar=N^{-1/3}$ and a bounded pairwise potential having a coupling constant $1/N$.
This corresponds to the Hamiltonian $\frac{1}{N^{1/3}} H^N_{N^{1/3},\ls \infty}$, which differs from the standard mean-field regime.
Here, the novelty is that the limits (in $N$) of the $n$-point correlation functions satisfy a \emph{Vlasov hierarchy} -- an infinite set of coupled linear PDEs, in which the limit of each $n$-point correlation function depends on that of the $(n+\nls 1)$-th one.
The solution to this hierarchy has been proven to exist and to be unique.\newline
Subsequent works, such as \cite{BagoMa00} and~\cite{ErYa01, ErScYa07}, have obtained improvements in this direction.
With the exception of~\cite{ErScYa07} -- which considers an \emph{intermediate scaling}\footnote{This class of regimes is meant to interpolate the behaviour of the system between the mean-field ($\beta=0$) and the Gross--Pitaevskii approximation ($\beta=1$).} with the replacement of $V_\infty(\vec{\cdot})$ with $N^{3\beta} V_\infty(N^\beta \ls \vec{\cdot})$ for $\beta\nls \in\nls  [0, 1/2)$ -- these works study the standard mean-field Hamiltonian $H^N_{N,\ls \infty}$.
Starting from a factorised initial state $\psi_N=\varphi^{\ls \otimes N}$, they derive, for fixed $N$, a Schr\"odinger hierarchy for the (normalised) $k$-particle reduced density matrices $\gamma^{(k)}_{N,\ls t}\in\schatten{1}[\Lp{2}[\R^{3k}]]$ associated with the time evolution of $\psi_N$.
Moreover, they prove (under mild assumptions on the potential) the convergence of the solution for the $N$-finite hierarchy to a solution for the infinite-particle hierarchy.
They also show (for bounded potentials in~\cite[Corollary 5.3, Theorem 5.4]{BagoMa00} and for the Coulomb potential in~\cite{ErYa01}) the uniqueness of such a solution and the conservation over time of the factorisation; specifically\footnote{The precise topology of the convergence varies among the three papers.},
$$\gamma^{(k)}_{N,\ls t}\approx (|\varphi^{\ls t}\rangle\langle\varphi^{\ls t}|)^{\otimes\ls  k},\qquad N\to\infty,$$
where $\varphi^{\ls t}$ solves the Hartree equation (or the cubic nonlinear Schr\"odinger equation in the case of~\cite{ErScYa07} for $\beta>0$) with initial datum $\varphi^{\ls 0}=\varphi\in\Lp{2}[\R^3]$ regular enough\footnote{We emphasise that the physically relevant initial states are those that are both eigenstates of the initial Hamiltonian and exhibit a Bose--Einstein condensate. Specifically, for $\varphi^{\ls \otimes N}$ to represent such a state, the system must be initially confined.
This confinement, encoded \eg~by the Hamiltonian $H^N_{N,\ls 1}$ at $t\nls =\nls 0$, enforces the \say{localisation} of the initial one-particle wave function $\varphi\nls \in\nls \Lp{2}[\R^3]$ within a volume of order one.}.
The proof of the convergence in these cases relies mainly on compactness arguments.
We refer to~\cite[Section 1.10]{Go16} for a review on the subject.\newline
The BBGKY approach was later connected to the formalism of \emph{Wigner measures} (see~\cite[Section 6]{AmNi08}) by Ammari and Nier in~\cite{AmNi11}.

\subsubsection*{Semiclassical Analysis}
A distinct framework was developed by Ammari and Nier in~\cite{AmNi09, AmNi11, AmFaPa16}, adopting a regime sometimes referred to as \say{quasi-classical} by some authors, where the creation and annihilation operators (defined in~\eqref{def:creationAnnihilation}) are rescaled by a factor $\sqrt{\varepsilon\ls }\ls $, so that the \emph{canonical commutation relations} (see~\eqref{eq:CCR}) mimic the dependence on the small constant $\hbar$ typically appearing in the commutators.
In these models, the unitary evolution involves the Hamiltonian in second quantisation divided by $\varepsilon$.
Hence, since the quadratic kinetic term comes with a factor $\varepsilon$ and the quartic term associated with the interaction is multiplied by $\varepsilon^2$, there is complete congruence with the regime studied in~\cite{GiVe79A} for $\hbar\nls \sim\nls \varepsilon$, which is in turn equivalent to the mean-field scaling, setting $\varepsilon=1/N$.
In these problems, a proper set of initial (mixed) states is considered so that there exists a \emph{Wigner measure} describing the expectation of observables obtained via Weyl or Wick quantisation in the limit $\varepsilon\nls \to\nls 0$.
In~\cite{AmNi09} (later improved in~\cite{AmNi11} by relaxing the hypotheses), it has been proved that such a description in terms of Wigner measures is preserved globally in time by means of a pushforward with the classical flow associated with the Hartree equation.
This result implies (\cite[Theorem 1.1]{AmNi11}) the trace norm convergence (when $\varepsilon\nls \to\nls 0$) of all $k$-particle reduced density matrices to a (normalised) compact operator expressed in terms of the integral of the projection onto the factorisation of $k$ one-particle states, with respect to the time-dependent Wigner measure.\newline
This line of inquiry is closely related to the works of Fr\"ohlich \etal~\cite{FrGrSc07, FrKnPi07, FrKnSc09}, who also employed the mean-field regime.
In these cases, a convergence of expectations of $p$-particle observables along factorised states is given in the Heisenberg picture.
In the limit $N\nls \to\nls \infty$, the expectation of these observables remains close to being computed along states that are still obtained by factorising $p$ one-particle wave functions evolving according to the Hartree equation\footnote{Importantly, this does \emph{not} imply that the time evolution of a $p$-particle factorised state is close to another factorised state in the topology induced by $\Lp{2}[\R^{3p}]$.
This topic was particularly addressed in~\cite{Pi11}, where a clever algorithm is developed to count in a biassed way the particles over time that do not fit into the description of the order parameter.
A control of this number in terms of the same quantity at $t\nls =\nls 0$ is proven, showing, in particular, the preservation over time of Bose--Einstein condensation.}.
Moreover, in~\cite{FrKnPi07, FrKnSc09} an Egorov type theorem is proven to hold true (for bounded pairwise potentials in~\cite{FrKnPi07} and more singular ones in~\cite{FrKnSc09}) -- namely, the Wick quantisation of a time-evolved classical system yields a result that is almost the same as the time evolution of the associated many-body quantum system when $N$ is large enough.\newline
Furthermore,~\cite{AmFaPa16} proved that the convergence in trace norm of the time-dependent reduced density matrices occurs at the optimal rate $1/N$ locally in time, provided that the initial state has associated reduced density matrices converging in trace norm to the infinite-particle counterpart (given in terms of a Wigner measure) not slower than $1/N$.

\subsubsection*{Rate of convergence}
Providing quantitative bounds on the speed of convergence to the mean-field approximation is relevant, as they clarify how effective the Hartree theory is for a system composed of a large but finite number of particles.
In~\cite{RoSc09}, Rodnianski and Schlein were the first to establish a quantitative bound of this type.
In particular, for a wide class of pairwise potentials (including Coulomb), they proved the optimal rate of convergence $1/N$ for the evolution of reduced density matrices associated with coherent states and a non-optimal rate $1/\sqrt{N}$ in the case of factorised states.
This gap was later filled by~\cite{ErSc09, ChOo11, ChOoSc11} considering potentials in $\Lp{\infty}[\R^3]$ (\cite{ErSc09}), in $\Lp{3}[\R^3]+\Lp{\infty}[\R^3]$ (\cite{ChOo11}), and finally in the wider class originally taken into account by Rodnianski and Schlein (\cite{ChOoSc11}).
Even more singular potentials in $\Lp{2}[\R^3]+\Lp{\infty}[\R^3]$ were considered in~\cite{ChOoLe18}, where Strichartz’s estimates were exploited to obtain an optimal rate of the trace-norm convergence for factorised initial states.\newline
An alternative approach for the study of the mean-field dynamics of factorised states -- involving the analysis of the projection onto the time-dependent order parameter -- can be found in~\cite{KnPi10}.\newline
In~\cite{DiGi22}, the Thomas-Fermi regime is investigated, characterised by the interaction potential $g_N N^{3\beta-1}V_\infty(N^\beta \vec{\cdot})$, with $g_N\nls \gg\nls  1$, $\beta\nls \in\nls \big(0,\frac{1}{6}\big)$.
For this model, the authors establish the existence of Bose--Einstein condensation in the ground state of the trapped Hamiltonian and prove its persistence over short time scales upon the release of the trap.\newline
A more PDE-oriented work establishing results similar to those of~\cite{RoSc09} for the intermediate regime $\beta\nls \in\nls \big[0,\frac{1}{3}\big)$ is~\cite{DiLe23}, which pivots on the exploitation of a dispersive estimate for the Hartree equation.\newline
Moreover, following the ideas developed by Wu in~\cite{Wu61, Wu98}, a second-order correction to the Hartree theory of coherent states has been captured by effective time-dependent states obtained by means of the Weyl operator composed with a quadratic unitary transformation in~\cite{GrMaMa10, GrMaMa11}, where the latter work improves issues of the former concerning more singular potentials and global-in-time convergence (which occurs here in the topology induced by the norm of the Fock space).
These works were later improved by considering the same kind of approximation for the intermediate regime in~\cite{GrMa13, GrMa17} (for $\beta\nls \in\nls \big[0,\frac{1}{3}\big)$ and $\beta\nls \in\nls \big[0,\frac{2}{3}\big)$, respectively) and in~\cite{BoCeSc17} (achieving $\beta\nls <\nls 1$ globally in time, with different techniques).

\subsubsection*{Probabilistic interpretation} We also mention that one-particle observables in a factorised state $\varphi^{\ls \otimes N}$ can be interpreted as $N$ identically distributed independent random variables, each acting as the identity on the other $N\!-\nls 1$ one-particle sectors.
Therefore, the central limit theorem holds true (\cfr~\cite{ArKiSc13, BuSaSc14} and \cite{Ra20, Ra22}).
In particular, in \cite{ArKiSc13} it is proved that the sum of the deviations from the mean value of the $N$ one-particle observables aforementioned, rescaled with the factor $1/\sqrt{N}$, converges to a normal distribution in a distributional sense.
A similar outcome is obtained by~\cite{BuSaSc14} for several one-particle observables, but with an explicit rate of convergence, while~\cite{Ra20} generalises the results to the intermediate scaling for $0\nls \leq\nls \beta\nls <\nls  1$, and~\cite{Ra22} takes into account $k$-particle observables in coherent states.
Here, the relevant information is that, despite the fact that the time evolution of $\varphi^{\ls \otimes N}$ is no longer a factorised state, few correlations develop during the dynamics, so that the validity of the central limit theorem is preserved.
However, such correlations are strong enough to change the variance of the normal distribution, which is affected by the action of the Bogoliubov transformation describing the fluctuations around the Hartree approximation.\newline
Large deviation principles were discussed in~\cite{KiRaSc21} and later in~\cite{RaSe22} (enlarging the class of potentials adopted).\newline
Further probabilistic implications concerning corrections to the central limit theorem have been studied for the ground state of a trapped Bose gas in the mean-field regime in~\cite{BoPe23}.

\subsubsection*{Fluctuations around Hartree dynamics}
Finally, based on the results of~\cite{LeNaSeSo15}, many papers, such as~\cite{LeNaSc15, NaNa17A, NaNa17B, BrNaNaSc19, NaRaSo26} and~\cite{MiPePi19, BoPaPiSo20, BoPePiSo21}, investigated a class of initial states with a fixed number of particles $N$, composed of a superposition of $N\!-\nls k$ states associated with a single wave function $\varphi^{\ls 0}$, and $k$ excitations, orthogonal to $\varphi^{\ls 0}$, with $k$ varying between $0$ and $N$.
They prove that, as $N\nls \to\nls \infty$, the $N$-particle wave function remains close in norm to a superposition of factorised states composed of $N\!-\nls k$ wave functions $\varphi^{\ls t}$ evolving according to a modified Hartree equation, and $k$ excitations evolving according to a Bogoliubov-type Hamiltonian (which is quadratic in the creation and annihilation operators), for again $k\in\{0,\ldots,N\}$.
We stress that this kind of result is stronger than the convergence in trace norm of the reduced density matrices (\cfr~\cite[Corollary 2]{LeNaSc15}).
While the setting of~\cite{LeNaSc15, NaRaSo26} and~\cite{MiPePi19, BoPePiSo21} is the mean-field scaling (a global-in-time result is achieved in~\cite[Theorem 1.2]{NaRaSo26}),~\cite{NaNa17A, NaNa17B, BrNaNaSc19} and~\cite{BoPaPiSo20} work more generally within the intermediate regime for $\beta<\frac{1}{3}$, $\beta<\frac{1}{2}$, $\beta<1$, and $\beta<\frac{1}{12}$, respectively.
However,~\cite{MiPePi19} provides the evolution of the excitations in terms of a first quantised Hamiltonian for the fluctuations, while~\cite{BoPaPiSo20, BoPePiSo21} explore a procedure involving several Bogoliubov-type Hamiltonians altogether in order to obtain a decomposition in terms of excitations evolving according to different generators.
This allows one to achieve arbitrary precision in the convergence in terms of powers of $1/N$ by increasing the number of terms considered in the decomposition.

\medskip

\noindent Along the same direction, Petrat, Pickl, and Soffer studied in~\cite{PePiSo20} (improving the results from~\cite{DeFrPiPi16}) the dynamics of fluctuations around the Hartree approximation for the Hamiltonian $H^{\lceil\rho \,|\Omega|\rceil}_{\rho,\ls \infty}$.
Indeed, the system is initially localised in a region $\Omega\nls \subset\nls \R^3$ by means of a trap that is subsequently removed.
This means they consider the inverse of the initial sample density $\rho$ as a coupling constant, and $|\Omega|$ represents a variable to be sent to infinity that parametrises the class of initial states.
Their result shows local-in-time convergence in the same fashion as~\cite{LeNaSc15}, in any double limit $\rho,|\Omega|\nls \to\nls \infty$ satisfying $|\Omega|^3 \nls \ll\nls  \rho$ (\cite[Theorem 2.2]{PePiSo20}), which precludes taking $|\Omega|\to\infty$ before $\rho$.

\bigskip

In short, the extensive literature on mean-field and related regimes has firmly established the validity of the Hartree approximation for the effective dynamics of Bose--Einstein condensates, with various techniques yielding results on the convergence of states, correlation functions, and observables.
Our contribution differs by analysing, for each fixed time $t$ (in a finite interval), the high-density regime on a large torus, where the coupling constant is scaled as $1/\varrho$.
The novel challenge we address is the specific order in which the limits $L\nls \to\nls \infty$ and $\varrho\nls \to\nls \infty$ must be taken, an aspect not covered by prior results.

\medskip

\noindent We conclude this section by summarising the strategy for pursuing our goals.

\subsection{Outline of the Strategy}

We work in the grand canonical picture, which provides a natural framework to analyse sequences of states with an indefinite particle number.
Within this setting, we introduce the notions of quasi-vacuum and quasi-coherent states (Definitions~\ref{def:Q-Omega},~\ref{def:Q-CCS}, respectively), which extend the standard vacuum and coherent states to the high-density scaling considered here.
The initial datum is chosen starting from these building blocks in such a way that it exhibits quasi-complete condensation and permits the control of fluctuations over time, as the expected particle number grows together with the size of the system.

\medskip

\noindent More precisely, the goal is to approximate the many-body evolution of the class of initial quasi-canonical coherent states $\varphi^{\ls 0}_{\varrho,\ls L}$ of the form $\varphi^{\ls 0}_{\varrho,\ls L}\!=\mathcal{W}(\Psi_{\!\varrho,\ls L})\ls \xi_{\varrho,\ls L}$, where $\mathcal{W}(\Psi_{\!\varrho,\ls L})$ is the Weyl operator associated with the initial order parameter $\Psi_{\!\varrho,\ls L}$, and $\xi_{\varrho,\ls L}$ is the quasi-vacuum state.
Imposing that the evolution preserves the structure
$$\varphi^{\ls t}_{\varrho,\ls L}=\mathcal{W}(\Psi^{\ls t}_{\!\varrho,\ls L})\ls \xi^{\ls t}_{\varrho,\ls L},\qquad\xi^{\ls t}_{\varrho,\ls L}=\mathcal{U}_{\varrho,\ls L}(t)\ls \xi_{\varrho,\ls L}$$
leads to the definition of the fluctuation dynamics $\mathcal{U}_{\varrho,\ls L}(t)$, which encodes the time evolution of the quasi-vacuum state.
The problem then reduces to proving that $\xi^{\ls t}_{\varrho,\ls L}$ -- referred to as the \emph{fluctuation state} hereafter -- remains a quasi-vacuum state over time in the iterated limit.

\medskip

\noindent To this end, we analyse the generator of the fluctuation dynamics, building on the ideas developed in~\cite{RoSc09, BeOlSc15}:
algebraic manipulations allow us to estimate both the generator and its time-derivative in terms of the expected number of fluctuations and their associated energy (Corollaries~\ref{th:generatorAPEstimate},~\ref{th:generatorDerivativeAPEstimates}).
By carefully controlling the nonlinearity of the Hartree equation (Propositions~\ref{th:nonlinearityControl},~\ref{th:LaplaceNonlinearityControl}), we are able to propagate both the convergence of the macroscopic order parameter (Proposition~\ref{th:propagationExistenceBECgeneral}) and key upper bounds derived from Assumptions~\ref{ass:tailCondition} and~\ref{ass:kineticTailCondition} over a finite time interval.
This enables us to close a combined Gr\"onwall estimate for the expectation of the number operator and the generator of the fluctuation dynamics (Lemma~\ref{th:numberPlusGeneratorGronwall}). 
The energy quasi-self-consistency (see Definition~\ref{def:Q-sC}) of the initial quasi-canonical coherent state $\varphi^{\ls 0}_{\varrho,\ls L}$ is crucial to provide the control of deviations from the Hartree energy functional~\eqref{def:HartreeFunctional}, which in turn guarantees that the energy of the quasi-vacuum state $\xi_{\varrho,\ls L}$ is small enough (Proposition~\ref{th:energyOfVoid}).\newline
Combining these ingredients -- the estimates on the generator, the control of the Hartree nonlinearity, and the initial energy bound -- we show that the expected number of fluctuations remains negligible compared to the system size when the density is large, at least for finite times (Lemma~\ref{th:fluctuationsNumber}).
As a result, we establish the convergence of the one-particle reduced density matrix towards the rank-one projection onto the Hartree evolution of the order parameter, with the rate of convergence determined by the decay properties of the initial data (Theorem~\ref{th:gamma1Convergence}).

\smallskip

\noindent The main challenge in our model lies in the nature of the coupling constant of the pairwise interaction that keeps intact the \emph{extensivity} of the potential energy, whereas a factor $\frac{1}{N}$ makes it \emph{intensive}.
As a consequence, the use of operator inequalities and energy estimates coming from the conservation of the $H^1(\Lambda_L)$ norm of the solution to the normalised Hartree equation~\cite[eq. (1.5)]{RoSc09} are no longer helpful in our scaling, since they strongly rely on intensivity.
In particular, we were unable to utilise the truncation method applied in~\cite{RoSc09} nor the refined method to control $\frac{\ls\mathcal{N}^2\!}{N}$ applied in~\cite{BuSaSc14} for our purposes, since any attempt to control any power of $\mathcal{N}$ larger than $1$ corresponds to the emergence of volume factors larger than $\oBig{L^3}$.
Instead, to control cubic terms, we used estimates based on controlling the growth of the generator itself, familiar from~\cite{BeOlSc15}.
Since in our model the $L^\infty$-norm is the only controllable quantity up to some finite times and we cannot hope to achieve small dependences in $L$ of the expectation of the number of fluctuations, the general strategy and appearance of the estimates differ from those of~\cite{RoSc09, BeOlSc15}.
This can be seen, for example, in the structure of the quadratic term in the generator~\eqref{def:quadraticGeneratorQ}.

\bigskip

The remainder of this paper is structured as follows.

\smallskip

\noindent In Section~\ref{sec:setting}, we provide the basics of Fock spaces that will serve as the environment for our discussion.
In this framework, we introduce the definition of a quasi-complete Bose--Einstein condensate -- the main object of our interest -- and then we formalise the obtained results.\newline
In Section~\ref{sec:Generator}, we recover known features of the generator associated with the fluctuation dynamics in order to collect the properties we will be using.\newline
In Section~\ref{sec:qcBECproperties}, we investigate the objects introduced in Section~\ref{sec:qcBEC} in greater depth and clarify their interconnections.\newline
In Section~\ref{sec:hartreeTorus}, we discuss some features of the Hartree equation on the torus, such as well-posedness, the representation in momentum space, and the control of its nonlinearity.\newline
In Section~\ref{sec:excitationsControl}, we develop the proof of the main results, focusing on controlling the expectation of the number of fluctuations.


\section{Setting and Statement of Results}\label{sec:setting}

In this section, we introduce the necessary notions to understand the framework of our problem and formalise the results.

\subsection{Second Quantisation}\label{sec:SQ}

To treat the sequence of $N$-body Hamiltonians as a single operator acting on one Hilbert space, as $N$ and $L$ both grow to infinity in the thermodynamic limit, we work within the \emph{grand canonical picture}, which is characterised by the use of quantum states with indefinite particle number.\newline
To this end, we recall the construction of a symmetric Fock space.

\subsubsection*{Fock Spaces}

Given a complex, separable Hilbert space $\hilbert*$, and $n$ vectors $\phi_1,\ldots,\phi_n\!\in\nls \hilbert*$, define the multi-antilinear functional
\begin{equation*}
    \begin{split}
        &\maps{\phi_1\nls \otimes\cdots\otimes\phi_n}{\hilbert*^n;\C}\\[-5pt]
        &\nquad (f_1,\ldots, f_n)\:\longmapsto \;\prod_{i\ls =1}^n \,\scalar{f_i}{\phi_i}[\hilbert*].
    \end{split}
\end{equation*}
Let $D_n$ denote the set of linear combinations of such functionals endowed with the inner product
\begin{equation}\label{def:innerProductInTensorProduct}
    \scalar{\phi_1\nls \otimes\cdots\otimes\phi_n}{\psi_1\nls \otimes\cdots\otimes\psi_n}[\hilbert*^{\otimes\ls n}]\vcentcolon=\prod_{i\ls =1}^n \,\scalar{\phi_i}{\psi_i}[\hilbert*].\vspace{-0.25cm}
\end{equation}
The tensor product $\overset{n\text{ times}}{\hilbert*\otimes\cdots\otimes \hilbert*}=\vcentcolon\hilbert*^{\otimes\ls n}$, with $\hilbert*^{\otimes\ls 0}\nls \vcentcolon=\C$, is defined as the completion of $D_n$ under the norm induced by~\eqref{def:innerProductInTensorProduct}.\newline
To account for particle indistinguishability, we introduce a \emph{unitary representation} of the symmetric group $\mathfrak{S}_n$ (\ie~the group of permutations of $n$ elements) defined by
\begin{equation*}
\begin{split}
    &\maps{U}{\mathfrak{S}_n;\bounded{\hilbert*^{\otimes\ls n}}}\\
    &\qquad\pi\;\longmapsto\: U_\pi,
\end{split}
\end{equation*}
where $U_\pi$ is the \emph{permutation operator}
$$(U_\pi \ls \phi)(f_1,\ldots f_n)=\phi(f_{\pi^{-1}(1)},\ldots,f_{\pi^{-1}(n)}),\qquad \phi\in\hilbert*^{\otimes \ls  n}.$$
For a system of $n$ indistinguishable particles, the associated Hamiltonian must commute with all $\{U_\pi\}_{\pi\ls \in\ls \mathfrak{S}_N}$.\newline
The symmetric Fock space over $\hilbert*$ is
\begin{equation*}
    \fockS[\hilbert*]\vcentcolon=\bigoplus_{n\ls \in\,\No} S_n\, \hilbert*^{\otimes \ls  n}=\Big\{(\psi^{(n)})_{n\ls \in\,\No}\,\Big|\: \psi^{(n)}\!\in\nls  S_n\,\hilbert*^{\otimes\ls n}, \, \sum_{n\ls \in\,\No}\norm{\psi^{(n)}}[\hilbert*^{\otimes\ls n}]^2<\infty\Big\},
\end{equation*}
where $S_n\!\in\nls \bounded{\hilbert*^{\otimes\ls  n}}$ is the \emph{symmetrisation operator}, namely the orthogonal projection 
$$S_0=1,\quad S_1=\Id_{\hilbert*},\qquad S_n=\frac{1}{n!}\!\sum_{\pi\ls \in\ls \mathfrak{S}_n}\! U_\pi,\quad n\geq 2.$$
\begin{note}
    The time evolution of elements in $S_n\,\hilbert*^{\otimes\ls n}$ remains in that subspace, since $U_\pi$ commutes with the $n$-body Hamiltonian on $\hilbert*^{\otimes\ls n}$, and therefore $S_n$ is conserved along the evolution.
\end{note}
Note that $\fockS[\hilbert*]$ is equipped with the inner product
$$\scalar{\psi}{\phi}[\fockS[\hilbert*]]\vcentcolon= \!\sum_{n\ls \in\,\No} \scalar{\psi^{(n)}}{\phi^{(n)}}[\hilbert*^{\otimes\ls n}],\qquad \psi=(\psi^{(n)})_{n\ls \in\,\No}, \, \phi=(\phi^{(n)})_{n\ls \in\,\No}.$$
As a matter of fact, $\fockS[\hilbert*]$ is a complex, separable Hilbert space.\newline
For our problem, we have $\hilbert*=\Lp{2}[\Lambda_L]$, from which it follows that $S_n\,\hilbert*^{\otimes\ls n}$ is unitarily equivalent to $\LpS{2}[\Lambda_L^n]$.

\subsubsection*{Creation and Annihilation Operators}

The \emph{number operator} plays a central role in the Fock space
\begin{equation}\label{def:numberOperator}
    \begin{split}
        &\maps{\mathcal{N}}{\fockS[\hilbert*]; \fockS[\hilbert*]}\\
        &\nquad(\psi^{(n)})_{n\ls \in\,\No}\,\longmapsto\: (n\ls \psi^{(n)})_{n\ls \in\,\No},
    \end{split}
\end{equation}
which is self-adjoint on the domain
$$\dom{\mathcal{N}}=\Big\{\nls (\psi^{(n)})_{n\ls \in\,\No}\!\in\nls \fockS[\hilbert*]\,\Big|\ls  \sum_{n\ls \in\,\N} \nls  n^2 \ls \norm{\psi^{(n)}}[\hilbert*^{\otimes\ls n}]^2\!<\infty\Big\}.$$
Clearly, $N$-particle vectors $(\psi^{(n)}\delta_{n,\ls N})_{n\ls \in\,\No}\!\in\nls \fockS[\hilbert*]$ are eigenvectors of $\mathcal{N}$ with eigenvalue $N$.
In general, the number of particles is a \emph{random variable} in the grand canonical picture, where $\norm{\phi^{(n)}}[\hilbert*^{\otimes\ls n}]^2$ gives the probability of finding $n$ particles in the system described by the unit vector $(\phi^{(n)})_{n\ls \in\,\No}\!\in\nls \fockS[\hilbert*]$.
The unique non-zero element of the Fock space (up to a phase) in the kernel of $\mathcal{N}$ is called the \emph{vacuum state}
\begin{equation}\label{def:vacuum}
    \Omega\vcentcolon=(\ls 1,0,0,\ldots\ls )\ls .
\end{equation}
Next, define the continuous maps
\begin{align*}
    &\maps{a_n}{\hilbert*;\bounded{\hilbert*^{\otimes\ls n},\hilbert*^{\otimes\ls n-1}}} && \maps{\adj{a}_n}{\hilbert*;\bounded{\hilbert*^{\otimes\ls n},\hilbert*^{\otimes\ls n+1}}}\\
    &\quad f\longmapsto a_n(f) && \mspace{27mu} f\longmapsto \adj{a}_n(f),
\end{align*}
where, given $f\nls \in\hilbert*$ and $\phi_1\nls \otimes\cdots\otimes \phi_n\nls \in\nls  D_n$, the action in $D_n$ is
\begin{subequations}
\begin{align}
    a_n(f) \ls \phi_1\nls \otimes\cdots\otimes \phi_n &= \scalar{f}{\phi_1}[\hilbert*] \,\phi_2\otimes\cdots\otimes \phi_n, \quad n\geq 1,\qquad a_0(f) z = 0, \quad \forall z\in\C,\\
    \adj{a}_n(f)\ls \phi_1\nls \otimes\cdots\otimes \phi_n&=f\nls \otimes \phi_1\nls \otimes\cdots\otimes\phi_n.
\end{align}
\end{subequations}
Note that, while both $a_n(f)$ and $\adj{a}_n(f)$ are linear operators, the map $\adj{a}_n$ is linear, whereas $a_n$ is antilinear.
However, the following bounds hold
\begin{align*}
    \norm{a_n(f)}[\linear{\hilbert*^{\otimes\ls n}\nls ,\, \hilbert*^{\otimes\ls n-1}}]\!\leq \norm{f}[\hilbert*], && \norm{\adj{a}_n(f)}[\linear{\hilbert*^{\otimes\ls n}\nls ,\, \hilbert*^{\otimes\ls n+1}}]\!\leq \norm{f}[\hilbert*].
\end{align*}
The bounded linear transformation (BLT) theorem ensures the existence of a unique norm-preserving extension of $a_n(f)$ and $\adj{a}_n(f)$ from $D_n$ to $\hilbert*^{\otimes\ls n}$.
Moreover, $\adj{a_n(f)}\!=\adj{a}_n(f)$ for all $f\nls \in\hilbert*$.\newline
These quantities serve as building blocks for the definition of the \emph{creation} and \emph{annihilation operators}, denoted by $\adj{a}(f)$ and $a(f)\!\in\nls \linear{\fockS[\hilbert*]}$, respectively.
Specifically, given $\psi=(\psi^{(n)})_{n\ls \in\,\No}\!\in\dom{\mathcal{N}^\frac{1}{2}}=\fdom{\mathcal{N}}$,
\begin{align}\label{def:creationAnnihilation}
    ( a(f)\ls \psi )^{(n)}=\sqrt{n+\nls 1\ls }\ls a_{n\ls +1}\mspace{-0.75mu}(f)\ls \psi^{(n+1)}, && (\adj{a}(f)\ls \psi)^{(n)}=\sqrt{n\ls }\ls  S_n\ls \adj{a}_{n\ls -1}\mspace{-0.75mu}(f)\ls \psi^{(n-1)}.
\end{align}
The adjoint of $a(f),\fdom{\mathcal{N}}$ is $\adj{a}(f), \fdom{\mathcal{N}}$; hence, they are both closed operators.
Furthermore, for all $\psi\in\fdom{\mathcal{N}}$
\begin{align}\label{eq:trivialBoundForCreationAnnihilation}
    \norm{a(f)\ls \psi}\leq\norm{f}[\hilbert*]\,\norm{\mathcal{N}^\frac{1}{2}\psi}, && \norm{\adj{a}(f)\ls \psi}\leq\norm{f}[\hilbert*]\,\norm{(\mathcal{N}\!+\nls 1)^\frac{1}{2}\psi},
\end{align}
and they satisfy the \emph{canonical commutation relations}
\begin{equation}\label{eq:CCR}
    \begin{cases}
        \big[a(f), \adj{a}(g)\big] \psi = \scalar{f}{g}[\hilbert*] \,\psi,\\
        \big[a(f), a(g)\big] \psi=\big[\adj{a}(f), \adj{a}(g)\big] \psi=0,
    \end{cases}
    \qquad \forall f,g\in\hilbert*, \, \psi\in \dom{\mathcal{N}}.
\end{equation}
For $f\nls \in\hilbert*$, we also introduce the self-adjoint operator
\begin{equation}\label{def:Sigal}
    \phi(f)=a(f)+\adj{a}(f),\qquad \dom{\phi(f)}=\fdom{\mathcal{N}}.
\end{equation}
Given a basis $\{f_k\}_{k\ls \in\,\N}\nls \subset \nls \hilbert*$ and a sequence $\mathfrak{n}\nls =\nls \{n_k\}_{k\ls \in\,\N}\nls \subset \nls \No$, where\footnote{The only way in which a sequence of numbers in $\No$ can be summable is if its elements eventually vanish.} $\mathfrak{n}\nls \in\nls \ell_1(\N)$, let $|\mathfrak{n}\rangle\nls \in\nls \fockS[\hilbert*]$ be the $\norm{\mathfrak{n}}[\ell_1(\N)]$-particle state
$$|\mathfrak{n}\rangle \vcentcolon= \frac{1}{\!\nls \sqrt{\nls \prod\limits_{k\ls \in\,\N} \!n_k!\ls }\ls } \bigg[\prod_{k\ls \in\,\N}\adj{a}(f_k)^{n_k}\nls \bigg]\ls \Omega.$$
The set $\big\{|\mathfrak{n}\rangle\! \in\nls \fockS[\hilbert*]\,|\;\mathfrak{n}\nls \in\nls \ell_1(\N)\nls\big\}$ forms an orthonormal basis for $\fockS[\hilbert*]$, and $\mathfrak{n}\nls \in\nls \ell_1(\N)$ is called the \emph{occupation number representation} of $|\mathfrak{n}\rangle\nls \in\nls \fockS[\hilbert*]$, indicating that $n_k\!\in\nls \No$ bosons occupy the one-particle state $f_k\!\in\nls \hilbert*$.
Within this basis, the annihilation and creation operators read
\begin{align*}
    a(f_k)|\mathfrak{n}\rangle = \sqrt{n_k\ls }\: |\{n_\ell\nls -\nls \delta_{\ell,\ls  k}\}_{\ell\ls \in\,\N}\rangle, && 
    \adj{a}(f_k)|\mathfrak{n}\rangle = \sqrt{n_k\nls +\nls 1\ls }\: |\{n_\ell\nls +\nls \delta_{\ell,\ls  k}\}_{\ell\ls \in\,\N}\rangle,
\end{align*}
yielding
\begin{equation}\label{eq:numberAsSumOfOccupations}
    \sum_{k\ls \in\,\N}\adj{a}(f_k)\ls  a(f_k)\mspace{2.5mu} |\mathfrak{n}\rangle = {\textstyle \sum\limits_{k\ls \in\,\N} \nls n_k\mspace{2.5mu} |\mathfrak{n}\rangle} = \mathcal{N} \mspace{2.5mu}|\mathfrak{n}\rangle,
\end{equation}
since $|\mathfrak{n}\rangle$ is an eigenstate of $\mathcal{N}$.

\medskip

\noindent For our case of interest $\hilbert*=\Lp{2}[\Lambda_L]$, the creation and annihilation operators act as follows
\begin{gather*}
    (a(f)\ls \psi)^{(n)}(\vec{x}_1,\ldots,\vec{x}_n)=\sqrt{n+1}\integrate[\Lambda_L]{\conjugate*{f(\vec{x})}\,\psi^{(n+1)}(\vec{x},\vec{x}_1,\ldots,\vec{x}_n);\!d\vec{x}},\\[-2.5pt]
    (\adj{a}(f)\ls \psi)^{(n)}(\vec{x}_1,\ldots,\vec{x}_n)=\frac{1}{\!\nls \sqrt{n\ls }\ls }\sum_{j=1}^n f(\vec{x}_j) \,\psi^{(n-1)}(\vec{x}_1,\ldots,\vec{x}_{j-1},\vec{x}_{j+1},\ldots,\vec{x}_n),
\end{gather*}
where $\psi=(\psi^{(n)})_{n\ls \in\,\No}$ is such that $\psi^{(n)}\!\in\nls \LpS{2}[\Lambda_L^n]$ and $\big\{\sqrt{n\ls }\ls  \norm{\psi^{(n)}}[\Lp{2}[\Lambda_L^n]]\big\}_{n\ls \in\,\No}\!\in\nls \ell_2(\No)$.
In this framework, we introduce the operator-valued distribution
\begin{equation}\label{def:annihilationDistribution}
    (a_{\vec{x}}\psi)^{(n)}(\vec{x}_1,\ldots,\vec{x}_n)\vcentcolon=\sqrt{n+\nls 1\ls }\ls \psi^{(n+1)}(\vec{x},\vec{x}_1,\ldots,\vec{x}_n),
\end{equation}
that satisfies
$$\integrate[\Lambda_L]{\conjugate*{f(\vec{x})}\,(a_{\vec{x}}\psi)^{(n)};\!d\vec{x}}=(a(f)\ls \psi)^{(n)}.$$
By definition~\eqref{def:annihilationDistribution}, the number operator satisfies the quadratic form identity
\begin{equation}\label{eq:NumberFromIntegral}
    \norm{\mathcal{N}^\frac{1}{2}\psi}^2=\integrate[\Lambda_L]{\norm{a_{\vec{x}} \psi}^2;\!d\vec{x}},\qquad \forall \psi\in\fdom{\mathcal{N}}.
\end{equation}

\subsubsection*{The Hamiltonian in Second Quantisation}

We now introduce the second quantisation of the Hamiltonian $H^N_{\nls \varrho,\ls L}$ in the symmetric Fock space $\fockS\big(\Lp{2}(\Lambda_L)\nls \big)$, denoted by $\mathcal{H}_{\varrho,\ls L}\!\in\nls \linear{\fockS\big(\Lp{2}[\Lambda_L]\big)}$. Its action for all vectors in its domain $\psi=\big(\psi^{(n)}\big)_{n\ls \in\ls \No}\!\in\dom{\mathcal{H}_{\varrho,\ls L}}$ is
\begin{subequations}\label{eqs:HamiltonianFock}
\begin{equation}\label{def:HamiltonianFock}
    \begin{split}
        (\mathcal{H}_{\varrho,\ls L}\ls \psi)^{(0)}=0,&\qquad (\mathcal{H}_{\varrho,\ls L}\ls \psi)^{(1)}=-\Delta\psi^{(1)},\\(\mathcal{H}_{\varrho,\ls L}\ls \psi)^{(n)}&=H^{\mspace{0.75mu}n}_{\nls \varrho,\ls L}\ls \psi^{(n)},\qquad n\geq 2,
    \end{split}
\end{equation}
\begin{equation*}
        \dom{\mathcal{H}_{\varrho,\ls L}}=
        \big\{\psi\nls \in\nls \dom{\mathcal{N}^{\ls 2}}\,|\; \psi^{(n)}\!\in\nls H^2(\Lambda^n_L),\,\Delta\psi\nls \in\nls \fockS\big(\Lp{2}[\Lambda_L]\nls \big)\big\}.
\end{equation*}
The associated Hermitian quadratic form can be written in terms of the operator-valued distribution~\eqref{def:annihilationDistribution}
\begin{align}\label{eq:formHamiltonianFock}
    \mathcal{H}_{\varrho,\ls L}[\psi]&=\integrate[\Lambda_L]{\norm{\nabla_{\!\vec{x}\,} a_{\vec{x}}\psi}^2;\!d\vec{x}}+\frac{1}{2\varrho}\integrate[\Lambda^2_L]{V_L(\vec{x}\nls -\nls \vec{y})\,\norm{a_{\vec{y}}a_{\vec{x}}\psi}^2;\!d\vec{x}d\vec{y}}\\
    &=\vcentcolon \mathcal{K}_L[\psi]+\tfrac{1}{\varrho}\ls \mathcal{V}_{\nls L}[\psi],\qquad \psi\in\fdom{\mathcal{H}_{\varrho,\ls L}}.\nonumber
\end{align}
\end{subequations}
Observe that $\fdom{\mathcal{V}_{\nls L}}\nls =\nls \dom{\mathcal{N}}$, while $\fdom{\mathcal{K}_L}\nls =\nls \big\{\psi\nls \in\nls \fockS\big(\Lp{2}[\Lambda_L]\nls \big)\,|\;\psi^{(n)}\!\in\nls  H^1(\Lambda_L^n),\,\nabla\psi\nls \in\nls \fockS\big(\Lp{2}[\Lambda_L]\nls \big)\big\}$.
Consequently, the form domain of the Hamiltonian corresponds to the intersection
$$\fdom{\mathcal{H}_{\varrho,\ls L}}\nls =\nls  \big\{\psi\nls \in\nls \dom{\mathcal{N}}\,|\;\psi^{(n)}\!\in\nls H^1(\Lambda^n_L),\,\nabla\psi\nls \in\nls \fockS\big(\Lp{2}[\Lambda_L]\nls \big)\big\}.$$
By construction, the Hamiltonian defined in~\eqref{eqs:HamiltonianFock} coincides with $H^N_{\nls \varrho,\ls L}$ when applied to any $N$-particle vector $\big(\psi^{(n)}\delta_{n,N}\big)_{n\ls \in\,\No}$ of $\dom{\mathcal{H}_{\varrho,\ls L}}$.
Furthermore, this definition does not couple different $n$-particle sectors; therefore, the number operator $\mathcal{N}$ commutes with $\mathcal{H}_{\varrho,\ls L}$, and $\dom{\mathcal{N}}$ is left invariant under the unitary evolution generated by $\mathcal{H}_{\varrho,\ls L}$, owing to Noether's theorem.

\bigskip

Hereafter, we focus on studying the dynamics generated by the Hamiltonian $\mathcal{H}_{\varrho,\ls L}$ in the high-density thermodynamic limit.
In particular, since the number of particles is not fixed in this grand canonical picture, we shall encode the thermodynamic limit by considering a proper class of initial data (a specific family of \emph{quasi-canonical coherent states} -- see Definition~\ref{def:Q-CCS} below) where the expectation of the number operator behaves asymptotically as $\varrho L^3$.

\smallskip

\noindent Having established the Fock space framework, we proceed to formalise the properties fulfilled by our initial state.

\subsection{Quasi-Complete Bose--Einstein Condensates}\label{sec:qcBEC}

We assume that the system is initially prepared in a state sufficiently close to a Bose--Einstein condensate (in a sense to be clarified) described by the \emph{order parameter} $\Psi_{\!\varrho,\ls L}\nls \in\nls H^1(\Lambda_L)$ satisfying\footnote{Equation~\eqref{eq:becNormalization} equivalently states that both $\limsup\limits_{L\to\infty}\frac{\norm{\Psi_{\mspace{-3.75mu}\varrho,\mspace{0.75mu} L}}[2]^2\!}{\varrho L^3}$ and $\liminf\limits_{L\to\infty} \frac{\norm{\Psi_{\mspace{-3.75mu}\varrho,\mspace{0.75mu} L}}[2]^2\!}{\varrho L^3}$ converge to $1$ as $\varrho\to\infty$.}
\begin{equation}\label{eq:becNormalization}
    \lim_{\varrho\to\infty}\limsup_{L\to\infty}\abs*{1-\frac{1}{\varrho L^3}\!\integrate[\Lambda_L]{\abs{\Psi_{\!\varrho,\ls L}(\vec{x})}^2;\!d\vec{x}}}=0.
\end{equation}
Roughly speaking, $\abs{\Psi_{\!\varrho,\ls L}}^2\!\in\! W^{1,\ls\frac{3}{2}}(\Lambda_L)\cap \Lp{3}[\Lambda_L]$ can be thought of as the density distribution of the Bose--Einstein condensate, which becomes complete in the high-density thermodynamic limit, as all particles eventually fall into the condensate component, by~\eqref{eq:becNormalization}.
\begin{note}\label{rmk:massMustBeConserved}
    The conservation of the number operator $\mathcal{N}$ ensures that the way we perform the thermodynamic limit remains consistent over time in the grand canonical picture.
    Specifically, although the system density $\varrho$ could, in principle, vary with time, both the volume and the expected number of particles are time-independent.
    Therefore, the density -- obtained as the asymptotic ratio between the expected number of particles in our initial datum and the volume of the box -- is constant by construction.
    This, in turn, implies that in order for any possible evolution $\Psi^{\ls t}_{\!\varrho,\ls L}$ of the order parameter $\Psi_{\!\varrho,\ls L}$ to describe a Bose--Einstein condensate that is eventually complete, the quantity $\norm{\Psi^{\ls t}_{\!\varrho,\ls L}}[2]^2$ must be conserved up to an error much smaller than $\varrho L^3$.
\end{note}
To specify the conditions imposed on $\Psi_{\!\varrho,\ls L}\ls $, we first provide some definitions.
These will clarify to what extent $\Psi_{\!\varrho,\ls L}$ can be regarded as a Bose--Einstein condensate.\newline
For an ideal Bose gas, the critical temperature increases with the density; consequently, Bose--Einstein condensation can occur at any temperature, provided the density $\varrho$ is sufficiently large.
Moreover, the condensate becomes progressively more complete as the density increases, in the sense that the fraction of particles outside the condensate becomes negligible when $\varrho\to\infty$.
In the regime $\varrho \ls a_{\mathrm{s}}^3\ll 1$ (which is the case of interest here, since $a_{\mathrm{s}}$ scales as $\varrho^{-1}$), interactions are expected to produce only a small correction to the critical temperature, so that the qualitative behaviour of the ideal gas remains a good guide.
This is the physical picture we heuristically aim to capture in our framework.\newline
In the literature, attention is typically restricted to states exhibiting Bose--Einstein condensation when they are endowed with a coherent structure.
Below, we introduce Definitions~\ref{def:Q-Omega} and~\ref{def:Q-CCS} in order to adapt this standard construction to the present setting, where both the thermodynamic limit and a subsequent high-density limit are considered.
Specifically, these notions are designed to describe states equipped with a proper \emph{quasi-coherent} structure meant to provide the suitable formulation of a \emph{quasi-complete} Bose--Einstein condensate (see Definition~\ref{def:QCBEC}) that mimics our heuristic physical picture.

\medskip

\noindent First, we define a broad class of vectors generalising some properties of the vacuum state.
\begin{defn}[Quasi-Vacuum States]\label{def:Q-Omega}
    Given a non-zero $f_{\varrho,\ls L}\!\in\nls \Lp{2}(\Lambda_L)$, a vector $\Omega_{\varrho,\ls L}\!\in\dom{\mathcal{N}}$ is a \emph{quasi-vacuum} state with respect to the one-particle wave function $f_{\varrho,\ls L}$ if
    \begin{enumerate}[label=\roman*), font=\itshape]
        \item $\norm{\Omega_{\varrho,\ls L}}=1$, for all $\varrho,L\nls >\nls 0$;
        \item $\lim\limits_{\varrho\to\infty}\limsup\limits_{\vphantom{{}_{q}}L\to\infty} \frac{\norm{\mathcal{N}^{\frac{1}{2}}\Omega_{\varrho,\ls L}}}{\vphantom{|^{[]}}\norm{f_{\varrho,\ls L}}[2]}=0$;
        \item given $\maps{\phi}{\Lp{2}[\Lambda_L];\linear{\fockS\big(\Lp{2}[\Lambda_L]\nls \big)\nls }}$ introduced in~\eqref{def:Sigal}, one has
        $$\lim_{\varrho\to\infty}\limsup_{L\to\infty}\ls  \abs*{\frac{\var_{\ls \Omega_{\varrho,\ls L}}\!\big[\ls \mathcal{N}\nls +\phi(f_{\varrho,\ls L})\big]}{\vphantom{|^1}\norm{f_{\varrho,\ls L}}[2]^2}-1\ls }=0.$$
    \end{enumerate}
\end{defn}
\begin{note}
    Since $\mathcal{N}$ and $\phi(f_{\varrho,\ls L})$ do not commute, the variance of the sum is strictly positive (they have no eigenvectors in common).
    For instance, for the exact vacuum state $\Omega$ (defined in~\eqref{def:vacuum}), one has $${\textstyle\var_{\ls \Omega}}\big[\ls \mathcal{N}\nls +\phi(f_{\varrho,\ls L})\big]\nls =\norm{f_{\varrho,\ls L}}[2]^2.$$
    Furthermore, the expectation value $\mathbb{E}_{\,\Omega_{\varrho,\ls L}}\nls \big[\ls \mathcal{N}\nls +\phi(f_{\varrho,\ls L})\big]$ is small compared to $\norm{f_{\varrho,\ls L}}[2]^2$ in the iterated limit by means of the second point of this definition.
    Specifically,
    $$\mathbb{E}_{\,\Omega_{\varrho,\ls L}}\nls \big[\ls \mathcal{N}\nls +\phi(f_{\varrho,\ls L})\big]\leq \norm{\mathcal{N}^{\frac{1}{2}}\Omega_{\varrho,\ls L}}^2\nls +2\ls \norm{f_{\varrho,\ls L}}[2]\,\norm{\mathcal{N}^\frac{1}{2}\Omega_{\varrho,\ls L}},$$
    which implies
    $$\lim_{\varrho\to\infty}\limsup_{L\to\infty}\ls \frac{\mathbb{E}_{\,\Omega_{\varrho,\ls L}}\nls \big[\ls \mathcal{N}\nls +\phi(f_{\varrho,\ls L})\big]}{\norm{f_{\varrho,\ls L}}[2]^2}\leq \lim_{\varrho\to\infty}\limsup_{L\to\infty}\left[\frac{\norm{\mathcal{N}^{\frac{1}{2}}\Omega_{\varrho,\ls L}}^2\!}{\norm{f_{\varrho,\ls L}}[2]^2}+2\,\frac{\norm{\mathcal{N}^\frac{1}{2}\Omega_{\varrho,\ls L}}}{\norm{f_{\varrho,\ls L}}[2]}\right]\!=0.$$
    Thus, Definition~\ref{def:Q-Omega} covers sequences of vectors where both the variance and the expectation value of the operator $\mathcal{N}\nls +\phi(f_{\varrho,\ls L})$ behave the same as in $\Omega$, up to an error smaller than $\norm{f_{\varrho,\ls L}}[2]^2\ls $.
    Strictly speaking, not all such sequences are included, as we require that the expectation of the number of particles has size smaller than $\norm{f_{\varrho,\ls L}}^2$, which is a stronger assumption than asking the same for $\mathcal{N}\nls +\phi(f_{\varrho,\ls L})$.
\end{note}
Before proceeding, we introduce the \emph{Weyl map} $\maps{\mathcal{W}}{\Lp{2}[\Lambda_L];\bounded{\fockS[\Lp{2}[\Lambda_L]]}}$
\begin{equation}\label{def:Weyl}
    \mathcal{W}:\, f \:\longmapsto\; \mathcal{W}(f)\vcentcolon=e^{-i\,\phi(if)}=e^{\adj{a}(f)\ls -\ls a(f)},
\end{equation}
where $\mathcal{W}(f)$ is unitary (indeed, $\phi(if),\fdom{\mathcal{N}}$ is self-adjoint) and leaves $\fdom{\mathcal{N}}$ invariant (for further details on the Weyl operator, see~\eg~\cite[Lemma~2.2]{RoSc09}).
Then, let $\xi_{\varrho,\ls L}\!\in\nls \dom{\mathcal{N}}$ be a quasi-vacuum state with respect to $\Psi_{\!\varrho,\ls L}\ls $, and consider the initial state\footnote{Identity~\eqref{eq:WeylNConjugation} shows that also the set $\dom{\mathcal{N}}$ is invariant under the action of the Weyl operator.} $\varphi^{\ls 0}_{\varrho,\ls L}\!\vcentcolon=\nls \mathcal{W}(\Psi_{\!\varrho,\ls L})\ls \xi_{\varrho,\ls L}\!\in\nls \dom{\mathcal{N}}$.
Its time evolution is
\begin{equation}\label{def:evolutionQ-CCS}
    \varphi^{\ls t}_{\varrho,\ls L}=e^{-i\,\mathcal{H}_{\varrho,\ls L}\ls t}\,\mathcal{W}(\Psi_{\!\varrho,\ls L})\ls \xi_{\varrho,\ls L},\qquad t\geq 0.
\end{equation}
If $\xi_{\varrho,\ls L}$ were exactly the vacuum $\Omega$, then $\mathcal{W}(\Psi_{\!\varrho,\ls L})\ls \Omega$ would be a vector in the Fock space representing a superposition of factorised states $\Psi_{\!\varrho,\ls L}^{\ls \otimes \ls n}$, each occurring with probability $e^{-\norm{\Psi_{\!\varrho,\mspace{0.75mu} L}}[2]^2}\mspace{0.75mu}\frac{\norm{\Psi_{\!\varrho,\mspace{0.75mu} L}}[2]^{2n}\!\nls }{n!}$.
More precisely, the observable associated with the number of particles would be a random variable following a \emph{Poisson distribution} with both mean and variance equal to $\norm{\Psi_{\!\varrho,\ls L}}[2]^2$.
Such properties stem from the fact that $\mathcal{W}(\Psi_{\!\varrho,\ls L})\ls \Omega$ -- known as a \emph{canonical coherent state} -- is an eigenvector of the annihilation operator $a(g)$ for any $g\nls \in\!\Lp{2}[\Lambda_L]$, with eigenvalue $\scalar{g}{\Psi_{\!\varrho,\ls L}}[2]$.\newline
These properties motivate the following definition.
\begin{defn}[Quasi-Coherent States]\label{def:Q-CCS}
    A vector $\phi_{\varrho,\ls L}\!\in\nls \dom{\mathcal{N}}$ is a \emph{quasi-coherent state} if
    \begin{enumerate}[label=\roman*), font=\itshape]
        \item $\vphantom{\lim\limits_{|}}\norm{\phi_{\varrho,\ls L}}=1$, for all $\varrho,L\nls >\nls 0$;
        \item $\lim\limits_{\varrho\to\infty}\limsup\limits_{\vphantom{|_q}L\to\infty} \ls \abs*{\frac{1}{\varrho L^3}\ls \mathbb{E}_{\phi_{\varrho,\ls L}}[\ls \mathcal{N}\ls ]-1}=0$;
        \item $\lim\limits_{\varrho\to\infty}\limsup\limits_{L\to\infty}\ls \abs*{\frac{1}{\varrho L^3}\nls \var_{\phi_{\varrho,\ls L}}[\ls \mathcal{N}\ls ]-1}=0$.
    \end{enumerate}
    Moreover, $\phi_{\varrho,\ls L}$ is a \emph{quasi-canonical coherent state} if, additionally, for any $g_{\varrho,\ls L}\nls \in\nls \Lp{2}(\Lambda_L)$ with $\norm{g_{\varrho,\ls L}}[2]=1$ for all $\varrho,L\nls >\nls 0$, there exists\footnote{A quasi-eigenvalue must satisfy $\limsup\limits_{\varrho\to\infty}\limsup\limits_{L\to\infty}\frac{\abs{z_{\varrho,\ls L}}}{\vphantom{\big|}\norm*{\mathcal{N}^{1/2}\phi_{\varrho,\ls L}}}\leq 1$.} $z_{\varrho,\ls L}\!\in\C\,$ such that
    $$\lim_{\varrho\to\infty}\limsup_{L\to\infty}\ls \frac{\norm*{\big(a(g_{\varrho,\ls L})-z_{\varrho,\ls L}\big)\ls \phi_{\varrho,\ls L}}}{\norm{\mathcal{N}^{\frac{1}{2}}\phi_{\varrho,\ls L}}}=0.$$
    In this case, $\phi_{\varrho,\ls L}$ is a \emph{quasi-eigenstate} of $a(g_{\varrho,\ls L})$ with \emph{quasi-eigenvalue} $z_{\varrho,\ls L}\ls $.
\end{defn}
\begin{note}
    By Proposition~\ref{th:WeylCoherent}, and because $\mathcal{N}$ commutes with $\mathcal{H}_{\varrho,\ls L}$ (\ie~$e^{-i\,\mathcal{H}_{\varrho,\ls L} \ls t}\mathcal{N}\subseteq \mathcal{N}e^{-i\,\mathcal{H}_{\varrho,\ls L} \ls t}$ for all $t\nls \in\nls \R$), $\varphi^{\ls t}_{\varrho,\ls L}$ introduced in~\eqref{def:evolutionQ-CCS} is a quasi-coherent state for all $t\nls \geq\nls 0$, and it is also quasi-canonical coherent at $t=0$.
\end{note}
To formalise the condition that $\Psi_{\!\varrho,\ls L}$ represents a Bose--Einstein condensate, we define the (normalised) one-particle reduced density matrix $\gamma^{(1)}_\psi\!\in\nls  \schatten{1}\big(\Lp{2}(\Lambda_L)\nls \big)$ for a vector $\psi\nls \in\nls \fdom{\mathcal{N}}\nls\smallsetminus\nls\ker{\mathcal{N}}$, as the integral operator with kernel
\begin{equation}\label{def:gamma1Kernel}
    \gamma^{(1)}_\psi(\vec{x},\vec{y})=\frac{\scalar{\ls a_{\vec{y}}\ls \psi}{a_{\vec{x}}\ls \psi}}{\norm{\mathcal{N}^{\frac{1}{2}}\psi}^2}.
\end{equation}
This operator is self-adjoint, non-negative, and has unit trace.
Generally speaking, if $\gamma^{(1)}_\psi$ has an eigenvalue close to $1$ (its largest possible value), the corresponding eigenfunction represents the one-particle wave function of the complete condensate.
We specify this notion in our framework with the following definition.
\begin{defn}[Quasi-Complete BEC]\label{def:QCBEC}
    A vector $\psi_{\varrho,\ls L}\nls \in\nls \fdom{\mathcal{N}}\nls\smallsetminus\nls \ker{\mathcal{N}}$, with $\norm{\psi_{\varrho,\ls L}}=1$ for all $\varrho,L\nls >\nls 0$, exhibits \emph{quasi-complete condensation} if there exist $\Phi_{\varrho,\ls L}\!\in\! H^1(\Lambda_L)$ satisfying~\eqref{eq:becNormalization} and $\Phi\nls \in\! H^1(\Lambda_1)$ with $\norm{\Phi}[\Lp{2}[\Lambda_1]]\nls =1$ such that
    $$\lim_{\varrho\to\infty}\limsup_{L\to\infty}\,\norm*{\tfrac{\Phi_{\nls\varrho,\mspace{0.75mu} L}(\ls \vec{\cdot}\ls )}{\norm{\Phi_{\nls\varrho,\mspace{0.75mu} L}}[2]\!\nls}-\tfrac{1}{L^{3/2}}\ls \Phi\nls \big(\mspace{-0.75mu}\tfrac{\vec{\cdot}}{L}\mspace{-0.75mu}\big)\nls }[H^1(\Lambda_L)]\! =0,\leqno i)$$
    \vspace{-0.35cm}
    $$\lim_{\varrho\to\infty}\limsup_{L\to\infty}\ls \abs*{\frac{\norm{\Phi_{\nls\varrho,\ls L}}[2]^2}{\varrho L^3}-\frac{\norm*{a\Big( \tfrac{\Phi_{\nls\varrho,\ls L}}{\norm{\Phi_{\nls\varrho,\ls L}}[2] \!\nls}\Big)\ls \psi_{\varrho,\ls L}}^2\!\nls }{\norm{\mathcal{N}^{\frac{1}{2}}\psi_{\varrho,\ls L}}^2}\:}\!\nls =0.\leqno ii)$$
    In this case, $\Phi_{\varrho,\ls L}$ is a \emph{quasi-complete Bose--Einstein condensate} for $\psi_{\varrho,\ls L}\ls $.
\end{defn}
\begin{note}
    If the system is prepared in an initial state exhibiting quasi-complete condensation, the first condition is demanding that $\frac{L^{\frac{3}{2}}}{\norm{\Psi_{\!\varrho,\mspace{0.75mu} L}}[2]\!\nls}\ls\Psi_{\!\varrho,\ls L}(L\,\vec{\cdot}\ls )$ converges in the iterated limit to a function $\Psi$ that belongs in $H^1(\Lambda_1)$.
    This means that from the macroscopic perspective (according to which a wave function is normalised to $1$ and defined on a region whose volume is of order $1$), the mass of the order parameter and its kinetic energy must be well-defined in the high-density thermodynamic limit.\newline
    We also stress that, due to~\eqref{eq:becNormalization}, the second condition is equivalent to
    $$\lim_{\varrho\to\infty}\limsup_{L\to\infty}\ls \left[1-\frac{\norm*{a\Big( \tfrac{\Psi_{\!\varrho,\mspace{0.75mu} L}}{\norm{\Psi_{\!\varrho,\mspace{0.75mu} L}}[2] \!\nls}\Big)\ls \varphi^{\ls 0}_{\varrho,\ls L}}^2\!\nls }{\norm{\mathcal{N}^{\frac{1}{2}}\varphi^{\ls 0}_{\varrho,\ls L}}^2}\:\right]\!\nls =0.$$
    This entails the existence of a constant $c_\varrho\!\in\nls [\ls 0,1]$ such that $\lim\limits_{\varrho\to\infty} \!c_\varrho=0$, and
    \begin{equation}
        \liminf_{L\to\infty}\frac{1}{\norm{\Psi_{\!\varrho,\ls L}}[2]^2\!}\: \scalar{\Psi_{\!\varrho,\ls L}}{\gamma^{(1)}_{\varphi^{\ls 0}_{\varrho,\ls L}}\!\Psi_{\!\varrho,\ls L}}[2]=\liminf_{L\to\infty}\frac{\norm*{a\Big(\tfrac{\Psi_{\!\varrho,\mspace{0.75mu} L}}{\norm{\Psi_{\!\varrho,\mspace{0.75mu} L}}[2]\!\nls}\Big)\ls \varphi^{\ls 0}_{\varrho,\ls L}}^2\!\nls }{\norm{\mathcal{N}^{\frac{1}{2}}\varphi^{\ls 0}_{\varrho,\ls L}}^2}= 1-c_\varrho.
    \end{equation}
    Roughly speaking, a macroscopic fraction of particles initially occupies the same one-particle state when $\varrho$ is large enough, approaching complete condensation as $\varrho\to\infty$.
    In particular, this means that a quasi-complete Bose--Einstein condensate is also a standard (incomplete) Bose--Einstein condensate when $\varrho$ is large enough (so that $c_\varrho \nls \lneq\nls  1$).
    However, the identification between the mass $\norm{\Psi_{\!\varrho, \ls L}}[2]^2$ of the effective order parameter and the condensate occupation extracted from the one-particle reduced density matrix is posed at the leading order in the high-density thermodynamic limit.
    More precisely, our assumptions ensure that both quantities encode the same macroscopic fraction of particles, up to errors much smaller than $\varrho L^3$.
\end{note}
Since our initial state is given by $\varphi^{\ls 0}_{\varrho,\ls L}\!=\mathcal{W}(\Psi_{\!\varrho,\ls L})\ls \xi_{\varrho,\ls L}$, we actually select a class of quasi-canonical coherent states among all possible vectors exhibiting quasi-complete condensation.
Indeed, as shown by combining Propositions~\ref{th:WeylCoherent} and~\ref{th:eigenfunctionOfaMeansQBEC}, quasi-canonical coherent states always exhibit quasi-complete condensation, provided that $\frac{L^{\frac{3}{2}}}{\norm{\Psi_{\!\varrho,\mspace{0.75mu} L}}[2]\!\nls}\ls\Psi_{\!\varrho,\ls L}(L\,\vec{\cdot}\ls )$ converges to some wave function $\Psi\in H^1(\Lambda_1)$ in the iterated limit.

\medskip

\noindent We also require our initial state to have energy close to the Hartree energy of the associated quasi-complete Bose--Einstein condensate.
\begin{defn}[Energy Quasi-Self-Consistency]\label{def:Q-sC}
    Consider a vector $\psi_{\varrho,\ls L}\nls \in\nls \fdom{\mathcal{H}_{\varrho,\ls L}}$ such that $\norm{\psi_{\varrho,\ls L}}=1$ for all $\varrho,L\nls >\nls 0$, which exhibits quasi-complete condensation with $\Phi_{\varrho,\ls L}\!\in\nls  H^1(\Lambda_L)$ a quasi-complete Bose--Einstein condensate for $\psi_{\varrho,\ls L}\ls $.
    We say that $\psi_{\varrho,\ls L}$ is \emph{energetically quasi-self-consistent} if
    \begin{equation*}
        \lim_{\varrho\to\infty}\limsup_{L\to\infty}\frac{1}{\varrho L^3}\ls \Big\lvert\mathcal{H}_{\varrho,\ls L}[\ls \psi_{\varrho,\ls L}]\nls -\mathscr{E}_{\varrho,\ls L}[\Phi_{\varrho,\ls L}]\Big\rvert=0,    
    \end{equation*}
    where 
    for any $\phi\in\nls  H^1(\Lambda_L)$ we have introduced the \emph{Hartree energy functional}
    \begin{equation}\label{def:HartreeFunctional}
        \mathscr{E}_{\varrho,\ls L}[\ls \phi\ls ]\vcentcolon=\integrate[\Lambda_L]{\abs{\nabla_{\!\vec{x}\,}\phi(\vec{x})}^2+\tfrac{1}{2\varrho}\big(V_L\!\ast\nls \abs{\ls \phi\ls }^2\big)\nls (\vec{x})\,\abs{\phi(\vec{x})}^2;\!d\vec{x}}.
    \end{equation}       
\end{defn}
We emphasise that
\begin{equation}
    \mathcal{H}_{\varrho,\ls L}\big[\mathcal{W}(\phi)\ls \Omega\big] = \mathscr{E}_{\varrho,\ls L}[\ls \phi\ls ],\qquad\forall \phi\in H^1(\Lambda_L).
\end{equation}
In particular, we have already mentioned that $\varphi^{\ls 0}_{\varrho,\ls L}\!=\nls \mathcal{W}(\Psi_{\!\varrho,\ls L})\ls \xi_{\varrho,\ls L}$ exhibits the quasi-complete Bose--Einstein condensate $\Psi_{\!\varrho,\ls L}$.
Calling for $\varphi^{\ls 0}_{\varrho,\ls L}$ to be energetically quasi-self-consistent is essential for quantifying the expected energy of the quasi-vacuum state $\xi_{\varrho,\ls L}$.
Specifically, by Proposition~\ref{th:energyOfVoid}, the energy quasi-self-consistency of $\varphi^{\ls 0}_{\varrho,\ls L}$ is required to ensure that $\mathcal{H}_{\varrho,\ls L}[\xi_{\varrho,\ls L}]$ is much smaller than $\varrho L^3$.
\begin{note}\label{rmk:HartreeEnergyMustBeConserved}
    Since the energy is conserved (in particular, $\mathcal{H}_{\varrho,\ls L}[\varphi^{\ls t}_{\varrho,\ls L}]=\mathcal{H}_{\varrho,\ls L}[\varphi^{\ls 0}_{\varrho,\ls L}]$), the fact that we assume $\varphi^{\ls 0}_{\varrho,\ls L}$ to be energetically quasi-self-consistent means that also $\varphi^{\ls t}_{\varrho,\ls L}$ is, and therefore the Hartree energy of $\Psi^{\ls t}_{\!\varrho,\ls L}$ must be preserved, up to an error much smaller than $\varrho L^3$.
\end{note}

\medskip

Our goal is to find an effective description of relevant degrees of freedom that approximates the dynamics of the many-body system.
To achieve this, we seek an evolution $\xi^{\ls t}_{\varrho,\ls L}$ for our initial quasi-vacuum state, satisfying $\varphi^{\ls t}_{\varrho,\ls L}\!=\mathcal{W}(\Psi^{\ls t}_{\!\varrho,\ls L})\ls \xi^{\ls t}_{\varrho,\ls L}\ls $, for some $\Psi^{\ls t}_{\!\varrho,\ls L}$ solving a suitable nonlinear PDE with initial datum $\Psi_{\!\varrho,\ls L}\ls $.
We claim that if $\varphi^{\ls t}_{\varrho,\ls L}$ remains a quasi-canonical coherent state (not only quasi-coherent), then $\varphi^{\ls t}_{\varrho,\ls L}$ exhibits quasi-complete condensation over time, with $\Psi^{\ls t}_{\!\varrho,\ls L}$ as a quasi-complete Bose--Einstein condensate (this is the content of Remark~\ref{rmk:qcCoherentExhibitsqcBEC}).
By construction, we set
$$\varphi^{\ls t}_{\varrho,\ls L}\!=\mathcal{W}(\Psi^{\ls t}_{\!\varrho,\ls L})\,\mathcal{U}_{\varrho,\ls L}(t)\ls \xi_{\varrho,\ls L},$$
leading to the definition of the so-called \emph{fluctuation dynamics}
\begin{equation}\label{def:fluctuationDynamics1}
    \mathcal{U}_{\varrho,\ls L}(t)\vcentcolon=\adj{\mathcal{W}\ls }\nls (\Psi^{\ls t}_{\!\varrho,\ls L})\ls e^{-i\,\mathcal{H}_{\varrho,\ls L}\ls  t}\ls \mathcal{W}(\Psi_{\!\varrho,\ls L}).
\end{equation}
In other words, we aim to prove that $\xi^{\ls t}_{\varrho,\ls L}\!=\mathcal{U}_{\varrho,\ls L}(t)\ls \xi_{\varrho,\ls L}$ is a quasi-vacuum state with respect to $\Psi^{\ls t}_{\!\varrho,\ls L}\ls $.
In this case, by Proposition~\ref{th:WeylCoherent}, $\varphi^{\ls t}_{\varrho,\ls L}$ turns out to be a quasi-canonical coherent state and a quasi-eigenfunction of $a\Big(\tfrac{\Psi^{\ls t}_{\mspace{-3.75mu}\varrho,\mspace{0.75mu} L}}{\norm{\Psi^{\ls t}_{\mspace{-3.75mu}\varrho,\mspace{0.75mu} L}}[2]\!\nls}\Big)$ with quasi-eigenvalue $\scalar{\ls \tfrac{\Psi^{\ls t}_{\mspace{-3.75mu}\varrho,\mspace{0.75mu} L}}{\norm{\Psi^{\ls t}_{\mspace{-3.75mu}\varrho,\mspace{0.75mu} L}}[2]\!\nls}}{\Psi^{\ls t}_{\!\varrho,\ls L}\ls }[2]=\norm{\Psi^{\ls t}_{\!\varrho,\ls L}}[2]$, which, in turn, meets the condition of Proposition~\ref{th:eigenfunctionOfaMeansQBEC}.
Consequently, $\varphi^{\ls t}_{\varrho,\ls L}$ would exhibit quasi-complete condensation with $\Psi^{\ls t}_{\!\varrho,\ls L}\ls $ as the quasi-complete Bose--Einstein condensate (the existence of the macroscopic limit at positive times is guaranteed by Proposition~\ref{th:propagationExistenceBECgeneral}), which represents the time evolution of the order parameter $\Psi_{\!\varrho,\ls L}\ls $.
In fact, to prove the preservation over time of the structure of a quasi-complete Bose--Einstein condensate, we show the convergence in Hilbert--Schmidt norm of the normalised one-particle reduced density matrix associated with our quasi-canonical coherent state $\varphi^{\ls t}_{\varrho,\ls L}$ towards the rank-one projection onto $\tfrac{\Psi^{\ls t}_{\mspace{-3.75mu}\varrho,\mspace{0.75mu} L}}{\norm{\Psi^{\ls t}_{\mspace{-3.75mu}\varrho,\mspace{0.75mu} L}}[2]\!\nls}$ (\cfr~Theorem~\ref{th:gamma1Convergence}).

\subsection{Main Results}\label{sec:Results}

We now list the assumptions to which we will resort in order to prove our results.
\begin{assumption}\label{ass:initialBEC}
    Given a wave function $\Psi_{\!\varrho,\ls L}\!\in\! H^1(\Lambda_L)$ satisfying~\eqref{eq:becNormalization}, we assume there exists $\Psi\nls \in\! H^1(\Lambda_1)$ satisfying $\norm{\Psi}[\Lp{2}[\Lambda_1]]=1$ and
    $$\lim_{\varrho\to\infty}\limsup_{L\to\infty}\,\norm*{\tfrac{\Psi_{\!\varrho,\mspace{0.75mu} L}(\ls \vec{\cdot}\ls )}{\norm{\Psi_{\!\varrho,\mspace{0.75mu} L}}[2]\!\nls}-\tfrac{1}{L^{3/2}}\ls \Psi\nls \big(\mspace{-0.75mu}\tfrac{\vec{\cdot}}{L}\mspace{-0.75mu}\big)\nls }[H^1(\Lambda_L)]\! =0.$$
\end{assumption}
This assumption ensures that the initial quasi-canonical coherent state $\varphi^{\ls 0}_{\!\varrho,\ls L}$ exhibits a quasi-complete Bose--Einstein condensate.
Then, we additionally require that $\varphi^{\ls 0}_{\!\varrho,\ls L}$ is energetically quasi-self-consistent.
\begin{assumption}\label{ass:energyConsistence}
    We consider a quasi-vacuum state $\xi_{\varrho,\ls L}\!\in\nls \fdom{\mathcal{H}_{\varrho,\ls L}}$
    with respect to $\Psi_{\!\varrho,\ls L}\!\in\nls H^1(\Lambda_L)$ such that the corresponding quasi-canonical coherent state\footnote{The Weyl operator leaves $\fdom{\mathcal{H}_{\varrho,\ls L}}$ invariant (\cfr~identity~\eqref{eq:HamiltonianWeylSandwich} below).\label{foo:invariantfdomH}} $\varphi^{\ls 0}_{\varrho,\ls L}\!=\mathcal{W}(\Psi_{\!\varrho,\ls L})\ls \xi_{\varrho,\ls L}$ is energetically quasi-self-consistent.
\end{assumption}
This is important to have control over the energy of the quasi-vacuum state with respect to $\Psi_{\!\varrho,\ls L}$, which is ultimately needed at the end of Lemma~\ref{th:fluctuationsNumber}.\newline
The following Assumption is not always taken into account in our results, as its sole technical consequence is to enlarge the time interval over which convergence to the many-body Schr\"odinger dynamics can be established.
However, its importance lies in its physical interpretation and in the role it plays in understanding the framework under consideration.
To state such an Assumption, let the operator $\langle\gamma^{(1)}_{\psi}\rangle\!\in \schatten{1}[\Lp{2}[\Lambda_L]]$ denote the translation-invariant component of the normalised one-particle reduced density matrix of a state $\psi\nls\in\nls\fdom{\mathcal{N}}\nls\smallsetminus\nls\ker{\mathcal{N}}$.
The expression of its kernel is
\begin{equation}\label{def:traslationInvariantProj}
    \langle\gamma^{(1)}_{\psi}\rangle(\vec{x},\vec{y})\vcentcolon=\frac{1}{L^3}\integrate[\Lambda_L]{\gamma^{(1)}_\psi(\vec{x}\nls +\nls\vec{z},\vec{y}\nls +\nls\vec{z}); \!d\vec{z}},\qquad \vec{x},\vec{y}\in\Lambda_L.
\end{equation}
\begin{assumption}\label{ass:translationalInvariance}
    The one-particle reduced density matrix associated with the quasi-canonical coherent state $\varphi^{\ls 0}_{\varrho,\ls L}$ is supposed to be translationally invariant in the high-density thermodynamic limit, namely
    $$\lim_{\varrho\to\infty}\limsup_{L\to\infty}\,\norm*{\gamma^{(1)}_{\varphi^{\ls 0}_{\varrho,\ls L}}\!\nls -\langle\gamma^{(1)}_{\varphi^{\ls 0}_{\varrho,\ls L}}\rangle}[\mathrm{HS}]\!\!=0.$$
\end{assumption}
This asserts that the initial one-particle reduced density matrix becomes spatially homogeneous in the high-density thermodynamic limit; equivalently, only its translation-invariant component survives asymptotically.
Such a requirement is particularly natural for equilibrium states in our setting, since the underlying Hamiltonian is translationally invariant and Bose--Einstein condensation does not break translational symmetry\footnote{Rather, Bose--Einstein condensation is a second-order phase transition that breaks the $U(1)$ symmetry generated by $\mathcal{N}$ (\cfr~\cite{DeNANa25}).}.
As a consequence, the infinite-particle system obtained in the limit $L\to\infty$ is expected to inherit translational invariance at equilibrium\footnote{In general, we do not expect a quasi-complete Bose--Einstein condensate to be exhibited by an equilibrium state.
However, a \emph{complete} one should be exhibited by some such state, even though a rigorous proof is lacking.}, in particular, as soon as the condensate becomes complete (that is, as $\varrho\to\infty$).\newline
However, from a broader perspective, different Hamiltonians -- possibly defined on domains of different shapes or subject to different boundary conditions -- may lead to distinct infinite-particle systems in principle.
Assumption~\ref{ass:translationalInvariance} precisely isolates the class of infinite-particle systems whose non-translation-invariant components become negligible as $\varrho\to\infty$ at the level of their one-particle structure.
Any such model, irrespective of its microscopic origin, is expected to exhibit the same effective macroscopic behaviour as ours.
Hence, once the condensate is complete, all these systems are expected to fall within the same universality class and, in particular, to give rise to the same effective dynamics.\newline
Moreover, Assumption~\ref{ass:translationalInvariance} admits a further physical interpretation: it selects the class of one-particle states into which bosons may condense.
More precisely, by Corollary~\ref{th:BECmomentumCharacterisation}, Assumption~\ref{ass:translationalInvariance} is equivalent to having a plane wave as the macroscopic counterpart of the quasi-complete Bose--Einstein condensate.
Consequently, at sufficiently high density, condensation occurs not in a generic quantum state but specifically in a momentum state.

\smallskip

\noindent Assumption~\ref{ass:translationalInvariance} should ultimately be regarded as a global requirement on the initial data.
In particular, it describes a perturbation of the translationally invariant solution in the spirit of~\cite{LeSa15} and related works, but formulated directly at the microscopic level, where one faces the additional challenge of controlling the many-body dynamics and preserving the relevant properties throughout the iterated high-density thermodynamic limit.\newline
It is important to note, however, that both this notion of perturbation and the notion of a macroscopic counterpart of a quasi-complete Bose--Einstein condensate are formulated in a topology involving the entire box $\Lambda_L$.
From a physical perspective, a more natural framework would require these convergences only locally, by testing against a suitable class of observables localising on length scales much smaller than $L$, as in~\cite[Theorem 3.7]{FrPoSc25}.
In such a setting, one should only expect convergence towards the effective dynamics in this weaker topology. Nevertheless, this is precisely the level of convergence needed to describe the behaviour of local observables, which are the physically relevant quantities.
Moreover, such a local formulation would substantially relax the universality selection principle discussed above.
Indeed, while distinct microscopic configurations may lead to different limits globally, they are expected to become indistinguishable in the thermodynamic limit when probed through local observables.
Therefore, in this local sense, they all share the same effective macroscopic description.
Although physically more natural, the analysis of convergence in such a local topology appears considerably more challenging than the global framework adopted in the present work, and lies beyond the scope of the present article.

\medskip

\noindent In order to state the two missing (technical) assumptions, we rewrite the order parameter in its macroscopically rescaled momentum representation (see Section~\ref{sec:HartreeInMomentumRepresentation} for additional details)
\begin{subequations}
\begin{gather}
    \Psi_{\!\varrho,\ls L}(\vec{x})=\norm{\Psi_{\!\varrho,\ls L}}[2]\!\sum_{\vec{n}\ls \in\,\Z^3}\nls  \frac{e^{\frac{2\pi\ls  i}{L}\mspace{2.25mu}\vec{n}\ls \cdot\,\vec{x}}\!\nls}{L^{3/2}} \:\alpha^{\ls 0}_{\varrho,\ls L}(\vec{n}),\\
    \alpha^{\ls 0}_{\varrho,\ls L}(\vec{n})=\!\integrate[\Lambda_1]{e^{-2\pi\ls i\mspace{2.25mu}\vec{n}\ls \cdot\,\vec{y}}\,\frac{L^{\frac{3}{2}}\ls\Psi_{\!\varrho,\ls L}(L\ls\vec{y})}{\norm{\Psi_{\!\varrho,\ls L}}[2]\!\nls};\!d\vec{y}},\qquad\vec{n}\nls \in\Z^3.\label{eq:BECmomentum}
\end{gather}
\end{subequations}
Notice that
$$\sum_{\vec{n}\ls \in\,\Z^3}\abs{\alpha^{\ls 0}_{\varrho,\ls L}(\vec{n})}^2=1,\qquad\forall \varrho,L\nls >\nls 0.$$
We also assume that the tail sum of the Fourier coefficients decays sufficiently fast at initial time when $\varrho$ and $L$ are large.
\begin{assumption}\label{ass:tailCondition}
    Let $\{\alpha^{\ls 0}_{\varrho,\ls L}(\vec{n})\nls \}_{\vec{n}\ls \in\,\Z^3}\!\subset\nls \C$ be the sequence defined in~\eqref{eq:BECmomentum}.
    We require
    $$\lim_{M\to\infty}\limsup_{\varrho\to\infty}\limsup_{L\to\infty}\! \sum_{\substack{\vec{m}\ls \in\,\Z^3 \ls :\\[1.5pt] \abs{\vec{m}}\ls >M}} \abs{\alpha^{\ls 0}_{\varrho,\ls L}(\vec{m})}=0.$$
\end{assumption}
We can clarify the role of Assumption~\ref{ass:tailCondition} by presenting one of its implications when combined with Assumption~\ref{ass:translationalInvariance}: if $\Psi_{\!\varrho,\ls L}$ stands for a quasi-complete Bose--Einstein condensate, Proposition~\ref{th:BECmomentum} ensures the $\ell_2$-convergence to a Kronecker delta (see equation~\eqref{eq:becFourierCoefficients1}); namely, there exist $\vec{k}_0\!\in\nls \Z^3$ and $\vartheta\in[0,2\pi)$ such that
\begin{equation*}
    \lim_{\varrho\to\infty}\limsup_{L\to\infty}\sum_{n\ls \in\,\Z^3}\abs*{\alpha^{\ls 0}_{\varrho,\ls L}(\vec{n})-e^{i\ls \vartheta}\delta_{\vec{n},\ls \vec{k}_0}}^2\!=0.
\end{equation*}
As pointed out in Remark~\ref{rmk:equivalenceOfConditions}, Assumption~\ref{ass:tailCondition} strengthens this convergence.
Specifically,
\begin{subequations}
\begin{equation}\label{eq:WienerControl}
    \lim_{\varrho\to\infty}\limsup_{L\to\infty}\sum_{\vec{n}\ls \in\,\Z^3}\abs*{\alpha^{\ls 0}_{\varrho,\ls L}(\vec{n})-e^{i\ls \vartheta}\delta_{\vec{n},\ls \vec{k}_0}}=0,
\end{equation}
which implies
\begin{equation}\label{eq:assTailConsequence}
    \lim_{\varrho\to\infty}\limsup_{L\to\infty}\sup_{\vec{y}\,\in\ls \Lambda_1}\abs*{\ls \tfrac{L^{3/2}}{\norm{\Psi_{\!\varrho,\mspace{0.75mu} L}}[2]\!\nls}\ls\Psi_{\!\varrho,\ls L}(L\ls \vec{y}) - e^{2\pi\ls i \mspace{2.25mu} \vec{k}_0\ls \cdot\,\vec{y}\,+\,i\ls \vartheta}}=0.
\end{equation}
\end{subequations}
Thus, Assumptions~\ref{ass:translationalInvariance} and \ref{ass:tailCondition} suffice to change the topology of the convergence required in Definition~\ref{def:QCBEC}, item \textit{i)}:
the macroscopic order parameter is uniformly close to a plane wave in the high-density thermodynamic limit.\newline
For later purposes, we define the shortcut
\begin{equation}\label{def:shortcutSumFourierCoefficiets}
    S^{\ls 0}_{\nls \varrho,\ls L}\vcentcolon=\tfrac{\norm{\Psi_{\mspace{-3.75mu}\varrho,\mspace{0.75mu} L}}[2]\nls}{\!\nls\sqrt{\varrho L^3\ls}\ls}\,\norm{\alpha^{\ls 0}_{\varrho,\ls L}}[\ell_1(\Z^3)],
\end{equation}
which fulfils, under Assumptions~\ref{ass:initialBEC} and~\ref{ass:tailCondition} (see Remark~\ref{rmk:S0Boundedness})
\begin{equation}\label{eq:S0boundedSequence}
    \limsup_{\varrho\to\infty}\limsup_{L\to\infty} S^{\ls 0}_{\nls \varrho,\ls L}<\infty.
\end{equation}
Additionally, taking account of Assumption~\ref{ass:translationalInvariance} yields $\lim\limits_{\varrho\to\infty}\limsup\limits_{L\to\infty} S^{\ls 0}_{\nls \varrho,\ls L}=1$.\newline
Next, we state our final assumption, which enforces an appropriate behaviour of the energy.
\begin{assumption}\label{ass:kineticTailCondition}
    Given the sequence $\{\alpha^{\ls 0}_{\varrho,\ls L}(\vec{n})\nls \}_{\vec{n}\ls \in\,\Z^3}\!\subset\nls \C$ introduced in~\eqref{eq:BECmomentum}, we require that there exists $c\nls >\nls 0$ such that
    $$\limsup_{\varrho\to\infty}\limsup_{L\to\infty}\!\sum_{\substack{\vec{m}\ls \in\,\Z^3 \ls :\\[1.5pt] \abs{\vec{m}}>\ls  c\ls  L}} \mspace{-6mu}\tfrac{\abs{\vec{m}}^2\!\nls }{L^2} \,\abs{\alpha^{\ls 0}_{\varrho,\ls L}(\vec{m})}<\infty.$$
\end{assumption}
This condition controls the second derivative of the order parameter (see Proposition~\ref{th:LaplaceNonlinearityControl})
\begin{equation}\label{eq:assKineticTailConsequence}
    \limsup_{\varrho\to\infty}\limsup_{L\to\infty}\,\tfrac{L^{3/2}}{\norm{\Psi_{\!\varrho,\mspace{0.75mu}L}}[2]\!\nls}\, \norm{\Delta \Psi_{\!\varrho,\ls L}}[\infty]<\infty,
\end{equation}
which is not guaranteed \emph{a priori}, as discussed in Remark~\ref{rmk:necessityOfKineticTailCondition}.
In particular, Assumptions~\ref{ass:initialBEC} and~\ref{ass:translationalInvariance} force the kinetic contribution to the energy per particle of the quasi-complete Bose--Einstein condensate to vanish when $\varrho\nls \to\nls \infty$ (as pointed out in Remark~\ref{rmk:consequencesOfCOnvergence}), namely when the condensate becomes complete.
Assumption~\ref{ass:kineticTailCondition} is an alternative condition in the same direction, restricting the maximum rate at which the magnitude of escaping momenta of particles outside the condensate can diverge (see Remarks~\ref{rmk:necessityOfKineticTailCondition},~\ref{rmk:vanishingKineticBEC}).

\medskip

\noindent Technically speaking, Assumptions~\ref{ass:tailCondition} and~\ref{ass:kineticTailCondition} are designed to provide appropriate regularity properties that allow bounds~\eqref{eq:S0boundedSequence} and~\eqref{eq:assKineticTailConsequence} to be propagated at positive times.\newline
In particular, these assumptions guarantee the well-posedness of the effective dynamics driving the time evolution of the quasi-complete Bose--Einstein condensate in a proper Banach space.
More precisely, we prove the global well-posedness of the Hartree equation in the \emph{Wiener--Sobolev space} $\mathfrak{A}^2(\Lambda_L)$ (introduced in~\eqref{def:Wiener}).
\begin{lemma}\label{th:GWP}
    Assume the potential satisfies the decay condition~\eqref{eq:potentialDecay} with $\FT{V}_\infty\!\geq\nls 0$ and\footnote{Calling for such a decay for the Fourier transform of $V_\infty$ implies $V_\infty$ is at least $C^{\mspace{0.75mu}4}(\R^3)$.} $\delta_2\!>\nls 4$.
    Then, given the weighted Wiener algebra $\mathfrak{A}^r(\Lambda_L)$ defined in~\eqref{def:Wiener} for $r\nls \geq\nls 0$, the Hartree equation
    \begin{equation}\label{eq:HartreePDE}
    \begin{dcases}
        i\ls  \partial_t\ls  \Psi_{\!\varrho,\ls L}^{\ls t} \nls = - \Delta \Psi_{\!\varrho,\ls L}^{\ls t} + \tfrac{1}{\varrho} \big(V_L\!\ast\nls  \abs{\Psi_{\!\varrho,\ls L}^{\ls t}}^2\big)\, \Psi_{\!\varrho,\ls L}^{\ls t}\quad\text{in }\,\Lambda_L,\\
        \Psi^{\ls 0}_{\!\varrho,\ls L}\nls =\Psi_{\!\varrho,\ls L}\in \mathfrak{A}^2(\Lambda_L)
    \end{dcases}
    \end{equation}
    admits a unique solution $t\longmapsto\Psi^{\ls t}_{\!\varrho,\ls L}\!\in C^{\mspace{0.75mu}1}\big(\mspace{0.75mu}[0,\infty),\ls \mathfrak{A}^0(\Lambda_L)\mspace{-0.75mu}\big)\mspace{-0.75mu}\cap C^{\ls 0}\big(\mspace{0.75mu}[0,\infty),\ls  \mathfrak{A}^2(\Lambda_L)\mspace{-0.75mu}\big)$ for each fixed $\varrho, L \nls >\nls 0$.
\end{lemma}
Hereafter, $\Psi^{\ls t}_{\!\varrho,\ls L}\!\in\nls \mathfrak{A}^2(\Lambda_L)$ shall denote the order parameter evolved through the Hartree flow.
\begin{note}
    We recall two of the conserved quantities of the Hartree equation on the torus $\Lambda_L\,$: for all $t\nls \geq\nls 0$
    \begin{align*}
        &\text{mass} \qquad \norm{\Psi^{\ls t}_{\!\varrho,\ls L}}[2]^2=\norm{\Psi_{\!\varrho,\ls L}}[2]^2,\\
        &\text{energy} \qquad \mathscr{E}_{\varrho,\ls L}[\Psi^{\ls t}_{\!\varrho,\ls L}]=\mathscr{E}_{\varrho,\ls L}[\Psi_{\!\varrho,\ls L}].
    \end{align*}
    Here, $\mathscr{E}_{\varrho,\ls L}$ is the Hartree functional~\eqref{def:HartreeFunctional}.
    Notice that the Hartree evolution meets the conditions pointed out in Remarks~\ref{rmk:massMustBeConserved}, \ref{rmk:HartreeEnergyMustBeConserved}.\newline
    Furthermore, as proven by Proposition~\ref{th:quantifyTotalEnergy}, Assumptions~\ref{ass:initialBEC},~\ref{ass:translationalInvariance}, and~\ref{ass:tailCondition} allow us to quantify the total energy of the quasi-complete Bose--Einstein condensate, \ie $$\lim_{\varrho\to\infty}\limsup_{L\to\infty}\abs*{\frac{\mathscr{E}_{\varrho,\ls L}[\Psi^{\ls t}_{\!\varrho,\ls L}]}{\varrho L^3}-\frac{1}{2}\,\FT{V}_\infty(\vec{0})}=0,\qquad\forall t\geq 0.$$
    Because of Assumption~\ref{ass:energyConsistence}, we stress that the Hartree energy of the order parameter is close to the expected energy of the quasi-coherent state $\varphi^{\ls t}_{\varrho,\ls L}$. 
\end{note}
Before proceeding, we observe that the quantity $\var_{\,\mathcal{U}_{\varrho,\ls L}(t)\ls \xi_{\varrho,\ls L}}[\ls \mathcal{N}\nls +\phi(\Psi^{\ls t}_{\!\varrho,\ls L})]$ is time-independent, where $\mathcal{U}_{\varrho,\ls L}(t)$ has been introduced in~\eqref{def:fluctuationDynamics1}.
Indeed, recalling that the following holds for all $f\nls \in\nls  \Lp{2}\big(\Lambda_L\nls \big)$
\begin{subequations}\label{eqs:WeylConjugation}
    \begin{gather}\label{eq:WeylaConjugation}
        \adj{\mathcal{W}\ls }\nls (f)\,a(g)\ls \mathcal{W}(f)\psi=a(g)\psi+\scalar{g}{f}[2]\ls \psi,\qquad\forall g\in\Lp{2}(\Lambda_L),\,\psi\in\fdom{\mathcal{N}},\\
        \adj{\mathcal{W}\ls }\nls (f)\ls \mathcal{N}\ls \mathcal{W}(f)\psi=\mathcal{N}\psi+\phi(f)\psi+\norm{f}[2]^2\,\psi,\qquad\psi\in\dom{\mathcal{N}},\label{eq:WeylNConjugation}
    \end{gather}
\end{subequations}
one has
$${\textstyle \var_{\,\mathcal{U}_{\varrho,\ls L}(t)\ls \xi_{\varrho,\ls L}}}[\ls \mathcal{N}\nls +\phi(\Psi^{\ls t}_{\!\varrho,\ls L})]={\textstyle \var_{\varphi^{\ls t}_{\varrho,\ls L}}}\big[\ls \mathcal{N}\nls -\norm{\Psi_{\!\varrho,\ls L}}[2]^2\, \Id\big],$$
with $\varphi^{\ls t}_{\varrho,\ls L}$ defined in~\eqref{def:evolutionQ-CCS}.
Then, since both $\mathcal{N}$ and $\norm{\Psi_{\!\varrho,\ls L}}[2]^2\, \Id$ commute with the Hamiltonian, one can exploit~\eqref{eq:WeylNConjugation} again, so that
\begin{equation}
    {\textstyle \var_{\,\mathcal{U}_{\varrho,\ls L}(t)\ls \xi_{\varrho,\ls L}}}[\ls \mathcal{N}\nls +\phi(\Psi^{\ls t}_{\!\varrho,\ls L})]={\textstyle \var_{\,\xi_{\varrho,\ls L}}}[\ls \mathcal{N}\nls +\phi(\Psi_{\!\varrho,\ls L})].
\end{equation}
Furthermore, since $\xi_{\varrho,\ls L}\!\in\nls \dom{\mathcal{N}}$ is a quasi-vacuum state with respect to $\Psi_{\!\varrho,\ls L}\ls $
$$\lim_{\varrho\to\infty}\limsup_{L\to\infty}\ls  \abs*{\frac{\var_{\,\mathcal{U}_{\varrho,\ls L}(t)\ls \xi_{\varrho,\ls L}}[\ls \mathcal{N}\nls +\phi(\Psi^{\ls t}_{\!\varrho,\ls L})]}{\norm{\Psi_{\!\varrho,\ls L}}[2]^2}-1\ls }=0.$$
Therefore, proving that $\mathcal{U}_{\varrho,\ls L}(t)\ls \xi_{\varrho,\ls L}$ is a quasi-vacuum state with respect to $\Psi^{\ls t}_{\!\varrho,\ls L}\ls $ reduces to showing that the expected number of particles in the state $\mathcal{U}_{\varrho,\ls L}(t)\ls \xi_{\varrho,\ls L}$ is much smaller than $\varrho L^3$.\newline
The kernel of $\gamma^{(1)}_{\varphi^{\ls t}_{\varrho,\ls L}}$ can be expressed in terms of $\mathcal{U}_{\varrho,\ls L}(t)\ls \xi_{\varrho,\ls L}\ls $, making use of~\eqref{eq:WeylaConjugation}
\begin{equation}\label{eq:gamma1MinusLimit}
    \begin{split}
        \gamma^{(1)}_{\varphi^{\ls t}_{\varrho,\ls L}}\!(\vec{x},\vec{y})-\frac{\Psi^{\ls t}_{\!\varrho,\ls L}(\vec{x}) \ls \conjugate*{\Psi^{\ls t}_{\!\varrho,\ls L}(\vec{y})}}{\norm{\mathcal{N}^{\frac{1}{2}}\mathcal{W}(\Psi_{\!\varrho,\ls L})\ls \xi_{\varrho,\ls L}}^2}=&\:\frac{\scalar{\ls a_{\vec{y}}\,\mathcal{U}_{\varrho,\ls L}(t)\ls \xi_{\varrho,\ls L}}{a_{\vec{x}}\,\mathcal{U}_{\varrho,\ls L}(t)\ls \xi_{\varrho,\ls L}}}{\norm{\mathcal{N}^{\frac{1}{2}}\mathcal{W}(\Psi_{\!\varrho,\ls L})\ls \xi_{\varrho,\ls L}}^2}\ls +\\
        &+\Psi^{\ls t}_{\!\varrho,\ls L}(\vec{x})\,\frac{\scalar{a_{\vec{y}}\,\mathcal{U}_{\varrho,\ls L}(t)\ls \xi_{\varrho,\ls L}}{\mathcal{U}_{\varrho,\ls L}(t)\ls \xi_{\varrho,\ls L}}}{\norm{\mathcal{N}^{\frac{1}{2}}\mathcal{W}(\Psi_{\!\varrho,\ls L})\ls \xi_{\varrho,\ls L}}^2}\ls +\\
        &+\conjugate*{\Psi^{\ls t}_{\!\varrho,\ls L}(\vec{y})}\,\frac{\scalar{\ls \mathcal{U}_{\varrho,\ls L}(t)\ls \xi_{\varrho,\ls L}}{a_{\vec{x}}\,\mathcal{U}_{\varrho,\ls L}(t)\ls \xi_{\varrho,\ls L}}}{\norm{\mathcal{N}^{\frac{1}{2}}\mathcal{W}(\Psi_{\!\varrho,\ls L})\ls \xi_{\varrho,\ls L}}^2}.
    \end{split}
\end{equation}
Taking the trace of both sides of the equation yields
\begin{equation*}
    \norm{\mathcal{N}^{\frac{1}{2}}\mathcal{W}(\Psi_{\!\varrho,\ls L})\ls \xi_{\varrho,\ls L}}^2-\norm{\Psi_{\!\varrho,\ls L}}[2]^2=\scalar{\ls \mathcal{U}_{\varrho,\ls L}(t)\ls \xi_{\varrho,\ls L}}{\big(\mathcal{N}\nls +\phi(\Psi^{\ls t}_{\!\varrho,\ls L})\big)\,\mathcal{U}_{\varrho,\ls L}(t)\ls \xi_{\varrho,\ls L}}.
\end{equation*}
Hence, since $\mathcal{W}(\Psi_{\!\varrho,\ls L})\ls \xi_{\varrho,\ls L}$ is quasi-coherent and $\Psi_{\!\varrho,\ls L}$ satisfies~\eqref{eq:becNormalization}
\begin{equation}\label{eq:vanishingExpectationNplusSegal}
    \lim_{\varrho\to\infty}\limsup_{L\to\infty}\ls \frac{1}{\varrho L^3}\nls \abs*{\mathbb{E}_{\,\mathcal{U}_{\varrho,\ls L}(t)\ls \xi_{\varrho,\ls L}}[\ls \mathcal{N}\nls +\phi(\Psi^{\ls t}_{\!\varrho,\ls L})]}=0.
\end{equation}
However, this alone does not establish that $\mathcal{U}_{\varrho,\ls L}(t)\ls \xi_{\varrho,\ls L}$ is a quasi-vacuum state with respect to $\Psi^{\ls t}_{\!\varrho,\ls L}\ls $, since we need the same statement for the expectation value to hold for the sole observable $\mathcal{N}$.
This is the content of the following lemma.
\begin{lemma}\label{th:fluctuationsNumber}
    Let $\xi_{\varrho,\ls L}\!\in\nls \fdom{\mathcal{H}_{\varrho,\ls L}}$ be a quasi-vacuum state with respect to $\Psi_{\!\varrho,\ls L}\!\in\nls \mathfrak{A}^2(\Lambda_L)$ such that Assumptions~\ref{ass:initialBEC},~\ref{ass:energyConsistence},~\ref{ass:tailCondition}, and~\ref{ass:kineticTailCondition} are fulfilled.
    Moreover, we introduce the shortcuts
    \begin{gather*}
        \mathfrak{n}_{\varrho,\ls L}\vcentcolon=\frac{\mathbb{E}_{\,\xi_{\varrho,\ls L}}[\ls \mathcal{N}\ls ]}{\norm{\Psi_{\!\varrho,\ls L}}[2]^2\!\nls}\,,\\
        \mathfrak{e}_{\varrho,\ls L}\vcentcolon=\frac{1}{\varrho L^3}\abs*{\mathcal{H}_{\varrho,\ls L}\big[\mathcal{W}(\Psi_{\!\varrho,\ls L})\ls \xi_{\varrho,\ls L}\big]\nls -\mathscr{E}_{\varrho,\ls L}[\Psi_{\!\varrho,\ls L}]},
    \end{gather*}
    both vanishing in the iterated limit, where $\mathscr{E}_{\varrho,\ls L}$ is the Hartree functional~\eqref{def:HartreeFunctional}.\newline
    Then, provided $S^{\ls 0}_{\nls \varrho,\ls L}$ defined by~\eqref{def:shortcutSumFourierCoefficiets} and a given $0\nls <\nls T_{\nls\varrho,\ls L}\nls <\nls \big(2\ls \FT{V}_\infty(\vec{0})\ls (S^{\ls 0}_{\nls \varrho,\ls L})^2\big)^{-1}$, for any fixed $\varrho, L > 0$ there exist $c_{\varrho,\ls L},\,\omega_{\varrho,\ls L} \nls >\nls 0$ such that
    \begin{gather}
        \limsup_{\varrho\to\infty}\limsup_{L\to\infty} \,c_{\varrho,\ls L}<\infty,\qquad \limsup_{\varrho\to\infty}\limsup_{L\to\infty} \,\omega_{\varrho,\ls L}<\infty,\nonumber\\
        \frac{1}{\varrho L^3}\,\mathbb{E}_{\,\mathcal{U}_{\varrho,\ls L}(t)\ls \xi_{\varrho,\ls L}}[\ls \mathcal{N}\ls ]\leq c_{\varrho,\ls L} \,e^{\omega_{\varrho,\ls L}\, t}\!\left(\!\sqrt{\mathfrak{n}_{\varrho,\ls L}\ls }+\mathfrak{n}_{\varrho,\ls L}+\mathfrak{e}_{\varrho,\ls L}+\frac{1}{\varrho}\right)\!,\qquad \forall t\in[0,T_{\nls\varrho,\ls L}],\label{eq:excitationsControl}
    \end{gather}
    provided $\mathcal{U}_{\varrho,\ls L}(t)$ defined by~\eqref{def:fluctuationDynamics1}, with $t\longmapsto\Psi_{\!\varrho,\ls L}^{\ls t}\!\in\nls  C^{\mspace{0.75mu}1}\big(\mspace{0.75mu}[0,\infty),\mathfrak{A}^0(\Lambda_L)\mspace{-0.75mu}\big)\mspace{-0.75mu}\cap C^{\ls 0}\big(\mspace{0.75mu}[0,\infty), \mathfrak{A}^2(\Lambda_L)\mspace{-0.75mu}\big)$ solving the Hartree equation~\eqref{eq:HartreePDE}.
\end{lemma}
On the basis of the result of Lemma~\ref{th:fluctuationsNumber}, one can show the convergence of the one-particle reduced density matrix, at least for a finite time interval (see Remark~\ref{rmk:finiteTimeIntervalValidity} for additional comments on this issue).
\begin{theo}\label{th:gamma1Convergence}
    Under the hypotheses of Lemma~\ref{th:fluctuationsNumber}, let $\varphi^{\ls t}_{\varrho,\ls L}\!\in\nls \fdom{\mathcal{H}_{\varrho,\ls L}}$ be the quasi-coherent state defined by~\eqref{def:evolutionQ-CCS}.
    Then, given $\gamma^{(1)}_{\varphi^{\ls t}_{\varrho,\ls L}}\!\!\in\nls \schatten{1}\big(\Lp{2}(\Lambda_L)\nls \big)$ the integral operator whose kernel is provided by~\eqref{def:gamma1Kernel}, for every $0\nls <\nls T\nls <\nls \big(2\ls \FT{V}_\infty(\vec{0}) \limsup\limits_{\varrho\to\infty}\limsup\limits_{L\to\infty}\, (S^{\ls 0}_{\nls\varrho,\ls L})^2\,\big)^{-1}$ 
    \begin{equation}\label{wts:HS-Convergence}
        \lim_{\varrho\to\infty}\limsup_{L\to\infty} \norm*{\gamma^{(1)}_{\varphi^{\ls t}_{\varrho,\ls L}}\!\!-\:\tfrac{|\Psi^{\ls t}_{\!\varrho,\ls L}\rangle\langle\Psi^{\ls t}_{\!\varrho,\ls L}|}{\vphantom{|^|}\norm{\Psi^{\ls t}_{\!\varrho,\mspace{0.75mu} L}}[2]^2}\,}[\mathrm{HS}] \!=\ls  0,\qquad\forall t\in[0,T].
    \end{equation}
    Moreover, by taking account of Assumption~\ref{ass:translationalInvariance}, the convergence is proved to occur over longer time intervals, as
    $$\lim\limits_{\varrho\to\infty}\limsup\limits_{L\to\infty}\, S^{\ls 0}_{\nls\varrho,\ls L}=1.$$
\end{theo}
\begin{note}\label{rmk:rateOfConvergence}
    The upper bound~\eqref{eq:excitationsControl} actually permits evaluating the rate of convergence of the one-particle reduced density matrix.
    Specifically, given $c_{\varrho,\ls L}, \omega_{\varrho,\ls L}$ and $T$ as in Lemma~\ref{th:fluctuationsNumber}, we set
    $$c=\limsup_{\varrho\to\infty}\limsup_{L\to\infty} \,c_{\varrho,\ls L},\qquad\omega=\limsup_{\varrho\to\infty}\limsup_{L\to\infty} \,\omega_{\varrho,\ls L}.$$
    Then, we have that the Hilbert--Schmidt distance $\norm*{\gamma^{(1)}_{\varphi^{\ls t}_{\varrho,\ls L}}\!\!-\:\tfrac{|\Psi^{\ls t}_{\!\varrho,\ls L}\rangle\langle\Psi^{\ls t}_{\!\varrho,\ls L}|}{\vphantom{|^|}\norm{\Psi^{\ls t}_{\!\varrho,\mspace{0.75mu} L}}[2]^2}\,}[\mathrm{HS}]\!$ is controlled from above by a quantity that in the iterated limit vanishes as
    $$2\sqrt{\mathfrak{n}_{\varrho,\ls L}}+\sqrt{6\ls c\ls }\ls e^{\omega\,t/2} \sqrt{\sqrt{\mathfrak{n}_{\varrho,\ls L}}+\mathfrak{e}_{\varrho,L}+\tfrac{1}{\varrho}\ls }\ls ,\qquad \forall t\in[0,T].$$
\end{note}
\begin{note}
    As pointed out in~\cite[Remark 1.4]{RoSc09} or~\cite[Footnote 3, p. 8]{BePoSc16}, one has that the trace norm is controlled from above by twice the Hilbert--Schmidt norm, since $\tr\gamma^{(1)}_{\varphi^{\ls t}_{\varrho,\ls L}}\!=1$ and $\frac{\tr|\Psi^{\ls t}_{\!\varrho,\ls L}\rangle\langle\Psi^{\ls t}_{\!\varrho,\ls L}|}{\vphantom{|^|}\norm{\Psi^{\ls t}_{\!\varrho,\mspace{0.75mu} L}}[2]^2}=1$ as well.
    As a consequence,~\eqref{wts:HS-Convergence} implies
    \begin{equation}\label{wts:TN-Convergence}
        \lim_{\varrho\to\infty}\limsup_{L\to\infty}\norm*{\gamma^{(1)}_{\varphi^{\ls t}_{\varrho,\ls L}}\!\!-\:\tfrac{|\Psi^{\ls t}_{\!\varrho,\ls L}\rangle\langle\Psi^{\ls t}_{\!\varrho,\ls L}|}{\vphantom{|^|}\norm{\Psi^{\ls t}_{\!\varrho,\mspace{0.75mu} L}}[2]^2}\,}[\mathrm{Tr}] \!=\ls 0.
    \end{equation}
    Additionally, for any \emph{intensive} one-particle observable $J_{\varrho,\ls L}\!\in\nls \bounded{\Lp{2}[\Lambda_L]}$ (meaning that $\norm{J_{\varrho,\ls L}}[\linear{\Lp{2}[\Lambda_L]}]$ does not depend on $L$) such that its operator norm does not grow too fast in $\varrho$, one has that its expectation can be approximately computed by replacing the one-particle reduced density matrix with the projection onto the (normalised) quasi-complete Bose--Einstein condensate, since
    $$\tr\abs*{\ls J_{\varrho,\ls L}\ls\Big(\nls\gamma^{(1)}_{\varphi^{\ls t}_{\varrho,\ls L}}\!\!-\:\tfrac{|\Psi^{\ls t}_{\!\varrho,\ls L}\rangle\langle\Psi^{\ls t}_{\!\varrho,\ls L}|}{\vphantom{|^|}\norm{\Psi^{\ls t}_{\!\varrho,\mspace{0.75mu} L}}[2]^2}\Big)\nls} \!\leq\ls  2\ls \norm{J_{\varrho,\ls L}}[\linear{\Lp{2}[\Lambda_L]}]\norm*{\gamma^{(1)}_{\varphi^{\ls t}_{\varrho,\ls L}}\!\!-\:\tfrac{|\Psi^{\ls t}_{\!\varrho,\ls L}\rangle\langle\Psi^{\ls t}_{\!\varrho,\ls L}|}{\vphantom{|^|}\norm{\Psi^{\ls t}_{\!\varrho,\mspace{0.75mu} L}}[2]^2}\,}[\mathrm{HS}]\!.$$
\end{note}
\begin{note}
    The result of Theorem~\ref{th:gamma1Convergence} can be rephrased by replacing $\frac{\Psi^{\ls t}_{\!\varrho,\ls L}}{\norm{\Psi^{\ls t}_{\!\varrho,\ls L}}[2]\!\nls}$ with its macroscopic counterpart $\tfrac{1}{L^{3/2}} \Psi^{\ls t}\big(\tfrac{\vec{\cdot}}{L}\big)$, with $\Psi^{\ls t} \,:\, \vec{y}\:\longmapsto\, e^{-i\ls\FT{V}_\infty(\vec{0}) \ls\abs{\Psi(\vec{y})}^2\ls t}\,\Psi(\vec{y})\in\mathfrak{A}^2(\Lambda_1)$, since
    $$\norm*{\ls \frac{|\Psi^{\ls t}_{\!\varrho,\ls L}\rangle\langle\Psi^{\ls t}_{\!\varrho,\ls L}|}{\norm{\Psi^{\ls t}_{\!\varrho,\ls L}}[2]^2\vphantom{|^|}}-\frac{|\Psi^{\ls t}(\tfrac{\vec{\cdot}}{L})\rangle\langle\Psi^{\ls t}(\tfrac{\vec{\cdot}}{L})|}{L^3}\ls }[\mathrm{HS}]^2\!\!\leq 3\norm*{\tfrac{\Psi^{\ls t}_{\!\varrho,\ls L}}{\norm{\Psi^{\ls t}_{\!\varrho,\ls L}}[2]\!\nls} -\tfrac{1}{L^{3/2}}\Psi^{\ls t}\big(\tfrac{\vec{\cdot}}{L}\big)}[2]^4+6\norm*{\tfrac{\Psi^{\ls t}_{\!\varrho,\ls L}}{\norm{\Psi^{\ls t}_{\!\varrho,\ls L}}[2]\!\nls} -\tfrac{1}{L^{3/2}}\Psi^{\ls t}\big(\tfrac{\vec{\cdot}}{L}\big)}[2]^2,$$
    which is small eventually (\cfr~Proposition~\ref{th:propagationExistenceBECgeneral}).
    In this case, the rate of convergence changes, but one could quantitatively relate this additional error in terms of the convergence of the initial quasi-complete Bose--Einstein condensate to its macroscopic counterpart.\newline
    Notice that in the case that condensation occurs in a momentum state (\ie, when Assumption~\ref{ass:translationalInvariance} is taken into account), the normalised one-particle reduced density matrix is close to
    $\frac{|\Psi^{\ls t}(\frac{\vec{\cdot}}{L})\rangle\langle\Psi^{\ls t}(\frac{\vec{\cdot}}{L})|}{L^3}$, which is time independent, as the profile of the macroscopic counterpart of the order parameter is a plane wave (\cfr~Proposition~\ref{th:propagationExistenceBEC}): the two conjugated time-dependent phase factors disappear in the integral kernel of $\frac{|\Psi^{\ls t}(\frac{\vec{\cdot}}{L})\rangle\langle\Psi^{\ls t}(\frac{\vec{\cdot}}{L})|}{L^3}$.
\end{note}
\begin{note}\label{rmk:notProvingTD}
    We cannot expect to prove the convergence of the one-particle reduced density matrix as $L$ grows large, that is, regardless of the value of $\varrho$, since this would imply solving the full thermodynamic problem.
\end{note}
\begin{cor}\label{th:propertiesPreservation}
    As a consequence of Lemma~\ref{th:fluctuationsNumber} and Proposition~\ref{th:propagationExistenceBECgeneral}, one has that $\mathcal{U}_{\varrho,\ls L}(t)\ls \xi_{\varrho,\ls L}$ is a quasi-vacuum state with respect to $\Psi^{\ls t}_{\!\varrho,\ls L}\ls $, and $\varphi^{\ls t}_{\varrho,\ls L}\!=e^{-i\,\mathcal{H}_{\varrho,\ls L}\ls t}\,\mathcal{W}(\Psi_{\!\varrho,\ls L})\ls \xi_{\varrho,\ls L}$ is quasi-canonical coherent and exhibits the quasi-complete Bose--Einstein condensate $\Psi^{\ls t}_{\!\varrho,\ls L}\ls $ for all $t\in[0,T]$ with $T<\big(2\ls\FT{V}_\infty(\vec{0})\ls\norm{\FT{\Psi}}[\ell_1(\Z^3)]\big)^{-1}$, where $\Psi\in \mathfrak{A}^2(\Lambda_1)$ is the macroscopic counterpart of the order parameter satisfying Assumption~\ref{ass:initialBEC}.
\end{cor}

\subsubsection*{Conclusion}
The results established in this section ultimately demonstrate that a Bose--Einstein condensate -- whose completeness increases with the density of the system -- is preserved by the many-body Schr{\"o}dinger evolution, at least over finite time intervals.
In particular, we have shown that the rank-one projection onto the time-dependent order parameter and the reduced one-particle density matrix behave the same way asymptotically.
This entails that the obtained approximation eventually becomes progressively more accurate in the thermodynamic limit as the density increases, without the need to invoke the unphysical limit of infinite density; in fact, the rate of convergence provides a quantitative measure of the validity of the approximation for any fixed -- yet sufficiently large -- finite density.

\medskip

\noindent The effective dynamics derived here should be understood as a leading-order macroscopic description.
Indeed, it consists in a closed nonlinear one-body equation for the order parameter $\Psi_{\!\varrho,\ls L}$, whose mass is asymptotically consistent with the condensate occupation extracted from the one-particle reduced density matrix; moreover, it ensures that the many-body Schr\"odinger evolution preserves the quasi-coherent structure in which this mode remains dominant.
On the other hand, subleading corrections are not resolved: we do not introduce an explicit decomposition into condensate and fluctuation sectors, nor do we investigate the time dependence of the number of particles outside the condensate.
Although the derived nonlinear Hartree equation governs the evolution of the order parameter associated with the \emph{incompletely occupied} condensate at finite density, our analysis only captures it up to errors much smaller than $\varrho L^3$.
This is precisely the scale relevant for the condensate, but it is too coarse to resolve the dynamics of quantities whose natural size is subleading with respect to $\varrho L^3$, such as fluctuations.
As a consequence, the evolution of the condensate fraction cannot be reliably inferred at this level of precision\footnote{Notice that the Hartree equation preserves the mass $\norm{\Psi^{\ls t}_{\!\varrho,\ls L}}[2]^2$, but the number of particles occupying the condensate is not expected to be conserved.}.
Obtaining a faithful description of the evolution of the Bose--Einstein condensate occupation that becomes accurate at sufficiently large density requires quantitative control of the fluctuations on a much finer scale and hence a more refined (Bogoliubov-type) analysis.\newline
In this sense, the result can be interpreted as the identification of a family of initial data for which the condensate is dynamically stable: not only do we establish the propagation of complete Bose--Einstein condensates, but the present setting also accommodates perturbations to them, allowing for non-zero kinetic energy of the order parameter and residual correlations of the many-body quantum state already at the initial time.
Rather than being excluded \emph{ab initio}, such perturbations are controlled in a coarse-grained sense throughout the evolution, in the sense that they are proved to remain subleading with respect to $\varrho L^3$.\newline
More precisely, we have considered the subset of quasi-canonical coherent states given by
\begin{equation*}
    \begin{split}
        D=\Big\{\phi_{\varrho,\ls L}\nls\in\fdom{\mathcal{H}_{\varrho,\ls L}} \,\Big\vert\; &\exists \Omega_{\varrho,\ls L}\nls\in\fdom{\mathcal{H}_{\varrho,\ls L}},\,\Phi_{\varrho,\ls L}\nls\in\mathfrak{A}^2(\Lambda_L)\, :\\
        &\Omega_{\varrho,\ls L}\:\text{is a quasi-vacuum state with respect to}\:\Phi_{\varrho,\ls L},\;\phi_{\varrho,\ls L}\nls =\mathcal{W}(\Phi_{\!\varrho,\ls L})\ls\Omega_{\varrho,\ls L}\Big\}.
    \end{split}
\end{equation*}
The elements of this class of many-body quantum states exhibit quasi-complete Bose--Einstein condensation and therefore represent a specific family of perturbations of the usual canonical coherent states exhibiting complete condensation.
Roughly speaking (leaving aside the technical assumptions), our contribution is the proof  that $D$ is preserved by the many-body Schr\"odinger dynamics, \ie~$e^{-i\ls\mathcal{H}_{\varrho,\ls L}\ls t}D\subseteq D$ (\cfr~Corollary~\ref{th:propertiesPreservation}).

\subsection*{Notation Adopted}
For reader's convenience, we collect here some of the notation used in the paper.
\begin{itemize}
    \item Given the $n$-dimensional Euclidean space $(\R^n,\ls \cdot\ls )$, $\vec{x}$ denotes a vector in $\R^n$ and $\abs{\vec{x}}$ its magnitude.
    \item For any $p\in \nls [1,\infty]$ and Borel set $\Omega\subset\R^n$, $\Lp{p}(\Omega)$ is the Banach space of $p\ls $-integrable functions with respect to the Lebesgue measure.
    We write $\norm{\ls \cdot\ls }[p]\nls \vcentcolon= \nls \norm{\ls \cdot\ls }[\Lp{\ls p}[\Lambda_L]]$.\newline
    If $\Omega$ is countable, $\ell_p(\Omega)$ is the Banach space of $p$-summable sequences.
    \item For a Borel set $\Omega\subset\R^n$, $W^{r,\ls p}(\Omega)$ is the Bessel potential space of functions with $r\nls >\nls 0$ fractional derivatives (defined via Fourier transform) in $\Lp{p}[\Omega]$.\newline
    Moreover, $H^r(\Omega)\vcentcolon= W^{r,\ls 2}(\Omega)$ denotes the Hilbert--Sobolev space of order $r$.
    \item If $\hilbert*$ is a complex Hilbert space, $\scalar{\cdot\ls }{\!\cdot}_{\hilbert*}\ls $ and  $\norm{\ls \cdot\ls }[\hilbert*]\!\vcentcolon= \!\sqrt{\scalar{\cdot\ls }{\!\cdot}_{\hilbert*}}$ denote its inner product and the induced norm, respectively.\newline
    We simply write $\scalar{\cdot\ls }{\!\cdot}$ and $\norm{\ls \cdot\ls }$ when $\hilbert*=\fockS\big(\Lp{2}[\Lambda_L]\big)$.
    \item Given two Hilbert spaces $\hilbert*_1$, $\hilbert*_2$, $\linear{\hilbert*_1,\hilbert*_2}$ and $\bounded{\hilbert*_1,\hilbert*_2}$ denote, respectively, the set of linear operators and the Banach space of bounded operators from $\hilbert*_1$ to $\hilbert*_2$.\newline
    We also set $\linear{\hilbert*_1}\vcentcolon=\linear{\hilbert*_1,\hilbert*_1}$ and $\bounded{\hilbert*_1}\vcentcolon=\bounded{\hilbert*_1,\hilbert*_1}$.
    \item Given two Hilbert spaces $\hilbert*_1$, $\hilbert*_2$ and an operator $A,\dom{A}\!\in\!\linear{\hilbert*_1,\hilbert*_2}$, the linear subspace $\dom{A}\nls \subseteq\nls \hilbert*_1$ stands for its domain, and $\fdom{A}\supseteq \dom{A}$ its form domain.
    \item For a Hilbert space $\hilbert*$ and a compact operator $K\!\nls \in\!\bounded{\hilbert*}$, we write $K\!\nls \in\nls \schatten{p}(\hilbert*)$ if it has finite $p$-Schatten norm $\norm{K}[\schatten{p}[\hilbert*]]\!=\big(\tr{\abs{K}^p}\big)^{1/p}\!<\infty$, with $p\nls \in\nls [1,\infty)$, and $\norm{K}[\schatten{\infty}[\hilbert*]]=\norm{K}[\linear{\hilbert*}]$.\newline
    In particular, $\norm{K}[\mathrm{HS}]\vcentcolon=\norm{K}[\schatten{2}[\Lp{2}[\Lambda_L]]]$ and $\norm{K}[\mathrm{Tr}]\vcentcolon=\norm{K}[\schatten{1}[\Lp{2}[\Lambda_L]]]$ denote the Hilbert--Schmidt norm and the trace norm, respectively.
    \item Given a complex Hilbert space $\hilbert*$ and the self-adjoint operator $A,\dom{A}\in\linear{\hilbert*}$, $\mathbb{E}_\psi[A\ls ]\!\vcentcolon=\!\scalar{\psi}{A\psi}[\hilbert*]$ denotes the expectation value of the observable associated with $A,\dom{A}$ in a quantum state $\psi\!\in\nls \dom{A}$, and $\var_\psi[A\ls ]\! \vcentcolon= \nls \norm{A\psi}[\hilbert*]^2\!-\nls (\mathbb{E}_\psi[A\ls ])^2$ stands for its variance.
\end{itemize}


\section{Generator of the Fluctuation Dynamics}\label{sec:Generator}

The expectation of the number of fluctuations plays a crucial role in our analysis.
This section is therefore dedicated to the detailed study of the fluctuation dynamics.\newline
Motivated by~\eqref{def:fluctuationDynamics1}, we define the two-parameter unitary evolution
\begin{equation}\label{def:fluctuationDynamics2}
    \mathcal{U}_{\varrho,\ls L}(t,s)\vcentcolon=\adj{\mathcal{W}\ls }\nls (\Psi^{\ls t}_{\!\varrho,\ls L})\ls e^{-i\,\mathcal{H}_{\varrho,\ls L}(t\ls -\ls s)}\ls \mathcal{W}(\Psi^{\ls s}_{\!\varrho,\ls L}),\qquad t,s\geq 0.
\end{equation}
For any $t_0,t,s\nls \geq\nls  0$, these operators satisfy
\begin{gather*}
    \mathcal{U}_{\varrho,\ls L}(s,s)=\Id,\\
    \adj{\mathcal{U}_{\varrho,\ls L}(t,s)}\!=\mathcal{U}_{\varrho,\ls L}(s,t),\\
    \mathcal{U}_{\varrho,\ls L}(t,t_0)=\mathcal{U}_{\varrho,\ls L}(t,s)\,\mathcal{U}_{\varrho,\ls L}(s,t_0).
\end{gather*}
Clearly, one has $\mathcal{U}_{\varrho,\ls L}(t)\nls =\mathcal{U}_{\varrho,\ls L}(t,0)$.
Moreover, $\mathcal{U}_{\varrho,\ls L}(\ls \cdot\ls ,\cdot\ls )$ is strongly continuous in both variables, permitting the definition of the \emph{generator} of the fluctuation dynamics
\begin{align}
    i\ls \partial_t\,\mathcal{U}_{\varrho,\ls L}(t,s)&=\left[i\ls \partial_t\adj{\mathcal{W}\ls }\nls (\Psi^{\ls t}_{\!\varrho,\ls L})\right]\nls \mathcal{W}(\Psi^{\ls t}_{\!\varrho,\ls L})\,\mathcal{U}_{\varrho,\ls L}(t,s)+\adj{\mathcal{W}\ls }\nls (\Psi^{\ls t}_{\!\varrho,\ls L})\ls \mathcal{H}_{\varrho,\ls L} \mathcal{W}(\Psi^{\ls t}_{\!\varrho,\ls L})\,\mathcal{U}_{\varrho,\ls L}(t,s)\nonumber\\
    &=\vcentcolon \mathcal{L}_{\varrho,\ls L}(t)\,\mathcal{U}_{\varrho,\ls L}(t,s),\label{def:generator}
\end{align}
where this makes sense on an appropriate domain of the Fock space (corresponding to $\dom{\mathcal{H}_{\varrho,\ls L}}$), with the time derivative taken in the strong-operator topology.
The derivation of the expression for such a generator is well known in the literature (see \eg~\cite[Chapter 3, p. 19]{BePoSc16}); however, we outline the key steps for completeness.

\smallskip

\noindent The first term can be computed by exploiting the infinitesimal form of the \emph{Baker--Campbell--Hausdorff formula}, that is, for our case
\begin{equation*}
    i\ls \partial_t \,e^{i\ls \phi(i\ls \Psi^{\ls t}_{\!\varrho,\ls L})}= -e^{i\ls \phi(i\ls \Psi^{\ls t}_{\!\varrho,\ls L})}\ls \phi(i\ls \partial_t\ls \Psi^{\ls t}_{\!\varrho,\ls L}) + \frac{e^{i\ls \phi(i\ls \Psi^{\ls t}_{\!\varrho,\ls L})}\mspace{-9mu}}{2}\;\big[a(\Psi^{\ls t}_{\!\varrho,\ls L})\nls -\adj{a}(\Psi^{\ls t}_{\!\varrho,\ls L}), \,\phi(i\ls \partial_t\Psi^{\ls t}_{\!\varrho,\ls L})\big].
\end{equation*}
Then, combining the above equation with~\eqref{eq:CCR} and~\eqref{eq:WeylaConjugation} yields
\begin{equation}
    \left[i\ls \partial_t\adj{\mathcal{W}\ls }\nls (\Psi^{\ls t}_{\!\varrho,\ls L})\right]\nls \mathcal{W}(\Psi^{\ls t}_{\!\varrho,\ls L})=-\phi(i\ls \partial_t\Psi^{\ls t}_{\!\varrho,\ls L})-\Re\ls \scalar{\Psi^{\ls t}_{\!\varrho,\ls L}}{i\ls \partial_t\Psi^{\ls t}_{\!\varrho,\ls L}}[2]\ls \Id.
\end{equation}
The second term is obtained by considering~\eqref{eq:formHamiltonianFock} and applying the distributional version of equation~\eqref{eq:WeylaConjugation}, iteratively, \ie
\begin{equation}
    \adj{\mathcal{W}\ls }\nls (f)\,a_{\vec{x}}\ls \mathcal{W}(f)=a_{\vec{x}}+f(\vec{x}).
\end{equation}
The result can be written in the sense of quadratic forms, for vectors $\psi\in\fdom{\mathcal{H}_{\varrho,\ls L}}$
\begin{equation}\label{eq:HamiltonianWeylSandwich}
    \mathcal{H}_{\varrho,\ls L}\!\left[\mathcal{W}(\Psi^{\ls t}_{\!\varrho,\ls L})\ls \psi\right]\nls =\nls \left(\mathcal{H}_{\varrho,\ls L}+\mathcal{Q}^{(1)}_{\varrho,\ls L}(t)+\mathcal{C}^{(2)}_{\varrho,\ls L}(t)+\mathcal{Q}^{(2)}_{\varrho,\ls L}(t)+\mathcal{Q}^{(3)}_{\varrho,\ls L}(t)\right)\![\psi]+\mathscr{E}_{\varrho,\ls L}[\Psi^{\ls t}_{\!\varrho,\ls L}]\,\norm{\psi}^2,
\end{equation}
where $\mathscr{E}_{\varrho,\ls L}$ is the Hartree functional~\eqref{def:HartreeFunctional}, and the Hermitian quadratic forms introduced here have the following expressions in terms of the distributional-valued operator~\eqref{def:annihilationDistribution} for $\psi\nls \in\nls \dom{\mathcal{N}}$
\begin{equation}
    \mathcal{Q}^{(1)}_{\varrho,\ls L}(t)[\psi]\vcentcolon=2\ls \Re\integrate[\Lambda_L]{\Big[\nls -\nls \Delta_{\vec{x}}\Psi^{\ls t}_{\!\varrho,\ls L}(\vec{x})+\tfrac{1}{\varrho}\big(V_L\!\ast\nls \abs{\Psi^{\ls t}_{\!\varrho,\ls L}}^2\big)\nls (\vec{x})\,\Psi^{\ls t}_{\!\varrho,\ls L}(\vec{x})\Big]\scalar{a_{\vec{x}}\psi}{\psi};\!d\vec{x}},
\end{equation}

\vspace{-0.425cm}

\begin{subequations}\label{eqs:generatorTerms}
    \begin{gather}
        \begin{split}
            \mathcal{C}^{(2)}_{\varrho,\ls L}(t)[\psi]\vcentcolon=&\:\frac{1}{\varrho}\nls \integrate[\Lambda_L]{\big(V_L\!\ast\nls \abs{\Psi^{\ls t}_{\!\varrho,\ls L}}^2\big)\nls (\vec{x})\,\norm{a_{\vec{x}}\psi}^2;\!d\vec{x}}\,+\\
            &+\frac{1}{\varrho}\nls \integrate[\Lambda^2_L]{V_L(\vec{x}\nls -\nls \vec{y})\,\conjugate*{\Psi^{\ls t}_{\!\varrho,\ls L}(\vec{x})}\ls \Psi^{\ls t}_{\!\varrho,\ls L}(\vec{y})\,\scalar{a_{\vec{y}}\psi}{a_{\vec{x}}\psi};\!d\vec{x}d\vec{y}},\label{def:quadraticGeneratorC}
        \end{split}\\
        \mathcal{Q}^{(2)}_{\varrho,\ls L}(t)[\psi]\vcentcolon=\frac{1}{\varrho}\,\Re\!\integrate[\Lambda^2_L]{V_L(\vec{x}\nls -\nls \vec{y})\,\Psi^{\ls t}_{\!\varrho,\ls L}(\vec{x})\Psi^{\ls t}_{\!\varrho,\ls L}(\vec{y})\,\scalar{a_{\vec{y}}a_{\vec{x}}\psi}{\psi};\!d\vec{x}d\vec{y}},\label{def:quadraticGeneratorQ}\\
        \mathcal{Q}^{(3)}_{\varrho,\ls L}(t)[\psi]\vcentcolon=\frac{2}{\varrho}\,\Re\!\integrate[\Lambda^2_L]{V_L(\vec{x}\nls -\nls \vec{y})\,\Psi^{\ls t}_{\!\varrho,\ls L}(\vec{y})\,\scalar{a_{\vec{y}}a_{\vec{x}}\psi}{a_{\vec{x}}\psi};\!d\vec{x}d\vec{y}}.\label{def:cubicGeneratorQ}
    \end{gather}
\end{subequations}
Therefore, since $\Psi^{\ls t}_{\!\varrho,\ls L}$ solves the Hartree equation~\eqref{eq:HartreePDE}, one has a simplification by virtue of the identity
\begin{equation}\label{eq:linearTermSimplified}
    \mathcal{Q}^{(1)}_{\varrho,\ls L}(t)[\psi]=\scalar{\psi}{\phi(i\ls \partial_t\Psi^{\ls t}_{\!\varrho,\ls L})\ls \psi},\qquad\psi\in\fdom{\mathcal{N}}.
\end{equation}
We stress that the term $\mathcal{Q}^{(1)}_{\varrho,\ls L}(t)$ is too large to be controlled in our setting.
This simplification explains our need to select $\Psi^{\ls t}_{\!\varrho,\ls L}$ as a solution to the Hartree equation~\eqref{eq:HartreePDE}.\newline
Taking account of~\eqref{eq:linearTermSimplified}, the Hermitian quadratic form associated with the generator of the fluctuation dynamics becomes, for $\psi\in\fdom{\mathcal{H}_{\varrho,\ls L}}$
\begin{equation}\label{eq:generatorDecomposition}
    \mathcal{L}_{\varrho,\ls L}(t)[\psi]=\left(\mathcal{H}_{\varrho,\ls L}+\mathcal{C}^{(2)}_{\varrho,\ls L}(t)+\mathcal{Q}^{(2)}_{\varrho,\ls L}(t)+\mathcal{Q}^{(3)}_{\varrho,\ls L}(t)\nls \right)\![\psi]-\tfrac{1}{2\varrho}\scalar{V_L\!\ast\nls \abs{\Psi^{\ls t}_{\!\varrho,\ls L}}^2}{\abs{\Psi^{\ls t}_{\!\varrho,\ls L}}^2}[2]\ls \norm{\psi}^2.
\end{equation}
Since the last term will not play any role, we also define the operator
\begin{equation}\label{def:effectiveGenerator}
    \mathcal{G}_{\varrho,\ls L}(t)\vcentcolon=\mathcal{L}_{\varrho,\ls L}(t)+\tfrac{1}{2\varrho}\scalar{V_L\!\ast\nls \abs{\Psi^{\ls t}_{\!\varrho,\ls L}}^2}{\abs{\Psi^{\ls t}_{\!\varrho,\ls L}}^2}[2]\ls \Id,\qquad\dom{\mathcal{G}_{\varrho,\ls L}}=\dom{\mathcal{H}_{\varrho,\ls L}}.
\end{equation}
The Hermitian quadratic forms~\eqref{eqs:generatorTerms} satisfy the following \emph{a priori} bounds.
\begin{prop}\label{th:generatorTermsAPEstimates}
    Given the Hermitian quadratic forms defined in equations~\eqref{eqs:generatorTerms}, one has for all $\psi\nls \in\nls \dom{\mathcal{N}}$
    \begin{subequations}\label{eqs:generatorTermsEstimates}
        \begin{gather}
            \abs*{\ls \mathcal{C}^{(2)}_{\varrho,\ls L}(t)[\psi]}\leq \tfrac{2}{\varrho}\,\norm{V_L\!\ast\nls \abs{\Psi^{\ls t}_{\!\varrho,\ls L}}^2}[\infty]\,\norm{\mathcal{N}^{\frac{1}{2}}\psi}^2,\label{eq:APestimateQuadraticGeneratorC}\\
            \abs*{\mathcal{Q}^{(2)}_{\varrho,\ls L}(t)[\psi]}\leq \tfrac{1}{\varrho}\,\norm{V_L\!\ast\nls \abs{\Psi^{\ls t}_{\!\varrho,\ls L}}^2}[\infty]\,\norm{\mathcal{N}^{\frac{1}{2}}\psi}^2+\tfrac{\norm{\Psi^{\ls t}_{\mspace{-3.75mu}\varrho,\mspace{0.75mu} L}}[2]\!\nls}{\varrho}\,  \norm{V_L^{\ls 2}\nls \ast\nls \abs{\Psi^{\ls t}_{\!\varrho,\ls L}}^2}[\infty]^{\frac{1}{2}}\, \norm{\mathcal{N}^{\frac{1}{2}}\psi}\,\norm{\psi},\label{eq:APestimateQuadraticGeneratorQ}\\
            \abs*{\ls \mathcal{Q}^{(3)}_{\varrho,\ls L}(t)[\psi]}\leq\tfrac{\sqrt{8}}{\varrho}\ls \norm{V_L\!\ast\nls \abs{\Psi^{\ls t}_{\!\varrho,\ls L}}^2}[\infty]^{\frac{1}{2}} \sqrt{\mathcal{V}_{\nls L}[\psi]\ls }\,\norm{\mathcal{N}^{\frac{1}{2}}\psi},\label{eq:APestimateCubicGeneratorQ}
        \end{gather}
    \end{subequations}
    where $\mathcal{V}_{\nls L}[\ls \cdot\ls ]$ stands for the Hermitian quadratic form associated with the potential of the Hamiltonian~\eqref{eq:formHamiltonianFock}.
    \begin{proof}
        Concerning~\eqref{def:quadraticGeneratorC}, one has
        \begin{align*}
            \abs*{\mathcal{C}^{(2)}_{\varrho,\ls L}(t)[\psi]}\leq&\:\frac{1}{\varrho}\nls \integrate[\Lambda_L]{\big(V_L\!\ast\nls \abs{\Psi^{\ls t}_{\!\varrho,\ls L}}^2\big)\nls (\vec{x})\,\norm{a_{\vec{x}}\psi}^2;\!d\vec{x}}\,+\\
            &+\frac{1}{\varrho}\nls \integrate[\Lambda^2_L]{V_L(\vec{x}\nls -\nls \vec{y})\,\abs{\Psi^{\ls t}_{\!\varrho,\ls L}(\vec{x})}\,\abs{\Psi^{\ls t}_{\!\varrho,\ls L}(\vec{y})}\:\norm{a_{\vec{y}}\psi}\,\norm{a_{\vec{x}}\psi};\!d\vec{x}d\vec{y}}\\
            \leq&\:\frac{2}{\varrho}\nls \integrate[\Lambda_L]{\big(V_L\!\ast\nls \abs{\Psi^{\ls t}_{\!\varrho,\ls L}}^2\big)\nls (\vec{x})\,\norm{a_{\vec{x}}\psi}^2;\!d\vec{x}},
        \end{align*}
        having adopted a Cauchy--Schwarz inequality and exploited the symmetry of exchange between $\vec{x}$ and $\vec{y}$ in the integrand.
        Inequality~\eqref{eq:APestimateQuadraticGeneratorC} is then obtained by exploiting~\eqref{eq:NumberFromIntegral}.
        \newline
        Next, taking account of~\eqref{def:quadraticGeneratorQ}
        \begin{equation*}
            \mathcal{Q}^{(2)}_{\varrho,\ls L}(t)[\psi]=\frac{1}{\varrho}\,\Re\!\integrate[\Lambda_L]{\conjugate*{\Psi^{\ls t}_{\!\varrho,\ls L}(\vec{x})\!}\;\scalar{\ls \adj{a}\big(V_L(\vec{x}-\vec{\cdot})\ls \Psi^{\ls t}_{\!\varrho,\ls L}\big)\ls \psi}{a_{\vec{x}}\psi};\!d\vec{x}},
        \end{equation*}
        hence,
        \begin{align*}
            \abs*{\mathcal{Q}^{(2)}_{\varrho,\ls L}(t)[\psi]}&\leq\frac{1}{\varrho}\nls  \integrate[\Lambda_L]{\abs{\Psi^{\ls t}_{\!\varrho,\ls L}(\vec{x})}\:\norm{\adj{a}\big(V_L(\vec{x}-\vec{\cdot})\ls \Psi^{\ls t}_{\!\varrho,\ls L}\big)\ls \psi}\,\norm{a_{\vec{x}}\psi};\!d\vec{x}}\\
            &\leq\frac{1}{\varrho}\nls  \integrate[\Lambda_L]{\abs{\Psi^{\ls t}_{\!\varrho,\ls L}(\vec{x})}\,\Big[\norm{a\big(V_L(\vec{x}-\vec{\cdot})\ls \Psi^{\ls t}_{\!\varrho,\ls L}\big)\ls \psi}+\nls \sqrt{\nls \big(V^{\ls 2}_L\nls \ast\nls \abs{\Psi^{\ls t}_{\!\varrho,\ls L}}^2\big)\nls (\vec{x})\ls }\ls \norm{\psi}\Big]\norm{a_{\vec{x}}\psi};\!d\vec{x}}.
        \end{align*}
        In the last step we have used
        \begin{equation}
            \norm{\adj{a}(f)\ls \psi}^2\!=\nls \norm{a(f)\ls \psi}^2\nls +\norm{f}[2]^2\ls \norm{\psi}^2,
        \end{equation}
        which is due to~\eqref{eq:CCR}.
        Then, making use of
        \begin{equation}
            \norm{a(f)\ls \psi}\leq \integrate[\Lambda_L]{\abs{f(\vec{x})}\,\norm{a_{\vec{x}}\psi};\!d\vec{x}},
        \end{equation}
        we obtain
        \begin{align*}
            \abs*{\mathcal{Q}^{(2)}_{\varrho,\ls L}(t)[\psi]}\leq&\:\frac{1}{\varrho}\nls  \integrate[\Lambda^2_L]{V_L(\vec{x}\nls -\nls \vec{y})\,\abs{\Psi^{\ls t}_{\!\varrho,\ls L}(\vec{x})}\,\abs{\Psi^{\ls t}_{\!\varrho,\ls L}(\vec{y})}\:\norm{a_{\vec{y}}\psi}\,\norm{a_{\vec{x}}\psi};\!d\vec{x}d\vec{y}}\,+\\
            &+\frac{1}{\varrho}\,\norm{V^{\ls 2}_L\nls \ast\nls \abs{\Psi^{\ls t}_{\!\varrho,\ls L}}^2}[\infty]^{\frac{1}{2}}\,\norm{\psi}\nls \integrate[\Lambda_L]{\abs{\Psi^{\ls t}_{\!\varrho,\ls L}(\vec{x})}\:\norm{a_{\vec{x}}\psi};\!d\vec{x}}\\[-2.5pt]
            \leq&\:\frac{1}{\varrho}\nls  \integrate[\Lambda_L]{\big(V_L\!\ast\nls \abs{\Psi^{\ls t}_{\!\varrho,\ls L}}^2\big)\nls (\vec{x})\:\norm{a_{\vec{x}}\psi}^2;\!d\vec{x}}+\nls \tfrac{\norm{\Psi^{\ls t}_{\mspace{-3.75mu}\varrho,\mspace{0.75mu} L}}[2]\!\nls}{\varrho}\, \norm{V^{\ls 2}_L\nls \ast\nls \abs{\Psi^{\ls t}_{\!\varrho,\ls L}}^2}[\infty]^{\frac{1}{2}}\,\norm{\psi}\,\sqrt{\nls \integrate[\Lambda_L]{\norm{a_{\vec{x}}\psi}^2;\!d\vec{x}}\ls },
        \end{align*}
        which yields~\eqref{eq:APestimateQuadraticGeneratorQ}.\newline
        In conclusion, by means of a Cauchy--Schwarz inequality, one exploits the non-negativity of the potential to estimate $\mathcal{Q}^{(3)}_{\varrho,\ls L}(t)[\psi]$ in terms of $\mathcal{V}_{\nls L}[\psi]$
        \begin{align*}
            \abs*{\mathcal{Q}^{(3)}_{\varrho,\ls L}(t)[\psi]}&\leq\frac{2}{\varrho}\,\sqrt{\integrate[\Lambda^2_L]{V_L(\vec{x}\nls -\nls \vec{y})\:\norm{a_{\vec{y}}a_{\vec{x}}\psi}^2;\!d\vec{x}d\vec{y}}\ls }\ls \sqrt{\integrate[\Lambda^2_L]{V_L(\vec{x}\nls -\nls \vec{y})\,\abs{\Psi^{\ls t}_{\!\varrho,\ls L}(\vec{y})}^2\:\norm{a_{\vec{x}}\psi}^2;\!d\vec{x}d\vec{y}}\ls }\\
            &\leq \frac{\sqrt{8}}{\varrho}\sqrt{\mathcal{V}_{\nls L}[\psi]\ls }\ls \norm{V_L\!\ast\nls \abs{\Psi^{\ls t}_{\!\varrho,\ls L}}^2}[\infty]^{\frac{1}{2}}\,\sqrt{\integrate[\Lambda_L]{\norm{a_{\vec{x}}\psi}^2;\!d\vec{x}}\ls }\ls ,
        \end{align*}
        which completes the proof.
        
    \end{proof}
\end{prop}
Combining these estimates, one gets an \emph{a priori} bound on the generator in terms of the Hamiltonian.
\begin{cor}\label{th:generatorAPEstimate}
    Let $\mathcal{G}_{\varrho,\ls L}(t)[\ls \cdot\ls ]$ be the Hermitian quadratic form associated with the operator defined in~\eqref{def:effectiveGenerator}.
    Then, for all $\varepsilon, \varsigma\nls >\nls 0$ and $\psi\in\fdom{\mathcal{H}_{\varrho,\ls L}}$, one has
    \begin{equation*}
        \abs*{\mathcal{G}_{\varrho,\ls L}(t)[\psi]\! -\mspace{-2.25mu}\mathcal{H}_{\varrho,\ls L}[\psi]}\nls\leq\mspace{-0.75mu}\varepsilon\ls \mathcal{H}_{\varrho,\ls L}[\psi]\mspace{0.75mu}+\tfrac{1}{\varrho}\Big[\norm{V_L\nls\ast\nls \abs{\Psi^{\ls t}_{\!\varrho,\ls L}}^2}[\infty]\big(3+\mspace{-0.75mu}\tfrac{2}{\varepsilon}\big)\nls +\nls \tfrac{\varsigma}{2}\ls \norm{V^{\ls 2}_L\nls \ast\nls \abs{\Psi^{\ls t}_{\!\varrho,\ls L}}^2}[\infty]\Big] \norm{\mathcal{N}^{\frac{1}{2}}\psi}^2\mspace{-2.25mu} +\mspace{-0.75mu}\tfrac{\,\norm{\Psi_{\mspace{-3.75mu}\varrho,\mspace{0.75mu} L}}[2]^2\!}{2\ls\varsigma \varrho}\, \norm{\psi}^2\!.
    \end{equation*}
    \begin{proof}
        The result follows directly from Proposition~\ref{th:generatorTermsAPEstimates} by simply applying Young's inequality for products.

    \end{proof}
\end{cor}
In Section~\ref{sec:excitationsControl}, the time derivative of the Hermitian quadratic form $\mathcal{G}_{\varrho,\ls L}(t)[\ls \cdot\ls ]$ will be central to the proof of our main result.
Clearly, one has the decomposition
\begin{equation}\label{eq:derivativeEffectiveGenerator}
    \partial_t\ls \mathcal{G}_{\varrho,\ls L}(t)[\psi]=\vcentcolon\dot{\mathcal{G}}_{\varrho,\ls L}(t)[\psi]=\left(\partial_t\mspace{2.25mu}\mathcal{C}^{(2)}_{\varrho,\ls L}(t)+\partial_t\mspace{0.75mu}\mathcal{Q}^{(2)}_{\varrho,\ls L}(t)+\partial_t\mspace{0.75mu}\mathcal{Q}^{(3)}_{\varrho,\ls L}(t)\nls \right)\![\psi],\qquad\psi\in\dom{\mathcal{N}}.
\end{equation}
In the following proposition, we compute the expressions of these objects.
\begin{prop}\label{th:timeDerivativeIdentities}
    Assume $t\longmapsto\Psi_{\!\varrho,\ls L}^{\ls t}\!\in C^{\mspace{0.75mu}1}\big(\mspace{0.75mu}[0,\infty),\mathfrak{A}^0(\Lambda_L)\mspace{-0.75mu}\big)\mspace{-0.75mu}\cap C^{\ls 0}\big(\mspace{0.75mu}[0,\infty), \mathfrak{A}^2(\Lambda_L)\mspace{-0.75mu}\big)$ solves the Hartree equation~\eqref{eq:HartreePDE} and take account of definitions~\eqref{eqs:generatorTerms}.
    Then, for every $T>0$ one has for all $\psi\nls \in\nls \dom{\mathcal{N}}$ and $t\in[0,T]$
    \begin{subequations}\label{eqs:generatorTermsDerivative}\begin{gather*}%
        \begin{split}%
            &\partial_t\mspace{2.25mu}\mathcal{C}^{(2)}_{\varrho,\ls L}(t)[\psi]=-\frac{2}{\varrho}\nls \integrate[\Lambda_L]{\big(V_L\!\ast\Im\ls (\ls \conjugate*{\Psi}^{\ls t}_{\!\varrho,\ls L}\ls \Delta\Psi^{\ls t}_{\!\varrho,\ls L})\big)\nls (\vec{x})\,\norm{a_{\vec{x}}\psi}^2;\!d\vec{x}}\,+\\
            &\qquad+\frac{2}{\varrho}\,\Im\!\integrate[\Lambda^2_L]{V_L(\vec{x}\nls -\nls \vec{y})\,\conjugate*{\Psi^{\ls t}_{\!\varrho,\ls L}(\vec{x})}\!\left[-\Delta_{\vec{y}}\Psi^{\ls t}_{\!\varrho,\ls L}(\vec{y})\nls +\tfrac{1}{\varrho} \big(V_L\!\ast\nls \abs{\Psi^{\ls t}_{\!\varrho,\ls L}}^2\big)(\vec{y})\Psi^{\ls t}_{\!\varrho,\ls L}(\vec{y})\right]\!\scalar{a_{\vec{y}}\psi}{a_{\vec{x}}\psi};\!d\vec{x}d\vec{y}},
        \end{split}\\
        \partial_t\mspace{0.75mu}\mathcal{Q}^{(2)}_{\varrho,\ls L}(t)[\psi]\nls =\frac{2}{\varrho}\,\Im\!\integrate[\Lambda^2_L]{V_L(\vec{x}\nls -\nls \vec{y})\,\Psi^{\ls t}_{\!\varrho,\ls L}(\vec{x})\!\left[-\Delta_{\vec{y}}\Psi^{\ls t}_{\!\varrho,\ls L}(\vec{y})\nls +\tfrac{1}{\varrho} \big(V_L\!\ast\nls \abs{\Psi^{\ls t}_{\!\varrho,\ls L}}^2\big)(\vec{y})\Psi^{\ls t}_{\!\varrho,\ls L}(\vec{y})\right]\!\scalar{a_{\vec{y}}a_{\vec{x}}\psi}{\psi};\!d\vec{x}d\vec{y}},\label{def:quadraticGeneratorQDerivative}\\
        \partial_t\mspace{0.75mu}\mathcal{Q}^{(3)}_{\varrho,\ls L}(t)[\psi]=\frac{2}{\varrho}\,\Im\!\integrate[\Lambda^2_L]{V_L(\vec{x}\nls -\nls \vec{y})\!\left[-\Delta_{\vec{y}}\Psi^{\ls t}_{\!\varrho,\ls L}(\vec{y})\nls +\tfrac{1}{\varrho} \big(V_L\!\ast\nls \abs{\Psi^{\ls t}_{\!\varrho,\ls L}}^2\big)(\vec{y})\Psi^{\ls t}_{\!\varrho,\ls L}(\vec{y})\right]\!\scalar{a_{\vec{y}}a_{\vec{x}}\psi}{a_{\vec{x}}\psi};\!d\vec{x}d\vec{y}}.\label{def:cubicGeneratorQDerivative}
    \end{gather*}
    \end{subequations}
    \begin{proof}
        The result is obtained by differentiating with respect to time inside the integrals appearing in the quantities given by~\eqref{eqs:generatorTerms}, and by recalling that $\Psi^{\ls t}_{\!\varrho,\ls L}$ solves the Hartree equation~\eqref{eq:HartreePDE}.
        To this end, we ensure that the Leibniz integral rule holds in these cases by exhibiting integrable majorants (uniformly in time) of the time-derivative of the integrands (\cfr~\eg~\cite[Theorem 2.27]{Fo99}).\newline
        Concerning $\partial_t\mspace{2.25mu}\mathcal{C}^{(2)}_{\varrho,\ls L}(t)$, the integrand of the first term can be bounded from above by exploiting Young's convolution inequality, so that an integrable majorant in $\Lambda_L$ is (\cfr~\eqref{eq:NumberFromIntegral})
        $$\vec{x}\;\longmapsto\,\norm{V_L}[1]\!\!\sup\limits_{s\,\in\ls [0,T]}\!\nls \left(\ls \norm{\Psi^{\ls s}_{\!\varrho,\ls L}}[\infty]\ls \norm{\Delta\Psi^{\ls s}_{\!\varrho,\ls L}}[\infty]\nls \right)\!\norm{a_{\vec{x}}\psi}^2.$$
        We stress that, since $t\longmapsto\Psi_{\!\varrho,\ls L}^{\ls t}$ belongs to $C^{\mspace{0.75mu}1}\big(\mspace{0.75mu}[0,\infty),\mathfrak{A}^0(\Lambda_L)\mspace{-0.75mu}\big)\mspace{-0.75mu}\cap C^{\ls 0}\big(\mspace{0.75mu}[0,\infty), \mathfrak{A}^2(\Lambda_L)\mspace{-0.75mu}\big)$, both $\norm{\Psi^{\ls t}_{\!\varrho,\ls L}}[\infty]$ and $\norm{\Delta\Psi^{\ls t}_{\!\varrho,\ls L}}[\infty]$ cannot blow up in a finite time interval for fixed $\varrho,L\nls >\nls 0$.\newline
        The second term is bounded by $$(\vec{x},\vec{y})\;\longmapsto\,\sup\limits_{s\,\in\ls [0,\,T]}\!\nls \left(\norm{\Psi^{\ls s}_{\!\varrho,\ls L}}[\infty]\ls \norm{\Delta\Psi^{\ls s}_{\!\varrho,\ls L}}[\infty]+\tfrac{1}{\varrho}\norm{V_L}[1] \norm{\Psi^{\ls s}_{\!\varrho,\ls L}}[\infty]^4\right)\!V_L(\vec{x}\nls -\nls \vec{y})\,\abs{\scalar{a_{\vec{y}}\psi}{a_{\vec{x}}\psi}},$$
        which is integrable in $\Lambda_L^2$, as shown by means of a Cauchy--Schwarz inequality.

        \smallskip
        
        \noindent Analogously, for $\partial_t\mspace{0.75mu}\mathcal{Q}^{(2)}_{\varrho,\ls L}$ we exhibit the time-independent majorant $$(\vec{x},\vec{y})\;\longmapsto\,\sup\limits_{s\,\in\ls [0,\,T]}\!\nls \left(\norm{\Psi^{\ls s}_{\!\varrho,\ls L}}[\infty]\ls \norm{\Delta\Psi^{\ls s}_{\!\varrho,\ls L}}[\infty]+\tfrac{1}{\varrho}\norm{V_L}[1] \norm{\Psi^{\ls s}_{\!\varrho,\ls L}}[\infty]^4\right)\!V_L(\vec{x}\nls -\nls \vec{y})\,\abs{\scalar{a_{\vec{y}}a_{\vec{x}}\psi}{\psi}},$$
        whose integral in $\Lambda_L^2$ is finite, recalling the definition of the potential part $\mathcal{V}_{\nls L}[\psi]$ in~\eqref{eq:formHamiltonianFock}.

        \smallskip

        \noindent In conclusion, taking the supremum over $\vec{y}\nls \in\nls \Lambda_L$ and $t\nls \in\nls  [0,T]$ of the square bracket in the integrand of $\partial_t\mspace{0.75mu}\mathcal{Q}^{(3)}_{\varrho,\ls L}$, one similarly obtains a time-independent integrable majorant in $\Lambda_L^2$.
        
    \end{proof}
\end{prop}
We conclude this section by providing an \emph{a priori} estimate for $\dot{\mathcal{G}}_{\varrho,\ls L}(t)[\ls \cdot\ls ]$.
\begin{cor}\label{th:generatorDerivativeAPEstimates}
    Consider the quantities derived in Proposition~\ref{th:timeDerivativeIdentities}.
    Then, for every $T\nls >\nls 0$, one has for all $\psi\in\dom{\mathcal{N}}$ and $t\in [0,T]$
    \begin{gather*}
        \abs*{\partial_t\mspace{2.25mu}\mathcal{C}^{(2)}_{\varrho,\ls L}(t)[\psi]}\leq \left(\tfrac{4}{\varrho}\sqrt{\norm{V_L\!\ast\nls \abs{\Psi^{\ls t}_{\!\varrho,\ls L}}^2}[\infty]\,\norm{V_L\!\ast\nls \abs{\Delta\Psi^{\ls t}_{\!\varrho,\ls L}}^2}[\infty]\!}\,+\tfrac{2}{\varrho^2}\ls \norm{V_L\!\ast\nls \abs{\Psi^{\ls t}_{\!\varrho,\ls L}}^2}[\infty]^2\right)\nls \norm{\mathcal{N}^{\frac{1}{2}}\psi}^2,\\[5pt]
        \begin{split}
            \abs*{\partial_t\mspace{0.75mu}\mathcal{Q}^{(2)}_{\varrho,\ls L}(t)[\psi]}\leq &\:\left(\tfrac{2}{\varrho}\sqrt{\norm{V_L\!\ast\nls \abs{\Psi^{\ls t}_{\!\varrho,\ls L}}^2}[\infty]\,\norm{V_L\!\ast\nls \abs{\Delta\Psi^{\ls t}_{\!\varrho,\ls L}}^2}[\infty]\!}\,+\tfrac{2}{\varrho^2}\ls \norm{V_L\!\ast\nls \abs{\Psi^{\ls t}_{\!\varrho,\ls L}}^2}[\infty]^2\right)\nls \norm{\mathcal{N}^{\frac{1}{2}}\psi}^2+\\
            &+2\ls\tfrac{\norm{\Psi^{\ls t}_{\!\nls\varrho,\mspace{0.75mu} L}}[2]\!}{\varrho}\nls \left(\norm{V^{\ls 2}_L\nls \ast\nls \abs{\Delta\Psi^{\ls t}_{\!\varrho,\ls L}}^2}[\infty]^{\frac{1}{2}}+\tfrac{1}{\varrho}\ls \norm{V_L\!\ast\nls \abs{\Psi^{\ls t}_{\!\varrho,\ls L}}^2}[\infty]\,\norm{V^{\ls 2}_L\nls \ast\nls \abs{\Psi^{\ls t}_{\!\varrho,\ls L}}^2}[\infty]^{\frac{1}{2}}\right)\!\norm{\mathcal{N}^{\frac{1}{2}}\psi}\,\norm{\psi},
        \end{split}\\
        \abs*{\partial_t\mspace{0.75mu}\mathcal{Q}^{(3)}_{\varrho,\ls L}(t)[\psi]}\leq\tfrac{\sqrt{8}}{\varrho}\left(\norm{V_L\!\ast\nls \abs{\Delta\Psi^{\ls t}_{\!\varrho,\ls L}}^2}[\infty]^{\frac{1}{2}}+\tfrac{1}{\varrho}\ls \norm{V_L\!\ast\nls \abs{\Psi^{\ls t}_{\!\varrho,\ls L}}^2}[\infty]^{\frac{3}{2}} \right)\!\sqrt{\mathcal{V}_{\nls L}[\psi]\ls }\,\norm{\mathcal{N}^{\frac{1}{2}}\psi},
    \end{gather*}
    where $\mathcal{V}_{\nls L}[\ls \cdot\ls ]$ stands for the Hermitian quadratic form associated with the potential of the Hamiltonian~\eqref{eq:formHamiltonianFock}.
    \begin{proof}
        The argument follows exactly the same steps as the proof of Proposition~\ref{th:generatorTermsAPEstimates}.

    \end{proof}
\end{cor}


\section{Properties of Quasi-Complete Bose--Einstein Condensates}\label{sec:qcBECproperties}

In this section, we collect propositions that clarify the relationship among the objects defined in Section~\ref{sec:qcBEC}, validating key properties of quasi-vacuum and quasi-canonical coherent states.

\smallskip

\noindent The first result states that applying the Weyl operator~\eqref{def:Weyl} on a quasi-vacuum state yields a quasi-canonical coherent state.
\begin{prop}\label{th:WeylCoherent}
    Let $f_{\varrho,\ls L}\!\in\nls \Lp{2}(\Lambda_L)$ satisfy~\eqref{eq:becNormalization}, and let $\Omega_{\varrho,\ls L}\!\in\nls \dom{\mathcal{N}}$ be a quasi-vacuum state with respect to $f_{\varrho,\ls L}\ls $.
    Then, $\phi_{\varrho,\ls L}\!=\mathcal{W}(f_{\varrho,\ls L})\ls \Omega_{\varrho,\ls L}$ is a quasi-canonical coherent state.\newline
    Moreover, for any $g_{\varrho,\ls L}\nls \in\!\Lp{2}(\Lambda_L)$ such that $\norm{g_{\varrho,\ls L}}[2]=1$ for all $\varrho,L\nls >\nls 0$, $\phi_{\varrho,\ls L}$ is a quasi-eigenstate of $a(g_{\varrho,\ls L})$ with quasi-eigenvalue $\scalar{g_{\varrho,\ls L}}{f_{\varrho,\ls L}}[2]\ls $.
    \begin{proof}
        We first show that items~\textit{ii)} and~\textit{iii)} of Definition~\ref{def:Q-CCS} hold.\newline
        Using equations~\eqref{eqs:WeylConjugation}, we obtain
        $$\mathbb{E}_{\ls\phi_{\varrho,\ls L}}[\ls \mathcal{N}\ls ]=\scalar{\Omega_{\varrho,\ls L}}{\adj{\mathcal{W}\ls }\nls (f_{\varrho,\ls L})\ls \mathcal{N}\ls \mathcal{W}(f_{\varrho,\ls L})\ls \Omega_{\varrho,\ls L}}=\norm{\mathcal{N}^{\frac{1}{2}}\Omega_{\varrho,\ls L}}^2+\scalar{\Omega_{\varrho,\ls L}}{\phi(f_{\varrho,\ls L})\ls \Omega_{\varrho,\ls L}}+\norm{f_{\varrho,\ls L}}[2]^2.$$
        Therefore,
        \begin{align*}
            \frac{1}{\varrho L^3}\ls \mathbb{E}_{\ls\phi_{\varrho,\ls L}}[\ls \mathcal{N}\ls ]&\leq \frac{\norm{\mathcal{N}^{\frac{1}{2}}\Omega_{\varrho,\ls L}}^2\!}{\varrho L^3}+2\,\frac{\norm{\mathcal{N}^{\frac{1}{2}}\Omega_{\varrho,\ls L}}}{\sqrt{\varrho L^3}}\frac{\norm{f_{\varrho,\ls L}}[2]}{\!\nls\sqrt{\varrho L^3\ls}\ls}+\frac{\norm{f_{\varrho,\ls L}}[2]^2}{\varrho L^3}\\
            &\leq 1+\frac{\norm{\mathcal{N}^{\frac{1}{2}}\Omega_{\varrho,\ls L}}^2\!}{\varrho L^3}+2\,\frac{\norm{\mathcal{N}^{\frac{1}{2}}\Omega_{\varrho,\ls L}}}{\sqrt{\varrho L^3}}\frac{\norm{f_{\varrho,\ls L}}[2]}{\!\nls\sqrt{\varrho L^3\ls}\ls}+\abs*{\frac{\norm{f_{\varrho,\ls L}}[2]^2}{\varrho L^3}-1},
        \end{align*}
        which implies (recall the limit superior is submultiplicative and commutes with non-decreasing continuous functions such as the square root)
        $$\lim_{\varrho\to\infty}\limsup_{L\to\infty}\abs*{\frac{1}{\varrho L^3}\ls \mathbb{E}_{\phi_{\varrho,\ls L}}[\ls \mathcal{N}\ls ]-1}=0,$$
        because of item~$ii)$ of Definition~\ref{def:Q-Omega}.
        Analogously,
        \begin{align*}
            {\textstyle \var_{\ls \phi_{\varrho,\ls L}}[\ls \mathcal{N}\ls ]}&=\nls \norm{\mathcal{N}\ls \mathcal{W}(f_{\varrho,\ls L})\ls \Omega_{\varrho,\ls L}}^2\nls -\big(\scalar{\Omega_{\varrho,\ls L}}{\adj{\mathcal{W}\ls }\nls (f_{\varrho,\ls L})\ls \mathcal{N}\ls \mathcal{W}(f_{\varrho,\ls L})\ls \Omega_{\varrho,\ls L}}\big)^2\\
            &=\nls \norm{(\mathcal{N}\nls +\phi(f_{\varrho,\ls L})+\norm{f_{\varrho,\ls L}}[2]^2\, \Id)\ls \Omega_{\varrho,\ls L}}^2\nls -\big(\scalar{\Omega_{\varrho,\ls L}}{(\mathcal{N}\nls +\phi(f_{\varrho,\ls L})+\norm{f_{\varrho,\ls L}}[2]^2\, \Id)\ls \Omega_{\varrho,\ls L}}\big)^2.
        \end{align*}
        A direct computation shows that some terms cancel in the difference, obtaining
        \begin{equation}
            {\textstyle \var_{\ls \phi_{\varrho,\ls L}}[\ls \mathcal{N}\ls ]}=\var_{\ls \Omega_{\varrho,\ls L}}[\ls \mathcal{N}\nls +\phi(f_{\varrho,\ls L})],
        \end{equation}
        and, due to item~\textit{iii)} of Definition~\ref{def:Q-Omega}, one gets\footnote{Here we have exploited the identity $\lim\limits_{\varrho\to\infty}\limsup\limits_{L\to\infty}\frac{\varrho L^3}{\norm{f_{\mspace{-0.75mu}\varrho,\mspace{0.75mu} L}}[2]^2}=\left(\lim\limits_{\varrho\to\infty}\liminf\limits_{L\to\infty}\frac{\norm{f_{\mspace{-0.75mu}\varrho,\mspace{0.75mu} L}}[2]^2}{\varrho L^3}\right)^{\!-1}\!\!=1$.}
        $$\lim_{\varrho\to\infty}\limsup_{L\to\infty} \abs*{\frac{\var_{\ls \phi_{\varrho,\ls L}}[\ls \mathcal{N}\ls ]}{\varrho L^3}-1}=0.$$
        Next, we notice that for all $g_{\varrho,\ls L}\nls \in\nls \Lp{2}(\Lambda_L)$ with $\norm{g_{\varrho,\ls L}}[2]=1$, identity~\eqref{eq:WeylaConjugation} implies
        \begin{equation}\label{eq:aMinusEigenvalueOnQCS}
            \norm*{\big(a(g_{\varrho,\ls L})-\scalar{g_{\varrho,\ls L}}{f_{\varrho,\ls L}}[2]\nls \big)\ls \phi_{\varrho,\ls L}}=\norm{a(g_{\varrho,\ls L})\ls \Omega_{\varrho,\ls L}},
        \end{equation}
        from which we deduce
        $$\frac{\norm*{\big(a(g_{\varrho,\ls L})\nls -\mspace{-0.75mu}\scalar{g_{\varrho,\ls L}}{f_{\varrho,\ls L}}[2]\nls \big)\ls \phi_{\varrho,\ls L}}}{\norm{\mathcal{N}^{\frac{1}{2}}\phi_{\varrho,\ls L}}}\leq\nls  \frac{\norm{\mathcal{N}^{\frac{1}{2}}\Omega_{\varrho,\ls L}}}{\norm{f_{\varrho,\ls L}}[2]}\:\frac{\sqrt{\varrho L^3\ls }}{\norm{\mathcal{N}^{\frac{1}{2}}\phi_{\varrho,\ls L}}}\:\frac{\norm{f_{\varrho,\ls L}}[2]\!}{\!\nls\sqrt{\varrho L^3\ls}\ls}.$$
        Since $\phi_{\varrho,\ls L}$ is a quasi-coherent state, its expectation value is close to $\varrho L^3$.
        More precisely, the lower limit of the expectation of the number of particles in $\phi_{\varrho,\ls L}$ differs from its upper limit by a quantity that vanishes as $\varrho\to\infty$.
        Indeed, for a generic positive double sequence one has
        \begin{align*}
            \lim_{n\to\infty}\limsup_{m\to\infty} f_{n,m}=1\quad\land\quad\lim_{n\to\infty}&\liminf_{m\to\infty} f_{n,m}=1 \\
            &\iff\quad \lim_{n\to\infty}\limsup_{m\to\infty} \frac{1}{f_{n,m}}=1\quad\land\quad\lim_{n\to\infty}\liminf_{m\to\infty} \frac{1}{f_{n,m}}=1.
        \end{align*}
\hide{In other words, let
        $$\underline{\ell}_{\,\varrho}=\liminf_{L\to\infty}\abs*{\frac{\norm{\mathcal{N}^{\frac{1}{2}}\phi_{\varrho,\ls L}}^2\!\nls }{\varrho L^3}-1},\qquad \overline{\ell}_{\ls \varrho}=\limsup_{L\to\infty}\abs*{\frac{\norm{\mathcal{N}^{\frac{1}{2}}\phi_{\varrho,\ls L}}^2\!\nls }{\varrho L^3}-1},$$
        satisfying $\underline{\ell}_{\,\varrho},\,\overline{\ell}_{\,\varrho}\!\longrightarrow\nls  0$ as $\varrho\!\to\!\infty$, so that for all $\epsilon\nls >\nls 0$ we know there exists $L_\epsilon\nls >\nls 0$ such that $\forall L \nls >\nls  L_\epsilon$
        \begin{equation*}
            \underline{\ell}_{\,\varrho}\nls -\epsilon<\abs*{\frac{\norm{\mathcal{N}^{\frac{1}{2}}\phi_{\varrho,\ls L}}^2\!\nls }{\varrho L^3}-1}<\overline{\ell}_{\ls \varrho}+\epsilon.
        \end{equation*}
        Equivalently,
        \begin{equation*}
            \frac{\norm{\mathcal{N}^{\frac{1}{2}}\phi_{\varrho,\ls L}}^2\!\!}{\varrho L^3}\in\begin{dcases}
                \big(\nls \max\nls \big\{0,\,1\!-\nls \epsilon\nls -\!\overline{\ell}_{\ls \varrho}\big\},1\mspace{-2.25mu}+\nls \epsilon\nls -\nls \underline{\ell}_{\,\varrho}\big)\cup\big(1\!-\nls \epsilon\nls +\nls \underline{\ell}_{\,\varrho},1\mspace{-2.25mu}+\nls \epsilon\nls +\!\overline{\ell}_{\ls \varrho}\big),\qquad &\text{if }\,\epsilon\nls <\nls \underline{\ell}_{\,\varrho},\\
                \big(\nls \max\nls \big\{0,\,1\!-\nls \epsilon\nls -\!\overline{\ell}_{\ls \varrho}\big\},1\mspace{-2.25mu}+\nls \epsilon\nls +\!\overline{\ell}_{\ls \varrho}\big), &\text{otherwise}.
            \end{dcases}
        \end{equation*}
        In any case, for all $\epsilon\nls >\nls 0$ and $L$ large enough we have $\frac{1}{\!\nls \sqrt{1+\ls \epsilon\,+\overline{\ell}_{\ls \varrho}\ls }\ls }\nls <\frac{\sqrt{\varrho L^3\ls }}{\vphantom{\big|}\norm*{\mathcal{N}^{1/2}\phi_{\varrho,\ls L}}} < \frac{1}{\!\nls \sqrt{\max\{0,\,1-\ls \epsilon\,-\overline{\ell}_{\ls \varrho}\}\ls }\ls }$.
        Hence,
        \begin{gather*}
            \frac{1}{\!\nls \sqrt{1+\ls \epsilon\ls }\ls }\leq \liminf_{\varrho\to\infty}\ls \limsup_{L\to\infty}\frac{\sqrt{\varrho L^3\ls }}{\norm{\mathcal{N}^{\frac{1}{2}}\phi_{\varrho,\ls L}}}\leq \limsup_{\varrho\to\infty}\ls \limsup_{L\to\infty}\frac{\sqrt{\varrho L^3\ls }}{\norm{\mathcal{N}^{\frac{1}{2}}\phi_{\varrho,\ls L}}} \leq \frac{1}{\!\nls \sqrt{\max\{0,\,1-\ls \epsilon\}\ls }\ls },\\
            \frac{1}{\!\nls \sqrt{1+\ls \epsilon\ls }\ls }\leq \liminf_{\varrho\to\infty}\ls \liminf_{L\to\infty}\frac{\sqrt{\varrho L^3\ls }}{\norm{\mathcal{N}^{\frac{1}{2}}\phi_{\varrho,\ls L}}}\leq \limsup_{\varrho\to\infty}\ls \liminf_{L\to\infty}\frac{\sqrt{\varrho L^3\ls }}{\norm{\mathcal{N}^{\frac{1}{2}}\phi_{\varrho,\ls L}}} \leq \frac{1}{\!\nls \sqrt{\max\{0,\,1-\ls \epsilon\}\ls }\ls }.
        \end{gather*}
        Since the inequality holds for all $\epsilon\nls >\nls 0$, this means that the limit in $\varrho$ exists both for the upper and the lower limit in $L$, obtaining
        \begin{equation}\label{eq:limsupOfInverse}
            \lim_{\varrho\to\infty}\limsup_{L\to\infty}\abs*{\frac{\mathbb{E}_{\phi_{\varrho,\ls L}}[\mathcal{N}]}{\varrho L^3}-1}=0\quad\implies\quad\lim_{\varrho\to\infty}\limsup_{L\to\infty}\abs*{\frac{\sqrt{\varrho L^3\ls }}{\norm{\mathcal{N}^{\frac{1}{2}}\phi_{\varrho,\ls L}}} -1}=0.
        \end{equation}
        Therefore, combining this result with~\eqref{eq:aMinusEigenvalueOnQCS},
        \begin{equation*}
            \begin{split}
                \limsup_{\varrho\to\infty}\limsup_{L\to\infty}\frac{\norm*{\big(a(g_{\varrho,\ls L})\nls -\mspace{-0.75mu}\scalar{g_{\varrho,\ls L}}{f_{\varrho,\ls L}}[2]\nls \big)\ls \phi_{\varrho,\ls L}}}{\norm{\mathcal{N}^{\frac{1}{2}}\phi_{\varrho,\ls L}}}\leq&\ls  \lim_{\varrho\to\infty}\limsup_{L\to\infty}\ls \frac{\norm{\mathcal{N}^{\frac{1}{2}}\Omega_{\varrho,\ls L}}}{\norm{f_{\varrho,\ls L}}[2]}\ls \times\\
                &\times\!\lim_{\varrho\to\infty}\limsup_{L\to\infty}\left(\nls 1+\abs*{\frac{\sqrt{\varrho L^3\ls }}{\norm{\mathcal{N}^{\frac{1}{2}}\phi_{\varrho,\ls L}}} -1}\right)\!,
            \end{split}
        \end{equation*}
        that vanishes in the iterated limit.
}

    \end{proof}
\end{prop}
\begin{note}
    Item~\textit{ii)} of Definition~\ref{def:Q-Omega}  is required to prove quasi-canonical coherence in Proposition~\ref{th:WeylCoherent}, whereas the analogous weaker condition involving $\mathbb{E}_{\,\Omega_{\varrho,\ls L}}[\ls \mathcal{N}\nls +\phi(f_{\varrho,\ls L})]$ suffices only for quasi-coherence.
\end{note}
With the following proposition, we analyse how Definition~\ref{def:QCBEC} of a quasi-complete Bose--Einstein condensate could be read in terms of its Fourier modes by taking account of Assumption~\ref{ass:translationalInvariance}.
\begin{prop}\label{th:BECmomentum}
    Consider $\psi_{\varrho,\ls L}\!\in\nls \fdom{\mathcal{N}}\nls \smallsetminus\nls \ker{\mathcal{N}}$, with $\norm{\psi_{\varrho,\ls L}}\nls =\nls 1$ for all $\varrho,L\nls >\nls 0$, exhibiting quasi-complete condensation, and let $\Phi_{\!\varrho,\ls L}\!\in\nls H^1(\Lambda_L)$ be a quasi-complete Bose--Einstein condensate for $\psi_{\varrho,\ls L}\ls $ with a macroscopic counterpart $\Phi\in H^1(\Lambda_1)$.
    Assume there exists an orthonormal basis $\{\mathrm{e}_{L;\,\vec{n}}\}_{\vec{n}\ls \in\,\Z^3}$ of $\Lp{2}(\Lambda_L)$ and a sequence $\{\alpha(\vec{n})\nls \}_{\vec{n}\ls \in\,\Z^3}$ such that $\norm{\alpha}[\ell_2(\Z^3)]=1$ and
    \begin{gather}\label{eq:tailConvergenceFourierCoefficients}
        \lim_{L\to\infty}\sum_{\vec{n}\ls \in\,\Z^3} \abs*{\scalar{\mathrm{e}_{L;\,\vec{n}}}{\tfrac{1}{\!\nls \sqrt{L^3\ls }\ls } \Phi\big(\tfrac{\vec{\cdot}}{L}\big)}[2]-\alpha(\vec{n})}^2\!=0,\\
        \lim_{\varrho\to\infty}\limsup_{L\to\infty}\, \frac{1}{\norm{\mathcal{N}^\frac{1}{2}\psi_{\varrho,\ls L}}^2\!\nls }\sum_{\substack{\vec{n},\ls \vec{m}\ls \in\,\Z^3\\[-1.5pt] \vec{n}\neq\ls  \vec{m}}}\!\abs{\scalar{\psi_{\varrho,\ls L}}{\adj{a}(\mathrm{e}_{L;\,\vec{m}})\ls a(\mathrm{e}_{L;\,\vec{n}})\ls \psi_{\varrho,\ls L}}}^2=0.\label{eq:translationalInvarianceFourier}
    \end{gather}
    Then, there exists a unique $\vec{k}_0\!\in\nls \Z^3$ and $\vartheta\in[0,2\pi)$ such that
    \begin{equation}\label{eq:becFourierCoefficients1}
        \lim_{\varrho\to\infty}\limsup_{L\to\infty}\ls \sum_{\vec{n}\ls \in\,\Z^3}\abs*{\ls\scalar{\mathrm{e}_{L;\,\vec{n}}}{\tfrac{\Phi_{\nls\varrho,\mspace{0.75mu} L}}{\norm{\Phi_{\nls\varrho,\mspace{0.75mu} L}}[2]\!}\,}[2]\nls -e^{i\ls \vartheta}\delta_{\vec{n},\ls \vec{k}_0}}^2\!=0.
    \end{equation}
    \begin{note}\label{rmk:likeFourierBasis}
        In case the complete orthonormal system $\{\mathrm{e}_{L;\,\vec{n}}\}_{\vec{n}\ls \in\,\Z^3}$ is the Fourier basis~\eqref{def:FourierBasis}, condition~\eqref{eq:tailConvergenceFourierCoefficients} holds true because the Fourier coefficients of the macroscopic order parameter compose a square summable sequence independent of $L$.
        Moreover, Assumption~\ref{ass:translationalInvariance} in momentum space representation exactly reproduces hypothesis~\eqref{eq:translationalInvarianceFourier}.
        In fact, one can check that
        $$\scalar{\mathrm{e}_{L;\,\vec{m}}}{\langle\gamma^{(1)}_{\varphi^{\ls 0}_{\varrho,\ls L}}\rangle\,\mathrm{e}_{L;\,\vec{n}}}[2]=\delta_{\vec{m},\ls\vec{n}}\,\scalar{\mathrm{e}_{L;\,\vec{m}}}{\gamma^{(1)}_{\varphi^{\ls 0}_{\varrho,\ls L}}\mathrm{e}_{L;\,\vec{n}}}[2].$$
    \end{note}
    \begin{proof}[Proof of Proposition~\ref{th:BECmomentum}]
        We decompose the quantity $\frac{\Phi_{\nls\varrho,\mspace{0.75mu} L}}{\norm{\Phi_{\nls\varrho,\mspace{0.75mu} L}}[2]\!\nls}$ in terms of the orthonormal basis $\{\mathrm{e}_{L;\,\vec{n}}\}_{\vec{n}\ls \in\,\Z^3}\ls $, so that
        \begin{equation}\label{eq:basisDecomposition}
            \lim_{\varrho\to\infty}\limsup_{L\to\infty}\abs*{\frac{\norm*{a\Big(\tfrac{\Phi_{\nls\varrho,\mspace{0.75mu} L}}{\norm{\Phi_{\nls\varrho,\mspace{0.75mu} L}}[2]\!\nls}\Big)\psi_{\varrho,\ls L}}^2\!\nls }{\norm{\mathcal{N}^\frac{1}{2}\psi_{\varrho,\ls L}}^2} -\!\sum_{\vec{n}\ls \in\,\Z^3}\abs*{\scalar{\mathrm{e}_{L;\,\vec{n}}}{\tfrac{\Phi_{\nls\varrho,\mspace{0.75mu} L}}{\norm{\Phi_{\nls\varrho,\mspace{0.75mu} L}}[2]\!\nls}\,}[2]}^2\frac{\norm{a(\mathrm{e}_{L;\,\vec{n}})\ls \psi_{\varrho,\ls L}}^2}{\norm{\mathcal{N}^\frac{1}{2}\psi_{\varrho,\ls L}}^2}}=0,
        \end{equation}
        since the cross terms vanish in the limit by hypothesis~\eqref{eq:translationalInvarianceFourier}.
        Indeed,
        \begin{align*}
            \bigg\lvert\sum_{\substack{\vec{n},\ls \vec{m}\ls \in\,\Z^3\\[-1.5pt] \vec{n}\neq\ls \vec{m}}}\!&\scalar{\mathrm{e}_{L;\,\vec{m}}}{\tfrac{\Phi_{\nls\varrho,\mspace{0.75mu} L}}{\norm{\Phi_{\nls\varrho,\mspace{0.75mu} L}}[2]\!\nls}\,}[2]\,\scalar{\tfrac{\Phi_{\nls\varrho,\mspace{0.75mu} L}}{\norm{\Phi_{\nls\varrho,\mspace{0.75mu} L}}[2]\!\nls}}{\mathrm{e}_{L;\,\vec{n}}}[2]\,\frac{\scalar{a(\mathrm{e}_{L;\,\vec{m}})\ls \psi_{\varrho,\ls L}}{a(\mathrm{e}_{L;\,\vec{n}})\ls \psi_{\varrho,\ls L}}}{\norm{\mathcal{N}^\frac{1}{2}\psi_{\varrho,\ls L}}^2}\bigg\rvert\\
            \leq &\sqrt{\sum_{\substack{\vec{n},\ls \vec{m}\ls \in\,\Z^3\\[-1.5pt] \vec{n}\neq\ls \vec{m}}}\mspace{-12mu}\abs*{\scalar{\mathrm{e}_{L;\,\vec{m}}}{\tfrac{\Phi_{\nls\varrho,\mspace{0.75mu} L}}{\norm{\Phi_{\nls\varrho,\mspace{0.75mu} L}}[2]\!\nls}\,}[2]\ls\scalar{\tfrac{\Phi_{\nls\varrho,\mspace{0.75mu} L}}{\norm{\Phi_{\nls\varrho,\mspace{0.75mu} L}}[2]\!\nls}}{\mathrm{e}_{L;\,\vec{n}}}[2]}^2\ls }\frac{1}{\norm{\mathcal{N}^\frac{1}{2}\psi_{\varrho,\ls L}}^2\!\nls }\,\sqrt{\sum_{\substack{\vec{n},\ls \vec{m}\ls \in\,\Z^3\\[-1.5pt] \vec{n}\neq\ls \vec{m}}}\!\!\abs{\scalar{a(\mathrm{e}_{L;\,\vec{m}})\ls \psi_{\varrho,\ls L}}{a(\mathrm{e}_{L;\,\vec{n}})\ls \psi_{\varrho,\ls L}}}^2\ls }\\
            \leq &\,\frac{1}{\norm{\mathcal{N}^\frac{1}{2}\psi_{\varrho,\ls L}}^2\!\nls }\,\sqrt{\sum_{\substack{\vec{n},\ls \vec{m}\ls \in\,\Z^3\\[-1.5pt] \vec{n}\neq\ls \vec{m}}}\!\abs{\scalar{\psi_{\varrho,\ls L}}{\adj{a}(\mathrm{e}_{L;\,\vec{m}})\ls a(\mathrm{e}_{L;\,\vec{n}})\ls \psi_{\varrho,\ls L}}}^2\ls },
        \end{align*}
        by Parseval's identity.
        Recalling~\eqref{eq:numberAsSumOfOccupations}, we have
        \begin{equation*}
            \norm{\mathcal{N}^\frac{1}{2}\psi_{\varrho,\ls L}}^2\!=\!\nls \sum_{\vec{n}\ls \in\,\Z^3}\norm{a(\mathrm{e}_{L;\,\vec{n}})\ls \psi_{\varrho,\ls L}}^2,
        \end{equation*}
        which we then use in decomposition~\eqref{eq:basisDecomposition}.
        Hence, since for any real sequences $\{a_n\}_{n\ls\in\ls\N},\{b_n\}_{n\ls\in\ls\N}$ the inequality $\liminf\limits_{n\to\infty} (a_n+b_n)\leq \liminf\limits_{n\to\infty} a_n+\limsup\limits_{n\to\infty} b_n$ holds, one has
        \begin{equation}\label{eq:doNotRuinLiminf}
            1=\!\lim_{\varrho\to\infty}\liminf_{L\to\infty}\frac{\norm*{a\Big(\tfrac{\Phi_{\nls\varrho,\mspace{0.75mu} L}}{\norm{\Phi_{\nls\varrho,\mspace{0.75mu} L}}[2]\!\nls}\Big)\psi_{\varrho,\ls L}}^2\!\nls }{\norm{\mathcal{N}^\frac{1}{2}\psi_{\varrho,\ls L}}^2}\leq\lim_{\varrho\to\infty}\liminf_{L\to\infty}\!\sum_{\vec{n}\ls \in\,\Z^3}\!\frac{\abs*{\scalar{\mathrm{e}_{L;\,\vec{n}}}{\tfrac{\Phi_{\nls\varrho,\mspace{0.75mu} L}}{\norm{\Phi_{\nls\varrho,\mspace{0.75mu} L}}[2]\!\nls}\,}[2]}^2\norm{a(\mathrm{e}_{L;\,\vec{n}})\ls \psi_{\varrho,\ls L}}^2\!\nls }{\vphantom{|^|}\sum\limits_{\vec{m}\ls \in\,\Z^3} \norm{a(\mathrm{e}_{L;\,\vec{m}})\ls \psi_{\varrho,\ls L}}^2}.
        \end{equation}
        Let us stress that the $\Lp{2}$-convergence of $\Phi_{\varrho,\ls L}$ in the high-density thermodynamic limit is equivalent to the following convergence to some $\{\alpha_L(\vec{n})\nls \}_{\vec{n}\ls \in\,\Z^3}\nls \in\ell_2(\Z^3)$ such that $\norm{\alpha_L}[\ell_2(\Z^3)]=1$:
        \begin{equation}\label{eq:L2convergenceForFourierCoefficients}
            \lim_{\varrho\to\infty}\limsup_{L\to\infty}\sum_{\vec{n}\ls \in\,\Z^3}\abs*{\scalar{\mathrm{e}_{L;\,\vec{n}}}{\tfrac{\Phi_{\nls\varrho,\mspace{0.75mu} L}}{\norm{\Phi_{\nls\varrho,\mspace{0.75mu} L}}[2]\!\nls}\,}[2]-\alpha_L(\vec{n})}^2\!=0.
        \end{equation}
        More precisely, $\alpha_L(n)\nls \vcentcolon=\nls \scalar{\mathrm{e}_{L;\,\vec{n}}}{\tfrac{1}{\!\nls \sqrt{L^3\ls }\ls }\Phi\big(\tfrac{\vec{\cdot}}{L}\big)}[2]$.
        This means that assumption~\eqref{eq:tailConvergenceFourierCoefficients} causes the $\ell_2$-convergence of $\scalar{\mathrm{e}_{L;\,\vec{\cdot}}}{\tfrac{\Phi_{\nls\varrho,\mspace{0.75mu} L}}{\norm{\Phi_{\nls\varrho,\mspace{0.75mu} L}}[2]\!\nls}\,}[2]$ towards $\alpha$.
        Therefore, setting for short the probability measure $\mathfrak{m}_{\varrho,\ls L}$ on $\Z^3$ as
        $$\mathfrak{m}_{\varrho,\ls L}(\Omega)=\sum_{\vec{n}\ls \in\,\Omega}\;\frac{\norm{a(\mathrm{e}_{L;\,\vec{n}})\ls \psi_{\varrho,\ls L}}^2}{\!\!\sum\limits_{\vec{m}\ls \in\,\Z^3} \norm{a(\mathrm{e}_{L;\,\vec{m}})\ls \psi_{\varrho,\ls L}}^2},\qquad \Omega\subset\Z^3,$$
        we rewrite inequality~\eqref{eq:doNotRuinLiminf} as
        \begin{equation*}
            1\leq\lim_{\varrho\to\infty}\liminf_{L\to\infty}\!\sum_{\vec{n}\ls \in\,\Z^3}\!\abs*{\scalar{\mathrm{e}_{L;\,\vec{n}}}{\tfrac{\Phi_{\nls\varrho,\mspace{0.75mu} L}}{\norm{\Phi_{\nls\varrho,\mspace{0.75mu} L}}[2]\!\nls}\,}[2]}^2\!\mathfrak{m}_{\varrho,\ls L}(\mspace{-0.75mu}\{\nls \vec{n}\nls \}\mspace{-0.75mu}).
        \end{equation*}
        Since $\Z^3$ is a countable set and $\mathfrak{m}_{\varrho,\ls L}$ is a probability measure (in particular, it takes values within a compact set), it is possible to exploit the Bolzano--Weierstrass theorem to exhibit two positive sequences $\{L_j\}_{j\ls \in\,\N}, \{\varrho_i\}_{i\ls \in\,\N}$, such that $\lim\limits_{i\to\infty}\varrho_i=\lim\limits_{j\to\infty}L_j=\infty$ and 
        \begin{equation}\label{eq:tyBolzanoWeierstrass}
            \exists \lim_{i\to\infty}\lim_{j\to\infty} \mathfrak{m}_{\varrho_i,\ls L_j}(\mspace{-0.75mu}\{\nls \vec{n}\nls \}\mspace{-0.75mu})\in[0,1],\qquad\forall \vec{n}\in\Z^3.
        \end{equation}

        To this end, let $\{\tilde{\varrho}_{\ls m}\}_{m\ls \in\ls \N}$ be any increasing, positive sequence diverging as $m$ approaches infinity.
        Let us also enumerate the sites of the lattice $\Z^3\!=\nls \{\vec{n}_1,\vec{n}_2,\ldots\}$ so that we can select, by Bolzano--Weierstrass, a diverging sequence $L^{\nls (1)}_{1;\,j}$ such that $\mathfrak{m}_{\tilde{\varrho}_1,\ls L^{\nls (1)}_{1;\,j}}(\vec{n}_1)$ converges for $j\nls \to\nls \infty$.
        Here, 
        The same can be done for $\vec{n}_2$, and the convergence on both sites can be guaranteed for fixed $\tilde{\varrho}_1$ if we choose $\{L^{\nls (1)}_{2\ls ;\,j}\}_{j\ls \in\ls \N}\mspace{-2.25mu}\subset\! \{L^{\nls (1)}_{1;\,j}\}_{j\ls \in\ls \N}$.
        This can be repeated by induction, so that along the diagonal sequence $L^{\nls (1)}_j\!\vcentcolon=L^{\nls (1)}_{j\ls ;\,j}$ we have the convergence of each $\mathfrak{m}_{\tilde{\varrho}_1,\ls L^{\nls (1)}_j}(\mspace{-0.75mu}\{\nls \vec{n}_k\nls \}\mspace{-0.75mu})$, for fixed $\vec{n}_k\!\in\nls \Z^3$.\newline
        Now we can reiterate the argument to find another set of sequences $\{L^{\nls (m_0\ls +1)}_j\}_{j\ls \in\ls \N}\mspace{-2.25mu}\subset\!\{L^{\nls (m_0)}_j\}_{j\ls \in\ls \N}$ for which $\mathfrak{m}_{\tilde{\varrho}_{\ls m},\ls L^{\nls (m)}_j}(\mspace{-2.25mu}\{\nls \vec{n}\nls \}\mspace{-2.25mu})$ converges as $j$ grows to infinity, for any fixed $m\leq m_0\!+\nls 1$, $\vec{n}\nls \in\nls \Z^3$.
        This proves that along the diagonal sequence $L_j\!\vcentcolon=L^{\nls (j)}_{j}$, the limit $\lim\limits_{j\to\infty}\mathfrak{m}_{\tilde{\varrho}_{\ls m},\ls L_j}(\mspace{-2.25mu}\{\nls \vec{n}\nls \}\mspace{-2.25mu})$ exists for all $\vec{n}\nls \in\nls \Z^3$ and $m\nls \in\nls \N$, and is contained within the set $[0,1]$.
        However, we cannot ensure that the limit in $m\nls \to\nls \infty$ exists as well.
        Therefore, we must extract a diverging subsequence $\{m_{1;\,i}\}_{i\ls \in\ls \N}$ (for which $\tilde{\varrho}_{\ls m_{1;\ls i}}$ is still diverging because of the monotonicity of $\{\tilde{\varrho}_{\ls m}\}_{m\ls \in\ls \N}$) along which $\lim\limits_{j\to\infty}\mathfrak{m}_{\tilde{\varrho}_{\ls m_{1;\ls i}},\ls L_j}(\mspace{-2.25mu}\{\nls \vec{n}_1\nls \}\mspace{-2.25mu})$ converges as $i\nls \to\nls \infty$.
        Repeating the same construction above, we can find the diagonal sequence $\varrho_i\vcentcolon=\tilde{\varrho}_{\ls m_{i,\ls i}}$ which makes~\eqref{eq:tyBolzanoWeierstrass} true.
        We stress that $L_j$ and $\varrho_i$ are independent: although $L_j$ depends on $\tilde{\varrho}_m$, it is not related to the specific way in which $\varrho_i$ is extracted.

        \medskip
        
        \noindent Since the limit achieved along any subsequence is larger than the limit inferior, from~\eqref{eq:doNotRuinLiminf} we have
        \begin{align*}
            1&\leq\liminf_{\varrho\to\infty}\liminf_{L\to\infty}\!\sum_{\vec{n}\ls \in\,\Z^3}\!\left[\abs{\alpha(\vec{n})}^2\!+\abs*{\scalar{\mathrm{e}_{L;\,\vec{n}}}{\tfrac{\Phi_{\nls\varrho,\mspace{0.75mu} L}}{\norm{\Phi_{\nls\varrho,\mspace{0.75mu} L}}[2]\!\nls}\,}[2]-\alpha(\vec{n})}^2 \!\!+\right.\\
            &\mspace{180mu}\left.\vphantom{\abs*{\tfrac{|}{|}}^2}+2\ls \Re{\:\conjugate*{\alpha(\vec{n})}\nls \left(\nls  \scalar{\mathrm{e}_{L;\,\vec{n}}}{\tfrac{\Phi_{\nls\varrho,\mspace{0.75mu} L}}{\norm{\Phi_{\nls\varrho,\mspace{0.75mu} L}}[2]\!\nls}\,}[2]-\alpha(\vec{n})\!\right)}\! \right]\nls \mathfrak{m}_{\varrho,\ls L}(\mspace{-0.75mu}\{\nls \vec{n}\nls \}\mspace{-0.75mu})\\
            &\leq \limsup_{i\to\infty}\limsup_{j\to\infty}\!\sum_{\vec{n}\ls \in\,\Z^3}\!\left[\abs{\alpha(\vec{n})}^2\!+\abs*{\scalar{\mathrm{e}_{L_j;\,\vec{n}}}{\tfrac{\Phi_{\nls\varrho_i,\mspace{0.75mu} L_j}}{\norm{\Phi_{\nls\varrho_i,\mspace{0.75mu} L_j}}[2]\!\nls}\,}[2]-\alpha(\vec{n})}^2 \!\!+\right.\\
            &\mspace{180mu}\left.\vphantom{\abs*{\scalar{\mathrm{e}_{L_j;\,\vec{n}}}{\tfrac{\Phi_{\nls\varrho_i,\mspace{0.75mu} L_j}}{\norm{\Phi_{\nls\varrho_i,\mspace{0.75mu} L_j}}[2]\!\nls}\,}[2]-\alpha(\vec{n})}^2}+2\ls \Re{\:\conjugate*{\alpha(\vec{n})}\nls \left(\nls  \scalar{\mathrm{e}_{L_j;\,\vec{n}}}{\tfrac{\Phi_{\nls\varrho_i,\mspace{0.75mu} L_j}}{\norm{\Phi_{\nls\varrho_i,\mspace{0.75mu} L_j}}[2]\!\nls}\,}[2]-\alpha(\vec{n})\!\right)} \right]\nls \mathfrak{m}_{\varrho_i,\ls L_j}(\mspace{-0.75mu}\{\nls \vec{n}\nls \}\mspace{-0.75mu}).
        \end{align*}
        Similarly, the limit along a specific sequence can be estimated from above by the limit superior, so that 
        \begin{align*}
            1\leq&\,\limsup_{\varrho\to\infty}\limsup_{L\to\infty}\!\sum_{\vec{n}\ls \in\,\Z^3}\!\left[\abs*{\scalar{\mathrm{e}_{L;\,\vec{n}}}{\tfrac{\Phi_{\nls\varrho,\mspace{0.75mu} L}}{\norm{\Phi_{\nls\varrho,\mspace{0.75mu} L}}[2]\!\nls}\,}[2]-\alpha(\vec{n})}^2 \!\!+2\ls \Re{\:\conjugate*{\alpha(\vec{n})}\nls \left(\nls\scalar{\mathrm{e}_{L;\,\vec{n}}}{\tfrac{\Phi_{\nls\varrho,\mspace{0.75mu} L}}{\norm{\Phi_{\nls\varrho,\mspace{0.75mu} L}}[2]\!\nls}\,}[2]-\alpha(\vec{n})\!\right)}\! \right]\nls \mathfrak{m}_{\varrho,\ls L}(\mspace{-0.75mu}\{\nls \vec{n}\nls \}\mspace{-0.75mu})\,+\\
            &+\!\!\sum_{\substack{\vec{m}\ls \in\,\Z^3\ls :\\[1.5pt] \abs{\vec{m}}\leq M}}\limsup_{i\to\infty}\limsup_{j\to\infty}\,\abs{\alpha(\vec{m})}^2\,\mathfrak{m}_{\varrho_i,\ls L_j}(\mspace{-0.75mu}\{\nls \vec{m}\nls \}\mspace{-0.75mu})+\limsup_{i\to\infty}\limsup_{j\to\infty}\!\!\sum_{\substack{\vec{m}\ls \in\,\Z^3\ls :\\[1.5pt] \abs{\vec{m}}> M}} \abs{\alpha(\vec{m})}^2\,\mathfrak{m}_{\varrho_i,\ls L_j}(\mspace{-0.75mu}\{\nls \vec{m}\nls \}\mspace{-0.75mu})\\
            \leq & \,\lim_{\varrho\to\infty}\limsup_{L\to\infty}\left[\sum_{\vec{n}\ls \in\,\Z^3}\!\abs*{\scalar{\mathrm{e}_{L;\,\vec{n}}}{\tfrac{\Phi_{\nls\varrho,\mspace{0.75mu} L}}{\norm{\Phi_{\nls\varrho,\mspace{0.75mu} L}}[2]\!\nls}\,}[2]-\alpha(\vec{n})}^2 \!\!+2\ls \norm{\alpha(\vec{n})}[\ell_2(\Z^3)]\sqrt{\sum_{\vec{n}\ls \in\,\Z^3}\!\abs*{\scalar{\mathrm{e}_{L;\,\vec{n}}}{\tfrac{\Phi_{\nls\varrho,\mspace{0.75mu} L}}{\norm{\Phi_{\nls\varrho,\mspace{0.75mu} L}}[2]\!\nls}\,}[2]-\alpha(\vec{n})}^2\ls }\,\right]\!+\\
            &+\!\!\sum_{\substack{\vec{m}\ls \in\,\Z^3\ls :\\[1.5pt] \abs{\vec{m}}\leq M}}\abs{\alpha(\vec{m})}^2\lim_{i\to\infty}\lim_{j\to\infty}\mathfrak{m}_{\varrho_i,\ls L_j}(\mspace{-0.75mu}\{\nls \vec{m}\nls \}\mspace{-0.75mu})+\!\!\sum_{\substack{\vec{m}\ls \in\,\Z^3\ls :\\[1.5pt] \abs{\vec{m}}> M}} \abs{\alpha(\vec{m})}^2.
        \end{align*}        
        In the last inequalities, we have made use of the sub-additivity of the limit superior.
        Thus, because of~(\ref{eq:tailConvergenceFourierCoefficients},~\ref{eq:L2convergenceForFourierCoefficients}), we can take the limit $M\nls \to\nls \infty$ in both sides to obtain
        \begin{equation*}
            1\leq \!\nls \sum_{\vec{n}\ls \in\,\Z^3}\abs{\alpha(\vec{n})}^2\lim_{i\to\infty}\lim_{j\to\infty}\mathfrak{m}_{\varrho_i,\ls L_j}(\mspace{-0.75mu}\{\nls \vec{n}\nls \}\mspace{-0.75mu}).
        \end{equation*}
        Now, we can adopt a H\"older inequality to get
        \begin{equation*}
            1\leq\! \sum_{\vec{n}\ls \in\,\Z^3}\abs{\alpha(\vec{n})}^2\lim_{i\to\infty}\lim_{j\to\infty}\mathfrak{m}_{\varrho_i,\ls L_j}(\mspace{-0.75mu}\{\nls \vec{n}\nls \}\mspace{-0.75mu})\leq \norm*{\ls \abs{\alpha}^2}[\ell_1(\Z^3)]\norm*{\ls \lim_{i\to\infty}\lim_{j\to\infty}\mathfrak{m}_{\varrho_i,\ls L_j}(\mspace{-0.75mu}\{\vec{\cdot}\}\mspace{-0.75mu})}[\ell_\infty(\Z^3)]\!\!\leq 1.
        \end{equation*}
        In other words, the definition of a quasi-complete Bose--Einstein condensate has required the previous H\"older inequality to be saturated.
        For conjugate exponents $\{1,\infty\}$, this can be true if and only if the sequence in $\ell_\infty(\Z^3)$ attains its maximum (that is, the value $1$) for all $\vec{n}$ such that the sequence in $\ell_1(\Z^3)$ is non-zero at $\vec{n}$.
        This means that\footnote{Without extracting converging subsequences, we would have only obtained $\limsup\limits_{\varrho\to\infty}\limsup\limits_{L\to\infty}\, \mathfrak{m}_{\varrho,\ls L}(\mspace{-0.75mu}\{\nls \vec{k}_0\nls \}\mspace{-0.75mu})\nls =\nls 1$.
        Since $\mathfrak{m}_{\varrho,\ls L}$ is a probability measure, this implies $0\nls \leq\nls  \limsup\limits_{\varrho\to\infty}\limsup\limits_{L\to\infty} \,\mathfrak{m}_{\varrho,\ls L}(\mspace{-0.75mu}\{\nls \vec{n}\nls \}\mspace{-0.75mu})\nls \leq\nls  1-\liminf\limits_{\varrho\to\infty}\liminf\limits_{L\to\infty}\, \mathfrak{m}_{\varrho,\ls L}(\mspace{-0.75mu}\{\nls \vec{k}_0\nls \}\mspace{-0.75mu})$ for all $\vec{n}\neq\vec{k}_0$, which is not enough to conclude the concentration of the mass in a single mode, unless the iterated limit of $\mathfrak{m}_{\varrho,\ls L}(\{\vec{k}_0\})$ exists.}
        $$\exists* \vec{k}_0\nls \in\Z^3:\quad\lim_{i\to\infty}\lim_{j\to\infty}\mathfrak{m}_{\varrho_i,\ls L_j}(\mspace{-0.75mu}\{\vec{k}_0\}\mspace{-0.75mu})=1,\qquad\lim_{i\to\infty}\lim_{j\to\infty}\mathfrak{m}_{\varrho_i,\ls L_j}(\mspace{-0.75mu}\{\vec{n}\}\mspace{-0.75mu})=0,\qquad \forall \vec{n}\neq \vec{k}_0.$$
        Therefore, $\alpha$ must be supported only on $\{\vec{k}_0\}$ with absolute value equal to $1$ because of its normalisation, and it vanishes anywhere else.
        In other words,~\eqref{eq:becFourierCoefficients1} has been proven.

    \end{proof}
\end{prop}
\begin{note}\label{rmk:consequencesOfCOnvergence}
    By taking account of Assumptions~\ref{ass:initialBEC} and~\ref{ass:translationalInvariance}, equation~\eqref{eq:becFourierCoefficients1} can be strengthen:
    $$\lim_{\varrho\to\infty}\limsup_{L\to\infty}\ls \sum_{\vec{n}\ls \in\,\Z^3}\!\big(1+\tfrac{4\pi^2\abs{\vec{n}}^2\!}{L^2}\ls \big)\nls \abs*{\scalar{\mathrm{e}_{L;\,\vec{n}}}{\tfrac{\Phi_{\nls\varrho,\mspace{0.75mu} L}}{\norm{\Phi_{\nls\varrho,\mspace{0.75mu} L}}[2]\!\nls}\,}[2]\nls -e^{i\ls \vartheta}\delta_{\vec{n},\ls \vec{k}_0}}^2\!=0.$$
    In particular, this forces the kinetic energy per particle to vanish in the high-density thermodynamic limit:
    \begin{align*}
        \lim_{\varrho\to\infty}\limsup_{L\to\infty}\ls &\sum_{\vec{n}\ls \in\,\Z^3}\tfrac{4\pi^2\abs{\vec{n}}^2\!}{L^2}  \abs*{\scalar{\mathrm{e}_{L;\,\vec{n}}}{\tfrac{\Phi_{\nls\varrho,\mspace{0.75mu} L}}{\norm{\Phi_{\nls\varrho,\mspace{0.75mu} L}}[2]\!\nls}\,}[2]}^2\!=\\[-5pt]
        &=\lim_{\varrho\to\infty}\limsup_{L\to\infty}\Bigg[\tfrac{4\pi^2\abs{\vec{k}_0}^2\!}{L^2} \abs*{\scalar{\mathrm{e}_{L;\,\vec{k}_0}}{\tfrac{\Phi_{\nls\varrho,\mspace{0.75mu} L}}{\norm{\Phi_{\nls\varrho,\mspace{0.75mu} L}}[2]\!\nls}\,}[2]}^2\!+\!\!\sum_{\substack{\vec{n}\ls \in\,\Z^3\ls :\\[1.5pt] \vec{n}\neq \vec{k}_0}}\!\!\tfrac{4\pi^2\abs{\vec{n}}^2\!}{L^2} \abs*{\scalar{\mathrm{e}_{L;\,\vec{n}}}{\tfrac{\Phi_{\nls\varrho,\mspace{0.75mu} L}}{\norm{\Phi_{\nls\varrho,\mspace{0.75mu} L}}[2]\!\nls}\,}[2]}^2\!\Bigg]\\[-7.5pt]
        &\leq \lim_{\varrho\to\infty}\limsup_{L\to\infty}\Bigg[\tfrac{4\pi^2\abs{\vec{k}_0}^2\!}{L^2}+\!\!\sum_{\vec{n}\ls \in\,\Z^3}\!\!\tfrac{4\pi^2\abs{\vec{n}}^2\!}{L^2} \abs*{\scalar{\mathrm{e}_{L;\,\vec{n}}}{\tfrac{\Phi_{\nls\varrho,\mspace{0.75mu} L}}{\norm{\Phi_{\nls\varrho,\mspace{0.75mu} L}}[2]\!\nls}\,}[2]\nls -e^{i\ls \vartheta} \delta_{\vec{n},\ls \vec{k}_0}}^2\!\Bigg]\!=0.
    \end{align*}
    This is the correct expected behaviour, as \emph{complete} Bose--Einstein condensates have no kinetic energy.
\end{note}
Owing to Proposition~\ref{th:BECmomentum}, we are able to characterise in which cases the quasi-complete Bose--Einstein condensate localises on momentum states.
\begin{cor}\label{th:BECmomentumCharacterisation}
    Consider $\psi_{\varrho,\ls L}\!\in\nls \fdom{\mathcal{N}}\nls \smallsetminus\nls \ker{\mathcal{N}}$, with $\norm{\psi_{\varrho,\ls L}}\nls =\nls 1$ for all $\varrho,L\nls >\nls 0$, exhibiting quasi-complete condensation, and let $\Phi_{\!\varrho,\ls L}\!\in\nls H^1(\Lambda_L)$ be a quasi-complete Bose--Einstein condensate for $\psi_{\varrho,\ls L}\ls $ with a macroscopic counterpart $\Phi\in H^1(\Lambda_1)$.
    Then,
    $$\exists \vec{k}_0\in\Z^3,\ls\vartheta\in[0,2\pi)\, \text{ s.t. }\, \Phi\,:\:\vec{y}\longmapsto e^{2\pi \ls i\,\vec{k}_0\ls\cdot\,\vec{y}\,+\,i\ls\vartheta} \quad\iff\quad \lim_{\varrho\to\infty}\limsup_{L\to\infty}\norm*{\gamma^{(1)}_{\psi_{\varrho,\mspace{0.75mu} L}}\!\!-\langle\gamma^{(1)}_{\psi_{\varrho,\mspace{0.75mu} L}}\rangle}[\mathrm{HS}]\!=0,$$
    where the kernels of the integral trace-class operators $\gamma^{(1)}_{\psi_{\varrho,\mspace{0.75mu} L}}$ and $\langle\gamma^{(1)}_{\psi_{\varrho,\mspace{0.75mu} L}}\rangle$ are given by~\eqref{def:gamma1Kernel} and~\eqref{def:traslationInvariantProj}, respectively.
    \begin{proof}
        The $\Leftarrow\nls )$ direction of the proof has been entirely discussed in Proposition~\ref{th:BECmomentum}.\newline
        Concerning the $\Rightarrow\mspace{-2.25mu})$ direction, take into account Remark~\ref{rmk:likeFourierBasis} and the identity
        \begin{equation}
            \tr{\abs{A}^2}=\!\sum_{i,\, j\ls\in\ls\N} \abs{\scalar{e_i}{A\, e_j}[\hilbert*]}^2,\qquad \forall A\nls\in\nls\schatten{2}[\hilbert*],\, \{e_i\}_{i\ls\in\ls\N}\,\text{ orthonormal basis for }\,\hilbert*,
        \end{equation}
        to write
        \begin{equation}\label{eq:offDiagonalGamma}
            \norm*{\gamma^{(1)}_{\psi_{\varrho,\mspace{0.75mu} L}}\!\!-\langle\gamma^{(1)}_{\psi_{\varrho,\mspace{0.75mu} L}}\rangle}[\mathrm{HS}]^2=\!\!\!\sum_{\substack{\vec{n},\ls\vec{m}\ls\in\,\Z^3\\[-1.5pt] \vec{n}\neq\ls \vec{m}}}\! \abs*{\scalar{\mathrm{f}_{L;\,\vec{m}}}{\gamma^{(1)}_{\psi_{\varrho,\mspace{0.75mu} L}}\mathrm{f}_{L;\,\vec{n}}}}^2\!.
        \end{equation}
        Here, $\{\mathrm{f}_{L;\,\vec{n}}\}_{\vec{n}\ls\in\ls\Z^3}$ is the Fourier basis~\eqref{def:FourierBasis}.
        Plug the defining kernel of the one-particle reduced density matrix in~\eqref{eq:offDiagonalGamma} to obtain
        \begin{align*}
            \norm*{\gamma^{(1)}_{\psi_{\varrho,\mspace{0.75mu} L}}\!\!-\langle\gamma^{(1)}_{\psi_{\varrho,\mspace{0.75mu} L}}\rangle}[\mathrm{HS}]^2&=\!\!\!\sum_{\substack{\vec{n},\ls\vec{m}\ls\in\,\Z^3:\\[-1.5pt] \vec{n}\neq\ls \vec{m}}}\!\!\! \frac{\abs*{\scalar{a(\mathrm{f}_{L;\,\vec{n}})\ls\psi_{\varrho,\ls L}}{a(\mathrm{f}_{L;\,\vec{m}})\ls\psi_{\varrho,\ls L}}}^2\!\nls}{\norm{\mathcal{N}^{\frac{1}{2}}\psi_{\varrho, \ls L}}^4}\\
            &=2\!\sum_{\substack{\vec{n}\ls\in\,\Z^3:\\ \vec{n}\neq\ls \vec{k}_0}}\! \frac{\abs*{\scalar{a(\mathrm{f}_{L;\,\vec{n}})\ls\psi_{\varrho,\ls L}}{a(\mathrm{f}_{L;\,\vec{k}_0})\ls\psi_{\varrho,\ls L}}}^2\!\nls}{\norm{\mathcal{N}^{\frac{1}{2}}\psi_{\varrho, \ls L}}^4}\;+\nquad\sum_{\substack{\vec{n},\ls\vec{m}\ls\in\,\Z^3:\\  \vec{n}\neq\ls \vec{k}_0,\,\vec{m}\neq\ls \vec{k}_0,\\ \vec{n}\neq\ls \vec{m}}}\nquad \frac{\abs*{\scalar{a(\mathrm{f}_{L;\,\vec{n}})\ls\psi_{\varrho,\ls L}}{a(\mathrm{f}_{L;\,\vec{m}})\ls\psi_{\varrho,\ls L}}}^2\!\nls}{\norm{\mathcal{N}^{\frac{1}{2}}\psi_{\varrho, \ls L}}^4}.
        \end{align*}
        By exploiting~\eqref{eq:trivialBoundForCreationAnnihilation} and the Cauchy--Schwarz inequality,
        \begin{equation}\label{eq:vanishingOffDiagonalGamma}
            \norm*{\gamma^{(1)}_{\psi_{\varrho,\mspace{0.75mu} L}}\!\!-\langle\gamma^{(1)}_{\psi_{\varrho,\mspace{0.75mu} L}}\rangle}[\mathrm{HS}]^2\leq 2\!\sum_{\substack{\vec{n}\ls\in\,\Z^3:\\ \vec{n}\neq\ls \vec{k}_0}}\! \frac{\norm{a(\mathrm{f}_{L;\,\vec{n}})\ls\psi_{\varrho,\ls L}}^2\!\nls}{\norm{\mathcal{N}^{\frac{1}{2}}\psi_{\varrho, \ls L}}^2}\,+\,\Bigg(\sum_{\substack{\vec{m}\ls\in\,\Z^3:\\  \vec{m}\neq\ls \vec{k}_0}}\! \frac{\norm*{a(\mathrm{f}_{L;\,\vec{m}})\ls\psi_{\varrho,\ls L}}^2\!\nls}{\norm{\mathcal{N}^{\frac{1}{2}}\psi_{\varrho, \ls L}}^2}\:\Bigg)^{\! 2}\!.
        \end{equation}
        Since $\Phi_{\varrho,\ls L}$ is a quasi-complete Bose--Einstein condensate with macroscopic counterpart $\Phi$, one has
        \begin{align*}
            1-\frac{\norm{a(\mathrm{f}_{L;\,\vec{k}_0})\ls\psi_{\varrho,\ls L}}^2\!\nls}{\norm{\mathcal{N}^{\frac{1}{2}}\psi_{\varrho, \ls L}}^2\!\nls}=&\, 1-\frac{\norm*{a\Big( \tfrac{\Phi_{\nls\varrho,\ls L}}{\norm{\Phi_{\nls\varrho,\ls L}}[2] \!\nls}\Big)\ls \psi_{\varrho,\ls L}}^2\!\nls }{\norm{\mathcal{N}^{\frac{1}{2}}\psi_{\varrho,\ls L}}^2}-\frac{\norm*{a\Big(\mathrm{f}_{L;\,\vec{k}_0}- \tfrac{\Phi_{\nls\varrho,\ls L}}{\norm{\Phi_{\nls\varrho,\ls L}}[2] \!\nls}\Big)\ls \psi_{\varrho,\ls L}}^2\!\nls }{\norm{\mathcal{N}^{\frac{1}{2}}\psi_{\varrho,\ls L}}^2}\,+\\
            &-\frac{2\ls\Re\,\scalar{a\Big(\tfrac{\Phi_{\nls\varrho,\ls L}}{\norm{\Phi_{\nls\varrho,\ls L}}[2] \!\nls}\Big)\ls \psi_{\varrho,\ls L}}{a\Big(\mathrm{f}_{L;\,\vec{k}_0}- \tfrac{\Phi_{\nls\varrho,\ls L}}{\norm{\Phi_{\nls\varrho,\ls L}}[2] \!\nls}\Big)\ls \psi_{\varrho,\ls L}}}{\norm{\mathcal{N}^{\frac{1}{2}}\psi_{\varrho,\ls L}}^2}\\
            \leq &\,1-\frac{\norm*{a\Big( \tfrac{\Phi_{\nls\varrho,\ls L}}{\norm{\Phi_{\nls\varrho,\ls L}}[2] \!\nls}\Big)\ls \psi_{\varrho,\ls L}}^2\!\nls }{\norm{\mathcal{N}^{\frac{1}{2}}\psi_{\varrho,\ls L}}^2}+2\norm*{\ls\mathrm{f}_{L;\,\vec{k}_0}-\tfrac{\Phi_{\nls\varrho,\ls L}}{\norm{\Phi_{\nls\varrho,\ls L}}[2] \!\nls}\,}[2].
        \end{align*}
        Hence,
        $$\lim_{\varrho\to\infty}\limsup_{L\to\infty}\left[ 1-\frac{\norm{a(\mathrm{f}_{L;\,\vec{k}_0})\ls\psi_{\varrho,\ls L}}^2\!\nls}{\norm{\mathcal{N}^{\frac{1}{2}}\psi_{\varrho, \ls L}}^2\!\nls}\:\right]\!=0.$$
        By recalling~\eqref{eq:numberAsSumOfOccupations} one clearly has
        $$\sum_{\substack{\vec{n}\ls\in\,\Z^3:\\ \vec{n}\neq\ls \vec{k}_0}}\! \frac{\norm{a(\mathrm{f}_{L;\,\vec{n}})\ls\psi_{\varrho,\ls L}}^2\!\nls}{\norm{\mathcal{N}^{\frac{1}{2}}\psi_{\varrho, \ls L}}^2}=1-\frac{\norm{a(\mathrm{f}_{L;\,\vec{k}_0})\ls\psi_{\varrho,\ls L}}^2\!\nls}{\norm{\mathcal{N}^{\frac{1}{2}}\psi_{\varrho, \ls L}}^2\!\nls},$$
        and therefore~\eqref{eq:vanishingOffDiagonalGamma} implies that the distance between the (normalised) one-particle reduced density matrix and its translational invariant projection vanishes in the iterated limit.

    \end{proof}
\end{cor}
Next, we establish that quasi-canonical coherent states exhibit quasi-complete Bose--Einstein condensation.
\begin{prop}\label{th:eigenfunctionOfaMeansQBEC}
    Let $\phi_{\varrho,\ls L}\!\nls \in\nls \fdom{\mathcal{N}}\nls \smallsetminus\nls \ker{\mathcal{N}}$ be a quasi-eigenfunction of $a\Big(\tfrac{f_{\nls\varrho,\mspace{0.75mu} L}}{\norm{f_{\nls\varrho,\mspace{0.75mu} L}}[2]\!\nls}\Big)$, with $f_{\varrho,\ls L}\!\nls \in\!H^1(\Lambda_L)$ s.t.
    \begin{enumerate}[label=\roman*)]
        \item $f_{\varrho,\ls L}$ satisfies~\eqref{eq:becNormalization};
        \item $\exists f \nls \in\! H^1(\Lambda_1)$ for which $\lim\limits_{\varrho\to\infty}\limsup\limits_{L\to\infty}\,\norm*{\tfrac{f_{\nls\varrho,\mspace{0.75mu} L}}{\norm{f_{\nls\varrho,\mspace{0.75mu} L}}[2]\!\nls}-\tfrac{1}{L^{3/2}}f\big(\tfrac{\vec{\cdot}}{L}\big)}[H^1(\Lambda_L)]\!\!=0$.
    \end{enumerate}
    Moreover, the quasi-eigenvalue $z_{\varrho,\ls L}\!\in\C$ associated with $\phi_{\varrho,\ls L}$ satisfies
    $$\lim_{\varrho\to\infty}\limsup_{L\to\infty}\abs*{1-\frac{\abs{z_{\varrho,\ls L}}^2\,\norm{\phi_{\varrho,\ls L}}^2\!\nls }{\norm{\mathcal{N}^{\frac{1}{2}}\phi_{\varrho,\ls L}}^2}\,}=0.$$
    Then, $\phi_{\varrho,\ls L}$ exhibits quasi-complete condensation, and $f_{\varrho,\ls L}$ is a quasi-complete Bose--Einstein condensate for $\phi_{\varrho,\ls L}\ls $.
    \begin{note}
        The assumption on the quasi-eigenvalue of the annihilation operators ensures that its square behaves the same as the expectation of the number of particles in the state $\phi_{\varrho,\ls L}$.
        This is what happens for exact canonical coherent states, where Bose--Einstein condensation occurs.
    \end{note}
    \begin{proof}[Proof of Proposition~\ref{th:eigenfunctionOfaMeansQBEC}]
        One has
        
        \begin{equation*}
            \begin{split}
                \norm*{a\Big(\tfrac{f_{\nls\varrho,\mspace{0.75mu} L}}{\norm{f_{\nls\varrho,\mspace{0.75mu} L}}[2]\!\nls}\Big)\ls \phi_{\varrho,\ls L}}^2\!=\norm*{\left[a\Big(\tfrac{f_{\nls\varrho,\mspace{0.75mu} L}}{\norm{f_{\nls\varrho,\mspace{0.75mu} L}}[2]\!\nls}\Big)\nls -\nls z_{\varrho,\ls L}\right]\! \phi_{\varrho,\ls L}}^2\!+\abs{z_{\varrho,\ls L}}^2\ls \norm{\phi_{\varrho,\ls L}}^2+\\
                +2\ls \Re\left(\conjugate{z}_{\varrho,\ls L}\scalar{\phi_{\varrho,\ls L}}{\!\left[a\Big(\tfrac{f_{\nls\varrho,\mspace{0.75mu} L}}{\norm{f_{\nls\varrho,\mspace{0.75mu} L}}[2]\!\nls}\Big)\nls -\nls z_{\varrho,\ls L}\right]\! \phi_{\varrho,\ls L}}\right)\!.
            \end{split}
        \end{equation*}
        Thus,
        \begin{align*}
            \limsup_{\varrho\to\infty}\limsup_{L\to\infty} \left(\nls 1-\frac{\norm*{a\Big(\tfrac{f_{\nls\varrho,\mspace{0.75mu} L}}{\norm{f_{\nls\varrho,\mspace{0.75mu} L}}[2]\!\nls}\Big)\ls \phi_{\varrho,\ls L}}^2\!\nls }{\norm{\mathcal{N}^{1/2}\ls\phi_{\varrho,\ls L}}^2}\,\right)\nls \leq \lim_{\varrho\to\infty}\limsup_{L\to\infty}\abs*{1-\frac{\abs{z_{\varrho,\ls L}}^2\norm{\phi_{\varrho,\ls L}}^2\!}{\norm{\mathcal{N}^\frac{1}{2}\phi_{\varrho,\ls L}}^2}}+\\
            -\!\lim_{\varrho\to\infty}\liminf_{L\to\infty}\frac{\norm*{\left[a\Big(\tfrac{f_{\nls\varrho,\mspace{0.75mu} L}}{\norm{f_{\nls\varrho,\mspace{0.75mu} L}}[2]\!\nls}\Big)\nls -\nls z_{\varrho,\ls L}\right]\! \phi_{\varrho,\ls L}}^2\!}{\norm{\mathcal{N}^{1/2}\ls\phi_{\varrho,\ls L}}^2}\,+\\
            +2\limsup_{\varrho\to\infty}\limsup_{L\to\infty}\frac{\abs{z_{\varrho,\ls L}}\ls \norm{\phi_{\varrho,\ls L}}\,\norm*{\left[a\Big(\tfrac{f_{\nls\varrho,\mspace{0.75mu} L}}{\norm{f_{\nls\varrho,\mspace{0.75mu} L}}[2]\!\nls}\Big)\nls -\nls z_{\varrho,\ls L}\right]\! \phi_{\varrho,\ls L}}}{\norm{\mathcal{N}^{1/2}\ls\phi_{\varrho,\ls L}}^2}.
        \end{align*}
        By definition of quasi-eigenfunction,
        $$\lim_{\varrho\to\infty}\limsup_{L\to\infty} \frac{\norm*{\left[a\Big(\tfrac{f_{\nls\varrho,\mspace{0.75mu} L}}{\norm{f_{\nls\varrho,\mspace{0.75mu} L}}[2]\!\nls}\Big)\nls -z_{\varrho,\ls L}\right]\! \phi_{\varrho,\ls L}}}{\norm{\mathcal{N}^{1/2}\ls\phi_{\varrho,\ls L}}}=0,$$
        and therefore the quasi-eigenvalue must obey the bound
        $$\limsup_{\varrho\to\infty}\limsup_{L\to\infty} \frac{\abs{z_{\varrho,\ls L}}\,\norm{\phi_{\varrho,\ls L}}}{\norm{\mathcal{N}^{\frac{1}{2}}\phi_{\varrho,\ls L}}}\leq 1.$$
        Additionally, our assumption on $z_{\varrho,\ls L}$ implies that the limit in $\varrho$ exists in the previous equation and the equality sign holds (to be precise we ask that the same must be true for the lower limit in $L$, as well).
        Thus,
        $$\lim_{\varrho\to\infty}\limsup_{L\to\infty} \left(\nls 1-\frac{\norm*{a\Big(\tfrac{f_{\nls\varrho,\mspace{0.75mu} L}}{\norm{f_{\nls\varrho,\mspace{0.75mu} L}}[2]\!\nls}\Big)\ls \phi_{\varrho,\ls L}}^2\!\nls }{\norm{\mathcal{N}^{1/2}\ls\phi_{\varrho,\ls L}}^2}\,\right)\nls =0.$$

    \end{proof}
\end{prop}
\begin{note}\label{rmk:qcCoherentExhibitsqcBEC}
    By combining Propositions~\ref{th:WeylCoherent} and~\ref{th:eigenfunctionOfaMeansQBEC}, one obtains that $\mathcal{W}(\Psi_{\!\varrho,\ls L})\ls\xi_{\varrho,\ls L}$ exhibits a quasi-complete Bose--Einstein condensate provided that $\xi_{\varrho,\ls L}\!\in\nls\dom{\mathcal{N}}$ is a quasi-vacuum state with respect to the order parameter $\Psi_{\!\varrho,\ls L}\!\in\nls H^1(\Lambda_L)$ satisfying~\eqref{eq:becNormalization} and having macroscopic counterpart $\Psi\in H^1(\Lambda_1)$.
\end{note}

\noindent We conclude this section by quantifying the energy of a quasi-vacuum state under suitable assumptions.
\begin{prop}\label{th:energyOfVoid}
    Let $\Omega_{\varrho,\ls L}\!\in\nls \fdom{\mathcal{H}_{\varrho,\ls L}}$ be a quasi-vacuum state with respect to $\Phi_{\varrho,\ls L}\!\in\nls W^{2,\ls \infty}(\Lambda_L)$, satisfying~\eqref{eq:becNormalization} and $\lim\limits_{\varrho\to\infty}\limsup\limits_{L\to\infty} \,\norm*{\tfrac{\Phi_{\nls\varrho,\mspace{0.75mu} L}}{\norm{\Phi_{\nls\varrho,\mspace{0.75mu} L}}[2]\!\nls}\ls -\tfrac{1}{L^{3/2}}\ls \Phi\big(\frac{\vec{\cdot}}{L}\big)}[H^1(\Lambda_L)]\!\! =0$ for some $\Phi\nls \in\! H^1(\Lambda_1)$.\newline
    Moreover, assume that
    \begin{enumerate}[label=\roman*)]
        \item $\limsup\limits_{\varrho\to\infty}\limsup\limits_{L\to\infty}\frac{L^{3/2}\,\norm{\Phi_{\nls\varrho,\mspace{0.75mu} L}}[\infty]\!\nls }{\norm{\Phi_{\nls\varrho,\mspace{0.75mu} L}}[2]}\,<\infty$;
        \item $\limsup\limits_{\varrho\to\infty}\limsup\limits_{L\to\infty}\frac{L^{3/2}\,\norm{\Delta\Phi_{\nls\varrho,\mspace{0.75mu} L}}[\infty]\!\nls }{\norm{\Phi_{\nls\varrho,\mspace{0.75mu} L}}[2]}\,<\infty$;
        \item the quasi-canonical coherent state $\mathcal{W}(\Phi_{\varrho,\ls L})\ls \Omega_{\varrho,\ls L}\nls\in \fdom{\mathcal{H}_{\varrho,\ls L}}$ is energetically quasi-self-consistent.
    \end{enumerate}
    Then,
    $$\lim_{\varrho\to\infty}\limsup_{L\to\infty} \frac{1}{\varrho L^3} \,\mathcal{H}_{\varrho,\ls L}[\Omega_{\varrho,\ls L}]=0.$$
    \begin{proof}
        Considering the quantity
        $\mathcal{H}_{\varrho,\ls L}[\ls \mathcal{W}(\Phi_{\varrho,\ls L})\ls \Omega_{\varrho,\ls L}]$, we want to estimate the expectation of the energy of the quasi-vacuum state.
        This computation has already been carried out in~\eqref{eq:HamiltonianWeylSandwich}; therefore, taking into account Proposition~\ref{th:generatorTermsAPEstimates}
        \begin{equation*}
            \begin{split}
                \mathcal{H}_{\varrho,\ls L}[\Omega_{\varrho,\ls L}]\leq &\:\mathcal{H}[\ls \mathcal{W}(\Phi_{\varrho,\ls L})\ls \Omega_{\varrho,\ls L}]-\mathscr{E}_{\varrho,\ls L}[\Phi_{\varrho,\ls L}]+\tfrac{1}{2\ls \varrho}\mathcal{V}_{\nls L}[\Omega_{\varrho,\ls L}]+\tfrac{\,\norm{\Phi_{\nls\varrho,\mspace{0.75mu} L}}[2]^2\!}{2\ls\varrho}\, +\\
                &+\nls \left(\tfrac{7}{\varrho}\ls \norm{V_L\!\ast\nls \abs{\Phi_{\varrho,\ls L}}^2}[\infty]\nls +\tfrac{1}{2\ls \varrho}\ls \norm{V^{\ls 2}_L\nls \ast\nls \abs{\Phi_{\varrho,\ls L}}^2}[\infty]\right)\nls \mathbb{E}_{\,\Omega_{\varrho,\ls L}}[\ls \mathcal{N}\ls ]\ls +\\
                &+\left(\norm{\Delta\Phi_{\varrho,\ls L}}[\infty]\nls +\tfrac{1}{\varrho}\ls \norm{V_L\!\ast\nls \abs{\Phi_{\varrho,\ls L}}^2}[\infty]\norm{\Phi_{\varrho,\ls L}}[\infty]\right)\!\sqrt{L^3\,\mathbb{E}_{\,\Omega_{\varrho,\ls L}}[\ls \mathcal{N}\ls ]\ls }\ls .
            \end{split}
        \end{equation*}
        In the last row, we have estimated the linear contribution (in terms of the operator-valued distribution~\eqref{def:annihilationDistribution}) coming from~\eqref{eq:HamiltonianWeylSandwich}.
        Then, making use of Young's inequality for convolutions and recalling $\norm{V_L}[1]\nls \leq\nls  \FT{V}_\infty(\vec{0})=\mathfrak{b}$
        \begin{align*}
                \mathcal{H}_{\varrho,\ls L}[\Omega_{\varrho,\ls L}]\leq &\:\mathcal{H}[\ls \mathcal{W}(\Phi_{\varrho,\ls L})\ls \Omega_{\varrho,\ls L}]-\mathscr{E}_{\varrho,\ls L}[\Phi_{\varrho,\ls L}]+\tfrac{1}{2}\mathcal{H}_{\varrho,\ls L}[\Omega_{\varrho,\ls L}]+\tfrac{\,\norm{\Phi_{\nls\varrho,\mspace{0.75mu} L}}[2]^2\!}{2\ls\varrho}\, +\\
                &+\nls \left(7\ls \mathfrak{b}\ls \tfrac{\norm{\Phi_{\varrho,\ls L}}[\infty]^2\!\nls }{\varrho}\nls +\tfrac{\norm{V_L}[2]^2}{2}\,\tfrac{\norm{\Phi_{\varrho,\ls L}}[\infty]^2\!\nls }{\varrho}\right)\nls \mathbb{E}_{\,\Omega_{\varrho,\ls L}}[\ls \mathcal{N}\ls ]\ls +\\
                &+\left(\tfrac{\norm{\Delta\Phi_{\varrho,\ls L}}[\infty]\!\nls }{\sqrt{\varrho\ls }}\nls +\mathfrak{b}\ls \tfrac{\norm{\Phi_{\varrho,\ls L}}[\infty]^3\!\nls }{\varrho^{\ls 3/2}}\right)\!\sqrt{\varrho L^3\,\mathbb{E}_{\,\Omega_{\varrho,\ls L}}[\ls \mathcal{N}\ls ]\ls }\\[5pt]
                \implies \frac{1}{\varrho L^3}\ls \mathcal{H}_{\varrho,\ls L}[\Omega_{\varrho,\ls L}]\leq &\:\tfrac{2}{\varrho L^3\!\nls }\abs*{\mathcal{H}[\ls \mathcal{W}(\Phi_{\varrho,\ls L})\ls \Omega_{\varrho,\ls L}]-\mathscr{E}_{\varrho,\ls L}[\Phi_{\varrho,\ls L}]}+\tfrac{1}{\varrho}\tfrac{\,\norm{\Phi_{\nls\varrho,\mspace{0.75mu} L}}[2]^2\!}{\varrho L^3}\, +\\
                &+\nls \left(14\ls \mathfrak{b}+\norm{V_L}[2]^2\right)\nls \tfrac{L^3\,\norm{\Phi_{\varrho,\ls L}}[\infty]^2\!\nls }{\norm{\Phi_{\nls\varrho,\mspace{0.75mu} L}}[2]^2}\,\tfrac{\norm{\Phi_{\nls\varrho,\mspace{0.75mu} L}}[2]^4}{\vphantom{|^|}(\varrho L^3)^2}\,\tfrac{\mathbb{E}_{\,\Omega_{\varrho,\ls L}}[\ls \mathcal{N}\ls ]}{\vphantom{|^|}\norm{\Phi_{\nls\varrho,\mspace{0.75mu} L}}[2]^2}\ls +\\
                &+2\nls \left(\tfrac{L^{3/2}\,\norm{\Delta\Phi_{\varrho,\ls L}}[\infty]\!\nls }{\vphantom{|^|}\norm{\Phi_{\nls\varrho,\mspace{0.75mu} L}}[2]}\nls +\mathfrak{b}\ls \tfrac{L^{9/2}\,\norm{\Phi_{\varrho,\ls L}}[\infty]^3\!\nls }{\norm{\Phi_{\nls\varrho,\mspace{0.75mu} L}}[2]^3}\ls\tfrac{\norm{\Phi_{\nls\varrho,\mspace{0.75mu} L}}[2]^2}{\varrho L^3}\right)\!\tfrac{\norm{\Phi_{\nls\varrho,\mspace{0.75mu} L}}[2]^2}{\varrho L^3} \sqrt{\tfrac{\mathbb{E}_{\,\Omega_{\varrho,\ls L}}[\ls \mathcal{N}\ls ]}{\vphantom{|^|}\norm{\Phi_{\nls\varrho,\mspace{0.75mu} L}}[2]^2}\ls }\ls .
        \end{align*}
        Hypotheses~\textit{i) -- iii)} ensure that the r.h.s. of the previous equation vanishes in the iterated limit, since $\Omega_{\varrho,\ls L}$ is a quasi-vacuum state with respect to $\Phi_{\varrho,\ls L}$.
        
    \end{proof}
\end{prop}
\begin{note}\label{rmk:potentialL2}
    In the proof of the previous proposition, we took the boundedness of $\{\norm{V_L}[2]\}_{L>\ls 0}$ for granted.
    This is actually the case, since an immediate estimate yields
    $$\norm{V_L}[2]\leq \nls \sqrt{\norm{V_L}[1]\norm{V_L}[\infty]\!\nls }\:\leq\nls \sqrt{\,\mathfrak{b}\ls \norm{V_\infty}[\Lp{\infty}[\R^3]]\!\nls }\;+\oBig{L^{-3-\delta_1}}.$$
    This is enough already, but for the sake of completeness, this estimate can be refined as follows
    \begin{align*}
        \norm{V_L}[2]^2&=\norm{V_\infty}[\Lp{2}[\R^3]]^2+\!\! \sum_{\substack{\vec{k}\ls \in\,\Z^3 \ls :\\[1.5pt] \vec{k}\ls \neq\ls  \vec{0}}}\!\left[\integrate[\abs{\vec{y}}\ls \leq\ls \abs{\vec{k}}L/2]{V_\infty(\vec{y})\ls V_\infty(\vec{y}\nls -\nls \vec{k}L);\mspace{-52mu}d\vec{y}}+\integrate[\abs{\vec{y}}\ls >\ls \abs{\vec{k}}L/2]{V_\infty(\vec{y})\ls V_\infty(\vec{y}\nls -\nls \vec{k}L);\mspace{-52mu}d\vec{y}}\right]\\[-5pt]
        &\leq \norm{V_\infty}[\Lp{2}[\R^3]]^2+\!\! \sum_{\substack{\vec{k}\ls \in\,\Z^3 \ls :\\[1.5pt] \vec{k}\ls \neq\ls  \vec{0}}}\!\left[\integrate[\abs{\vec{y}}\ls \leq\ls \abs{\vec{k}}L/2]{V_\infty(\vec{y})\,\frac{C}{(1+\abs{\vec{k}L\nls -\nls \vec{y}})^{\ls 3\ls +\ls \delta_1}\nquad}\quad;\mspace{-52mu}d\vec{y}}+\integrate[\abs{\vec{y}}\ls >\ls \abs{\vec{k}}L/2]{V_\infty(\vec{y}\nls -\nls \vec{k}L)\,\frac{C}{(1+\abs{\vec{y}})^{\ls 3\ls +\ls \delta_1}\nquad}\quad;\mspace{-52mu}d\vec{y}}\right]\\[-2.5pt]
        &\leq\norm{V_\infty}[\Lp{2}[\R^3]]^2\nls +2\ls \mathfrak{b}\!\!\sum_{\substack{\vec{k}\ls \in\,\Z^3 \ls :\\[1.5pt] \vec{k}\ls \neq\ls  \vec{0}}}\frac{C}{(1+\abs{\vec{k}}L/2)^{\ls 3\ls +\ls \delta_1}\nquad}\quad\leq \norm{V_\infty}[\Lp{2}[\R^3]]^2\nls +\oBig{L^{-3-\delta_1}}.
    \end{align*}
\end{note}


\section{The Hartree Equation on the Torus}\label{sec:hartreeTorus}
In this section, we study in detail the properties of the time evolution on the three-dimensional torus driven by the Hartree equation~\eqref{eq:HartreePDE}.
Specifically, we first investigate the well-posedness in a suitable Banach space; then we provide the formulation of~\eqref{eq:HartreePDE} in terms of the momentum distribution, and finally, we propagate Assumption~\ref{ass:initialBEC} on the initial data for positive times.

\subsection{Global Well-Posedness}\label{sec:HartreeWellPosedness}

We introduce the \emph{polynomially-weighted Wiener algebra} for $r\geq 0$
\begin{equation}\label{def:Wiener}
    \begin{split}
    \mathfrak{A}^r(\Lambda_L)\vcentcolon=\Big\{\Phi\in\Lp{\infty}[\Lambda_L] \,\Big|\: \norm{\Phi}[\mathfrak{A}^r(\Lambda_L)]\nls \vcentcolon=\!\!\sum_{\vec{m}\ls \in\,\Z^3}\!\nls \big(1\nls +\mspace{-0.75mu}\tfrac{2^r\nls \pi^r\!}{L^r}\,\abs{\vec{m}}^r\big)\ls \abs{\FT{\Phi}_{\vec{m}}}<\infty\Big\},\\
    \text{where}\qquad\FT{\Phi}_{\vec{m}}=\frac{1}{L^3}\!\integrate[\Lambda_L]{e^{-i\ls \frac{2\pi}{L} \vec{x}\,\cdot\ls \vec{m}}\,\Phi(\vec{x}); \!d\vec{x}}.
    \end{split}
\end{equation}
Clearly, $\big(\mathfrak{A}^r(\Lambda_L), \norm{\ls \cdot\ls }[\mathfrak{A}^r(\Lambda_L)]\big)$ is a Banach space because it is isomorphic to $\ell_1(\Z^3)$, which is complete.
\begin{note}\label{rmk:derivativeLowerWienerOrder}
    For any three-dimensional multi-index $\alpha$ and $\Phi\nls \in\nls \mathfrak{A}^{|\alpha|}(\Lambda_L)$, one has
    $$\norm{\partial^\alpha\Phi}[\mathfrak{A}^0(\Lambda_L)]\nls \leq \norm{\Phi}[\mathfrak{A}^{|\alpha|}(\Lambda_L)].$$
    In particular, this implies that $\mathfrak{A}^r(\Lambda_L)\subset C^{\lfloor r\rfloor}(\Lambda_L)$, since the Wiener algebra $\mathfrak{A}^0(\Lambda_L)$ is a subset of $C^{\ls 0}(\Lambda_L)$, given the uniform convergence in $\Lambda_L$ of the Fourier series.
    In general, there holds the following chain of continuous embeddings
    $$H^{\mspace{0.75mu} r\ls +\frac{3}{2}+\ls\epsilon}(\Lambda_L)\hookrightarrow \mathfrak{A}^r(\Lambda_L)\hookrightarrow W^{r,\ls \infty}(\Lambda_L),\qquad \forall \epsilon>0.$$
    One could refer to the space $\mathfrak{A}^r(\Lambda_L)$ as the \emph{Wiener--Sobolev algebra} of order $r$, as a function $\Phi$ belongs to this space when $\Phi$ and its first $r$ derivatives lie in the Wiener algebra $\mathfrak{A}^0(\Lambda_L)$.
\end{note}
In the following, given $\Psi_{\!\varrho,\ls L}\!\in\nls  \mathfrak{A}^2(\Lambda_L)\nls \subset\nls  C^2(\Lambda_L)$, we aim to prove that there exists a unique function
$t\longmapsto\Psi^{\ls t}_{\!\varrho,\ls L}\!\in C^{\mspace{0.75mu}1}\big(\mspace{0.75mu}[0,\infty),\ls \mathfrak{A}^0(\Lambda_L)\mspace{-0.75mu}\big)\mspace{-0.75mu}\cap C^{\ls 0}\big(\mspace{0.75mu}[0,\infty),\ls  \mathfrak{A}^2(\Lambda_L)\mspace{-0.75mu}\big)$ solving the Hartree equation~\eqref{eq:HartreePDE} on the torus $\Lambda_L$ with initial datum $\Psi^{\ls 0}_{\!\varrho,\ls L}\!=\Psi_{\!\varrho,\ls L}$.
The overarching strategy -- relying on the Banach fixed point theorem -- is standard (we refer, for instance, to~\cite[Chapter 3]{Tao06}).
On the other hand, none of the Strichartz estimates typically used to probe well-posedness work in our setting, as no $\Lp{p}$-norm can be adopted but $\Lp{\infty}$, because of the present scaling.
Nevertheless, although the strategy may not be new, the regularity conditions on the pairwise interaction are non-trivial and may strongly depend on the Banach space in which the well-posedness is investigated.
Therefore, in the absence\footnote{The most closely related work is~\cite{CaMo14}, which proves the global well-posedness of the Hartree equation in the whole space for square integrable functions whose Fourier transform is integrable in $\R^3$.} of a precise reference addressing the specifics of our problem, we present a complete presentation of the argument to ensure the exposition is self-contained.
\begin{note}\label{rmk:finiteLifeSpan}
    By Assumptions~\ref{ass:initialBEC},~\ref{ass:tailCondition}, and~\ref{ass:kineticTailCondition}, we have (see equation~\eqref{eq:SumLowerBound} below and Remarks~\ref{rmk:S0Boundedness} and~\ref{rmk:T0Boundedness})
    $$\tfrac{1}{\!\nls \sqrt{\varrho\ls }\ls } \norm{\Psi_{\!\varrho,\ls L}}[\mathfrak{A}^2(\Lambda_L)]\geq \tfrac{\norm{\Psi_{\mspace{-3.75mu}\varrho,\mspace{0.75mu} L}}[2]\nls}{\!\nls\sqrt{\varrho L^3\ls}\ls}, \qquad \limsup_{\varrho\to\infty}\limsup_{L\to\infty} \, \tfrac{1}{\!\nls \sqrt{\varrho\ls }\ls }\norm{\Psi_{\!\varrho,\ls L}}[\mathfrak{A}^2(\Lambda_L)]<\infty.$$
\end{note}
We first prove that the Banach space $\mathfrak{A}^r(\Lambda_L)$ endowed with the pointwise product is a Banach algebra.
\begin{prop}\label{th:WienerBanachAlgebra}
    Let $r\nls \geq\nls 0$ and consider the Banach space $\mathfrak{A}^r(\Lambda_L)$ defined by~\eqref{def:Wiener}.
    Then, there exists $c_r\!>\nls 0$ such that
    $$\norm{\Phi \Psi}[\mathfrak{A}^r(\Lambda_L)]\leq c_r \norm{\Phi}[\mathfrak{A}^r(\Lambda_L)]\norm{\Psi}[\mathfrak{A}^r(\Lambda_L)],\qquad \forall\Phi,\Psi\in\mathfrak{A}^r(\Lambda_L).$$
    Specifically, $c_0=\frac{1}{2}$ and $c_2=\frac{4}{3}$.
    \begin{proof}
        Consider the elementary inequality
        \begin{equation}\label{eq:WienerAlgebraConstant}
            1+\big(\tfrac{2\pi}{L}\abs{\vec{m}}\big)^r\!\leq c_r \big(1+\tfrac{2^r\pi^r\!}{L^r}\,\abs{\vec{n}}^r\big)\big(1+\tfrac{2^r\pi^r\!}{L^r}\,\abs{\vec{m}\nls -\nls \vec{n}}^r\big),\qquad \forall \vec{n},\vec{m}\in\Z^3.
        \end{equation}
        In particular, one has the sharp constant
        $$c_r=\max_{a,\ls b\,\in\ls [0,\ls \infty)}\frac{1+(a+b)^r}{(1+a^r)(1+b^r)},\hide{=\max\{1,\frac{4^{r-1}}{2^r-1}\} }$$
        which is equal to $\frac{4}{3}$ in case $r=2$, attained at $a=b=\frac{1}{\!\nls \sqrt{2\ls }\ls }$.
        By explicit computation,
        \begin{align*}
            \norm{\Phi \Psi}[\mathfrak{A}^r(\Lambda_L)]&=\! \sum_{\vec{m}\ls \in\,\Z^3} \!\big(1+\tfrac{2^r\pi^r}{L^r}\abs{\vec{m}}^r\big)\,\bigg\lvert\nls \sum_{\vec{n}\ls \in\,\Z^3} \!\FT{\Phi}_{\vec{n}}\,\FT{\Psi}_{\vec{m}-\vec{n}}\ls \bigg\rvert\\
            &\leq \!\!\!\nls \sum_{\vec{n},\ls \vec{m}\ls \in\,\Z^3}\!\!\!\big(1+\tfrac{2^r\pi^r}{L^r}\abs{\vec{m}}^r\big)\,\abs{\FT{\Phi}_{\vec{n}}}\,\abs{\FT{\Psi}_{\vec{m}-\vec{n}}},
        \end{align*}
        and the result follows by~\eqref{eq:WienerAlgebraConstant}.

    \end{proof}
\end{prop}
Given $\mathrm{T}\nls >\nls 0$, we define the Banach space
\begin{equation}\label{def:BanachSpaceOfMaps}
    \X\vcentcolon= \big\{\Phi: t\longmapsto \Phi^t\in C^{\ls 0}\big(\mspace{0.75mu}[0,\mathrm{T}],\ls  \mathfrak{A}^2(\Lambda_L)\mspace{-0.75mu}\big)\nls \big\},\qquad\text{with }\,\norm{\Phi}[\X] = \!\sup_{s\,\in\ls  [0,\ls \mathrm{T}]} \norm{\Phi^{\ls s}}[\mathfrak{A}^2(\Lambda_L)].
\end{equation}
Notice that $\Phi\in\X$ implies $\norm{\Phi}[\X]<\infty$ by continuity (whereas the converse is obviously false in general).\newline
The Duhamel formula associated with the Hartree equation in $\Lambda_L$ is
\begin{equation}\label{eq:DuhamelHartree}
    \Phi^{\ls t}=e^{i\ls t\ls  \Delta}\ls \Phi^{\ls 0}-\ls \frac{i}{\varrho}\integrate[0;t]{e^{i\ls (t-s)\ls \Delta} \big(V_L\!\ast\nls \abs{\Phi^{\ls s}}^2\big)\ls \Phi^{\ls s}; \,ds}.
\end{equation}
We observe that the map $F_{\nls \varrho,\ls L}(\Phi):\,t\:\longmapsto F_{\nls \varrho,\ls L}(\Phi)^{\ls t}\vcentcolon=\frac{1}{\varrho}\big(V_L\!\ast\nls \abs{\Phi^{\ls t}}^2\big)\Phi^{\ls t}$ is in $\X$, provided $\Phi\in\X$:
\begin{align*}
    \norm{F_{\nls \varrho,\ls L}(\Phi)}[\X]&=\tfrac{1}{\varrho}\!\sup_{s\,\in\ls  [0,\ls \mathrm{T}]} \norm{\big(V_L\!\ast\nls \abs{\Phi^{\ls s}}^2\big)\Phi^{\ls s}}[\mathfrak{A}^2(\Lambda_L)]\leq \tfrac{4}{3\ls \varrho}\!\sup_{s\,\in\ls  [0,\ls \mathrm{T}]} \norm{V_L\!\ast\nls \abs{\Phi^{\ls s}}^2}[\mathfrak{A}^2(\Lambda_L)]\norm{\Phi}[\X]\\
    &\leq \tfrac{4}{3\ls \varrho}\ls \FT{V}_\infty(\vec{0})\!\sup_{s\,\in\ls  [0,\ls \mathrm{T}]} \norm{\ls \abs{\Phi^{\ls s}}^2}[\mathfrak{A}^2(\Lambda_L)]\norm{\Phi}[\X]\leq\tfrac{16}{9\ls \varrho}\ls \mathfrak{b}\,\norm{\Phi}[\X]^3.
\end{align*}
In the last row, we exploited the identity
$$\frac{1}{L^3}\!\integrate[\Lambda_L]{e^{-i\ls \frac{2\pi}{L} \vec{x}\,\cdot\ls \vec{m}}\,\big(V_L\!\ast \mspace{-0.75mu}\Phi\big)(\vec{x});\!d\vec{x}}=\FT{V}_\infty\nls \big(\tfrac{2\pi}{L}\ls \vec{m}\big)\ls \FT{\Phi}_{\vec{m}},\qquad\forall \Phi\in\mathfrak{A}^0(\Lambda_L),$$
combined with the estimate $\norm{\FT{V}_\infty}[\Lp{\infty}(\R^3)]\nls \leq \norm{V_\infty}[\Lp{1}[\R^3]]=\mathfrak{b}$.\newline
Analogously, the map
\begin{equation}\label{def:contractionMap}
    G_{\nls \varrho,\ls L}(\Phi):\,t\:\longmapsto G_{\nls \varrho,\ls L}(\Phi)^{\ls t}\vcentcolon=e^{i\ls t\ls  \Delta}\ls \Phi^{\ls 0}-i\!\integrate[0;t]{e^{i\ls (t-s)\ls \Delta} F_{\nls \varrho,\ls L}(\Phi)^{\ls s}; \,ds}
\end{equation}
is in $\X$ as well, since
\begin{align*}
    \norm{G_{\nls \varrho,\ls L}(\Phi)}[\X]&\leq \norm{\Phi^{\ls 0}}[\mathfrak{A}^2(\Lambda_L)]+\sup_{t\ls \in\ls [0,\ls \mathrm{T}]}\norm*{\integrate[0;t]{e^{i\ls (t-s)\ls \Delta} F_{\nls \varrho,\ls L}(\Phi)^{\ls s};\,ds}\,}[\mathfrak{A}^2(\Lambda_L)]\\
    &\leq \norm{\Phi^{\ls 0}}[\mathfrak{A}^2(\Lambda_L)]+\sup_{t\ls \in\ls [0,\ls \mathrm{T}]}\integrate[0;t]{\norm{F_{\nls \varrho,\ls L}(\Phi)^{\ls s}}[\mathfrak{A}^2(\Lambda_L)];\,ds}.
\end{align*}
Indeed, $e^{i \ls t\ls \Delta}$ is an isometry on $\mathfrak{A}^r(\Lambda_L)$ for all $r\nls \geq\nls  0$ and $t\in\R$.
Thus,
\begin{equation}\label{eq:contractionMapControl}
    \norm{G_{\nls \varrho,\ls L}(\Phi)}[\X]\leq \norm{\Phi^{\ls 0}}[\mathfrak{A}^2(\Lambda_L)]+\tfrac{16\ls \mathfrak{b}}{9\ls \varrho}\,\mathrm{T}\,\norm{\Phi}[\X]^3.
\end{equation}
Because of~\eqref{eq:DuhamelHartree}, we aim to prove that $G_{\nls \varrho,\ls L}(\Phi)=\Phi$ admits a unique solution in the set $\{\Phi\nls \in\nls \X \,|\: \Phi^{\ls 0}\nls =\Psi_{\!\varrho,\ls L}\}$.
To this end, we start by making use of the Banach fixed point theorem on a suitable complete metric space (with respect to the distance induced by the norm $\norm{\ls \cdot\ls }[\X]$) contained in $\X$.
This requires showing that $\maps{G_{\nls \varrho,\ls L}}{\X;\X}$ is a strict contraction on that space.
This is the content of the following Lemma.
\begin{lemma}\label{th:contractiveMap}
    Let $\X$ be defined by~\eqref{def:BanachSpaceOfMaps} and $$B_{\ls \Phi_0}(R)\nls \vcentcolon=\nls \{\Phi\nls \in\nls \X\,|\:\Phi^{\ls 0}\!=\nls \Phi_0,\, \norm{\Phi}[\X]\!\leq\nls R\},\qquad R\nls >\nls 0,\,\Phi_0\!\in\nls \mathfrak{A}^2(\Lambda_L).$$
    Then, the map $\maps{G_{\nls \varrho,\ls L}}{B_{\ls \Phi_0}\big(\tau\ls \norm{\Phi_0}[\mathfrak{A}^2(\Lambda_L)]\big);B_{\ls \Phi_0}\big(\tau\ls \norm{\Phi_0}[\mathfrak{A}^2(\Lambda_L)]\big)}$ defined by~\eqref{def:contractionMap} is a strict contraction for all $\tau\nls >\nls 1$ if
    \begin{equation}
        \mathrm{T}<\frac{3\min\{3\ls (\tau\nls -\nls 1),\tau\}}{16\ls \mathfrak{b}\,\tau^3}\,\frac{\varrho}{\:\norm{\Phi_0}[\mathfrak{A}^2(\Lambda_L)]^2\nqquad}.
    \end{equation}
    \begin{proof}
        First, we check that $G_{\nls \varrho,\ls L}(\Phi)$ is in $B_{\ls \Phi_0}\big(\tau\ls \norm{\Phi_0}[\mathfrak{A}^2(\Lambda_L)]\big)$, provided $\Phi\nls \in\! B_{\ls \Phi_0}\big(\tau\ls \norm{\Phi_0}[\mathfrak{A}^2(\Lambda_L)]\big)$.
        Because of~\eqref{eq:contractionMapControl}, 
        $$\mathrm{T}\leq \frac{9\ls (\tau\nls -\nls 1)}{16\ls \mathfrak{b}\,\tau^3}\,\frac{\varrho}{\:\norm{\Phi_0}[\mathfrak{A}^2(\Lambda_L)]^2\nqquad}\quad\implies\quad \norm{G_{\nls \varrho,\ls L}(\Phi)}[\X]\leq \tau\ls \norm{\Phi_0}[\mathfrak{A}^2(\Lambda_L)].$$
        Next, given $\Phi_1,\Phi_2\!\in\nls B_{\ls \Phi_0}\big(\tau\ls \norm{\Phi_0}[\mathfrak{A}^2(\Lambda_L)]\big)$, we have
        \begin{equation*}
            F_{\nls \varrho,\ls L}(\Phi_1)-F_{\nls \varrho,\ls L}(\Phi_2)= \tfrac{1}{\varrho}\big(V_L\!\ast\nls (\abs{\Phi_1}^2\nls -\abs{\Phi_2}^2)\big)\ls \Phi_1+\tfrac{1}{\varrho}\big(V_L\!\ast\nls \abs{\Phi_2}^2\big)\ls (\Phi_1\nls -\Phi_2).
        \end{equation*}
        Therefore,
        \begin{align*}
            \norm{F_{\nls \varrho,\ls L}(\Phi_1)-F_{\nls \varrho,\ls L}(\Phi_2)}[\X]&\leq\tfrac{1}{\varrho}\norm*{\big(V_L\!\ast\nls (\abs{\Phi_1}^2\nls -\abs{\Phi_2}^2)\big)\ls \Phi_1}[\X]+\tfrac{1}{\varrho}\norm*{\big(V_L\!\ast\nls \abs{\Phi_2}^2\big)\ls (\Phi_1\nls -\Phi_2)}[\X]\\
            &\leq \tfrac{4\ls \mathfrak{b}}{3\ls \varrho}\left(\norm{\ls \abs{\Phi_1}^2\nls -\abs{\Phi_2}^2}[\X]\,\norm{\Phi_1}[\X]+\norm{\ls \abs{\Phi_2}^2}[\X]\,\norm{\Phi_1\nls -\Phi_2}[\X]\right)\\
            &\leq\tfrac{16\ls \mathfrak{b}}{9\ls \varrho}\,\norm{\Phi_1\nls -\Phi_2}[\X]\left(\norm{\Phi_1}[\X]^2+\norm{\Phi_1}[\X]\,\norm{\Phi_2}[\X]+\norm{\Phi_2}[\X]^2\right)\!,
        \end{align*}
        since $\abs{w}^2\nls -\abs{z}^2=(w-z)\conjugate{w}+\conjugate*{(w-z)}z\,$ for all $w,z\in\C$.
        Hence, we have obtained
        $$\norm{F_{\nls \varrho,\ls L}(\Phi_1)-F_{\nls \varrho,\ls L}(\Phi_2)}[\X]\leq \tfrac{16\ls \mathfrak{b}}{3\ls \varrho}\,\tau^2\,\norm{\Phi_0}[\mathfrak{A}^2(\Lambda_L)]^2\,\norm{\Phi_1\nls -\Phi_2}[\X].$$
        Consequently,
        \begin{align*}
            \norm{G_{\nls \varrho,\ls L}(\Phi_1)-G_{\nls \varrho,\ls L}(\Phi_2)}[\X]&\leq \sup_{t\ls \in\ls [0,\ls \mathrm{T}]}\integrate[0;t]{\norm{F_{\nls \varrho,\ls L}(\Phi_1)^{\ls s}\!-F_{\nls \varrho,\ls L}(\Phi_2)^{\ls s}}[\mathfrak{A}^2(\Lambda_L)];\,ds}\\
            &\leq\tfrac{16\ls \mathfrak{b}}{3\ls \varrho}\,\tau^2\,\mathrm{T}\,\norm{\Phi_0}[\mathfrak{A}^2(\Lambda_L)]^2\,\norm{\Phi_1\nls -\Phi_2}[\X].
        \end{align*}
        In order for $G_{\nls \varrho,\ls L}$ to be a strict contraction, we must require
        $$\mathrm{T}<\frac{3}{16\ls \mathfrak{b}\,\tau^2}\,\frac{\varrho}{\:\norm{\Phi_0}[\mathfrak{A}^2(\Lambda_L)]^2\nqquad},$$
        and the result follows.
        
    \end{proof}
\end{lemma}
By means of Lemma~\ref{th:contractiveMap}, we can apply the Banach fixed point theorem, since the set $B_{\Phi_0}(R)$ is a complete metric space with respect to the distance induced by the norm $\norm{\ls \cdot\ls }[\X]$.
Therefore, we have obtained that the initial value problem for the Hartree equation on the torus has a unique solution $\Psi^{\ls t}_{\!\varrho,\ls L}$, for instance, in the space\footnote{This specific choice of the value of $\tau$ in Lemma~\ref{th:contractiveMap} maximises the lifespan of the solution.} $\big\{\Phi\nls \in\nls \X\,|\:\Phi^{\ls 0}\nls =\Psi_{\!\varrho,\ls L},\, \norm{\Phi}[\X]\!\leq\nls \tfrac{3}{2}\norm{\Psi_{\!\varrho,\ls L}}[\mathfrak{A}^2(\Lambda_L)]\big\}$, with
$$0\leq t\leq\mathrm{T}<\frac{\varrho}{12\ls \mathfrak{b}\,\norm{\Psi_{\!\varrho,\ls L}}[\mathfrak{A}^2(\Lambda_L)]^2\nqquad}.$$
By Remark~\ref{rmk:finiteLifeSpan}, the lifespans of each solution in the family $\{\Psi^{\ls t}_{\!\varrho,\ls L}\}_{\varrho, \,L \ls >\ls  0}$ eventually share in common all the intervals $[0,T]$ in the iterated limit, with $T\nls < T_\ast$ given by
\begin{equation}\label{eq:liminfLifeSpan}
    T_\ast=\frac{1}{12\ls \mathfrak{b}}\ls \frac{1}{\big(\limsup\limits_{\varrho\to\infty}\limsup\limits_{L\to\infty} \ls \frac{1}{\!\nls \sqrt{\varrho\ls }\ls }\norm{\Psi^{\ls 0}_{\!\varrho,\ls L}}[\mathfrak{A}^2(\Lambda_L)]\big)^{\nls 2\vphantom{|}}}.
\end{equation}
Next, we must ensure that there still exists a unique solution in the space $\{\Phi\in \X \,|\:\Phi^{\ls 0}=\Psi_{\!\varrho,\ls L}\}$ and that there are no other solutions to the Hartree equation with a larger $\X$-norm.
\begin{cor}\label{th:LWP}
    There exists a unique solution to the Hartree equation~\eqref{eq:HartreePDE}
    $$t\longmapsto\Psi^{\ls t}_{\!\varrho,\ls L}\!\in C^{\mspace{0.75mu}1}\big(\mspace{0.75mu}[0,\mathrm{T}],\ls \mathfrak{A}^0(\Lambda_L)\mspace{-0.75mu}\big)\mspace{-0.75mu}\cap C^{\ls 0}\big(\mspace{0.75mu}[0,\mathrm{T}],\ls  \mathfrak{A}^2(\Lambda_L)\mspace{-0.75mu}\big)$$
    for some $\mathrm{T}\nls >\nls 0$, with initial datum $\Psi^{\ls 0}_{\!\varrho,\ls L}\nls =\Psi_{\!\varrho,\ls L}\!\in\nls \mathfrak{A}^2(\Lambda_L)$.
    \begin{proof}
        The existence of $\Psi^{\ls t}_{\!\varrho,\ls L}$ in $C^{\ls 0}\big(\mspace{0.75mu}[0,\mathrm{T}],\ls  \mathfrak{A}^2(\Lambda_L)\mspace{-0.75mu}\big)$ has already been provided by the Banach fixed point theorem for $\mathrm{T}$ small enough.
        By contradiction, suppose there are two solutions $\Psi_{\!\varrho,\ls L},\Phi_{\!\varrho,\ls L}\!\in \nls C^{\ls 0}\big(\mspace{0.75mu}[0,\mathrm{T}],\ls  \mathfrak{A}^2(\Lambda_L)\mspace{-0.75mu}\big)$ with the same initial datum $\Psi^{\ls 0}_{\!\varrho,\ls L}=\Phi^{\ls 0}_{\!\varrho,\ls L}=\Psi_{\!\varrho,\ls L}$.
        By means of the Duhamel formula~\eqref{eq:DuhamelHartree},
        \begin{align*}
            \norm{\Phi^{\ls t}_{\!\varrho,\ls L}\nls -\Psi^{\ls t}_{\!\varrho,\ls L}}[\mathfrak{A}^2(\Lambda_L)]&\!\leq \integrate[0;t]{\norm{F(\Phi_{\!\varrho,\ls L})^{\ls s}\!-F(\Psi_{\!\varrho,\ls L})^{\ls s}}[\mathfrak{A}^2(\Lambda_L)];\,ds}\\
            &\leq \tfrac{16\ls \mathfrak{b}}{9\ls \varrho} \!\sup_{s\ls \in\ls [0,\ls \mathrm{T}]}\!\big(\norm{\Phi^{\ls s}_{\!\varrho,\ls L}}[\mathfrak{A}^2(\Lambda_L)]^2\!+\norm{\Phi^{\ls s}_{\!\varrho,\ls L}}[\mathfrak{A}^2(\Lambda_L)]\norm{\Psi^{\ls s}_{\!\varrho,\ls L}}[\mathfrak{A}^2(\Lambda_L)]\!+\norm{\Psi^{\ls s}_{\!\varrho,\ls L}}[\mathfrak{A}^2(\Lambda_L)]^2\big)\times\\[-10pt]
            &\mspace{345mu}\times\!\nls \integrate[0;t]{\norm{\Phi^{\ls s}_{\!\varrho,\ls L}\nls -\Psi^{\ls s}_{\!\varrho,\ls L}}[\mathfrak{A}^2(\Lambda_L)];\,ds},
        \end{align*}
        and Gr\"onwall's lemma yields $\Phi^{\ls t}_{\!\varrho,\ls L}=\Psi^{\ls t}_{\!\varrho,\ls L}$ for all $t\in[0,\mathrm{T}]$.\newline
        Ultimately, due to Remark~\ref{rmk:derivativeLowerWienerOrder} and the control of the nonlinearity $\norm{F_{\nls \varrho,\ls L}(\Phi)^{\ls t}}[\mathfrak{A}^0(\Lambda_L)]\!\leq\nls  \tfrac{16\ls \mathfrak{b}}{9\ls \varrho}\ls \norm{\Phi^{\ls t}}[\mathfrak{A}^2(\Lambda_L)]^3\!<\infty$, we have $\partial_t\Psi^{\ls t}_{\!\varrho,\ls L}\!\in\nls \mathfrak{A}^0(\Lambda_L)$ for all $t\in[0,\mathrm{T}]$, and the proof is complete.

    \end{proof}
\end{cor}
We emphasise that the space of solutions we are investigating is continuous in time in a certain compact interval $[0,\mathrm{T}]$ and there is no loss of the $\mathfrak{A}^2$-regularity along the evolution.
Because of these properties, a \emph{blow-up alternative} principle holds (Proposition~\ref{th:blowupAlternative} below): if such solutions exist only up to some finite maximal time $T_{\mathrm{max}}$, the associated $\mathfrak{A}^2$-norm must diverge at that point.
Importantly, local well-posedness implies that
\begin{equation}\label{def:maxT}
    T_{\mathrm{max}}\vcentcolon=\sup\big\{T\nls >\nls 0 \, \big|\: \exists* \,t\longmapsto\Psi^{\ls t}_{\!\varrho,\ls L}\!\in\nls  C^{\ls 0}\big([0,T], \mathfrak{A}^2(\Lambda_L)\nls \big) \text{ solving the Hartree equation \eqref{eq:HartreePDE}}\big\}
\end{equation}
is strictly positive.
This remains true also when $L\nls \to\nls \infty$ and $\varrho$ is large, since we have $T_{\mathrm{max}}\!\geq T_\ast\!>\nls 0$, owing to Assumptions~\ref{ass:tailCondition} and~\ref{ass:kineticTailCondition} (see Remark~\ref{rmk:finiteLifeSpan}).\newline
If we were unable to prove global well-posedness, or at least to quantify $T_{\mathrm{max}}$ in terms of $\varrho$ and $L$, the positivity of $T_{\ast}$ would stand as the sole ingredient we have to ensure that the lifespan of the solution does not shrink in the iterated limit.\newline
Let us prove the blow-up alternative principle.
\begin{prop}\label{th:blowupAlternative}
    Let $T_{\mathrm{max}}$ be defined by~\eqref{def:maxT} and consider $t\nls \in\nls [0, T_{\mathrm{max}})$ so that $\Psi^{\ls t}_{\!\varrho,\ls L}$ is the unique solution to the Hartree equation~\eqref{eq:HartreePDE} with initial datum $\Psi_{\!\varrho,\ls L}\!\in\nls \mathfrak{A}^2(\Lambda_L)$.
    Then, either
    \begin{itemize}
        \item $T_{\mathrm{max}}=\infty$;
        \item $T_{\mathrm{max}}<\infty$ and $\lim\limits_{t\to T_{\mathrm{max}}^-}\norm{\Psi^{\ls t}_{\!\varrho,\ls L}}[\mathfrak{A}^2(\Lambda_L)]\nls =\infty$, \:for all $\varrho, L \nls >\nls 0$.
    \end{itemize}
    \begin{proof}
        Assume $T_{\mathrm{max}}\nls <\infty$ and, by contradiction,
        $$\limsup_{t\to T_{\mathrm{max}}^-} \,\norm{\Psi^{\ls t}_{\!\varrho,\ls L}}[\mathfrak{A}^2(\Lambda_L)]\nls = M_{\varrho,\ls L} \nls \in (0,\infty).$$
        In particular, this means that $\forall\epsilon\nls >\nls 0$ there exists $t_0\in[0,T_{\mathrm{max}})$ such that $\!\!\sup\limits_{t\ls \in\ls [t_0,\ls T_{\mathrm{max}})}\norm{\Psi^{\ls t}_{\!\varrho,\ls L}}[\mathfrak{A}^2(\Lambda_L)]\!<\nls  M_{\varrho,\ls L}\nls +\epsilon$.
        Fix $\epsilon\nls =\nls 1$ for the sake of simplicity.
        Then, select $t_1\!\in\nls [t_0,T_{\mathrm{max}})$ close enough to the upper bound of the interval, namely such that $t_1>\max\big\{t_0, \,T_{\mathrm{max}}-\frac{\varrho}{12\ls \mathfrak{b}\,(M_{\varrho,\ls L}+1)^2\!\vphantom{|^|}}\big\}$.
        We consider the initial value problem starting at $t_1\!\in\nls (t_0,T_{\mathrm{max}})$ with initial datum $\Psi^{\ls t_1}_{\!\varrho,\ls L}\!\in\nls \mathfrak{A}^2(\Lambda_L)$.
        Clearly, the new initial datum is controlled regardless of the choice of $t_1$, since $\norm{\Psi^{\ls t_1}_{\!\varrho,\ls L}}[\mathfrak{A}^2(\Lambda_L)]\nls < M_{\varrho,\ls L}\nls +1$.
        By Lemma~\ref{th:contractiveMap}, the map $G_{\nls \varrho,\ls L}$ identified by~\eqref{def:contractionMap} is a strict contraction on the complete metric space $$\big\{t\longmapsto\Phi^{\ls t}_{\!\varrho,\ls L}\!\in C^{\ls 0}\big([t_1, t_1\nls +T],\mathfrak{A}^2(\Lambda_L)\big) \,\big|\: \Phi^{\ls t_1}_{\!\varrho,\ls L}\nls =\Psi^{\ls t_1}_{\!\varrho,\ls L},\,\sup_{t\ls \in\ls [t_1,\,t_1+\ls T]} \norm{\Phi^{\ls t}_{\!\varrho,\ls L}}[\mathfrak{A}^2(\Lambda_L)]\leq\tfrac{3}{2}\norm{\Psi^{\ls t_1}_{\!\varrho,\ls L}}[\mathfrak{A}^2(\Lambda_L)]\big\},$$
        in case $T<\frac{\varrho}{\vphantom{|^{|^1}}12\ls \mathfrak{b}\,\norm{\Psi^{\ls t_1}_{\!\varrho,\ls L}}[\mathfrak{A}^2(\Lambda_L)]^2\nqquad}$.\newline
        Hence, choosing $T=\frac{\varrho}{12\ls \mathfrak{b}\,(M_{\!\varrho,\ls L}+1)^2\!\vphantom{|^|}}<\frac{\varrho}{\vphantom{|^{|^1}} 12\ls \mathfrak{b}\,\norm{\Psi^{\ls t_1}_{\!\varrho,\ls L}}[\mathfrak{A}^2(\Lambda_L)]^2\nqquad},\quad$ we have a unique solution $\Phi^{\ls t}_{\!\varrho,\ls L}$ for $t\in[t_1, \ls  t_1\nls +T]$, where $t_1\nls +T\nls >\nls T_{\mathrm{max}}$.
        By uniqueness, $\Psi^{\ls t}_{\!\varrho,\ls L}\!=\Phi^{\ls t}_{\!\varrho,\ls L}$ for all $t\in[t_1, T_{\mathrm{max}})$, and we have shown a contradiction with the maximality of the lifespan.
        Consequently,
        $$\lim_{t\to T_{\mathrm{max}}^-} \norm{\Psi^{\ls t}_{\!\varrho,\ls L}}_{\mathfrak{A}^2(\Lambda_L)}=\infty,\qquad \forall \varrho, L \nls >\nls 0.$$
        
    \end{proof}
\end{prop}
In the following, we prove that the $\mathfrak{A}^2$-norm of the solution we found in Corollary~\ref{th:LWP} actually stays finite for each finite time; hence, a blow-up can occur only at $T_{\mathrm{max}}\!=\nls \infty$, so that our local-in-time solution can be promoted to a global one, by gluing several solutions at different time steps by continuity.
\begin{proof}[Proof of Lemma~\ref{th:GWP}]
    Let $t\nls \in\nls [0,T_{\mathrm{max}})$, with $T_{\mathrm{max}}$ defined by~\eqref{def:maxT}.
    Then, by the Duhamel formula~\eqref{eq:DuhamelHartree}, we have
    \begin{align*}
        &\norm{\Psi^{\ls t}_{\!\varrho,\ls L}}[\mathfrak{A}^2(\Lambda_L)]\!\leq \norm{\Psi_{\!\varrho,\ls L}}[\mathfrak{A}^2(\Lambda_L)]+\frac{1}{\varrho}\integrate[0;t]{\norm{\big(V_L\!\ast\nls \abs{\Psi^{\ls s}_{\!\varrho,\ls L}}^2\big)\ls \Psi^{\ls s}_{\!\varrho,\ls L}}[\mathfrak{A}^2(\Lambda_L)];\,ds}\\
        &\mspace{36mu}\leq \norm{\Psi_{\!\varrho,\ls L}}[\mathfrak{A}^2(\Lambda_L)]+\frac{2}{\varrho}\integrate[0;t]{\Big[\norm{V_L\!\ast\nls \abs{\Psi^{\ls s}_{\!\varrho,\ls L}}^2}[\mathfrak{A}^0(\Lambda_L)]\,\norm{\Psi^{\ls s}_{\!\varrho,\ls L}}[\mathfrak{A}^2(\Lambda_L)]\!+\norm{V_L\!\ast\nls \abs{\Psi^{\ls s}_{\!\varrho,\ls L}}^2}[\mathfrak{A}^2(\Lambda_L)]\,\norm{\Psi^{\ls s}_{\!\varrho,\ls L}}[\mathfrak{A}^0(\Lambda_L)]\Big];\,ds}.
    \end{align*}
    This follows from the elementary inequality
    $1+\abs{\vec{a}}^2\!\leq 2\ls (1\nls +\abs{\vec{b}}^2)+2\ls (1\nls +\abs{\vec{a}\nls -\nls \vec{b}}^2)$ for $\vec{a},\vec{b}\nls \in\nls \R^3$ used in the same fashion as in Proposition~\ref{th:WienerBanachAlgebra}.
    Next, let us estimate the quantity $\norm{V_L\!\ast\nls \abs{\Phi^{\ls s}}^2}[\mathfrak{A}^r(\Lambda_L)]$ for $r\in\{0,2\}$ and $\Phi^{\ls s}\!\in\nls \mathfrak{A}^2(\Lambda_L)$.
    Computing the Fourier coefficients, one has (\cfr~equations~(\ref{eq:nonlinearityFourierCoefficients},~\ref{def:autocorrelation}) below)
    \begin{align*}
        \norm{V_L\!\ast\nls \abs{\Phi^{\ls s}}^2}[\mathfrak{A}^r(\Lambda_L)]&=\!\nls \sum_{\vec{m}\ls \in\,\Z^3} \!\big(1\nls +\mspace{-0.75mu}\tfrac{2^r\pi^r\!}{L^r}\,\abs{\vec{m}}^r\big)\ls \abs*{\FT{V}_\infty\nls \big(\tfrac{2\pi}{L}\ls \vec{m}\big)}\bigg\lvert\mspace{-2.25mu} \sum_{\vec{n}\ls \in\,\Z^3}\! \conjugate*{\FT{\Phi}^{\ls s}_{\vec{n}}}\,\FT{\Phi}^{\ls s}_{\vec{n}+\vec{m}}\bigg\rvert\\
        &\leq\sqrt{\sum_{\vec{m}\ls \in\,\Z^3}\!\big(1\nls +\mspace{-0.75mu}\tfrac{2^r\pi^r\!}{L^r}\,\abs{\vec{m}}^r\big)^2\abs*{\FT{V}_\infty\nls \big(\tfrac{2\pi}{L}\ls \vec{m}\big)\nls }\ls }\ls \sqrt{\sum_{\vec{m}\ls \in\,\Z^3}\!\abs*{\FT{V}_\infty\nls \big(\tfrac{2\pi}{L}\ls \vec{m}\big)\nls }\bigg\lvert\mspace{-2.25mu}\sum_{\vec{n}\ls \in\,\Z^3}\! \conjugate*{\FT{\Phi}^{\ls s}_{\vec{n}}}\,\FT{\Phi}^{\ls s}_{\vec{n}+\vec{m}}\bigg\rvert^2}.
    \end{align*}
    Given the assumption $\FT{V}_\infty\geq 0$, the second factor in the last inequality can be bounded from above in terms of the energy $\mathscr{E}_{\varrho,\ls L}[\Psi^{\ls s}_{\!\varrho,\ls L}]$ (\cfr~Proposition~\ref{th:energyInFourierCoefficients}), which is a conserved quantity:
    \begin{equation*}
        \norm{V_L\!\ast\nls \abs{\Psi^{\ls s}_{\!\varrho,\ls L}}^2}[\mathfrak{A}^r(\Lambda_L)]\leq \varrho\,\sqrt{\ls 2\,e_{\varrho,\ls L}\!\!\sum_{\vec{m}\ls \in\,\Z^3}\!\big(1\nls +\mspace{-0.75mu}\tfrac{2^r\pi^r\!}{L^r}\,\abs{\vec{m}}^r\big)^2\,\FT{V}_\infty\nls \big(\tfrac{2\pi}{L}\ls \vec{m}\big)\ls }\ls ,
    \end{equation*}
    where $e_{\varrho,\ls L}\!>\nls 0$ stands for the energy per particle $\mathscr{E}_{\varrho,\ls L}[\Psi^{\ls s}_{\!\varrho,\ls L}]/\varrho L^3$ (which is close to $\frac{1}{2}\FT{V}_\infty(\vec{0})$ in the high-density thermodynamic limit, as shown in Proposition~\ref{th:quantifyTotalEnergy}).
    Finally, taking into account the decay of the potential~\eqref{eq:potentialDecay}, provided $\delta_2\!>\nls 2r$, one has for all $t\nls \in\nls [0,T_{\mathrm{max}})$
    $$\norm{V_L\!\ast\nls \abs{\Psi^{\ls t}_{\!\varrho,\ls L}}^2}[\mathfrak{A}^r(\Lambda_L)]\leq \varrho\,\sqrt{\ls 2\,e_{\varrho,\ls L}\,C\!\!\sum_{\vec{m}\ls \in\,\Z^3}\!\frac{\big(1\nls +\mspace{-0.75mu}\tfrac{2^r\pi^r\!}{L^r}\,\abs{\vec{m}}^r\big)^2\!\nls }{\big(1+\frac{2\pi}{L}\,\abs{\vec{m}}\big)^{\ls 3\ls +\ls \delta_2}\nquad}\mspace{9mu}}\mspace{6mu}=\vcentcolon \sqrt{e_{\varrho,\ls L}\ls }\ls  \varrho\, c_{L,\ls r}(V_\infty)<\infty.$$
    We point out that $c_{L,\ls r}(V_\infty)=\oBig{L^{\frac{3}{2}}}$, for all $0\leq r<\delta_2/2$ when $L\to\infty$.\newline
    Making use of this bound, we have obtained for all $t\nls \in\nls [0,T_{\mathrm{max}})$
    \begin{align*}
        \norm{\Psi^{\ls t}_{\!\varrho,\ls L}}[\mathfrak{A}^2(\Lambda_L)]&\leq \norm{\Psi_{\!\varrho,\ls L}}[\mathfrak{A}^2(\Lambda_L)]+2\sqrt{e_{\varrho,\ls L}\ls }\!\integrate[0;t]{\Big[c_{L,\ls 0}(V_\infty)\,\norm{\Psi^{\ls s}_{\!\varrho,\ls L}}[\mathfrak{A}^2(\Lambda_L)]\!+c_{L,\ls 2}(V_\infty)\,\norm{\Psi^{\ls s}_{\!\varrho,\ls L}}[\mathfrak{A}^0(\Lambda_L)]\Big];\,ds},\\
        &\leq \norm{\Psi_{\!\varrho,\ls L}}[\mathfrak{A}^2(\Lambda_L)]+2\sqrt{e_{\varrho,\ls L}\ls }\ls \big[c_{L,\ls 0}(V_\infty)\nls +c_{L,\ls 2}(V_\infty)\big]\!\integrate[0;t]{\norm{\Psi^{\ls s}_{\!\varrho,\ls L}}[\mathfrak{A}^2(\Lambda_L)];\,ds}.
    \end{align*}
    since $\norm{\ls \cdot\ls }[\mathfrak{A}^2(\Lambda_L)]$ is a stronger norm than $\norm{\ls \cdot\ls }[\mathfrak{A}^0(\Lambda_L)]$.
    By Gr\"onwall's inequality,
    $$\norm{\Psi^{\ls t}_{\!\varrho,\ls L}}[\mathfrak{A}^2(\Lambda_L)]\leq \norm{\Psi_{\!\varrho,\ls L}}[\mathfrak{A}^2(\Lambda_L)]\, e^{\ls 2\big[c_{L,\ls 0}(V_\infty)+\ls c_{L,\ls 2}(V_\infty)\big]\sqrt{e_{\varrho,\ls L}\ls }\ls t},$$
    which proves the result by means of Proposition~\ref{th:blowupAlternative}, since
    $$T_{\mathrm{max}}<\infty\quad\implies\quad\sup_{t\ls \in\ls [0,\ls T_{\mathrm{max}})} \norm{\Psi^{\ls t}_{\!\varrho,\ls L}}[\mathfrak{A}^2(\Lambda_L)]<\infty,$$
    which is a contradiction.
    Hence, $T_{\mathrm{max}}=\infty$ for any fixed $\varrho, L >0$.
    
\end{proof}

\bigskip

The above discussion on well-posedness suggests that the space of Fourier coefficients is a natural playground for the analysis of the Hartree equation.
Consequently, we proceed to reformulate equation~\eqref{eq:HartreePDE} in this setting.

\subsection{The Hartree Equation in Momentum Space}\label{sec:HartreeInMomentumRepresentation}

    Motivated by Proposition~\ref{th:BECmomentum}, we study the Hartree equation~\eqref{eq:HartreePDE} in the momentum representation.
    More precisely, since the system is defined on a torus, we use the Fourier basis $\{\mathrm{f}_{L;\,\vec{n}}\}_{\vec{n}\ls \in\,\Z^3}\nls \subset\nls  \Lp{2}[\Lambda_L]$ as a complete orthonormal set to decompose the time-dependent order parameter $\Psi^{\ls t}_{\!\varrho,\ls L}$, where 
    \begin{gather}\label{def:FourierBasis}
        \mathrm{f}_{L;\,\vec{n}}:\: \vec{x}\,\longmapsto\;\tfrac{1}{\!\nls \sqrt{L^3\ls }\ls }\,e^{\frac{2\pi\ls  i}{L} \mspace{2.25mu}\vec{n}\ls \cdot\,\vec{x}},\qquad \vec{x}\in\Lambda_L,\\
        \Psi^{\ls t}_{\!\varrho,\ls L}=\norm{\Psi_{\!\varrho,\ls L}}[2]\!\sum_{\vec{n}\ls \in\,\Z^3}\nls  \alpha^{\ls t}_{\varrho,\ls L}(\vec{n})\,\mathrm{f}_{L;\,\vec{n}}.\label{eq:decompositionBEC}
    \end{gather}
    Here, the Fourier coefficients $\alpha^{\ls t}_{\varrho,\ls L}\nls :\, \vec{n}\:\longmapsto\frac{1}{\norm{\Psi_{\!\varrho,\ls L}}[2]\!\nls}\,\scalar{\mathrm{f}_{L;\,\vec{n}}}{\Psi^{\ls t}_{\!\varrho,\ls L}}[2]$ have been chosen in such a way that \begin{equation}\label{eq:FourierCoefficientsNormalization}
        \sum_{\vec{n}\ls \in\,\Z^3} \abs{\alpha^{\ls t}_{\varrho,\ls L}(\vec{n})}^2=1,\qquad \forall t\geq 0,\, \varrho,L \nls >\nls 0.
    \end{equation}
    \begin{note}\label{rmk:unitTorusFourierCoefficients}
        By definition, we have
        $$\alpha^{\ls t}_{\varrho,\ls L}(\vec{n})=\scalar{\mathrm{f}_{1;\,\vec{n}}}{\tfrac{L^{3/2}}{\norm{\Psi_{\mspace{-3.75mu}\varrho,\mspace{0.75mu} L}}[2]\!\nls}\,\Psi^{\ls t}_{\!\varrho,\ls L}(L\,\vec{\cdot}\ls )}[\Lp{2}[\Lambda_1]].$$
        In other words, the Fourier coefficients we are adopting correspond to those associated with the macroscopic order parameter defined on the unit torus, whose norm is now $1$.
    \end{note}
    Concerning the nonlinearity of the Hartree equation, one has
    \begin{align*}
        \big(V_L\!\ast\nls \abs{\Psi^{\ls t}_{\!\varrho,\ls L}}^2\big)\nls (\vec{x})&=\nls \integrate[\Lambda_L]{V_L(\vec{x}\nls -\nls \vec{y})\abs{\Psi^{\ls t}_{\!\varrho,\ls L}(\vec{y})}^2;\!d\vec{y}}\\
        &=\nls \integrate[\Lambda_L]{\frac{1}{L^3}\nls \sum_{\vec{k}\ls \in\,\Z^3}e^{i\ls \frac{2\pi}{L} \vec{k}\,\cdot\ls (\vec{x}-\vec{y})}\,\FT{V}_\infty\nls \big(\tfrac{2\pi}{L}\ls \vec{k}\big)\: \tfrac{\norm{\Psi_{\mspace{-3.75mu}\varrho,\mspace{0.75mu} L}}[2]^2\!}{L^3} \mspace{-12mu}\sum_{\vec{m},\ls \vec{n}\ls \in\,\Z^3}\mspace{-9mu}e^{-i\ls \frac{2\pi}{L} \vec{y}\,\cdot\ls (\vec{m}-\vec{n})}\,\conjugate*{\alpha^{\ls t}_{\varrho,\ls L}(\vec{m})}\ls \alpha^{\ls t}_{\varrho,\ls L}(\vec{n});\!d\vec{y}}.
    \end{align*}
    For each fixed $\varrho,L\nls >\nls 0$, the well-posedness ensures that $\alpha^t_{\varrho,\ls L}\!\in\nls \ell_1(\Z^3)$, while the sum of the potential is finite, owing to the decay~\eqref{eq:potentialDecay}.
    This allows the use of Fubini's theorem, that yields
    \begin{equation*}
        \big(V_L\!\ast\nls \abs{\Psi^{\ls t}_{\!\varrho,\ls L}}^2\big)\nls (\vec{x})=\tfrac{\norm{\Psi_{\mspace{-3.75mu}\varrho,\mspace{0.75mu} L}}[2]^2\!}{L^3}\!\sum_{\vec{k}\ls \in\,\Z^3}e^{i\ls \frac{2\pi}{L} \vec{k}\,\cdot\,\vec{x}}\,\FT{V}_\infty\nls \big(\tfrac{2\pi}{L}\ls \vec{k}\big)\mspace{-12mu}\sum_{\vec{m},\ls \vec{n}\ls \in\,\Z^3}\mspace{-9mu}\conjugate*{\alpha^{\ls t}_{\varrho,\ls L}(\vec{m})}\ls \alpha^{\ls t}_{\varrho,\ls L}(\vec{n})\nls \integrate[\Lambda_L]{\frac{e^{-i\ls \frac{2\pi}{L} \vec{y}\,\cdot\ls (\vec{m}-\vec{n}+\vec{k})}}{L^3};\!d\vec{y}}.
    \end{equation*}
    Computing the integral, one gets
    \begin{equation}\label{eq:nonlinearityFourierCoefficients}
        \big(V_L\!\ast\nls \abs{\Psi^{\ls t}_{\!\varrho,\ls L}}^2\big)\nls (\vec{x})=\tfrac{\norm{\Psi_{\mspace{-3.75mu}\varrho,\mspace{0.75mu} L}}[2]^2\!}{L^3}\!\sum_{\vec{k}\ls \in\,\Z^3}e^{i\ls \frac{2\pi}{L} \vec{k}\,\cdot\,\vec{x}}\,\FT{V}_\infty\nls \big(\tfrac{2\pi}{L}\ls \vec{k}\big)\ls \beta^{\ls t}_{\varrho,\ls L}(\vec{k}).
    \end{equation}
    Here $\beta^{\ls t}_{\varrho,\ls L}\!\in\nls \ell_1(\Z^3)$ denotes the quantity
    \begin{equation}\label{def:autocorrelation}
        \beta^{\ls t}_{\varrho,\ls L}(\vec{k})\vcentcolon=\!\!\sum_{\vec{m}\ls \in\,\Z^3}\!\conjugate*{\alpha^{\ls t}_{\varrho,\ls L}(\vec{m})}\ls \alpha^{\ls t}_{\varrho,\ls L}(\vec{m}+\vec{k}),
    \end{equation}
    and we shall refer to it as the \emph{auto-correlation} of $\alpha^{\ls t}_{\varrho,\ls L}$.
    \begin{prop}\label{th:autocorrelationproperties}
        Provided $\beta^{\ls t}_{\varrho,\ls L}$ defined by~\eqref{def:autocorrelation}, one has for all $t\nls \geq\nls 0$ and $\varrho,L\nls >\nls 0$
        \begin{enumerate}[label=\roman*), itemsep=4pt]
            \item $\beta^{\ls t}_{\varrho,\ls L}(\vec{0})=1$;
            \item $\conjugate*{\beta^{\ls t}_{\varrho,\ls L}(\vec{k})}=\beta^{\ls t}_{\varrho,\ls L}(-\vec{k})$, for all $\vec{k}\in\Z^3$;
            \item $\norm{\beta^{\ls t}_{\varrho,\ls L}}[\ell_\infty(\Z^3)]\leq 1$;
            \item $\norm{\beta^{\ls t}_{\varrho,\ls L}}[\ell_2(\Z^3)]^2\nls =\frac{L^3}{\norm{\Psi_{\mspace{-3.75mu}\varrho,\mspace{0.75mu} L}}[2]^4}\ls\norm{\Psi^{\ls t}_{\!\varrho,\ls L}}[4]^4\ls $.
        \end{enumerate}
        \begin{proof}
            Point \textit{i)} is a consequence of normalisation~\eqref{eq:FourierCoefficientsNormalization}.\newline
            Point \textit{ii)} can be proven by changing the variable inside the sum $\vec{m}+\vec{k}\longmapsto \vec{m}'$.\newline
            Point \textit{iii)} can be deduced by means of the Cauchy--Schwarz inequality.\newline
            In conclusion, one has
            \begin{align*}
                \integrate[\Lambda_L]{\abs{\Psi^{\ls t}_{\!\varrho,\ls L}(\vec{x})}^4;\!d\vec{x}}\nls &=\tfrac{\norm{\Psi_{\mspace{-3.75mu}\varrho,\mspace{0.75mu} L}}[2]^4\!}{L^6}\!\integrate[\Lambda_L]{\!\bigg(\sum_{\vec{m},\ls \vec{n}\ls \in\,\Z^3}\mspace{-9mu}e^{-i\ls \frac{2\pi}{L} \vec{x}\,\cdot\ls (\vec{m}-\vec{n})}\,\conjugate*{\alpha^{\ls t}_{\varrho,\ls L}(\vec{m})}\ls \alpha^{\ls t}_{\varrho,\ls L}(\vec{n})\!\bigg)^{\!2};\!d\vec{x}}\\
                &=\tfrac{\norm{\Psi_{\mspace{-3.75mu}\varrho,\mspace{0.75mu} L}}[2]^4\!}{L^6}\mspace{-33mu}\sum_{\vec{m},\ls \vec{m}'\nls ,\ls \vec{n},\ls \vec{n}'\in\,\Z^3}\mspace{-27mu}\conjugate*{\alpha^{\ls t}_{\varrho,\ls L}(\vec{m})}\ls \conjugate*{\alpha^{\ls t}_{\varrho,\ls L}(\vec{m}')}\ls \alpha^{\ls t}_{\varrho,\ls L}(\vec{n})\ls \alpha^{\ls t}_{\varrho,\ls L}(\vec{n}')\!\integrate[\Lambda_L]{e^{-i\ls \frac{2\pi}{L} \vec{x}\,\cdot\ls (\vec{m}\ls +\ls \vec{m}'-\ls \vec{n}-\ls \vec{n}')};\!d\vec{x}}\\
                &=\tfrac{\norm{\Psi_{\mspace{-3.75mu}\varrho,\mspace{0.75mu} L}}[2]^4\!}{L^3}\mspace{-33mu} \sum_{\vec{m},\ls \vec{m}'\nls ,\ls \vec{n},\ls \vec{n}'\in\,\Z^3}\mspace{-27mu}\conjugate*{\alpha^{\ls t}_{\varrho,\ls L}(\vec{m})}\ls \conjugate*{\alpha^{\ls t}_{\varrho,\ls L}(\vec{m}')}\ls \alpha^{\ls t}_{\varrho,\ls L}(\vec{n})\ls \alpha^{\ls t}_{\varrho,\ls L}(\vec{n}') \,\delta_{\vec{m}+\vec{m}'\nls ,\,\vec{n}+\vec{n}'}.
            \end{align*}
            With the substitution $(\vec{m}', \vec{n}')\longmapsto (\vec{n}+\vec{k}', \vec{m}+\vec{k})$, one gets
            \begin{align*}
                \integrate[\Lambda_L]{\abs{\Psi^{\ls t}_{\!\varrho,\ls L}(\vec{x})}^4;\!d\vec{x}}\nls &=\tfrac{\norm{\Psi_{\mspace{-3.75mu}\varrho,\mspace{0.75mu} L}}[2]^4\!}{L^3}\mspace{-27mu} \sum_{\vec{m},\ls \vec{n},\ls \vec{k},\ls \vec{k}'\in\,\Z^3}\nquad\conjugate*{\alpha^{\ls t}_{\varrho,\ls L}(\vec{m})}\ls \conjugate*{\alpha^{\ls t}_{\varrho,\ls L}(\vec{n}+\vec{k}')}\ls \alpha^{\ls t}_{\varrho,\ls L}(\vec{n})\ls \alpha^{\ls t}_{\varrho,\ls L}(\vec{m}+\vec{k}) \,\delta_{\vec{k},\,\vec{k}'}\\
                &=\tfrac{\norm{\Psi_{\mspace{-3.75mu}\varrho,\mspace{0.75mu} L}}[2]^4\!}{L^3}\!\sum_{\vec{k}\ls \in\,\Z^3}\sum_{\vec{m}\ls \in\,\Z^3}\conjugate*{\alpha^{\ls t}_{\varrho,\ls L}(\vec{m})}\ls \alpha^{\ls t}_{\varrho,\ls L}(\vec{m}+\vec{k})\!\sum_{\vec{n}\ls \in\,\Z^3}\conjugate*{\alpha^{\ls t}_{\varrho,\ls L}(\vec{n}+\vec{k})}\ls \alpha^{\ls t}_{\varrho,\ls L}(\vec{n})\\
                &=\tfrac{\norm{\Psi_{\mspace{-3.75mu}\varrho,\mspace{0.75mu} L}}[2]^4\!}{L^3}\!\sum_{\vec{k}\ls \in\,\Z^3}\abs{\beta^{\ls t}_{\varrho,\ls L}(\vec{k})}^2,
            \end{align*}
            which proves item~\textit{iv)}.

        \end{proof}
    \end{prop}
    We now derive the Hartree equation in momentum space.
    \begin{prop}\label{th:HartreeInMomentaRepresentaion}
        Given $t\longmapsto\Psi^{\ls t}_{\!\varrho,\ls L}\!\in C^{\mspace{0.75mu}1}\big(\mspace{0.75mu}[0,\infty),\ls \mathfrak{A}^0(\Lambda_L)\mspace{-0.75mu}\big)\mspace{-0.75mu}\cap C^{\ls 0}\big(\mspace{0.75mu}[0,\infty),\ls  \mathfrak{A}^2(\Lambda_L)\mspace{-0.75mu}\big)$ solving the Hartree equation~\eqref{eq:HartreePDE}, one has that the Fourier coefficients $\alpha^{\ls t}_{\varrho,\ls L}\!\in\nls \ell_1(\Z^3)$ defined by~\eqref{eq:decompositionBEC} fulfil
        \begin{equation}\label{eq:HartreeMomenta}
            i\ls \partial_t\ls \alpha^{\ls t}_{\varrho,\ls L}(\vec{n})=\tfrac{4\pi^2\abs{\vec{n}}^2\!\nls }{L^2}\,\alpha^{\ls t}_{\varrho,\ls L}(\vec{n})+\tfrac{\norm{\Psi_{\mspace{-3.75mu}\varrho,\mspace{0.75mu} L}}[2]^2}{\varrho L^3}\!\sum_{\vec{k}\ls \in\,\Z^3}\alpha^{\ls t}_{\varrho,\ls L}(\vec{n}\nls -\nls \vec{k})\,\FT{V}_\infty\nls \big(\tfrac{2\pi}{L}\ls \vec{k}\big)\ls \beta^{\ls t}_{\varrho,\ls L}(\vec{k}),
        \end{equation}
        where $\beta^{\ls t}_{\varrho,\ls L}\!\in\nls \ell_1(\Z^3)$ is the auto-correlation of $\alpha^{\ls t}_{\varrho,\ls L}$ introduced in~\eqref{def:autocorrelation}.
        \begin{note}\label{rmk:uselessPotential}
            This representation of the Hartree equation emphasises how the nonlinearity enables the coupling between low and high momenta, driven by the term involving all Fourier coefficients.
            Although the decay of the potential might seem to suppress this momentum transfer, significant suppression only occurs when $\abs{\vec{k}}\nls \gg\nls  L$, which is not that relevant, since our goal is to consider $L$ large.\newline
            This is a significant feature that distinguishes our model from the standard mean-field regime (where $L\sim 1$).
        \end{note}
        \begin{proof}[Proof of Proposition~\ref{th:HartreeInMomentaRepresentaion}]
            First of all, we discuss the kinetic term.
            The well-posedness entails the uniform convergence in $\Lambda_L$ of the Fourier series~\eqref{eq:decompositionBEC}, since $\Psi^{\ls t}_{\!\varrho,\ls L}\!\in \nls \mathfrak{A}^2(\Lambda_L)$.
            Furthermore, the Fourier series associated with $\Delta \Psi^{\ls t}_{\!\varrho,\ls L}\!\in\nls \mathfrak{A}^0(\Lambda_L)$ converges uniformly in $\Lambda_L$ as well.
            Hence,
            \begin{equation}\label{eq:LaplaceBEC}
                -\Delta\Psi^{\ls t}_{\!\varrho,\ls L}=4\pi^2\ls\norm{\Psi_{\!\varrho,\ls L}}[2]\!\sum_{\vec{n}\ls \in\,\Z^3}\nls  \tfrac{\abs{\vec{n}}^2\!\nls }{L^2}\,\mathrm{f}_{L;\,\vec{n}} \,\alpha^{\ls t}_{\varrho,\ls L}(\vec{n}).
            \end{equation}
            Concerning the potential term, we can exploit the absolute convergence of the series~\eqref{eq:decompositionBEC} and~\eqref{eq:nonlinearityFourierCoefficients} to compute their product
            \begin{align*}
                \frac{1}{\varrho} \big(V_L\!\ast\nls \abs{\Psi^{\ls t}_{\!\varrho,\ls L}}^2\big)\ls  \Psi^{\ls t}_{\!\varrho,\ls L}&=\tfrac{\norm{\Psi_{\mspace{-3.75mu}\varrho,\mspace{0.75mu} L}}[2]^3}{\varrho L^3}\!\!\sum_{\vec{n},\ls \vec{k}\ls \in\,\Z^3}\!\!\mathrm{f}_{L;\,\vec{k}+\vec{n}}\,\FT{V}_\infty\nls \big(\tfrac{2\pi}{L}\ls \vec{k}\big)\ls \beta^{\ls t}_{\varrho,\ls L}(\vec{k})\,\alpha^{\ls t}_{\varrho,\ls L}(\vec{n})\\
                &=\tfrac{\norm{\Psi_{\mspace{-3.75mu}\varrho,\mspace{0.75mu} L}}[2]^3}{\varrho L^3}\!\sum_{\vec{m}\ls \in\,\Z^3}\mathrm{f}_{L;\,\vec{m}}\!\sum_{\vec{k}\ls \in\,\Z^3}\!\FT{V}_\infty\nls \big(\tfrac{2\pi}{L}\ls \vec{k}\big)\ls \beta^{\ls t}_{\varrho,\ls L}(\vec{k})\,\alpha^{\ls t}_{\varrho,\ls L}(\vec{m}\nls -\nls \vec{k}).
            \end{align*}
            To complete the proof, we have to justify the interchanging between the time derivative and summation.
            This requires the following identity to hold
            \begin{equation}\label{wts:derivativeInsideIntegral}
                i\ls \partial_t\ls  \alpha^{\ls t}_{\varrho,\ls L}(\vec{n})=\tfrac{1}{\norm{\Psi_{\mspace{-3.75mu}\varrho,\mspace{0.75mu} L}}[2]\!\nls}\ls\scalar{\mathrm{f}_{L;\,\vec{n}}}{i\ls \partial_t \Psi^{\ls t}_{\!\varrho,\ls L}}[2].
            \end{equation}
            If it is the case, the series $\norm{\Psi_{\!\varrho,\ls L}}[2]\!\!\nls \sum\limits_{\vec{n}\ls \in\,\Z^3}\nls  \mathrm{f}_{L;\,\vec{n}} \,i\ls \partial_t\ls \alpha^{\ls t}_{\varrho,\ls L}(\vec{n})$ converges to $ i\ls \partial_t\Psi^{\ls t}_{\!\varrho,\ls L}$ in $\Lp{2}[\Lambda_L]$.\newline
            Since the map $t\longmapsto\partial_t\Psi^{\ls t}_{\!\varrho,\ls L}\!\in\nls  C^{\ls 0}\big(\mspace{0.75mu}[0,\infty),\ls  \mathfrak{A}^0(\Lambda_L)\mspace{-0.75mu}\big)$ is continuous both in time and in space -- being $\mathfrak{A}^0(\Lambda_L)$ a subset of $C^{\ls 0}(\Lambda_L)$ -- we can apply the Leibniz rule to prove the validity of~\eqref{wts:derivativeInsideIntegral}.
            This means that one has
            \begin{equation}\label{eq:HartreeDecomposedInMomenta}
                \begin{split}
                    \norm{\Psi_{\!\varrho,\ls L}}[2]\!\sum\limits_{\vec{n}\ls \in\,\Z^3}\nls  \mathrm{f}_{L;\,\vec{n}} \,i\ls \partial_t\ls \alpha^{\ls t}_{\varrho,\ls L}(\vec{n})=&\:4\pi^2\ls\norm{\Psi_{\!\varrho,\ls L}}[2]\!\sum_{\vec{n}\ls \in\,\Z^3}\nls  \tfrac{\abs{\vec{n}}^2\!\nls }{L^2}\,\mathrm{f}_{L;\,\vec{n}} \,\alpha^{\ls t}_{\varrho,\ls L}(\vec{n})\ls +\\
                    &+\nls \tfrac{\norm{\Psi_{\mspace{-3.75mu}\varrho,\mspace{0.75mu} L}}[2]^3}{\varrho L^3}\!\sum_{\vec{n}\ls \in\,\Z^3}\mathrm{f}_{L;\,\vec{n}}\!\sum_{\vec{k}\ls \in\,\Z^3}\alpha^{\ls t}_{\varrho,\ls L}(\vec{n}\nls -\nls \vec{k})\,\FT{V}_\infty\nls \big(\tfrac{2\pi}{L}\ls \vec{k}\big)\ls \beta^{\ls t}_{\varrho,\ls L}(\vec{k}),
                \end{split}
            \end{equation}
            in the sense that both sides of the equation converge to the same limit in the $\Lp{2}$-topology, owing to the fact that $\Psi^{\ls t}_{\!\varrho,\ls L}$ solves the Hartree equation~\eqref{eq:HartreePDE}.\newline
            In conclusion, since $\{\mathrm{f}_{L;\,\vec{n}}\}_{\vec{n}\ls \in\,\Z^3}\!$ is a complete orthonormal system, identity~\eqref{eq:HartreeDecomposedInMomenta} implies equality between the coefficients and the proof is complete.

        \end{proof}
    \end{prop}
    \begin{prop}\label{th:energyInFourierCoefficients}
       Consider $t\longmapsto\Psi^{\ls t}_{\!\varrho,\ls L}\!\in C^{\mspace{0.75mu}1}\big(\mspace{0.75mu}[0,\infty),\ls \mathfrak{A}^0(\Lambda_L)\mspace{-0.75mu}\big)\mspace{-0.75mu}\cap C^{\ls 0}\big(\mspace{0.75mu}[0,\infty),\ls  \mathfrak{A}^2(\Lambda_L)\mspace{-0.75mu}\big)$ solving the Hartree equation~\eqref{eq:HartreePDE}.
       Given the Hartree energy functional $\mathscr{E}_{\varrho,\ls L}$ introduced in~\eqref{def:HartreeFunctional}, one has
       \begin{equation}\label{eq:energtFourierCoefficients}
            \mathscr{E}_{\varrho,\ls L}[\Psi^{\ls t}_{\!\varrho,\ls L}]=\norm{\Psi_{\!\varrho,\ls L}}[2]^2\!\sum_{\vec{n}\ls \in\,\Z^3} \!\left[\tfrac{4\pi^2\abs{\vec{n}}^2\!\nls }{L^2}\,\abs{\alpha^{\ls t}_{\varrho,\ls L}(\vec{n})}^2+\tfrac{\norm{\Psi_{\mspace{-3.75mu}\varrho,\mspace{0.75mu} L}}[2]^2}{2\ls\varrho L^3}\ls \FT{V}_\infty\big(\tfrac{2\pi}{L}\ls \vec{n}\big)\abs{\beta^{\ls t}_{\varrho,\ls L}(\vec{n})}^2\right]\!,
        \end{equation}
        where $\alpha^{\ls t}_{\varrho,\ls L},\ls \beta^{\ls t}_{\varrho,\ls L}\!\in\nls \ell_1(\Z^3)$ have been defined in~\eqref{eq:decompositionBEC} and~\eqref{def:autocorrelation}, respectively.
        \begin{proof}
            For the kinetic term, one has
            \begin{align*}
                \integrate[\Lambda_L]{\abs{\nabla_{\!\vec{x}}\Psi^{\ls t}_{\!\varrho,\ls L}(\vec{x})}^2;\!d\vec{x}}&=\tfrac{\norm{\Psi_{\mspace{-3.75mu}\varrho,\mspace{0.75mu} L}}[2]^2\!}{L^3}\!\integrate[\Lambda_L]{\mspace{-12mu}\sum_{\vec{m},\ls \vec{n}\ls \in\,\Z^3}\!\!\tfrac{4\pi^2\!\nls }{L^2}\,\vec{m}\nls \cdot\nls \vec{n}\;\conjugate*{\alpha^{\ls t}_{\varrho,\ls L}(\vec{m})}\ls \alpha^{\ls t}_{\varrho,\ls L}(\vec{n})\,e^{-i\ls \frac{2\pi}{L} \vec{x}\,\cdot\ls (\vec{m}\ls -\ls \vec{n})} ;\!d\vec{x}}\\
                &=\tfrac{\norm{\Psi_{\mspace{-3.75mu}\varrho,\mspace{0.75mu} L}}[2]^2\!}{L^3}\mspace{-12mu}\sum_{\vec{m},\ls \vec{n}\ls \in\,\Z^3}\!\!\tfrac{4\pi^2\!\nls }{L^2}\,\vec{m}\nls \cdot\nls \vec{n}\;\conjugate*{\alpha^{\ls t}_{\varrho,\ls L}(\vec{m})}\ls \alpha^{\ls t}_{\varrho,\ls L}(\vec{n})\!\integrate[\Lambda_L]{e^{-i\ls \frac{2\pi}{L} \vec{x}\,\cdot\ls (\vec{m}\ls -\ls \vec{n})} ;\!d\vec{x}}.
            \end{align*}
            Here, the exchange of the integral and the sum is justified by Fubini's theorem, since the Fourier series associated with $\nabla\Psi^{\ls t}_{\!\varrho,\ls L}\!\in\nls  \mathfrak{A}^0(\Lambda_L)$ converges absolutely.
            Hence,
            \begin{equation*}
                \integrate[\Lambda_L]{\abs{\nabla_{\!\vec{x}}\Psi^{\ls t}_{\!\varrho,\ls L}(\vec{x})}^2;\!d\vec{x}}=\norm{\Psi_{\!\varrho,\ls L}}[2]^2\mspace{-12mu}\sum_{\vec{m},\ls \vec{n}\ls \in\,\Z^3}\!\!\tfrac{4\pi^2\!\nls }{L^2}\,\vec{m}\nls \cdot\nls \vec{n}\;\conjugate*{\alpha^{\ls t}_{\varrho,\ls L}(\vec{m})}\ls \alpha^{\ls t}_{\varrho,\ls L}(\vec{n})\,\delta_{\vec{m},\,\vec{n}}.
            \end{equation*}
            On the other hand, for the component involving the potential, one has
            \begin{align*}
                \frac{1}{2\varrho}\nls \integrate[\Lambda_L]{\big(V_L\!\ast\nls \abs{\Psi^{\ls t}_{\!\varrho,\ls L}}^2\big)\nls (\vec{x})\,\abs{\Psi^{\ls t}_{\!\varrho,\ls L}(\vec{x})}^2;\!d\vec{x}}=\frac{\norm{\Psi_{\!\varrho,\ls L}}[2]^2\!\nls}{2\ls\varrho L^3}\nls \integrate[\Lambda_L]{\abs{\Psi^{\ls t}_{\!\varrho,\ls L}(\vec{x})}^2\!\sum_{\vec{k}\ls \in\,\Z^3}e^{i\ls \frac{2\pi}{L} \vec{k}\,\cdot\,\vec{x}}\,\FT{V}_\infty\nls \big(\tfrac{2\pi}{L}\ls \vec{k}\big)\ls \beta^{\ls t}_{\varrho,\ls L}(\vec{k});\!d\vec{x}}\\
                =\tfrac{\norm{\Psi_{\mspace{-3.75mu}\varrho,\mspace{0.75mu} L}}[2]^4\!}{2\ls \varrho L^6}\nls \integrate[\Lambda_L]{\mspace{-15mu}\sum_{\vec{m},\ls \vec{n},\ls \vec{k}\ls \in\,\Z^3}\nquad e^{i\ls \frac{2\pi}{L} \vec{x}\,\cdot\ls (\vec{k}\ls +\ls \vec{n}-\vec{m})}\,\conjugate*{\alpha^{\ls t}_{\varrho,\ls L}(\vec{m})}\ls \alpha^{\ls t}_{\varrho,\ls L}(\vec{n})\ls \FT{V}_\infty\nls \big(\tfrac{2\pi}{L}\ls \vec{k}\big)\ls \beta^{\ls t}_{\varrho,\ls L}(\vec{k});\!d\vec{x}}\\
                =\tfrac{\norm{\Psi_{\mspace{-3.75mu}\varrho,\mspace{0.75mu} L}}[2]^4\!}{2\ls \varrho L^6}\!\sum_{\vec{k}\ls \in\,\Z^3}\!\FT{V}_\infty\nls \big(\tfrac{2\pi}{L}\ls \vec{k}\big)\ls \beta^{\ls t}_{\varrho,\ls L}(\vec{k})\mspace{-12mu}\sum_{\vec{m},\ls \vec{n}\ls \in\,\Z^3}\mspace{-6mu}\conjugate*{\alpha^{\ls t}_{\varrho,\ls L}(\vec{m})}\ls \alpha^{\ls t}_{\varrho,\ls L}(\vec{n})\!\integrate[\Lambda_L]{e^{i\ls \frac{2\pi}{L} \vec{x}\,\cdot\ls (\vec{k}\ls +\ls \vec{n}-\vec{m})};\!d\vec{x}}.
            \end{align*}
            Once again, the absolute convergence of each series allows the use of Fubini's theorem.
            Therefore,
            \begin{align*}
                \frac{1}{2\varrho}\nls \integrate[\Lambda_L]{\big(V_L\!\ast\nls \abs{\Psi^{\ls t}_{\!\varrho,\ls L}}^2\big)\nls (\vec{x})\,\abs{\Psi^{\ls t}_{\!\varrho,\ls L}(\vec{x})}^2;\!d\vec{x}}&=\tfrac{\norm{\Psi_{\mspace{-3.75mu}\varrho,\mspace{0.75mu} L}}[2]^4\!}{2\ls \varrho L^3}\!\sum_{\vec{k}\ls \in\,\Z^3}\!\FT{V}_\infty\nls \big(\tfrac{2\pi}{L}\ls \vec{k}\big)\ls \beta^{\ls t}_{\varrho,\ls L}(\vec{k})\mspace{-12mu}\sum_{\vec{m},\ls \vec{n}\ls \in\,\Z^3}\mspace{-6mu}\conjugate*{\alpha^{\ls t}_{\varrho,\ls L}(\vec{m})}\ls \alpha^{\ls t}_{\varrho,\ls L}(\vec{n})\,\delta_{\vec{m},\,\vec{k}+\vec{n}}\\
                &=\tfrac{\norm{\Psi_{\mspace{-3.75mu}\varrho,\mspace{0.75mu} L}}[2]^4\!}{2\ls \varrho L^3}\!\sum_{\vec{k}\ls \in\,\Z^3}\!\FT{V}_\infty\nls \big(\tfrac{2\pi}{L}\ls \vec{k}\big)\ls \beta^{\ls t}_{\varrho,\ls L}(\vec{k})\ls \conjugate*{\beta^{\ls t}_{\varrho,\ls L}(\vec{k})},
            \end{align*}
            and the proof is concluded.
            
        \end{proof}
    \end{prop}
    When condensation occurs in a momentum state, we can quantify the total energy $\mathscr{E}_{\varrho,\ls L}[\Psi^{\ls t}_{\!\varrho,\ls L}]$ since it is a conserved quantity.
    \begin{prop}\label{th:quantifyTotalEnergy}
        Take into account Assumptions~\ref{ass:initialBEC},~\ref{ass:translationalInvariance},~\ref{ass:tailCondition} and the hypotheses of Proposition~\ref{th:energyInFourierCoefficients}.
        Then, one has
        $$\lim_{\varrho\to\infty}\limsup_{L\to\infty}\abs*{\frac{\mathscr{E}_{\varrho,\ls L}[\Psi^{\ls t}_{\!\varrho,\ls L}]}{\varrho L^3}-\frac{\mathfrak{b}}{2}}=0,\qquad \forall t\geq 0,$$
        where $\mathfrak{b}=\FT{V}_\infty(\vec{0})$.
        \begin{proof}
            Once we prove that
            \begin{equation}\label{wts:convergenceOfBeta}
                \lim_{\varrho\to\infty}\limsup_{L\to\infty} \norm*{\ls \abs{\beta^{\ls 0}_{\varrho,\ls L}}^2-\delta_{\vec{0}}}[\ell_1(\Z^3)]=0,
            \end{equation}
            the result is proven by exploiting~\eqref{eq:becNormalization}, Proposition~\ref{th:energyInFourierCoefficients} and Remark~\ref{rmk:consequencesOfCOnvergence}:
            \begin{align*}
                \limsup_{\varrho\to\infty}\limsup_{L\to\infty} \abs*{\tfrac{1}{\varrho L^3} \,\mathscr{E}_{\varrho,\ls L}[\Psi^{\ls t}_{\!\varrho,\ls L}]-\tfrac{\mathfrak{b}}{2}}\leq & \lim_{\varrho\to\infty}\limsup_{L\to\infty}\!\sum_{\vec{n}\ls \in\,\Z^3} \!\tfrac{4\pi^2\abs{\vec{n}}^2\!\nls }{L^2}\,\abs{\alpha^{\ls 0}_{\varrho,\ls L}(\vec{n})}^2\ls +\\[-5pt]
                &+\tfrac{1}{2}\limsup_{\varrho\to\infty}\limsup_{L\to\infty}\abs*{\sum_{\vec{n}\ls \in\,\Z^3}\!\FT{V}_\infty\big(\tfrac{2\pi}{L}\ls \vec{n}\big)\abs{\beta^{\ls 0}_{\varrho,\ls L}(\vec{n})}^2\nls -\mathfrak{b}}\\
                \leq & \, \tfrac{1}{2}\limsup_{\varrho\to\infty}\limsup_{L\to\infty}\!\sum_{\vec{n}\ls \in\,\Z^3}\!\abs*{\FT{V}_\infty\big(\tfrac{2\pi}{L}\ls \vec{n}\big)\mspace{-0.75mu}}\ls\abs*{\mspace{0.75mu} \abs{\beta^{\ls 0}_{\varrho,\ls L}(\vec{n})}^2\nls -\delta_{\vec{n},\ls 
                \vec{0}}}\\
                \leq &\, \tfrac{\mathfrak{b}}{2}\limsup_{\varrho\to\infty}\limsup_{L\to\infty} \norm*{\ls \abs{\beta^{\ls 0}_{\varrho,\ls L}}^2\nls -\delta_{\vec{0}}}[\ell_1(\Z^3)]\nls .
            \end{align*}
            Therefore, we proceed to show~\eqref{wts:convergenceOfBeta}.
            To this end, we observe that for all $t\geq 0$, $\beta^{\ls t}_{\varrho,\ls L}(\vec{0})=1$ (item \textit{i)} of Proposition~\ref{th:autocorrelationproperties}), hence
            \begin{align*}
                \limsup_{\varrho\to\infty}\limsup_{L\to\infty} \norm*{\ls \abs{\beta^{\ls 0}_{\varrho,\ls L}}^2\nls -\delta_{\vec{0}}}[\ell_1(\Z^3)]\nls &=\limsup_{\varrho\to\infty}\limsup_{L\to\infty} \norm*{\beta^{\ls 0}_{\varrho,\ls L}}[\ell_2(\Z^3\smallsetminus\{\vec{0}\})]^2\\[-2.5pt]
                &\leq \limsup_{\varrho\to\infty}\limsup_{L\to\infty}\!\sum_{\substack{\vec{k}\ls \in\,\Z^3 \ls :\\[1.5pt] \vec{k}\neq\ls \vec{0}}}\abs*{\sum_{\vec{n}\ls \in\,\Z^3}\conjugate*{\alpha^{\ls 0}_{\varrho,\ls L}(\vec{n})}\ls \alpha^{\ls 0}_{\varrho,\ls L}(\vec{n}+\vec{k})}^2\\[-5pt]
                &\mspace{-168mu}\leq\limsup_{\varrho\to\infty}\limsup_{L\to\infty}\!\sum_{\substack{\vec{k}\ls \in\,\Z^3 \ls :\\[1.5pt] \vec{k}\neq\ls \vec{0}}}\abs*{\ls \sum_{\vec{n}\ls \in\,\Z^3}\!\conjugate*{r_{\varrho,\ls L}(\vec{n})}\ls r_{\varrho,\ls L}(\vec{n}+\vec{k})+e^{-i\ls \vartheta}\alpha^{\ls 0}_{\varrho,\ls L}(\vec{k}_0\nls +\vec{k})+e^{i\ls \vartheta}\,\conjugate*{\alpha^{\ls 0}_{\varrho,\ls L}(\vec{k}_0\nls -\vec{k})}\ls }^2\!\nls ,
            \end{align*}
            where $r_{\varrho,\ls L}(\vec{n})\vcentcolon=\alpha^{\ls 0}_{\varrho,\ls L}(\vec{n})-e^{i\ls \vartheta}\delta_{\vec{n},\ls \vec{k}_0}$.
            We recall that equation~\eqref{eq:becFourierCoefficients1} implies
            $$\lim_{\varrho\to\infty}\limsup_{L\to\infty}\,\norm{r_{\varrho,\ls L}}[\ell_2(\Z^3)]=0.$$
            We have obtained
            \begin{align*}
                \limsup_{\varrho\to\infty}\limsup_{L\to\infty} \norm*{\ls \abs{\beta^{\ls 0}_{\varrho,\ls L}}^2\nls -\delta_{\vec{0}}}[\ell_1(\Z^3)]\leq \limsup_{\varrho\to\infty}\limsup_{L\to\infty}\!\sum_{\substack{\vec{k}\ls \in\,\Z^3 \ls :\\[1.5pt] \vec{k}\neq\ls \vec{0}}} &\left\lvert\ls \sum_{\vec{n}\ls \in\,\Z^3}\!\conjugate*{r_{\varrho,\ls L}(\vec{n})}\ls r_{\varrho,\ls L}(\vec{n}+\vec{k})\ls +\right.\\
                &\left.\vphantom{\sum_{\Z^3}}\,+e^{-i\ls \vartheta}r_{\varrho,\ls L}(\vec{k}_0\nls +\vec{k})+e^{i\ls \vartheta}\,\conjugate*{r_{\varrho,\ls L}(\vec{k}_0\nls -\vec{k})}\ls \right\rvert^2\\
                &\mspace{-180mu}\leq 3\limsup_{\varrho\to\infty}\limsup_{L\to\infty}\, \norm{\ls \conjugate*{r_{\varrho,\ls L}}\ast r_{\varrho,\ls L}}[\ell_2(\Z^3)]^2\nls +6\lim_{\varrho\to\infty}\limsup_{L\to\infty}\,\norm{r_{\varrho,\ls L}}[\ell_2(\Z^3)]^2.
            \end{align*}
            Because of Young's inequality for convolutions, 
            $$\limsup_{\varrho\to\infty}\limsup_{L\to\infty} \norm*{\ls \abs{\beta^{\ls 0}_{\varrho,\ls L}}^2\nls -\delta_{\vec{0}}}[\ell_1(\Z^3)]\leq 3\limsup_{\varrho\to\infty}\limsup_{L\to\infty}\, \norm{r_{\varrho,\ls L}}[\ell_1(\Z^3)]^2\norm{r_{\varrho,\ls L}}[\ell_2(\Z^3)]^2.$$
            The proposition is therefore proven as soon as we show $\limsup\limits_{\varrho\to\infty}\limsup\limits_{L\to\infty}\, \norm{r_{\varrho,\ls L}}[\ell_1(\Z^3)]<\infty$.
            But we can even prove that such an iterated limit exists and is $0$, since the $\ell_2$-convergence implies pointwise convergence.
            Specifically, let $M>0$ be large enough so that $\abs{\vec{k}_0}\leq M$; then,
            $$\limsup_{\varrho\to\infty}\limsup_{L\to\infty}\,\norm{r_{\varrho,\ls L}}[\ell_1(\Z^3)]=\ls\limsup_{\varrho\to\infty}\limsup_{L\to\infty}\! \sum_{\substack{\vec{m}\ls \in\,\Z^3 \ls :\\[1.5pt] \abs{\vec{m}}\ls  \leq M}}\nls\abs{r_{\varrho,\ls L}(\vec{m})}+\ls\limsup_{\varrho\to\infty}\limsup_{L\to\infty}\! \sum_{\substack{\vec{m}\ls \in\,\Z^3 \ls :\\[1.5pt] \abs{\vec{m}}\ls  > M}}\nls\abs{\alpha^{\ls 0}_{\varrho,\ls L}(\vec{m})}.$$
            Assumption~\ref{ass:tailCondition} entails the control of high momenta, whereas we can make use of the sub-additivity of the limit superior for low momenta to compute the limit inside the sum (which is zero, by pointwise convergence).

        \end{proof}
    \end{prop}

    The next step consists in obtaining time-dependent estimates of the nonlinearity in terms of the initial datum.

    \subsection{Control of the Nonlinearity}\label{sec:controlOfNonlinearity}

    A key recurring quantity requiring estimation throughout our analysis is the supremum of the nonlinearity in the Hartree equation.
    Specifically, we must ensure this remains finite when $L$ and $\varrho$ are large, at least over a finite time interval. 
    In the momentum representation, identity~\eqref{eq:nonlinearityFourierCoefficients} yields the immediate upper bound
    \begin{equation}
        \tfrac{1}{\varrho}\,\norm{V_L\!\ast\nls \abs{\Psi^{\ls t}_{\!\varrho,\ls L}}^2}[\infty]\leq \tfrac{\norm{\Psi_{\mspace{-3.75mu}\varrho,\mspace{0.75mu} L}}[2]^2}{\varrho L^3}\!\sum_{\vec{k}\ls \in\,\Z^3}\!\abs*{\FT{V}_\infty\nls \big(\tfrac{2\pi}{L}\ls \vec{k}\big)\nls }\ls \abs{\beta^{\ls t}_{\varrho,\ls L}(\vec{k})}.
    \end{equation}
    Thus, controlling the r.h.s. suffices.
    We have already pointed out in Remark~\ref{rmk:uselessPotential} that the factor $\FT{V}_\infty\big(\tfrac{2\pi}{L}\ls \vec{k}\big)$ behaves effectively as a constant for large $L$; therefore, the decay of the potential is not helpful in controlling the nonlinearity.
    This suggests that we must rely solely on the summability of the auto-correlation of $\alpha^{\ls t}_{\varrho,\ls L}$.
    By taking into account Assumption~\ref{ass:translationalInvariance}, we can provide a heuristic idea of the mechanism we are after.
    In fact, Proposition~\ref{th:BECmomentum} demonstrates that quasi-complete condensation requires the momentum distribution to accumulate near a single mode when $\varrho$ and $L$ are large.
    This implies that the corresponding autocorrelation must converge to $\delta_{\vec{k},\ls \vec{0}}$, provided the decay in $\varrho$ and $L$ is sufficiently fast (so that the iterated limit can be computed inside the sum).
    Indeed, the nonlinearity is comparable to the energy per particle when the autocorrelation is close to being supported around zero (since the autocorrelation and its square behave the same), and a uniform control in $\varrho, L\nls >\nls 0$ can be provided due to energy conservation.
    Anyway, exchanging the series with the iterated limit is a very delicate step.

    \smallskip
    
    \noindent Motivated thus, we aim to establish an upper bound of the $\ell_1$-norm of $\alpha^{\ls t}_{\varrho,\ls L}$, since
    \begin{equation}
        \tfrac{1}{\varrho}\,\norm{V_L\!\ast\nls \abs{\Psi^{\ls t}_{\!\varrho,\ls L}}^2}[\infty]\leq \big(S^{\ls t}_{\nls \varrho,\ls L}\big)^2\ls \mathfrak{b},
    \end{equation}
    where we recall that $\mathfrak{b}=\FT{V}_\infty(\vec{0})$ and we have introduced the shortcut
    \begin{equation}\label{def:FourierCoefficientsSumShortcut}
        S^{\ls t}_{\nls \varrho,\ls L}\vcentcolon=\tfrac{\norm{\Psi_{\mspace{-3.75mu}\varrho,\mspace{0.75mu} L}}[2]\nls}{\!\nls\sqrt{\varrho L^3\ls}\ls}\,\norm{\alpha^{\ls t}_{\varrho,\ls L}}[\ell_1(\Z^3)].
    \end{equation}
    Indeed,
    \begin{equation}\label{eq:autoCorrelationControlledByS2}
        \norm{\beta^{\ls t}_{\varrho,\ls L}}[\ell_1(\Z^3)]\leq \nls \tfrac{\varrho L^3}{\norm{\Psi_{\mspace{-3.75mu}\varrho,\mspace{0.75mu} L}}[2]^2}\big(S^{\ls t}_{\nls \varrho,\ls L}\big)^2.
    \end{equation}
    Additionally, one has the lower bound
    \begin{equation}\label{eq:SumLowerBound}
        1=\norm{\alpha^{\ls t}_{\varrho,\ls L}}[\ell_2(\Z^3)]\leq \norm{\alpha^{\ls t}_{\varrho,\ls L}}[\ell_1(\Z^3)]=\tfrac{\!\nls\sqrt{\varrho L^3\ls}\ls}{\norm{\Psi_{\mspace{-3.75mu}\varrho,\mspace{0.75mu} L}}[2]\!\nls}\,S^{\ls t}_{\nls \varrho,\ls L},
    \end{equation}
    and the quantity~\eqref{def:FourierCoefficientsSumShortcut} controls directly the supremum of the order parameter
    \begin{equation}\label{eq:SumLinftyControlled}
        \norm{\Psi^{\ls t}_{\!\varrho,\ls L}}[\infty]\leq \sqrt{\varrho\ls}\ls S^{\ls t}_{\nls \varrho,\ls L}.
    \end{equation}
    \begin{note}
        The following Propositions~\ref{th:nonlinearityControl} and~\ref{th:LaplaceNonlinearityControl} have already been implicitly addressed in Lemma~\ref{th:contractiveMap}.
        Such a lemma guarantees we can indeed exhibit a unique solution whose supremum in $t\nls \in\nls  [0,T]$ of the $\mathfrak{A}^2$-norm is controlled in terms of the initial data, provided $T\!<\nls T_\ast$ -- defined in~\eqref{eq:liminfLifeSpan} -- when $\varrho$ and $L$ are large enough.\newline
        In the following, we derive slightly more refined estimates providing an explicit pointwise control in time, which works for a strictly longer time interval.
        However, by the result of these improved estimates, we are able to prove that the lifespan within which the solution is controlled in the iterated limit only depends on the $\mathfrak{A}^0(\Lambda_L)$ norm of the initial data -- which we know is bounded in the limit under Assumptions~\ref{ass:initialBEC} and~\ref{ass:tailCondition} (see Remark~\ref{rmk:S0Boundedness}).
    \end{note}
    \begin{prop}\label{th:nonlinearityControl}
        Let $t\longmapsto\Psi^{\ls t}_{\!\varrho,\ls L}\!\in C^{\mspace{0.75mu}1}\big(\mspace{0.75mu}[0,\infty),\ls \mathfrak{A}^0(\Lambda_L)\mspace{-0.75mu}\big)\mspace{-0.75mu}\cap C^{\ls 0}\big(\mspace{0.75mu}[0,\infty),\ls  \mathfrak{A}^2(\Lambda_L)\mspace{-0.75mu}\big)$ solving the Hartree equation~\eqref{eq:HartreePDE}, and assume $\Psi^{\ls 0}_{\!\varrho,\ls L}$ is such that
        \begin{itemize}
            \item[i)] $\forall \epsilon > 0\quad \exists M\nls >\nls 0\,:\quad\limsup\limits_{\varrho\to\infty}\limsup\limits_{L\to\infty} \!\sum\limits_{\substack{\vec{m}\ls \in\,\Z^3 \ls :\\[1.5pt] \abs{\vec{m}}\ls  > M}} \abs{\alpha^{\ls 0}_{\varrho,\ls L}(\vec{m})}<\epsilon$;
        \end{itemize}
        where $\alpha^{\ls t}_{\varrho,\ls L}$ has been introduced in~\eqref{eq:decompositionBEC}.
        Moreover, let $\ls\overline{\nls S}^{\,0}\!\vcentcolon=\limsup\limits_{\varrho\to\infty}\limsup\limits_{L\to\infty}\ls\frac{1}{\!\nls\sqrt{\varrho\ls}\ls}\norm{\Psi^{\ls 0}_{\!\varrho,\ls L}}[\mathfrak{A}^0(\Lambda_L)]$.\newline
        Then, for all $0\leq t< \big(2\ls \mathfrak{b}\ls (\ls\overline{\nls S}^{\,0})^2\big)^{-1}$, with $\mathfrak{b}=\FT{V}_\infty(\vec{0})$, we have
        $$\limsup_{\varrho\to\infty}\limsup_{L\to\infty} S^{\ls t}_{\nls \varrho,\ls L}\leq \frac{\overline{\nls S}^{\,0}}{\!\nls \sqrt{1-2\ls \mathfrak{b}\ls (\ls\overline{\nls S}^{\, 0})^2\,t\ls }\ls },$$
        where $S^{\ls t}_{\nls \varrho,\ls L}$ has been defined in~\eqref{def:FourierCoefficientsSumShortcut}.\newline
        If in addition
        \begin{itemize}
            \item[ii)] $\exists* \vec{k}_0\in\Z^3,\,\vartheta\in[0,2\pi)\,:\quad \lim\limits_{\varrho\to\infty}\limsup\limits_{L\to\infty}\ls \abs*{\alpha^{\ls 0}_{\varrho,\ls L}(\vec{m})-e^{i\ls \vartheta}\delta_{\vec{m},\ls \vec{k}_0}}\nls =0,\quad \forall\vec{m}\in\Z^3$;
        \end{itemize}
        then, $\overline{\nls S}^{\, 0}\nls =1$.
        \begin{note}\label{rmk:equivalenceOfConditions}
            First, observe that item \textit{i)} corresponds to Assumption~\ref{ass:tailCondition}, while item \textit{ii)} is a by-product of Assumption~\ref{ass:translationalInvariance}.\newline
            Moreover, condition \textit{i)} $\land$ \textit{ii)} is equivalent to
            $$\lim_{\varrho\to\infty}\limsup_{L\to\infty}\,\norm*{\alpha^{\ls 0}_{\varrho,\ls L}\nls -\ls e^{i\ls \vartheta}\delta_{\vec{k}_0}}[\ell_1(\Z^3)]\!=0.$$
            In particular, the $\Rightarrow\!)$ direction is a consequence of the sub-additivity of the limit superior, while for $\Leftarrow\nls )$ direction one has only to prove the validity of \textit{i)}, since \textit{ii)} is trivial, owing to the fact that $\ell_1$-convergence implies pointwise convergence.
            Note that for any $M$ large enough so that $\abs{\vec{k}_0}\nls \leq\nls  M$, one has
            \begin{align*}
                \limsup_{\varrho\to\infty}\limsup_{L\to\infty}\!\!\sum_{\substack{\vec{m}\ls \in\,\Z^3 \ls :\\[1.5pt] \abs{\vec{m}}\ls >M}} \!\abs{\alpha^{\ls 0}_{\varrho,\ls L}(\vec{m})}&\leq \limsup_{\varrho\to\infty}\limsup_{L\to\infty}\!\!\sum_{\substack{\vec{m}\ls \in\,\Z^3 \ls :\\[1.5pt] \vec{m}\neq\ls  \vec{k}_0}} \!\abs{\alpha^{\ls 0}_{\varrho,\ls L}(\vec{m})}\\[-2.5pt]
                &\leq \lim_{\varrho\to\infty}\limsup_{L\to\infty}\,\norm*{\alpha^{\ls 0}_{\varrho,\ls L}\nls -\ls e^{i\ls \vartheta}\delta_{\vec{k}_0}}[\ell_1(\Z^3)]\!=0.
            \end{align*}
            Hence,
            $$\lim_{M\to\infty}\limsup_{\varrho\to\infty}\limsup_{L\to\infty}\!\!\sum_{\substack{\vec{m}\ls \in\,\Z^3 \ls :\\[1.5pt] \abs{\vec{m}}\ls >M}} \!\abs{\alpha^{\ls 0}_{\varrho,\ls L}(\vec{m})}=0,$$
            which corresponds to condition \textit{i)}.
        \end{note}
        \begin{note}\label{rmk:S0Boundedness}
            By taking account of Assumptions~\ref{ass:initialBEC} and~\ref{ass:tailCondition}, one proves that $\ls\overline{\nls S}^{\, 0}\nls\vcentcolon=\limsup\limits_{\varrho\to\infty}\limsup\limits_{L\to\infty}S^{\ls 0}_{\nls \varrho,\ls L}\!<\nls\infty$.
            Indeed, let $\alpha\!\in\!\ell_2(\Z^3)$ be the sequence of the Fourier coefficients associated with $\Psi\!\in\! H^1(\Lambda_1)$.
            Then, by triangle inequality,
            \begin{align*}
                \ls\overline{\nls S}^{\, 0}\!&\leq\limsup_{\varrho\to\infty}\limsup_{L\to\infty}\Bigg[\,\sum_{\substack{\vec{m}\ls \in\,\Z^3 \ls :\\[1.5pt] \abs{\vec{m}}\leq M}}\!\abs{\alpha^{\ls 0}_{\varrho,\ls L}(\vec{m})\nls -\alpha(\vec{m})}\,+\!\!\sum_{\substack{\vec{m}\ls \in\,\Z^3 \ls :\\[1.5pt] \abs{\vec{m}}\leq M}}\!\abs{\alpha(\vec{m})}\,+\!\!\!\sum_{\substack{\vec{m}\ls \in\,\Z^3 \ls :\\[1.5pt] \abs{\vec{m}}>\ls  M}}\!\abs{\alpha^{\ls 0}_{\varrho,\ls L}(\vec{m})}\Bigg]\\[-5pt]
                &\leq\limsup_{\varrho\to\infty}\limsup_{L\to\infty}\Bigg[\ls\big\lvert\{\vec{m}\!\in\nls\Z^3 \,\big\vert\: \abs{\vec{m}}\nls\leq\nls M\}\big\rvert^{\frac{1}{2}}\big(\norm{\alpha^{\ls 0}_{\varrho,\ls L}\!-\alpha}[\ell_2(\Z^3)]+1\big)+\!\!\!\sum_{\substack{\vec{m}\ls \in\,\Z^3 \ls :\\[1.5pt] \abs{\vec{m}}>\ls  M}}\!\abs{\alpha^{\ls 0}_{\varrho,\ls L}(\vec{m})}\Bigg].
            \end{align*}
            For the first two term, we exploited a Cauchy--Schwarz inequality and the normalisation of $\alpha\nls\in\nls\ell_2(\Z^3)$.
            Hence,
            \begin{equation*}
                \ls\overline{\nls S}^{\, 0}\!\nls\leq\nls \big\lvert\{\vec{m}\!\in\nls\Z^3 \,\big\vert\: \abs{\vec{m}}\nls\leq\nls M\}\big\rvert^{\frac{1}{2}}\Big(1+\limsup_{\varrho\to\infty}\limsup_{L\to\infty}\norm*{\tfrac{\Psi_{\!\varrho,\mspace{0.75mu} L}(\ls \vec{\cdot}\ls )}{\norm{\Psi_{\!\varrho,\mspace{0.75mu} L}}[2]\!\nls}-\tfrac{1}{L^{3/2}}\ls \Psi\nls \big(\mspace{-0.75mu}\tfrac{\vec{\cdot}}{L}\mspace{-0.75mu}\big)\nls }[2]\Big)+\epsilon.
            \end{equation*}
            In the previous step, we made use of Assumption~\ref{ass:tailCondition} to obtain the smallness of the iterated limit of the series' tail, provided $M$ is large enough. 
            In other words, owing to Assumption~\ref{ass:initialBEC}, when $M$ is sufficiently large we get
            \begin{equation*}
                \ls\overline{\nls S}^{\, 0}\!\leq \big\lvert\{\vec{m}\!\in\nls\Z^3 \,\big\vert\: \abs{\vec{m}}\nls\leq\nls M\}\big\rvert^{\frac{1}{2}}+\epsilon<\infty.
            \end{equation*}
        \end{note}
        \begin{proof}[Proof of Proposition~\ref{th:nonlinearityControl}]
            First, we claim that the Duhamel formula associated with equation~\eqref{eq:HartreeMomenta} is given by
            \begin{equation}\label{eq:DuhamelFourierCoefficients}
                \alpha^{\ls t}_{\varrho,\ls L}(\vec{n})=e^{-i \ls \frac{4\pi^2\abs{\vec{n}}^2\!\nls }{L^2}\,t}\,\alpha^{\ls 0}_{\varrho,\ls L}(\vec{n})-i\ls\tfrac{\norm{\Psi_{\mspace{-3.75mu}\varrho,\mspace{0.75mu} L}}[2]^2\nls}{\varrho L^3}\!\integrate[0;t]{e^{-i\ls \frac{4\pi^2\abs{\vec{n}}^2\!\nls }{L^2}\ls (t-s)}\!\sum_{\vec{k}\ls \in\,\Z^3}\alpha^{\ls s}_{\varrho,\ls L}(\vec{n}\nls -\nls \vec{k})\,\FT{V}_\infty\nls \big(\tfrac{2\pi}{L}\ls \vec{k}\big)\ls \beta^{\ls s}_{\varrho,\ls L}(\vec{k});\ls ds}.
            \end{equation}
            This can be verified by multiplying both sides of the equation by the factor $e^{i\ls \frac{4\pi^2\abs{\vec{n}}^2\!\nls }{L^2}\,t}$, and differentiating with respect to time.\newline
            Summing absolute values over $\vec{n}\in\Z^3$ across both sides of~\eqref{eq:DuhamelFourierCoefficients} gives
            \begin{align*}
                S^{\ls t}_{\nls \varrho,\ls L}\leq S^{\ls 0}_{\nls \varrho,\ls L}+\tfrac{\norm{\Psi_{\mspace{-3.75mu}\varrho,\mspace{0.75mu} L}}[2]^2\nls}{\varrho L^3}\!\integrate[0;t]{\!\nls \sum_{\vec{k}\ls \in\,\Z^3} \abs*{\FT{V}_\infty\nls \big(\tfrac{2\pi}{L}\ls \vec{k}\big)\nls }\ls \abs{\beta^{\ls s}_{\varrho,\ls L}(\vec{k})}\,S^{\ls s}_{\nls \varrho,\ls L};\ls ds}.
            \end{align*}
            Hence, applying~\eqref{eq:autoCorrelationControlledByS2}
            \begin{equation*}
                S^{\ls t}_{\nls \varrho,\ls L}\leq S^{\ls 0}_{\nls \varrho,\ls L}+\mathfrak{b}\!\integrate[0;t]{(S^{\ls s}_{\nls \varrho,\ls L})^3;\ls ds}=\vcentcolon \tilde{S}^{\ls t}_{\nls \varrho,\ls L}.
            \end{equation*}
            This implies
            \begin{align*}
                S^{\ls t}_{\nls \varrho,\ls L}\nls \leq \tilde{S}^{\ls t}_{\nls \varrho,\ls L},\qquad S^{\ls 0}_{\nls \varrho,\ls L}\nls = \tilde{S}^{\ls 0}_{\nls \varrho,\ls L}, && \partial_t\mspace{0.75mu} \tilde{S}^{\ls t}_{\nls \varrho,\ls L} \!=(S^{\ls t}_{\nls \varrho,\ls L})^3\ls \mathfrak{b}\leq  (\tilde{S}^{\ls t}_{\nls \varrho,\ls L})^3\ls \mathfrak{b}.
            \end{align*}
            Equivalently, since $\tilde{S}^{\ls t}_{\nls \varrho,\ls L}\nls \geq S^{\ls t}_{\nls \varrho,\ls L}\gneq 0$ (\cfr~equation~\eqref{eq:SumLowerBound})
            $$\partial_t\nls  \left[-\frac{1}{2(\tilde{S}^{\ls t}_{\nls \varrho,\ls L})^2}-\mathfrak{b}\, t\right]\nls \leq 0,$$
            which means that the function of $t$ in the square bracket is non-increasing.
            Therefore,
            \begin{align}
                &-\frac{1}{2(\tilde{S}^{\ls t}_{\nls \varrho,\ls L})^2}-\mathfrak{b}\, t\leq -\frac{1}{2(S^{\ls 0}_{\nls \varrho,\ls L})^2}\nonumber\\
                \implies\quad &S^{\ls t}_{\nls \varrho,\ls L}\leq \tilde{S}^{\ls t}_{\nls \varrho,\ls L}\leq \frac{S^{\ls 0}_{\nls \varrho,\ls L}}{\!\nls \sqrt{1-2\ls (S^{\ls 0}_{\nls \varrho,\ls L})^2\ls \mathfrak{b}\,t\ls }\ls },\qquad t < \frac{1}{2\ls (S^{\ls 0}_{\nls \varrho,\ls L})^2\ls \mathfrak{b}}.\label{eq:SumNonlinearityControlInTime}
            \end{align}
            Due to the summability of $\alpha^{\ls t}_{\varrho,\ls L}$, we know that for fixed $\varrho,L\nls >\nls 0$ the sum $S^{\ls t}_{\nls \varrho,\ls L}$ cannot actually blow up at finite times.
            However, in principle $S^{\ls t}_{\nls \varrho,\ls L}$ could grow indefinitely with $\varrho$ and/or $L$, without further information.
            Denoting by $\ls\overline{\nls S}^{\, 0}\!\vcentcolon=\limsup\limits_{\varrho\to\infty}\limsup\limits_{L\to\infty} S^{\ls 0}_{\nls \varrho,\ls L}$, we can only deduce that
            $$\limsup_{\varrho\to\infty}\limsup_{L\to\infty} S^{\ls t}_{\nls \varrho,\ls L}\leq \frac{\ls\overline{\nls S}^{\, 0}}{\!\nls \sqrt{1-2\ls (\ls\overline{\nls S}^{\, 0})^2\ls \mathfrak{b}\,t\ls }\ls }, \qquad t \nls <\nls  \frac{1}{\vphantom{\big|}2\ls (\ls\overline{\nls S}^{\, 0})^2\ls \mathfrak{b}}.$$
            In the last step we exploited the fact that the limit superior commutes with non-decreasing, continuous functions.\newline
            To conclude, we prove that the convergence of the momentum distribution to a Kronecker delta minimises the value of $\ls\overline{\nls S}^{\, 0}\!$.
            To this end, we use~\eqref{eq:becNormalization} and split the sum
            $$\ls\overline{\nls S}^{\, 0}\!\leq\limsup_{\varrho\to\infty}\limsup_{L\to\infty} \!\sum_{\substack{\vec{m}\ls \in\,\Z^3 \ls :\\[1.5pt] \abs{\vec{m}} \leq M}} \abs{\alpha^{\ls 0}_{\varrho,\ls L}(\vec{m})}+\limsup_{\varrho\to\infty}\limsup_{L\to\infty} \!\sum_{\substack{\vec{m}\ls \in\,\Z^3 \ls :\\[1.5pt] \abs{\vec{m}}\ls  > M}} \abs{\alpha^{\ls 0}_{\varrho,\ls L}(\vec{m})}.$$
            For the first term, we make use of the sub-additivity of the limit superior, while the second is controlled by assumption, so that
            $$\ls\overline{\nls S}^{\, 0}\leq \!\!\sum_{\substack{\vec{m}\ls \in\,\Z^3 \ls :\\[1.5pt] \abs{\vec{m}} \leq M}} \limsup_{\varrho\to\infty}\limsup_{L\to\infty}\,\abs{\alpha^{\ls 0}_{\varrho,\ls L}(\vec{m})}+\epsilon.$$
            Furthermore, by hypothesis \textit{ii)} we have concentration in a single mode, namely
            $$\ls\overline{\nls S}^{\, 0}\leq \!\!\sum_{\substack{\vec{m}\ls \in\,\Z^3 \ls :\\[1.5pt] \abs{\vec{m}} \leq M}}\mspace{9mu} \underbrace{\mspace{-9mu}\limsup_{\varrho\to\infty}\limsup_{L\to\infty}\,\abs*{\abs{\alpha^{\ls 0}_{\varrho,\ls L}(\vec{m})}-\delta_{\vec{m},\ls \vec{k}_0}}\mspace{-9mu}}_{=\ls 0}\mspace{9mu}\,+\!\!\sum_{\substack{\vec{m}\ls \in\,\Z^3 \ls :\\[1.5pt] \abs{\vec{m}} \leq M}}\!\!\delta_{\vec{m},\ls \vec{k}_0}+\epsilon.$$
            Thus,
            $$1\leq\ls\overline{\nls S}^{\, 0}\leq 1+\epsilon,$$
            but since the inequality holds for any arbitrary $\epsilon\nls >\nls 0$, the result follows.

        \end{proof}
    \end{prop}
    To push forward the understanding of the heuristic picture, let us suppose Assumption~\ref{ass:translationalInvariance} holds true.
    In particular, we stress that the closeness of $\abs{\beta^{\ls 0}_{\varrho,\ls L}}$ to $\{\delta_{\vec{k},\ls \vec{0}}\}_{\vec{k}\ls \in\,\Z^3}$ does not, by itself, imply the existence of a quasi-complete Bose--Einstein condensate.
    For instance, take into account the example given by
    $$\alpha^{\ls 0}_{\varrho,\ls L}(\vec{n})=\delta_{\vec{n},\ls (\lfloor L\rfloor,\ls 0,\ls 0)},$$
    where we have only one mode directed along the $\mathrm{x}$-axis which increases with $L$.
    Clearly, this satisfies the conditions $\norm{\alpha^{\ls 0}_{\varrho,\ls L}}[\ell_2(\Z^3)]=1$ and $\beta^{\ls 0}_{\varrho,\ls L}(\vec{k})=\delta_{\vec{k},\ls \vec{0}}$, but there is no condensation, since this mode escapes to infinity, and therefore
    $$\lim_{L\to\infty} \,\abs{\alpha^{\ls 0}_{\varrho,\ls L}(\vec{n})}=0,\qquad\forall \vec{n}\in\Z^3,\,\varrho>0,$$
    meaning that condition~\eqref{eq:becFourierCoefficients1} is not fulfilled.
    On the other hand, it is not true that we want to forbid escaping modes in general.
    Indeed, a prototype of momentum distribution we allow in our problem is
    $$\alpha^{\ls 0}_{\varrho,\ls L}(\vec{n})=\sqrt{\tfrac{\varrho}{\varrho\mspace{2.25mu}+1}\ls }\ls \delta_{\vec{n},\ls \vec{k}_0}+\sqrt{\tfrac{1}{\varrho\mspace{2.25mu}+1}\ls }\ls \delta_{\vec{n},\ls  (\lfloor L\rfloor,\ls 0,\ls 0)},\qquad\vec{k}_0\nls \in\Z^3.$$
    In this case, the first term is the one that yields the complete Bose--Einstein condensate in the Fourier mode $\vec{k}_0$ when $\varrho$ is large, while the second term provides a non-vanishing contribution to the kinetic energy per particle when $\varrho$ is fixed:
    $$\lim_{L\to\infty}\sum_{\vec{n}\ls \in\,\Z^3}\!\tfrac{4\pi^2\abs{\vec{n}}^2\!\nls }{L^2}\,\abs{\alpha^{\ls 0}_{\varrho,\ls L}(\vec{n})}^2=\tfrac{4\pi^2}{\varrho\,+\ls 1}.$$
    Clearly, we expect this kind of excitation energy to vanish when $\varrho$ also goes to infinity and the condensate is complete.
    This means that we want to avoid too fast escaping momenta (\cfr~the hypotheses of Proposition~\ref{th:LaplaceNonlinearityControl} below) whose corresponding kinetic energy does not vanish in the iterated limit.\newline
    However, in order for $\alpha^{\ls 0}_{\varrho,\ls L}$ to be associated with a quasi-complete Bose--Einstein condensate occurring in a momentum state, it is necessary that the amplitudes associated with every mode -- but a certain $\vec{k}_0$ -- vanish in the iterated limit.

    \bigskip

    Another important object we need to control is the supremum of the convolution between the potential and the square of the Laplacian of the order parameter; by Young's convolution inequality
    \begin{equation}
        \tfrac{1}{\varrho}\,\norm{V_L\!\ast\nls \abs{\Delta\Psi^{\ls t}_{\!\varrho,\ls L}}^2}[\infty]\leq \!\tfrac{1}{\varrho}\,\norm{V_L}[1]\,\norm{\Delta\Psi^{\ls t}_{\!\varrho,\ls L}}[\infty]^2.
    \end{equation}
    On the one hand, we recall that the $\Lp{1}$-norm of the potential on the torus is estimated from above by the $\Lp{1}$-norm of the function $V_\infty$, \ie~$\mathfrak{b}$.
    On the other hand, introducing the shortcut    \begin{equation}\label{def:FourierCoefficientsKineticShortcut}
        T^{\ls t}_{\nls \varrho,\ls L}\vcentcolon=\tfrac{\norm{\Psi_{\mspace{-3.75mu}\varrho,\mspace{0.75mu} L}}[2]\nls}{\!\nls\sqrt{\varrho L^3\ls}\ls}\!\sum_{\vec{n}\ls \in\,\Z^3} \frac{4\pi^2\abs{\vec{n}}^2\!\nls }{L^2}\,\abs{\alpha^{\ls t}_{\varrho,\ls L}(\vec{n})},
    \end{equation}
    one has, by identity~\eqref{eq:LaplaceBEC}
    \begin{equation}\label{eq:KineticLinftyControlled}
        \norm{\Delta\Psi^{\ls t}_{\!\varrho,\ls L}}[\infty]\leq \sqrt{\varrho\ls }\ls  T^{\ls t}_{\nls \varrho,\ls L}.
    \end{equation}
    Thus, in order to prove that
    \begin{equation}
        \limsup_{\varrho\to\infty}\limsup_{L\to\infty} \tfrac{1}{\varrho}\,\norm{V_L\!\ast\nls \abs{\Delta\Psi^{\ls t}_{\!\varrho,\ls L}}^2}[\infty]<\infty,
    \end{equation}
    it is enough to prove the boundedness of $T^{\ls t}_{\nls \varrho,\ls L}$ as $\varrho$ and $L$ grow.
    This is the content of the following proposition.
    \begin{prop}\label{th:LaplaceNonlinearityControl}
        Let $t\longmapsto\Psi^{\ls t}_{\!\varrho,\ls L}\!\in C^{\mspace{0.75mu}1}\big(\mspace{0.75mu}[0,\infty),\ls \mathfrak{A}^0(\Lambda_L)\mspace{-0.75mu}\big)\mspace{-0.75mu}\cap C^{\ls 0}\big(\mspace{0.75mu}[0,\infty),\ls  \mathfrak{A}^2(\Lambda_L)\mspace{-0.75mu}\big)$ solving the Hartree equation~\eqref{eq:HartreePDE}, and suppose also that Assumptions~\ref{ass:tailCondition} and~\ref{ass:kineticTailCondition} hold.\newline
        Additionally, let $\overline{\nls S}^{\, 0}\!\vcentcolon=\limsup\limits_{\varrho\to\infty}\limsup\limits_{L\to\infty}\ls\frac{1}{\!\nls\sqrt{\varrho\ls}\ls}\norm{\Psi^{\ls 0}_{\!\varrho,\ls L}}[\mathfrak{A}^0(\Lambda_L)]$.
        Then, one has for all $0\nls \leq\nls t\nls <\nls \big(2\ls(\overline{\nls S}^{\, 0})^2\ls \mathfrak{b}\big)^{-1}$
        $$\limsup_{\varrho\to\infty}\limsup_{L\to\infty} T^{\ls t}_{\nls \varrho,\ls L}\nls <\infty, $$
        where $T^{\ls t}_{\nls \varrho,\ls L}$ has been defined in~\eqref{def:FourierCoefficientsKineticShortcut}.
        \begin{note}\label{rmk:necessityOfKineticTailCondition}
            In principle, it could seem that the \emph{extensivity} of the energy ensured by Proposition~\ref{th:quantifyTotalEnergy}, together with Assumptions~\ref{ass:translationalInvariance} and~\ref{ass:tailCondition}, could by themselves prove the result of Proposition~\ref{th:LaplaceNonlinearityControl} without further assumptions.
            Actually, this cannot be the case and therefore the additional hypothesis introduced in Proposition~\ref{th:LaplaceNonlinearityControl} is really needed.
            Indeed, consider the example
            $$\alpha^{\ls 0}_{\varrho,\ls L}(\vec{n})=\sqrt{\tfrac{\varrho}{\varrho\mspace{2.25mu}+1}\ls }\ls \delta_{\vec{n},\ls \vec{k}_0}+\sqrt{\tfrac{1}{\varrho\mspace{2.25mu}+1}\ls }\ls \delta_{\vec{n},\ls  (\lfloor\varrho^{\ls 3/8} L\rfloor,\ls 0,\ls 0)},\qquad\vec{k}_0\nls \in\Z^3.$$
            One can verify that this distribution of momenta is compatible with the definition of a quasi-complete Bose--Einstein condensate, and at the same time, the corresponding kinetic energy is finite for fixed $\varrho$ and vanishes when the density grows
            $$\lim_{L\to\infty} \sum_{\vec{n}\ls \in\,\Z^3}\!\tfrac{4\pi^2 \abs{\vec{n}}^2\!\nls }{L^2}\,\abs{\alpha^{\ls 0}_{\varrho,\ls L}(\vec{n})}^2=\tfrac{4\pi^2\varrho^{3/4}\!\nls }{\varrho\mspace{2.25mu}+1}.$$
            However, one has
            $$\lim_{L\to\infty} \sum_{\vec{n}\ls \in\,\Z^3}\!\tfrac{4\pi^2 \abs{\vec{n}}^2\!\nls }{L^2}\,\abs{\alpha^{\ls 0}_{\varrho,\ls L}(\vec{n})}=\tfrac{4\pi^2\varrho^{3/4}\!\nls }{\!\nls \sqrt{\varrho\mspace{2.25mu}+1\ls }\ls },$$
            which diverges for $\varrho$ -- as a matter of fact, the fastest possible escaping momentum we can allow in this example has magnitude of order $\varrho^{1/4} L\ls $.
        \end{note}
        \begin{note}\label{rmk:T0Boundedness}
            We stress that Assumptions~\ref{ass:tailCondition} and~\ref{ass:kineticTailCondition} entail the finiteness of $\overline{T}^{\ls 0}\nls\vcentcolon=\limsup\limits_{\varrho\to\infty}\limsup\limits_{L\to\infty}T^{\ls 0}_{\nls \varrho,\ls L}\ls $.
            This can be seen by splitting the sum into three pieces
            \begin{align*}
                \overline{T}^{\ls 0}\!&=\limsup_{\varrho\to\infty}\limsup_{L\to\infty}\Bigg[\,\sum_{\substack{\vec{m}\ls \in\,\Z^3 \ls :\\[1.5pt] \abs{\vec{m}}\leq M}}\!\!\tfrac{4\pi^2\abs{\vec{m}}^2\!\nls }{L^2}\,\abs{\alpha^{\ls 0}_{\varrho,\ls L}(\vec{m})}\,+\nquad\sum_{\substack{\vec{m}\ls \in\,\Z^3 \ls :\\[1.5pt] M<\abs{\vec{m}}\leq \,c\ls L}}\nquad\tfrac{4\pi^2\abs{\vec{m}}^2\!\nls }{L^2}\,\abs{\alpha^{\ls 0}_{\varrho,\ls L}(\vec{m})}\,+\!\!\!\sum_{\substack{\vec{m}\ls \in\,\Z^3 \ls :\\[1.5pt] \abs{\vec{m}}>\ls  c\ls L}}\!\!\tfrac{4\pi^2\abs{\vec{m}}^2\!\nls }{L^2}\,\abs{\alpha^{\ls 0}_{\varrho,\ls L}(\vec{m})}\Bigg]\\[-5pt]
                &\leq\limsup_{\varrho\to\infty}\limsup_{L\to\infty}\Bigg[\ls\tfrac{4\pi^2 M^2\!\nls }{L^2}\,\big\lvert\{\vec{m}\!\in\nls\Z^3 \,\big\vert\: \abs{\vec{m}}\nls\leq\nls M\}\big\rvert^{\frac{1}{2}}\!+\,4\pi^2 c^2\mspace{-24mu}\sum_{\substack{\vec{m}\ls \in\,\Z^3 \ls :\\[1.5pt] M<\abs{\vec{m}}\leq \,c\ls L}}\nquad \abs{\alpha^{\ls 0}_{\varrho,\ls L}(\vec{m})}\,+\!\!\!\sum_{\substack{\vec{m}\ls \in\,\Z^3 \ls :\\[1.5pt] \abs{\vec{m}}>\ls  c\ls L}}\!\!\tfrac{4\pi^2\abs{\vec{m}}^2\!\nls }{L^2}\,\abs{\alpha^{\ls 0}_{\varrho,\ls L}(\vec{m})}\Bigg].
            \end{align*}
            For the first term, we exploited a Cauchy--Schwarz inequality and the normalisation of $\alpha^{\ls 0}_{\varrho,\ls L}\!\in\ell_2(\Z^3)$.
            For the second term, we adopt the item~\textit{i)} of the hypotheses of Proposition~\ref{th:nonlinearityControl}, so that
            \begin{equation*}
                \overline{T}^{\ls 0}\!\leq 4\pi^2 c^2\,\epsilon\,+\limsup_{\varrho\to\infty}\limsup_{L\to\infty}\!\!\sum_{\substack{\vec{m}\ls \in\,\Z^3 \ls :\\[1.5pt] \abs{\vec{m}}>\ls  c\ls L}}\!\!\tfrac{4\pi^2\abs{\vec{m}}^2\!\nls }{L^2}\,\abs{\alpha^{\ls 0}_{\varrho,\ls L}(\vec{m})}.
            \end{equation*}
            In other words, we have proved that
            \begin{equation*}
                \overline{T}^{\ls 0}\!\leq4\pi^2\,\limsup_{\varrho\to\infty}\limsup_{L\to\infty}\!\!\sum_{\substack{\vec{m}\ls \in\,\Z^3 \ls :\\[1.5pt] \abs{\vec{m}}>\ls  c\ls L}}\!\!\tfrac{\abs{\vec{m}}^2\!\nls }{L^2}\,\abs{\alpha^{\ls 0}_{\varrho,\ls L}(\vec{m})},
            \end{equation*}
            which is finite by assumption.
        \end{note}
        \begin{proof}[Proof of Proposition~\ref{th:LaplaceNonlinearityControl}]
            Multiply both sides of equation~\eqref{eq:DuhamelFourierCoefficients} by the factor $\frac{4\pi^2\abs{\vec{n}}^2\!\nls }{L^2}\ls\frac{\norm{\Psi_{\mspace{-3.75mu}\varrho,\mspace{0.75mu} L}}[2]\!}{\!\nls\sqrt{\varrho L^3\ls}\ls}$ and then take the absolute value and sum over $\vec{n}\in\Z^3$, to obtain
            \begin{align*}
                T^{\ls t}_{\nls \varrho,\ls L}\leq T^{\ls 0}_{\nls \varrho,\ls L}\,+\tfrac{\norm{\Psi_{\mspace{-3.75mu}\varrho,\mspace{0.75mu} L}}[2]^3}{(\varrho L^3)^{3/2}}&\!\nls \integrate[0;t]{\!\nls \sum_{\vec{k}\ls \in\,\Z^3}\nls \abs*{\FT{V}_\infty\nls \big(\tfrac{2\pi}{L}\ls \vec{k}\big)\nls }\ls \abs{\beta^{\ls s}_{\varrho,\ls L}(\vec{k})}\!\sum_{\vec{n}\ls \in\,\Z^3}\!\nls \tfrac{4\pi^2\abs{\vec{n}-\vec{k}\ls +\vec{k}}^2\!\nls }{L^2}\,\abs{\alpha^{\ls s}_{\varrho,\ls L}(\vec{n}\nls -\nls \vec{k})};\ls ds}\\
                \leq T^{\ls 0}_{\nls \varrho,\ls L}\,+&\:2\ls\tfrac{\norm{\Psi_{\mspace{-3.75mu}\varrho,\mspace{0.75mu} L}}[2]^2\nls}{\varrho L^3}\! \integrate[0;t]{\!\nls \sum_{\vec{k}\ls \in\,\Z^3}\nls \abs*{\FT{V}_\infty\nls \big(\tfrac{2\pi}{L}\ls \vec{k}\big)\nls }\ls \abs{\beta^{\ls s}_{\varrho,\ls L}(\vec{k})}\,T^{\ls s}_{\nls \varrho,\ls L};\ls ds}\,+\\
                +&\:2\ls\tfrac{\norm{\Psi_{\mspace{-3.75mu}\varrho,\mspace{0.75mu} L}}[2]^2\nls}{\varrho L^3}\! \integrate[0;t]{\!\nls \sum_{\vec{k}\ls \in\,\Z^3}\!\tfrac{4\pi^2\abs{\vec{k}}^2\!\nls }{L^2}\ls \abs*{\FT{V}_\infty\nls \big(\tfrac{2\pi}{L}\ls \vec{k}\big)\nls }\ls \abs{\beta^{\ls s}_{\varrho,\ls L}(\vec{k})}\,S^{\ls s}_{\nls \varrho,\ls L};\ls ds}\\
                \leq T^{\ls 0}_{\nls \varrho,\ls L}\,+&\:2\ls \mathfrak{b}\!\integrate[0;t]{\big(S^{\ls s}_{\nls \varrho,\ls L}\big)^2\,T^{\ls s}_{\nls \varrho,\ls L};\ls ds}+8\pi^2\!\sup_{\vec{k}\ls \in\,\Z^3} \!\Big[\tfrac{\abs{\vec{k}}^2\!}{L^2}\ls \abs*{\FT{V}_\infty\nls \big(\tfrac{2\pi}{L}\ls \vec{k}\big)\nls }\nls \Big]\nls \integrate[0;t]{\big(S^{\ls s}_{\nls \varrho,\ls L}\big)^3;\ls ds},
            \end{align*}
            where in the last step we have made use of~\eqref{eq:autoCorrelationControlledByS2}.
            Taking into account the decay of the potential~\eqref{eq:potentialDecay}, let us compute the maximum of the function
            $$[0,\infty)\nls \ni p\:\longmapsto \,\frac{4\pi^2 p^2\!\nls }{L^2}\,\frac{C}{\;\big(1+\tfrac{2\pi}{L}\ls p\big)^{3\ls +\ls \delta_2}\nquad}\quad.$$
            This function is non-negative, vanishing at $0$ and infinity; consequently, the unique critical point $p_c=\frac{L}{\pi(1+\delta_2)}$ is where the supremum is attained.
            Thus,
            \begin{equation}
                \sup_{\vec{k}\ls \in\,\Z^3} \!\Big[\tfrac{4\pi^2\abs{\vec{k}}^2\!}{L^2}\ls \abs*{\FT{V}_\infty\nls \big(\tfrac{2\pi}{L}\ls \vec{k}\big)\nls }\nls \Big]\leq 4\ls C\,\frac{\;(1+\ls \delta_2)^{1\ls +\ls \delta_2}\nquad}{\;(3+\ls \delta_2)^{\ls 3\ls +\ls \delta_2}\nquad}\quad < \frac{4}{27} \ls  C,
            \end{equation}
            which implies
            \begin{equation*}
                T^{\ls t}_{\nls \varrho,\ls L}\leq \underbrace{T^{\ls 0}_{\nls \varrho,\ls L}+\tfrac{8}{27}\ls C\!\nls \integrate[0;t]{\big(S^{\ls s}_{\nls \varrho,\ls L}\big)^3;\ls ds}}_{=\vcentcolon f_{\varrho,\ls L}(t)}\nls +\,2\ls \mathfrak{b}\!\integrate[0;t]{\big(S^{\ls s}_{\nls \varrho,\ls L}\big)^2\,T^{\ls s}_{\nls \varrho,\ls L};\ls ds}.
            \end{equation*}
            Here we can apply the integral version of Gr\"onwall's lemma, and since $t\longmapsto f_{\varrho,\ls L}(t)$ is non-decreasing regardless of $\varrho, L\nls >\nls 0$, one has
            \begin{equation*}
                T^{\ls t}_{\nls \varrho,\ls L}\leq f_{\varrho,\ls L}(t) \exp\nls \left[2\mathfrak{b}\!\nls \integrate[0;t]{\big(S^{\ls s}_{\nls \varrho,\ls L}\big)^2;\ls ds}\right]\nls .
            \end{equation*}
            Exploiting the upper bound~\eqref{eq:SumNonlinearityControlInTime}, one has for $0\nls \leq\nls  t\nls <\frac{1}{\vphantom{|^|}2\ls (S^{\ls 0}_{\nls \varrho,\ls L})^2\ls \mathfrak{b}}$
            \begin{align}
                T^{\ls t}_{\nls \varrho,\ls L}&\leq \left[T^{\ls 0}_{\nls \varrho,\ls L}+\tfrac{8}{27}\ls C\!\nls \integrate[0;t]{\frac{\big(S^{\ls 0}_{\nls \varrho,\ls L}\big)^3}{\big(1-2\ls (S^{\ls 0}_{\nls \varrho,\ls L})^2\ls \mathfrak{b}\,s\big)^{3/2}};\ls ds}\right] \exp\nls \left[\integrate[0;t]{\frac{2\ls \big(S^{\ls 0}_{\nls \varrho,\ls L}\big)^2\ls \mathfrak{b}}{1-2\ls (S^{\ls 0}_{\nls \varrho,\ls L})^2\ls \mathfrak{b}\,s};\ls ds}\right]\nonumber\\
                &=\left[T^{\ls 0}_{\nls \varrho,\ls L}+\tfrac{8}{27\,\mathfrak{b}}\ls C\, S^{\ls 0}_{\nls \varrho,\ls L} \left(\tfrac{1}{\!\nls \sqrt{1-2\ls (S^{\ls 0}_{\nls \varrho,\ls L})^2\ls \mathfrak{b}\,t\ls }\ls }-1\right)\!\right] \exp\!\Big[\nls -\ln\!\left(1-2\ls (S^{\ls 0}_{\nls \varrho,\ls L})^2\ls \mathfrak{b}\,t\right)\!\Big]\nonumber\\
                \implies T^{\ls t}_{\nls \varrho,\ls L}&\leq\frac{T^{\ls 0}_{\nls \varrho,\ls L}}{1-2\ls (S^{\ls 0}_{\nls \varrho,\ls L})^2\ls \mathfrak{b}\,t}+\tfrac{8}{27\,\mathfrak{b}}\ls C\, S^{\ls 0}_{\nls \varrho,\ls L}\! \left[\tfrac{1}{\left(1-2\ls (S^{\ls 0}_{\nls \varrho,\ls L})^2\ls \mathfrak{b}\,t\right)^{3/2}}-\tfrac{1}{\vphantom{\big[}1-2\ls (S^{\ls 0}_{\nls \varrho,\ls L})^2\ls \mathfrak{b}\,t}\right]\nls .\label{eq:SumKineticNonlinearityControlInTime}
            \end{align}
            The well-posedness prevents $T^{\ls t}_{\nls \varrho,\ls L}$ from blowing up at finite times, for fixed $\varrho,L\nls >\nls 0$.
            Nevertheless, one can only deduce that for $0\nls\leq\nls t\nls <\nls \big(2\ls(\overline{\nls S}^{\, 0})^2\ls \mathfrak{b}\big)^{-1}$
            \begin{equation}
                \limsup_{\varrho\to\infty}\limsup_{L\to\infty} T^{\ls t}_{\nls \varrho,\ls L}\leq \frac{\overline{T}^{\ls 0}}{1-2\ls \mathfrak{b}\ls(\overline{\nls S}^{\, 0})^2\,t}+\tfrac{8}{27\,\mathfrak{b}}\: \overline{\nls S}^{\, 0} C\nls\left[\tfrac{1}{\left(1-2\ls \mathfrak{b}\ls(\overline{\nls S}^{\, 0})^2\,t\right)^{3/2}}-\tfrac{1}{\vphantom{[^{/}}1-2\ls \mathfrak{b}\ls(\overline{\nls S}^{\, 0})^2\,t}\right]\!,
            \end{equation}
            where $\overline{T}^{\ls 0}\nls \vcentcolon=\limsup\limits_{\varrho\to\infty}\limsup\limits_{L\to\infty}T^{\ls 0}_{\nls \varrho,\ls L}\ls $.
            We stress that the last step is justified by the increase of the function $\big[0,\frac{1}{a}\big)\nls \ni x\,\longmapsto\: (1-a\ls x)^{-3/2}-(1-a\ls x)^{-1}$, for any $a\nls >\nls 0$.\newline
            The boundedness of $\overline{T}^{\ls 0}$ has been discussed in Remark~\ref{rmk:T0Boundedness}.

        \end{proof}
    \end{prop}
    \begin{note}\label{rmk:vanishingKineticBEC}
        Assumptions~\ref{ass:tailCondition} and~\ref{ass:kineticTailCondition} force the kinetic energy per particle of the initial quasi-complete Bose--Einstein condensate to vanish when $\varrho$ is large.
        Indeed,
        \begin{align*}
            \sum_{\vec{m}\ls \in\,\Z^3}\!\tfrac{4\pi^2\abs{\vec{m}}^2\!\nls }{L^2} \,\abs{\alpha^{\ls 0}_{\varrho,\ls L}(\vec{m})}^2&=\sum_{\substack{\vec{m}\ls \in\,\Z^3 \ls :\\[1.5pt] \abs{\vec{m}}\leq M}}\!\!\tfrac{4\pi^2\abs{\vec{m}}^2\!\nls }{L^2} \,\abs{\alpha^{\ls 0}_{\varrho,\ls L}(\vec{m})}^2\,+\!\!\!\sum_{\substack{\vec{m}\ls \in\,\Z^3 \ls :\\[1.5pt] \abs{\vec{m}}>\ls  M}}\!\!\tfrac{4\pi^2\abs{\vec{m}}^2\!\nls }{L^2} \,\abs{\alpha^{\ls 0}_{\varrho,\ls L}(\vec{m})}^2\\
            &\leq \tfrac{4\pi^2 M^2\!\nls }{L^2} +\max_{\vec{n}\ls \in\,\Z^3}\nls \left\{\tfrac{4\pi^2\abs{\vec{n}}^2\!\nls }{L^2} \,\abs{\alpha^{\ls 0}_{\varrho,\ls L}(\vec{n})}\right\}\!\!\!\sum_{\substack{\vec{m}\ls \in\,\Z^3 \ls :\\[1.5pt] \abs{\vec{m}}>\ls  M}}\!\abs{\alpha^{\ls 0}_{\varrho,\ls L}(\vec{m})}.
        \end{align*}
        Due to $\limsup\limits_{\varrho\to\infty}\limsup\limits_{L\to\infty} T^{\ls 0}_{\nls \varrho,\ls L}<\infty$, the quantity $\big\{\tfrac{\abs{\vec{n}}^2\!\nls }{L^2} \,\abs{\alpha^{\ls 0}_{\varrho,\ls L}(\vec{n})}\big\}_{\vec{n}\ls \in\,\Z^3}$ has a maximum value that stays finite in the iterated limit.
        Hence, by taking the iterated high-density thermodynamic limit first and then $M\to\infty$, the result is obtained because of Assumption~\ref{ass:tailCondition}.
        \hide{This argument shows that Assumption~\ref{ass:kineticTailCondition} strengthens the convergence to zero of the kinetic energy per particle, already provided by the convergence of the order parameter in the high-density thermodynamic limit:
        $$\lim_{\varrho\to\infty}\limsup_{L\to\infty}\sum_{\vec{m}\ls \in\,\Z^3}\!\nls \big(1+\tfrac{4\pi^2\abs{\vec{m}}^2\!}{L^2}\ls \big) \nls \abs*{\ls \alpha^{\ls 0}_{\varrho,\ls L}(\vec{m})-e^{i\ls \vartheta}\delta_{\vec{m},\ls \vec{k}_0}}^2\!=0.$$
        }
    \end{note}
    \subsection{Evolution of the Macroscopic Counterpart}
    Now we are ready to prove\footnote{We would like to thank our colleague and friend Giuseppe Lipardi for helpful suggestions on this topic.} the existence of a macroscopic wave function playing the role of a Bose--Einstein condensate in the high-density thermodynamic limit for $t\nls >\nls 0$ when Assumption~\ref{ass:translationalInvariance} is considered.
    In this particular case, the initial macroscopic counterpart of the order parameter is a plane wave with momentum $\vec{k}_0\!\in\nls \Z^3$, owing to Proposition~\ref{th:BECmomentum}.
    In the following, we prove the closeness of the time-evolved macroscopic order parameter to a time-dependent plane wave with frequency $\mathfrak{b}$.
    \begin{prop}\label{th:propagationExistenceBEC}
        Take into account Assumptions~\ref{ass:initialBEC},~\ref{ass:translationalInvariance},~\ref{ass:tailCondition}, and~\ref{ass:kineticTailCondition}, and let
        $$\Psi^{\ls t}: \,\vec{y}\:\longmapsto e^{2\pi \ls i\, \vec{k}_0\ls  \cdot\, \vec{y}\ls -\ls i\ls \mathfrak{b}\ls  t\ls +\ls i\ls \vartheta}\in H^1(\Lambda_1),\quad t\in\R,\,\vartheta\in[0,2\pi)$$
        be the vector satisfying
        $$\lim_{\varrho\to\infty}\limsup_{L\to\infty}\,\norm*{\tfrac{\Psi_{\mspace{-3.75mu}\varrho,\mspace{0.75mu} L}}{\norm{\Psi_{\mspace{-3.75mu}\varrho,\mspace{0.75mu} L}}[2]\!\nls} -\tfrac{1}{L^{3/2}} \Psi^{\ls 0}\nls \big(\mspace{-0.75mu}\tfrac{\vec{\cdot}}{L}\mspace{-0.75mu}\big)\nls }[H^1(\Lambda_L)]^2\nls =0.$$
        Then, for all $t< (2\ls \mathfrak{b})^{-1}$
        $$\lim_{\varrho\to\infty}\limsup_{L\to\infty}\,\norm*{\tfrac{\Psi^{\ls t}_{\!\nls\varrho,\mspace{0.75mu} L}}{\vphantom{|^|}\norm{\Psi^{\ls t}_{\!\nls\varrho,\mspace{0.75mu} L}}[2]\!\nls} -\tfrac{1}{L^{3/2}} \Psi^{\ls t}\nls \big(\mspace{-0.75mu}\tfrac{\vec{\cdot}}{L}\mspace{-0.75mu}\big)\nls }[H^1(\Lambda_L)]^2\nls =0.$$
        \begin{proof}
            Define the wave function $$\Phi^{\ls t}_{\!\varrho,\ls L}:\,\vec{x}\:\longmapsto \sqrt{\varrho\ls }\ls  e^{\frac{2\pi \ls i}{L}\mspace{2.25mu} \vec{k}_0\ls  \cdot\, \vec{x}\ls -\ls i\ls \omega_L t\ls +\ls i\ls \vartheta}\in\mathfrak{A}^2(\Lambda_L), \qquad\text{with }\,\omega_L\vcentcolon=\tfrac{4\pi^2\abs{\vec{k}_0}^2\!}{L^2}+\norm{V_L}[1].$$
            By a direct computation, one has
            $$\norm*{\Phi^{\ls t}_{\!\varrho,\ls L}-\sqrt{\varrho\ls }\ls \Psi^{\ls t}\big(\tfrac{\vec{\cdot}}{L}\big)}[H^1(\Lambda_L)]^2\!=\varrho L^3\abs*{e^{-\ls i\ls \omega_L t}-e^{-\ls i\ls \mathfrak{b}\ls t}}^2\!\big(1+\tfrac{4\pi^2}{L^2}\abs{\vec{k}_0}^2\big).$$
            Therefore, we have for all $\varrho\nls >\nls 0$
            \begin{equation}
                \lim_{L\to\infty}\,\norm*{\ls\tfrac{\Phi^{\ls t}_{\!\varrho,\mspace{0.75mu} L}}{\!\nls\sqrt{\varrho L^3\ls}\ls}-\tfrac{1}{L^{3/2}}\ls \Psi^{\ls t}\big(\tfrac{\vec{\cdot}}{L}\big)}[H^1(\Lambda_L)]^2\!=0.
            \end{equation}
            Hence, in order to test the convergence, it suffices to prove that the $H^1$-distance between $\frac{\Psi^{\ls t}_{\!\nls\varrho,\mspace{0.75mu} L}}{\vphantom{|^|}\norm{\Psi^{\ls t}_{\!\nls\varrho,\mspace{0.75mu} L}}[2]\!\nls}$ and $\frac{\Phi^{\ls t}_{\!\varrho,\ls L}}{\!\nls\sqrt{\varrho L^3\ls}\ls}$ vanishes.
            As an advantage, $\Phi^{\ls t}_{\!\varrho,\ls L}$ is the unique solution to the Hartree equation~\eqref{eq:HartreePDE} with the plane wave $\sqrt{\varrho\ls }\ls \Psi^{\ls 0}\mspace{-0.75mu}\big(\tfrac{\vec{\cdot}}{L})\!\in\nls \mathfrak{A}^2(\Lambda_L)$ as initial datum.\newline
            Let $u^{\ls t}_{\varrho,\ls L}\!\in\nls \mathfrak{A}^2(\Lambda_L)$ denote the difference $u^{\ls t}_{\varrho,\ls L}\!\vcentcolon=\Psi^{\ls t}_{\!\varrho,\ls L}\!-\Phi^{\ls t}_{\!\varrho,\ls L}$.
            It is straightforward to check that $u^{\ls t}_{\varrho,\ls L}$ solves the following nonlinear PDE
            \begin{equation}
                \begin{dcases}
                    i\ls \partial_t\ls  u^{\ls t}_{\varrho,\ls L}=(-\Delta + \norm{V_L}[1])\ls u^{\ls t}_{\varrho,\ls L}+\frac{1}{\varrho} \Big[V_L\!\ast\nls  \Big(\abs{u^{\ls t}_{\varrho,\ls L}}^2\!+2\ls \Re\ls \big(\conjugate{\Phi}^{\ls t}_{\!\varrho,\ls L}\, u^{\ls t}_{\varrho,\ls L}\big)\!\Big)\nls \Big](\Phi^{\ls t}_{\!\varrho,\ls L}\nls +u^{\ls t}_{\varrho,\ls L}),\qquad\text{in }\,\Lambda_L\\
                    u^{\ls 0}_{\varrho,\ls L}=\Psi_{\!\varrho,\ls L}-\sqrt{\varrho\ls }\ls \Psi^{\ls 0}\big(\tfrac{\vec{\cdot}}{L}\big)\in\mathfrak{A}^2(\Lambda_L).
                \end{dcases}
            \end{equation}
            We recall that by Assumption~\ref{ass:initialBEC} and~\eqref{eq:becNormalization} (see equation~\ref{eq:becFourierCoefficients1})
            \begin{equation}\label{eq:initialMassIsCloseToZero}
                \lim_{\varrho\to\infty}\limsup_{L\to\infty}\,\tfrac{1}{\!\nls \sqrt{\varrho L^3\ls }\ls }\norm{u^{\ls 0}_{\varrho,\ls L}}[2]\leq \!\lim_{\varrho\to\infty}\limsup_{L\to\infty}\left[\norm*{\tfrac{\Psi_{\mspace{-3.75mu}\varrho,\mspace{0.75mu} L}}{\norm{\Psi_{\mspace{-3.75mu}\varrho,\mspace{0.75mu} L}}[2]\!\nls}-\tfrac{1}{L^{3/2}}\Psi^{\ls 0}\big(\tfrac{\vec{\cdot}}{L}\big)}[2]\!+\abs*{\tfrac{\norm{\Psi_{\mspace{-3.75mu}\varrho,\mspace{0.75mu} L}}[2]}{\!\nls\sqrt{\varrho L^3\ls}\ls}-1}\,\right]\!=0.
            \end{equation}
            Our first goal is to prove that the same result holds for $t\nls >\nls 0$.\newline
            We introduce the shortcut $v^{\ls t}_{\varrho,\ls L}\vcentcolon=\abs{u^{\ls t}_{\varrho,\ls L}}^2+2\ls \Re\ls \big(\conjugate{\Phi}^{\ls t}_{\!\varrho,\ls L}\ls  u^{\ls t}_{\varrho,\ls L}\big)$, so that
            \begin{align*}
                \frac{d}{d\ls t}\norm{u^{\ls t}_{\varrho,\ls L}}[2]^2&=2\ls \Re\,\scalar{u^{\ls t}_{\varrho,\ls L}}{\partial_t \ls u^{\ls t}_{\varrho,\ls L}}[2]=2\ls \Im\,\scalar{u^{\ls t}_{\varrho,\ls L}}{(-\Delta+\norm{V_L}[1]) \ls u^{\ls t}_{\varrho,\ls L}+\tfrac{1}{\varrho} \big(V_L\!\ast\nls v^{\ls t}_{\varrho,\ls L}\big)(\Phi^{\ls t}_{\!\varrho,\ls L}\mspace{-2.25mu}+u^{\ls t}_{\varrho,\ls L})}[2]\\
                &=2\ls \Im\,\scalar{u^{\ls t}_{\varrho,\ls L}}{\tfrac{1}{\varrho} \big(V_L\!\ast\nls v^{\ls t}_{\varrho,\ls L}\big)\ls \Phi^{\ls t}_{\!\varrho,\ls L}}[2].
            \end{align*}
            By means of Cauchy--Schwarz's inequality and Young's inequality for convolutions, one gets
            $$\frac{d}{d\ls t}\norm{u^{\ls t}_{\varrho,\ls L}}[2]^2\leq \tfrac{2}{\varrho}\ls \norm{u^{\ls t}_{\varrho,\ls L}}[2]\,\norm{V_L}[1]\ls \norm{v^{\ls t}_{\varrho,\ls L}}[2]\,\norm{\Phi^{\ls t}_{\!\varrho,\ls L}}[\infty]\leq \tfrac{2\ls \norm{V_L}[1]\!\nls}{\!\nls \sqrt{\varrho\ls }\ls }\, \norm{u^{\ls t}_{\varrho,\ls L}}[2]\ls \norm{v^{\ls t}_{\varrho,\ls L}}[2].$$
            Hence, let us estimate the $\Lp{2}$-norm of $v^{\ls t}_{\varrho,\ls L}$
            $$\norm{v^{\ls t}_{\varrho,\ls L}}[2]\leq\norm{u^{\ls t}_{\varrho,\ls L}}[4]^2+2\ls \norm{\conjugate{\Phi}^{\ls t}_{\!\varrho,\ls L}\ls u^{\ls t}_{\varrho,\ls L}}[2]\leq \big(\norm{u^{\ls t}_{\varrho,\ls L}}[\infty]\nls +2\ls \sqrt{\varrho\ls }\ls \big)\norm{u^{\ls t}_{\varrho,\ls L}}[2].$$
            Since $\norm{u^{\ls t}_{\varrho,\ls L}}[\infty]\leq \norm{\Psi^{\ls t}_{\!\varrho,\ls L}}[\infty]+\norm{\Phi^{\ls t}_{\!\varrho,\ls L}}[\infty]\leq \sqrt{\varrho\ls }\ls \big(S^{\ls t}_{\nls \varrho,\ls L}\nls +1\big)$ (owing to~\eqref{eq:SumLinftyControlled}),
            \begin{equation*}
                \norm{v^{\ls t}_{\varrho,\ls L}}[2]\leq\sqrt{\varrho\ls }\ls \big(S^{\ls t}_{\nls \varrho,\ls L}\nls +3\big)\norm{u^{\ls t}_{\varrho,\ls L}}[2].
            \end{equation*}
            Thus, recalling $\norm{V_L}[1]\leq\FT{V}_{\infty}(\vec{0})=\mathfrak{b}$,
            \begin{equation}
                \frac{d}{d\ls t}\norm{u^{\ls t}_{\varrho,\ls L}}[2]^2\leq 2\ls \mathfrak{b}\big(S^{\ls t}_{\nls \varrho,\ls L}\nls +3\big)\norm{u^{\ls t}_{\varrho,\ls L}}[2]^2\ls ,
            \end{equation}
            and by Gr{\"o}nwall's inequality
            \begin{equation}
                \norm{u^{\ls t}_{\varrho,\ls L}}[2]^2\leq \norm{u^{\ls 0}_{\varrho,\ls L}}[2]^2 \,\exp\!\bigg(6\ls \mathfrak{b}\ls t+2\mathfrak{b}\!\integrate[0;t]{S^{\ls s}_{\nls \varrho,\ls L};\ls  ds}\bigg).
            \end{equation}
            For an explicit estimate, we make use of~\eqref{eq:SumNonlinearityControlInTime} to obtain
            \begin{equation*}
                \norm{u^{\ls t}_{\varrho,\ls L}}[2]^2\leq \norm{u^{\ls 0}_{\varrho,\ls L}}[2]^2 \,\exp\!\bigg(6\ls \mathfrak{b}\ls t+\tfrac{2-2\sqrt{1-2\ls (S^{\ls 0}_{\nls \varrho,\ls L})^2\ls \mathfrak{b}\,t\ls }\ls }{S^{\ls 0}_{\nls \varrho,\ls L}}\bigg),\qquad 0\leq t<\big(2\ls (S^{\ls 0}_{\nls \varrho,\ls L})^2\ls \mathfrak{b}\big)^{-1}.
            \end{equation*}
            Since the argument of the exponential is an increasing function with respect to $S^{\ls 0}_{\nls \varrho,\ls L}$, one has for $t\nls <\nls  (2\ls \mathfrak{b})^{-1}$
            \begin{equation}\label{eq:convergenceBEC}
                \limsup_{\varrho\to\infty}\limsup_{L\to\infty}\,\tfrac{1}{\varrho L^3}\norm{u^{\ls t}_{\varrho,\ls L}}[2]^2\leq \lim_{\varrho\to\infty}\limsup_{L\to\infty}\,\tfrac{1}{\varrho L^3}\norm{u^{\ls 0}_{\varrho,\ls L}}[2]^2 \,\exp\!\Big(6\ls \mathfrak{b}\ls t+2-2\sqrt{1-2\ls \mathfrak{b}\,t\ls }\ls \Big),
            \end{equation}
            which implies the convergence of $\frac{L^{3/2}}{\vphantom{|^|}\norm{\Psi^{\ls t}_{\!\nls\varrho,\mspace{0.75mu} L}}[2]\!\nls}\ls\Psi^{\ls t}_{\!\varrho,\ls L}(L\,\vec{\cdot}\ls )$ in $\Lp{2}[\Lambda_1]$ because of~\eqref{eq:initialMassIsCloseToZero}.
            Additionally, one also has
            \begin{equation}\label{eq:neededInTheEnd}
                \limsup_{\varrho\to\infty}\limsup_{L\to\infty}\,\norm*{\tfrac{\Psi^{\ls t}_{\!\nls\varrho,\mspace{0.75mu} L}}{\vphantom{|^|}\norm{\Psi^{\ls t}_{\!\nls\varrho,\mspace{0.75mu} L}}[2]\!\nls}-\tfrac{1}{L^{3/2}}\Psi^{\ls t}\big(\tfrac{\vec{\cdot}}{L}\big)}[2]^2\!\leq \lim_{\varrho\to\infty}\limsup_{L\to\infty}\,\left[\tfrac{2}{\varrho L^3}\norm{u^{\ls t}_{\varrho,\ls L}}[2]^2+2\abs*{1-\tfrac{\norm{\Psi_{\mspace{-3.75mu}\varrho,\mspace{0.75mu} L}}[2]}{\!\nls\sqrt{\varrho L^3\ls}\ls}}^2\right]\!= 0
            \end{equation}
            for all $t<(2\ls\mathfrak{b})^{-1}$.
            
            \smallskip

            \noindent Concerning the convergence in $\Lp{2}[\Lambda_L]$ of the derivative, note that $\frac{1}{\!\nls\sqrt{\varrho L^3\ls }\ls }\ls\norm{\nabla \Phi^{\ls t}_{\!\varrho,\ls L}}[2]=\frac{2\pi}{L}\abs{\vec{k}_0}\xrightarrow[L\to\infty]{} 0$.
            Therefore, let us prove similarly that $\lim\limits_{\varrho\to\infty}\limsup\limits_{L\to\infty}\,\frac{1}{\varrho L^3}\ls\norm{\nabla \Psi^{\ls t}_{\!\varrho,\ls L}}[2]^2=0$.
            Since we already know that the kinetic energy at time $0$ vanishes in the iterated high-density thermodynamic limit, we can proceed to close a Gr\"onwall lemma.
            \begin{align*}
                \frac{d}{d\ls t}\norm{\nabla \Psi^{\ls t}_{\!\varrho,\ls L}}[2]^2&=-2\ls \Re\,\scalar{\partial_t \Psi^{\ls t}_{\!\varrho,\ls L}}{\Delta \Psi^{\ls t}_{\!\varrho,\ls L}}[2]=2\ls \Im\,\scalar{-\Delta \Psi^{\ls t}_{\!\varrho,\ls L}\nls +\tfrac{1}{\varrho} \big(V_L\!\ast\nls \abs{\Psi^{\ls t}_{\!\varrho,\ls L}}^2\big)\Psi^{\ls t}_{\!\varrho,\ls L}}{\Delta \Psi^{\ls t}_{\!\varrho,\ls L}}[2]\\
                &=2\ls \Im\,\scalar{\tfrac{1}{\varrho} \big(V_L\!\ast\nls \abs{\Psi^{\ls t}_{\!\varrho,\ls L}}^2\big)\Psi^{\ls t}_{\!\varrho,\ls L}}{\Delta \Psi^{\ls t}_{\!\varrho,\ls L}}[2]=-2\ls \Im\,\scalar{\tfrac{1}{\varrho} \big(\nabla V_L\!\ast\nls \abs{\Psi^{\ls t}_{\!\varrho,\ls L}}^2\big)\Psi^{\ls t}_{\!\varrho,\ls L}}{\nabla \Psi^{\ls t}_{\!\varrho,\ls L}}[2]\\
                &=-2\ls \Im\,\scalar{\tfrac{1}{\varrho} \big(V_L\!\ast\nls \nabla\abs{\Psi^{\ls t}_{\!\varrho,\ls L}}^2\big)\Psi^{\ls t}_{\!\varrho,\ls L}}{\nabla \Psi^{\ls t}_{\!\varrho,\ls L}}[2].
            \end{align*}
            By Cauchy--Schwarz and Young's inequality for convolutions,
            \begin{equation*}
                \frac{d}{d\ls t}\norm{\nabla \Psi^{\ls t}_{\!\varrho,\ls L}}[2]^2\leq \tfrac{2}{\varrho} \ls \norm{V_L}[1]\ls\norm{\Psi^{\ls t}_{\!\varrho,\ls L}}[\infty]\,\norm{\ls\nabla \abs{\Psi^{\ls t}_{\!\varrho,\ls L}}^2\ls}[2]\ls\norm{\nabla\Psi^{\ls t}_{\!\varrho,\ls L}}[2].
            \end{equation*}
            Since $\nabla\abs{\Psi^{\ls t}_{\!\varrho,\ls L}}^2=2\ls\Re\, \conjugate*{\Psi}^{\ls t}_{\!\varrho,\ls L}\nabla \Psi^{\ls t}_{\varrho,\ls L}$, one has
            \begin{equation*}
                \frac{d}{d\ls t}\norm{\nabla \Psi^{\ls t}_{\!\varrho,\ls L}}[2]^2\leq 4\ls \mathfrak{b}\,\tfrac{\norm{\Psi^{\ls t}_{\!\varrho,\ls L}}[\infty]^2\!\nls}{\varrho}\;\norm{\nabla\Psi^{\ls t}_{\!\varrho,\ls L}}[2]^2\leq 4\ls \mathfrak{b}\,\big(S^{\ls t}_{\nls\varrho,\ls L}\big)^2\,\norm{\nabla\Psi^{\ls t}_{\!\varrho,\ls L}}[2]^2\ls.
            \end{equation*}
            This proves that the kinetic energy of the time-dependent order parameter keeps to vanish in the iterated high-density thermodynamic limit for all $t<\frac{1}{2\ls\mathfrak{b}}$.

        \end{proof}
    \end{prop}
    Next, we generalise the result obtained in Proposition~\ref{th:propagationExistenceBEC} by dropping Assumption~\ref{ass:translationalInvariance}.
    As a consequence, the initial macroscopic profile $\Psi$ of the order parameter associated with the quasi-complete Bose--Einstein condensate is no longer constrained to be a plane wave.
    Let us first prove the following lemma.
    \begin{lemma}\label{th:NLSconvergence}
        Let $\Psi\nls\in \mathfrak{A}^{2}(\Lambda_1)$ and consider
        $$t\longmapsto \Psi^{\ls t}=e^{-\ls i\ls \mathfrak{b}\ls\abs{\Psi}^2\ls t}\,\Psi \in C^{\infty}\mspace{-0.75mu}\big(\R,\ls \mathfrak{A}^{2}(\Lambda_1)\mspace{-0.75mu}\big).$$
        Moreover, let $\Phi^{\ls t}_{\!\varrho,\ls L}$ be the solution to the Hartree equation in $\Lambda_L$ with initial datum $\Phi^{\, 0}_{\!\varrho,\ls L}\!=\nls\sqrt{\varrho\ls}\ls \Psi\big(\tfrac{\vec{\cdot}}{L}\big)\mspace{-2.25mu}\in\mspace{-0.75mu}\mathfrak{A}^2(\Lambda_L)$.
        Then, for all $t<\big(2\ls\mathfrak{b}\, \norm{\FT{\Psi}}[\ell_1(\Z^3)]^2\big)^{-1}\!$, one has
        $$\lim_{\varrho\to\infty}\limsup_{L\to\infty}\,\tfrac{1}{\!\nls\sqrt{\varrho L^3\ls}\ls}\norm*{\Phi^{\ls t}_{\!\varrho,\ls L}\nls -\sqrt{\varrho\ls }\ls\Psi^{\ls t}\nls \big(\mspace{-0.75mu}\tfrac{\vec{\cdot}}{L}\mspace{-0.75mu}\big)\nls }[H^1(\Lambda_L)]\nls =0.$$
        \begin{proof}
            We start by observing that the problem can be entirely formulated in $\Lambda_1$ by rescaling.
            Indeed, let us introduce $\Psi^{\ls t}_{\! L}=\frac{1}{\!\nls\sqrt{\varrho\ls}\ls}\ls\Phi^{\ls t}_{\!\varrho,\ls L}\big(L\vec{\cdot}\nls\big)$; then, it becomes clear that $\frac{1}{\!\nls\sqrt{\varrho\ls}\ls}\ls\Phi^{\ls t}_{\!\varrho,\ls L}\big(L\vec{\cdot}\nls\big)$ does not depend on $\varrho$ as it solves the Hartree equation
            \begin{equation*}
                \begin{dcases}
                    i\ls\partial_t\Psi^{\ls t}_{\! L}=-\tfrac{1}{L^2}\Delta\Psi^{\ls t}_{\! L}+\big(W_{\nls L}\nls\ast\abs{\Psi^{\ls t}_{\! L}}^2\ls\big)\Psi^{\ls t}_{\! L},\qquad \text{in }\,\Lambda_1,\\
                    \Psi^{\, 0}_{\! L}=\Psi,
                \end{dcases}
            \end{equation*}
            where $\displaystyle W_L\,:\:\vec{y}\,\longmapsto\!\sum\limits_{\vec{m}\ls\in\,\Z^3}\!\! e^{2\pi \ls i\, \vec{m}\ls\cdot\ls\vec{y}}\,\FT{V}_{\infty}\big(\tfrac{2\pi}{L}\ls\vec{m}\big)$.
            Therefore, the result is proven as soon as we show
            \begin{align*}
                \lim_{L\to\infty}\norm{\Psi^{\ls t}_{\! L}\nls -\Psi^{\ls t}}[\Lp{2}[\Lambda_1]]=0,  &&  \lim_{L\to\infty}\tfrac{1}{L}\ls\norm{\nabla\Psi^{\ls t}_{\! L}\nls -\nabla\Psi^{\ls t}}[\Lp{2}[\Lambda_1]]=0. 
            \end{align*}
            However, $\frac{1}{L}\ls \norm{\nabla \Psi^{\ls t}}[\Lp{2}[\Lambda_1]]$ is clearly vanishing in the thermodynamic limit, since $\Psi^{\ls t}$ does not depend on $L$.
            Hence, one has to prove that the kinetic energy of $\Psi^{\ls t}_{\! L}$ converges to zero in the limit as well, provided its initial datum does.
            This computation has been already carried out in Proposition~\ref{th:propagationExistenceBEC}, and can be reproduced here in the same way.
            Thus, it remains to prove
            $$\lim_{L\to\infty}\norm{\Psi^{\ls t}_{\! L}\nls -\Psi^{\ls t}}[\Lp{2}[\Lambda_1]]=0.$$
            To this end, notice that $\Psi^{\ls t}$ solves the nonlinear PDE
            \begin{equation*}
            \begin{dcases}
                i\ls\partial_t \Psi^{\ls t}\!=\mathfrak{b}\ls\abs{\Psi^{\ls t}}^2\ls\Psi^{\ls t},\qquad\text{in }\,\Lambda_1,\\
                \Psi^{\ls 0}=\Psi.
            \end{dcases}
            \end{equation*}
            It is straightforward to prove uniqueness in $\mathfrak{A}^2(\Lambda_1)$ by Gr\"onwall lemma since $\mathfrak{A}^2$ is a Banach algebra (\cfr~Corollary~\ref{th:LWP}).
            Additionally, the exhibition of an explicit global-in-time solution provides global well-posedness.\newline
            Let us rephrase, via the Duhamel formula, the Cauchy problems for these two wave functions having the same initial datum:
            \begin{gather*}
                \Psi^{\ls t}_{\! L}=e^{i\ls t\ls\Delta/L^2}\Psi-i\!\integrate[0;t]{e^{i\ls (t-s)\Delta/L^2}\nls\big(W_{\nls L}\nls \ast \abs{\Psi^{\ls s}_{\! L}}^2\big)\Psi^{\ls s}_{\! L};ds},\\
                \Psi^{\ls t}=\Psi-i\ls\mathfrak{b}\!\integrate[0;t]{\abs{\Psi^{\ls s}}^2\ls\Psi^{\ls s};ds}.
            \end{gather*}
            By subtracting (observe that $\abs{\Psi^{\ls t}}=\abs{\Psi}$)
            \begin{equation}
            \begin{split}
                \Psi^{\ls t}_{\! L}\nls-\nls\Psi^{\ls t}=&\,\big(e^{i\ls t\ls\Delta/L^2}\!\!-\Id_{\Lp{2}[\Lambda_1]}\big)\Psi-i\ls\mathfrak{b}\!\integrate[0;t]{\big(e^{i\ls (t-s)\ls\Delta/L^2}\!\!-\Id_{\Lp{2}[\Lambda_1]}\big)\abs{\Psi}^2\ls\Psi^{\ls s};ds}\,+\\
                &-i\!\integrate[0;t]{e^{i\ls (t-s)\ls\Delta/L^2}\Big[\Big(W_{\nls L}\nls \ast (\ls\abs{\Psi^{\ls s}_{\! L}}^2\nls-\abs{\Psi^{\ls s}}^2\big)\!\Big)\Psi^{\ls s}_{\! L}\,+\\
                &\mspace{159mu}+\big(W_{\nls L}\nls \ast\abs{\Psi}^2\big)\big(\Psi^{\ls s}_{\! L}-\Psi^{\ls s}\big)+\big(W_{\nls L}\nls\ast\abs{\Psi}^2-\mathfrak{b}\abs{\Psi}^2\big)\Psi^{\ls s}\Big];ds}.
            \end{split}
            \end{equation}
            Thus,
            \begin{align*}
                \norm{\Psi^{\ls t}_{\! L}\nls-\nls\Psi^{\ls t}}[\Lp{2}[\Lambda_1]]\leq &\, \norm*{\big(e^{i\ls t\ls\Delta/L^2}\!\!-\Id_{\Lp{2}[\Lambda_1]}\big)\Psi}[\Lp{2}[\Lambda_1]] \!\!+\ls \mathfrak{b}\!\integrate[0;t]{\norm*{\big(e^{i\ls (t-s)\ls\Delta/L^2}\!\!-\Id_{\Lp{2}[\Lambda_1]}\big)\abs{\Psi}^2\Psi^{\ls s}}[\Lp{2}[\Lambda_1]];ds}\!\!+\\
                &+\integrate[0;t]{\Big[\norm*{\Big(W_{\nls L}\nls \ast (\ls\abs{\Psi^{\ls s}_{\! L}}^2\nls-\abs{\Psi^{\ls s}}^2\big)\!\Big)\Psi^{\ls s}_{\! L}}[\Lp{2}[\Lambda_1]]\!+\\
                &\mspace{75mu}+\norm*{\big(W_{\nls L}\nls \ast\abs{\Psi}^2\big)\big(\Psi^{\ls s}_{\! L}-\Psi^{\ls s}\big)}[\Lp{2}[\Lambda_1]]\nls+\norm*{\big(W_{\nls L}\nls\ast\abs{\Psi}^2-\mathfrak{b}\abs{\Psi}^2\big)\Psi^{\ls s}}[\Lp{2}[\Lambda_1]]\nls\Big];ds}.
            \end{align*}
            By Cauchy--Schwarz and Young's inequality for convolutions (recall $\abs{w}^2\nls -\abs{z}^2=(w-z)\conjugate{w}+\conjugate*{(w-z)}z\,$ for all $w,z\in\C$)
            \begin{gather*}
                \norm*{\mspace{-0.75mu}\Big(\nls W_{\nls L}\mspace{-2.25mu} \ast\mspace{-0.75mu} (\ls\abs{\Psi^{\ls s}_{\! L}}^2\mspace{-3.75mu}-\nls\abs{\Psi^{\ls s}}^2\big)\mspace{-3.75mu}\Big)\Psi^{\ls s}_{\! L}}[\Lp{2}[\Lambda_1]]\!\leq \norm{\Psi^{\ls s}_{\! L}}[\Lp{\infty}[\Lambda_1]]\ls\norm{W_{\nls L}}[\Lp{1}[\Lambda_1]]\big(\norm{\Psi^{\ls s}_{\! L}}[\Lp{\infty}[\Lambda_1]]\nls+\norm{\Psi}[\Lp{\infty}[\Lambda_1]]\big)\norm{\Psi^{\ls s}_{\! L}\nls-\mspace{-0.75mu}\Psi^{\ls s}}[\Lp{2}[\Lambda_1]],\\
                \norm*{\big(W_{\nls L}\nls \ast\abs{\Psi}^2\big)\big(\Psi^{\ls s}_{\! L}-\Psi^{\ls s}\big)}[\Lp{2}[\Lambda_1]]\leq \norm{W_{\nls L}}[\Lp{1}[\Lambda_1]]\ls\norm{\Psi}[\Lp{\infty}[\Lambda_1]]^2\ls\norm{\Psi^{\ls s}_{\! L}\nls -\mspace{-0.75mu}\Psi^{\ls s}}[\Lp{2}[\Lambda_1]],\\
                \norm*{\big(W_{\nls L}\nls\ast\abs{\Psi}^2\nls-\mathfrak{b}\abs{\Psi}^2\big)\Psi^{\ls s}}[\Lp{2}[\Lambda_1]]\leq \norm{\Psi}[\Lp{\infty}[\Lambda_1]]\norm*{W_{\nls L}\nls\ast\abs{\Psi}^2\nls -\mathfrak{b}\abs{\Psi}^2}[\Lp{2}[\Lambda_1]]\nls .
            \end{gather*}
            Hence, (notice that $\norm{W_{\nls L}}[\Lp{1}[\Lambda_1]]=\norm{V_L}[1]\leq\mathfrak{b}$)
            \begin{align*}
                \norm{\Psi^{\ls t}_{\! L}\nls-\nls\Psi^{\ls t}}[\Lp{2}[\Lambda_1]]\leq & \sup_{0\ls <\ls s \ls <\ls t} \norm*{\big(e^{i\ls s\ls\Delta/L^2}\!\!-\Id_{\Lp{2}[\Lambda_1]}\big)\Psi}[\Lp{2}[\Lambda_1]] \!\!+\norm{\Psi}[\Lp{\infty}[\Lambda_1]]\norm*{W_{\nls L}\nls\ast\abs{\Psi}^2\nls -\mathfrak{b}\abs{\Psi}^2}[\Lp{2}[\Lambda_1]] t \,+\\
                &+\ls \mathfrak{b}\!\integrate[0;t]{\!\sup_{0\ls <\ls s' <\ls t}\norm*{\big(e^{i\ls s' \Delta/L^2}\!\!-\Id_{\Lp{2}[\Lambda_1]}\big)\abs{\Psi}^2\Psi^{\ls s}}[\Lp{2}[\Lambda_1]];ds}\!+\\
                &+\mathfrak{b}\!\integrate[0;t]{\!\left(\norm{\Psi^{\ls s}_{\! L}}[\Lp{\infty}[\Lambda_1]]^2\nls +\norm{\Psi}[\Lp{\infty}[\Lambda_1]]\ls\norm{\Psi^{\ls s}_{\! L}}[\Lp{\infty}[\Lambda_1]]\nls +\norm{\Psi}[\Lp{\infty}[\Lambda_1]]^2\right)\nls\norm{\Psi^{\ls s}_{\! L}\nls -\mspace{-0.75mu}\Psi^{\ls s}}[\Lp{2}[\Lambda_1]]; ds}.
            \end{align*}
            On the right hand side, we have an increasing function of $t$; hence, the integral form of Gr\"onwall's lemma implies
            $$\norm{\Psi^{\ls t}_{\! L}\nls-\nls\Psi^{\ls t}}[\Lp{2}[\Lambda_1]]\nls\leq \big(f_L(t)+g_L(t)\big)\exp{\nls\left[\mathfrak{b}\!\integrate[0;t]{h_L(s); ds}\right]},$$
            where
            \begin{gather*}
                f_L(t)\vcentcolon=\norm{\Psi}[\Lp{\infty}[\Lambda_1]]\norm*{W_{\nls L}\nls\ast\abs{\Psi}^2\nls -\mathfrak{b}\abs{\Psi}^2}[\Lp{2}[\Lambda_1]] t\ls ,\\
                g_L(t)\vcentcolon=\sup_{0\ls <\ls s \ls <\ls t} \norm*{\big(e^{i\ls s\ls\Delta/L^2}\!\!-\Id_{\Lp{2}[\Lambda_1]}\big)\Psi}[\Lp{2}[\Lambda_1]]\!\!+\ls\mathfrak{b}\!\integrate[0;t]{\!\sup_{0\ls <\ls s' <\ls t}\norm*{\big(e^{i\ls s' \Delta/L^2}\!\!-\Id_{\Lp{2}[\Lambda_1]}\big)\abs{\Psi}^2\Psi^{\ls s}}[\Lp{2}[\Lambda_1]];ds}\nls ,\\
                h_L(t)\vcentcolon=\norm{\Psi^{\ls t}_{\! L}}[\Lp{\infty}[\Lambda_1]]^2\nls +\norm{\Psi}[\Lp{\infty}[\Lambda_1]]\ls\norm{\Psi^{\ls t}_{\! L}}[\Lp{\infty}[\Lambda_1]]\nls +\norm{\Psi}[\Lp{\infty}[\Lambda_1]]^2.
            \end{gather*}
            Owing to~\eqref{eq:SumNonlinearityControlInTime}, we know that
            $\norm{\Psi^{\ls t}_{\! L}}[\Lp{\infty}[\Lambda_1]]\leq \frac{\norm{\FT{\Psi}}[\ell_1(\Z^3)]}{\!\nls\sqrt{1-2\ls\mathfrak{b}\ls\norm{\FT{\Psi}}[\ell_1(\Z^3)]^2 t\ls }\ls }$, resulting in a uniform control in $L$ of the function $h_L(t)$ whenever $t<\big(2\ls\mathfrak{b}\norm{\FT{\Psi}}[\ell_1(\Z^3)]^2\big)^{-1}$.
            We therefore proceed to prove that both $f_L$ and $g_L$ vanish pointwise in the thermodynamic limit.
            By Parseval's identity,
            $$\norm*{W_{\nls L}\nls\ast\abs{\Psi}^2\nls-\mathfrak{b}\ls\abs{\Psi}^2}[\Lp{2}[\Lambda_1]]^2=\!\sum_{\vec{m}\ls\in\,\Z^3}\big\lvert\widehat{W}_{\!\nls L}(\vec{m})-\mathfrak{b}\big\rvert^2\:\Big\lvert\widehat{\abs{\Psi}^2}(\vec{m})\Big\rvert^2\!,$$
            where $\widehat{W}_{\!\nls L}(\vec{m})=\FT{V}_\infty\big(\tfrac{2\pi}{L}\ls\vec{m}\big)$.
            Since $\FT{V}_\infty\big(\tfrac{2\pi}{L}\ls\vec{m}\big)\leq \FT{V}_\infty(\vec{0})=\mathfrak{b}$ for all $m\in\Z^3$, we know that $4\mathfrak{b}^2\Big\lvert\widehat{\abs{\Psi}^2}\Big\rvert^2$ is a uniform summable majorant, \ie,
            $$\norm*{W_{\nls L}\nls\ast\abs{\Psi}^2\nls-\mathfrak{b}\ls\abs{\Psi}^2}[\Lp{2}[\Lambda_1]]^2\leq 4\ls\mathfrak{b}^2\ls \norm{\Psi}[\Lp{4}[\Lambda_1]]^4.$$
            By dominated convergence theorem, we have $\lim\limits_{L\to\infty} f_L(t)=0$ for all $t\in\R$.\newline
            Then, we stress that for all $\phi\in \Lp{2}[\Lambda_1]$
            \begin{align*}
                \sup_{0\ls <\ls s \ls <\ls t} \norm*{\big(e^{i\ls s\ls\Delta/L^2}\!\!-\Id_{\Lp{2}[\Lambda_1]}\big)\phi\ls}[\Lp{2}[\Lambda_1]]^2&\leq \!\!\sum_{\vec{m}\ls\in\,\Z^3} \;\sup_{0\ls <\ls s \ls <\ls t}\,\abs*{e^{i\ls \frac{4\pi^2\!\nls}{L^2}\ls s\, \abs{\vec{m}}^2}\!\!-1}^2\abs{\FT{\phi}(\vec{m})}^2\\
                &=4\!\!\sum_{\vec{m}\ls\in\,\Z^3} \;\sup_{0\ls <\ls s \ls <\ls t}\,\sin^2\!\Big(\mspace{-0.75mu}\tfrac{2\pi^2\ls s\,\abs{\vec{m}}^2\!\nls}{L^2}\,\Big)\ls\abs{\FT{\phi}(\vec{m})}^2.
            \end{align*}
            In the first inequality, we adopted Parseval's identity and the sub-additivity of the supremum.
            One has the following elementary identity for all $a>0$
            $$\max_{0\ls <\ls y\ls <\ls x} \sin^2(a\ls y) = \begin{cases}
                \sin^2(a\ls x) \quad & \text{if } x<\frac{\pi}{2a};\\
                1, & \text{if } x\geq \frac{\pi}{2a}.
            \end{cases}$$
            This implies that we always fall into the first case eventually, when $L$ is large enough.
            More precisely,
            $$\lim_{L\to\infty}\, \max_{0\ls <\ls s \ls <\ls t}\,\sin^2\!\big(\tfrac{2\pi^2\ls s}{L^2}\,\abs{\vec{m}}^2\big)=0,\qquad\forall \vec{m}\nls\in\nls\Z^3,\, t>0.$$
            Since $\phi\in\nls\Lp{2}[\Lambda_1]$ and $\sin^2(\cdot)$ is uniformly bounded, dominated convergence theorem yields
            $$\lim_{L\to\infty}\,\sup_{0\ls < \ls s \ls <\ls t} \norm*{\big(e^{i\ls s\ls\Delta/L^2}\!\!-\Id_{\Lp{2}[\Lambda_1]}\big)\ls\phi\ls}[\Lp{2}[\Lambda_1]]=0,\qquad \forall \phi\in\Lp{2}[\Lambda_1],\, t>0.$$
            Indeed, notice that both $\Psi$ and $\abs{\Psi}^2 \Psi^{\ls s}$ are square integrable functions.
            To finish the proof, we apply dominated convergence theorem again, as
            $\sup\limits_{0\ls <\ls s'<\ls t}\norm*{\big(e^{i\ls s'\Delta/L^2}\!\!-\Id_{\Lp{2}[\Lambda_1]}\big)\ls\abs{\Psi}^2\Psi^{\ls s}\ls}[\Lp{2}[\Lambda_1]]\nls\!\leq\nls 2\ls \norm{\Psi}[\Lp{6}[\Lambda_1]]^3$ implies
            $$\lim_{L\to\infty} \integrate[0;t]{\!\sup\limits_{0\ls <\ls s'<\ls t}\norm*{\big(e^{i\ls s'\Delta/L^2}\!\!-\Id_{\Lp{2}[\Lambda_1]}\big)\ls\abs{\Psi}^2\Psi^{\ls s}\ls}[\Lp{2}[\Lambda_1]]; ds}\!=0,\qquad \forall t>0.$$
            
        \end{proof}
    \end{lemma}
    \begin{prop}\label{th:propagationExistenceBECgeneral}
        Take into account Assumptions~\ref{ass:initialBEC},~\ref{ass:tailCondition}, and~\ref{ass:kineticTailCondition}, and let $\Psi\nls\in \mathfrak{A}^{2}(\Lambda_1)$
        be the vector satisfying
        $$\lim_{\varrho\to\infty}\limsup_{L\to\infty}\,\norm*{\tfrac{\Psi_{\mspace{-3.75mu}\varrho,\mspace{0.75mu} L}}{\norm{\Psi_{\mspace{-3.75mu}\varrho,\mspace{0.75mu} L}}[2]\!\nls} -\tfrac{1}{L^{3/2}} \Psi\nls \big(\mspace{-0.75mu}\tfrac{\vec{\cdot}}{L}\mspace{-0.75mu}\big)\nls }[H^1(\Lambda_L)]^2\nls =0.$$
        Then, for all $0<t< \big(2\ls \mathfrak{b}\ls\norm{\FT{\Psi}}[\ell_1(\Z^3)]^2\big)^{\! -1}$
        $$\lim_{\varrho\to\infty}\limsup_{L\to\infty}\,\norm*{\tfrac{\Psi^{\ls t}_{\!\nls\varrho,\mspace{0.75mu} L}}{\vphantom{|^|}\norm{\Psi^{\ls t}_{\!\nls\varrho,\mspace{0.75mu} L}}[2]\!\nls} -\tfrac{1}{L^{3/2}} \Psi^{\ls t}\nls \big(\mspace{-0.75mu}\tfrac{\vec{\cdot}}{L}\mspace{-0.75mu}\big)\nls }[H^1(\Lambda_L)]\nls =0,$$
        where
        $$\Psi^{\ls t}: \,\vec{y}\:\longmapsto e^{-\ls i\ls \mathfrak{b}\ls\abs{\Psi(\vec{y})}^2\ls t}\,\Psi(\vec{y})\in \mathfrak{A}^2(\Lambda_1),\quad t\in\R.$$
        \begin{proof}
            First, notice that
            \begin{align*}
                \limsup_{\varrho\to\infty}\limsup_{L\to\infty}&\,\norm*{\tfrac{\Psi^{\ls t}_{\!\nls\varrho,\mspace{0.75mu} L}}{\vphantom{|^|}\norm{\Psi^{\ls t}_{\!\nls\varrho,\mspace{0.75mu} L}}[2]\!\nls} -\tfrac{1}{L^{3/2}} \Psi^{\ls t}\nls \big(\mspace{-0.75mu}\tfrac{\vec{\cdot}}{L}\mspace{-0.75mu}\big)\nls }[H^1(\Lambda_L)]\!\!\leq\limsup_{\varrho\to\infty}\limsup_{L\to\infty}\bigg[\norm{\Psi^{\ls t}_{\!\nls\varrho,\ls L}}[H^1(\Lambda_L)]\abs*{\tfrac{1}{\norm{\Psi_{\!\varrho,\ls L}}[2]\!\nls} -\tfrac{1}{\!\nls\sqrt{\varrho L^3\ls}\ls}}+\\
                &\mspace{369mu}+\norm*{\,\tfrac{\Psi^{\ls t}_{\!\nls\varrho,\mspace{0.75mu} L}}{\!\nls\sqrt{\varrho L^3\ls}\ls} -\tfrac{1}{L^{3/2}} \Psi^{\ls t}\nls \big(\mspace{-0.75mu}\tfrac{\vec{\cdot}}{L}\mspace{-0.75mu}\big)\nls }[H^1(\Lambda_L)]\bigg]\\
                &\leq\limsup_{\varrho\to\infty}\limsup_{L\to\infty}\bigg[\sqrt{\norm{\Psi_{\!\nls\varrho,\ls L}}[2]^2+\mathscr{E}_{\varrho,\ls L}[\Psi_{\varrho,\ls L}]\ls}\ls\abs*{\tfrac{1}{\norm{\Psi_{\!\varrho,\ls L}}[2]\!\nls} -\tfrac{1}{\!\nls\sqrt{\varrho L^3\ls}\ls}}+\\
                &\mspace{150mu}+\norm*{\,\tfrac{\Psi^{\ls t}_{\!\nls\varrho,\mspace{0.75mu} L}}{\!\nls\sqrt{\varrho L^3\ls}\ls} -\tfrac{1}{L^{3/2}} \Psi^{\ls t}\nls \big(\mspace{-0.75mu}\tfrac{\vec{\cdot}}{L}\mspace{-0.75mu}\big)\nls }[H^1(\Lambda_L)]\bigg],
            \end{align*}
            where $\mathscr{E}_{\varrho,\ls L}$ is the Hartree energy functional~\eqref{def:HartreeFunctional}.
            Hence,
            \begin{align*}
                \limsup_{\varrho\to\infty}\limsup_{L\to\infty}&\norm*{\tfrac{\Psi^{\ls t}_{\!\nls\varrho,\mspace{0.75mu} L}}{\vphantom{|^|}\norm{\Psi^{\ls t}_{\!\nls\varrho,\mspace{0.75mu} L}}[2]\!\nls} -\tfrac{1}{L^{3/2}} \Psi^{\ls t}\nls \big(\mspace{-0.75mu}\tfrac{\vec{\cdot}}{L}\mspace{-0.75mu}\big)\nls }[H^1(\Lambda_L)]\!\!\!\!\leq\limsup_{\varrho\to\infty}\limsup_{L\to\infty}\bigg[\abs*{1-\tfrac{\norm{\Psi_{\!\varrho,\mspace{0.75mu} L}}[2]\!}{\!\nls\sqrt{\varrho L^3\ls}\ls}}+\sqrt{\tfrac{\mathscr{E}_{\varrho,\ls L}[\Psi_{\!\varrho,\mspace{0.75mu} L}]}{\varrho L^3}\ls}\ls\abs*{\tfrac{\!\nls\sqrt{\varrho L^3\ls}\ls}{\norm{\Psi_{\!\varrho,\mspace{0.75mu} L}}[2]\!}-1}+\\
                &\mspace{354mu}+\tfrac{1}{\!\nls\sqrt{\varrho L^3\ls }\ls} \norm*{\Psi^{\ls t}_{\!\varrho,\ls L}-\sqrt{\varrho\ls}\ls\Psi^{\ls t}\nls\big(\mspace{-0.75mu}\tfrac{\vec{\cdot}}{L}\mspace{-0.75mu}\big)}[H^1(\Lambda_L)]\!\bigg].
            \end{align*}
            Since the Hartree energy per particle is finite by construction, identity~\eqref{eq:becNormalization} entails
            $$\limsup_{\varrho\to\infty}\limsup_{L\to\infty}\norm*{\tfrac{\Psi^{\ls t}_{\!\nls\varrho,\mspace{0.75mu} L}}{\vphantom{|^|}\norm{\Psi^{\ls t}_{\!\nls\varrho,\mspace{0.75mu} L}}[2]\!\nls} -\tfrac{1}{L^{3/2}} \Psi^{\ls t}\nls \big(\mspace{-0.75mu}\tfrac{\vec{\cdot}}{L}\mspace{-0.75mu}\big)\nls }[H^1(\Lambda_L)]\!\!\!\!\leq \limsup_{\varrho\to\infty}\limsup_{L\to\infty}\,\tfrac{1}{\!\nls\sqrt{\varrho L^3\ls }\ls} \norm*{\Psi^{\ls t}_{\!\varrho,\ls L}-\sqrt{\varrho\ls}\ls\Psi^{\ls t}\nls\big(\mspace{-0.75mu}\tfrac{\vec{\cdot}}{L}\mspace{-0.75mu}\big)}[H^1(\Lambda_L)].$$
            Next, consider the solution $\Phi^{\ls t}_{\!\varrho,\ls L}$ to the Hartree equation in $\Lambda_L$ with initial datum $\Phi^{\, 0}_{\!\varrho,\ls L}=\sqrt{\varrho\ls}\ls\Psi\big(\tfrac{\vec{\cdot}}{L}\big)\in\mathfrak{A}^2(\Lambda_L)$.
            By the triangle inequality and Lemma~\ref{th:NLSconvergence},
            \begin{equation}
                \limsup_{\varrho\to\infty}\limsup_{L\to\infty}\norm*{\tfrac{\Psi^{\ls t}_{\!\nls\varrho,\mspace{0.75mu} L}}{\vphantom{|^|}\norm{\Psi^{\ls t}_{\!\nls\varrho,\mspace{0.75mu} L}}[2]\!\nls} -\tfrac{1}{L^{3/2}} \Psi^{\ls t}\nls \big(\mspace{-0.75mu}\tfrac{\vec{\cdot}}{L}\mspace{-0.75mu}\big)\nls }[H^1(\Lambda_L)]\!\!\!\!\leq \limsup_{\varrho\to\infty}\limsup_{L\to\infty}\,\tfrac{1}{\!\nls\sqrt{\varrho L^3\ls }\ls} \norm*{\Psi^{\ls t}_{\!\varrho,\ls L}\!-\Phi^{\ls t}_{\!\varrho,\ls L}}[H^1(\Lambda_L)].
            \end{equation}
            The problem is therefore reduced to proving the continuity of the Hartree flow in $H^1(\Lambda_L)$, as both $\Psi^{\ls t}_{\!\varrho,\ls L}$ and $\Phi^{\ls t}_{\!\varrho,\ls L}$ solve the same PDE with different, but close, initial values.
            We have already argued in the proof of Proposition~\ref{th:propagationExistenceBEC} and Lemma~\ref{th:NLSconvergence} that the $\Lp{2}$-norm of the derivative of $\Phi^{\ls t}_{\!\varrho,\ls L}$ vanishes in the thermodynamic limit, while Assumptions~\ref{ass:tailCondition} and~\ref{ass:kineticTailCondition} ensure the same for the kinetic energy of $\Psi^{\ls t}_{\!\varrho,\ls L}$ in the high-density thermodynamic limit.
            We are left to prove the convergence in $\Lp{2}[\Lambda_L]$.
            To this end, we refer to the computations carried out in Proposition~\ref{th:propagationExistenceBEC}, where the plane wave is replaced with the more generic profile $\Phi^{\ls t}_{\!\varrho,\ls L}$: by introducing $u^{\ls t}_{\varrho,\ls L}\nls =\Psi^{\ls t}_{\!\varrho,\ls L}\!-\Phi^{\ls t}_{\!\varrho,\ls L}$ and $v^{\ls t}_{\varrho,\ls L}\nls =\abs{u^{\ls t}_{\varrho,\ls L}}^2+2\Re\ls (\conjugate{\Phi}^{\ls t}_{\!\varrho,\ls L}\ls u^{\ls t}_{\varrho,\ls L})$, one has
            \begin{gather}
                \frac{d}{d\ls t}\ls \norm{u^{\ls t}_{\varrho,\ls L}}[2]^2=2\ls \Im\ls\scalar{u^{\ls t}_{\varrho,\ls L}}{\tfrac{1}{\varrho} \big(V_L\!\ast v^{\ls t}_{\varrho,\ls L}\big)\ls\Phi^{\ls t}_{\!\varrho,\ls L}}[2],\\
                \norm{v^{\ls t}_{\varrho,\ls L}}[2]\leq (\norm{\Psi^{\ls t}_{\!\varrho,\ls L}}[\infty]+3\norm{\Phi^{\ls t}_{\!\varrho,\ls L}}[\infty])\ls\norm{u^{\ls t}_{\varrho,\ls L}}[2].\nonumber
            \end{gather}
            Therefore,
            \begin{equation}
                \frac{d}{d\ls t}\ls \norm{u^{\ls t}_{\varrho,\ls L}}[2]^2\leq \tfrac{2\ls \mathfrak{b}}{\varrho} \big(\norm{\Psi^{\ls t}_{\!\varrho,\ls L}}[\infty]+3\ls\norm{\Phi^{\ls t}_{\!\varrho,\ls L}}[\infty]\big)\norm{\Phi^{\ls t}_{\!\varrho,\ls L}}[\infty]\ls\norm{u^{\ls t}_{\varrho,\ls L}}[2]^2\ls .
            \end{equation}
            By exploiting Gr\"onwall's lemma,
            \begin{equation}
                \norm{u^{\ls t}_{\varrho,\ls L}}[2]^2\leq \norm{u^{\ls 0}_{\varrho,\ls L}}[2]^2\,\exp \bigg[\tfrac{2\mathfrak{b}}{\varrho}\!\integrate[0; t]{\norm{\Phi^{\ls s}_{\!\varrho,\ls L}}[\infty]\big(\norm{\Psi^{\ls s}_{\!\varrho,\ls L}}[\infty]+3\ls \norm{\Phi^{\ls s}_{\!\varrho,\ls L}}[\infty]\big); ds}\bigg].
            \end{equation}
            Because of~\eqref{eq:SumNonlinearityControlInTime}, the argument of the exponential is uniformly bounded in $\varrho$ and $L$ whenever $$t<\tfrac{1}{2\ls\mathfrak{b}} \min\Big\{\tfrac{1}{\ls\overline{\nls S}^{\, 0}},\,\tfrac{1}{\norm{\FT{\Psi}}[\ell_1(\Z^3)]\nquad\!\!\!\!}\mspace{21mu}\Big\}^2\!,$$
            with $\ls\overline{\nls S}^{\, 0}\!\vcentcolon=\limsup\limits_{\varrho\to\infty}\limsup\limits_{L\to\infty}\ls\frac{1}{\!\nls\sqrt{\varrho\ls}\ls}\norm{\Psi_{\!\varrho,\ls L}}[\mathfrak{A}^0(\Lambda_L)]$.
            We stress that $\ls\overline{\nls S}^{\, 0}\!\leq \norm{\FT{\Psi}}[\ell_1(\Z^3)]$ (this can be seen by repeating the estimates of Remark~\ref{rmk:S0Boundedness} by taking account of the summability of the Fourier coefficients of the limit profile).
            The proposition is proven by taking account of Assumption~\ref{ass:initialBEC}.

        \end{proof}
    \end{prop}
    We conclude this section with a couple of remarks.
    \begin{note}
        We emphasise that all the results contained in this section, as well as Proposition~\ref{th:BECmomentum}, do not depend on the choice of the initial state, which is supposed to be the application of the Weyl operator to a quasi-vacuum state (a particular case of quasi-canonical coherent state) for the rest of the analysis.
        Indeed, the only requirement on which these outcomes rely (besides the technical Assumptions~\ref{ass:tailCondition},~\ref{ass:kineticTailCondition}) is the fact that the initial state exhibits quasi-complete condensation, \ie~Assumption~\ref{ass:initialBEC} (\cfr~Definition~\ref{def:QCBEC}).
    \end{note}
    \begin{note}\label{rmk:finiteTimeIntervalValidity}
        Propositions~\ref{th:nonlinearityControl},~\ref{th:LaplaceNonlinearityControl},~\ref{th:propagationExistenceBEC} and~\ref{th:propagationExistenceBECgeneral} hold true for a finite time interval.
        At this stage, we cannot rule out that this constraint is just a technical issue, rather than a physical limitation.\newline
        These propositions would hold for all times if we had global-in-time control of the nonlinearity of the Hartree equation, specifically of $\norm{\beta^{\ls t}_{\varrho,\ls L}}[\ell_1(\Z^3)]$.
        Even by considering Assumption~\ref{ass:translationalInvariance}, what can be proven is the pointwise convergence of $\beta^{\ls t}_{\varrho,\ls L}$ to $\delta_{\vec{0}}\ls $, by computing a time derivative and then exploiting Gr\"onwall's lemma by means of energy conservation.
        However, to bound $\norm{\beta^{\ls t}_{\varrho,\ls L}}[\ell_1(\Z^3)]$ we need the same convergence in the stronger $\ell_1$-topology.\newline
        In other words, what is left to prove\footnote{Actually, it is sufficient to have control of high momenta $M\nls <\nls \abs{\vec{m}}\nls <\nls L^{1+\ls r}$, for some $r\nls >\nls 0$, since in that case there exists $\delta_2\! >\nls 0$ such that the potential suppresses the nonlinearity for $\abs{\vec{m}}\nls \geq\nls  L^{1+\ls r}$.} (maybe under additional suitable assumptions) is
        $$\lim_{M\to\infty}\limsup_{\varrho\to\infty}\limsup_{L\to\infty}\sum_{\substack{\vec{m}\ls \in\,\Z^3 \ls :\\[1.5pt] \abs{\vec{m}} > M}} \abs{\beta^{\ls t}_{\varrho,\ls L}(\vec{m})}=0,\qquad \forall t\nls \geq \nls  0.$$
    \end{note}


\section{Growth Control of Fluctuations}\label{sec:excitationsControl}

In this section, we discuss the development of the proof of Theorem~\ref{th:gamma1Convergence}.\newline
First, consider the operator
$$\adj{\mathcal{U}}_{\varrho,\ls L}(t)\big(\ls \mathcal{N}+\ls \mathcal{G}_{\varrho,\ls L}(t)\big)\ls \mathcal{U}_{\varrho,\ls L}(t),\,\dom{\mathcal{H}_{\varrho,\ls L}}\in\linear{\fockS\big(\Lp{2}[\Lambda_L]\big)},$$
where $\mathcal{U}_{\varrho,\ls L}\!\in\nls \bounded{\fockS\big(\Lp{2}[\Lambda_L]\big)}$ and $\mathcal{G}_{\varrho,\ls L},\dom{\mathcal{H}_{\varrho,\ls L}}$ have been defined in~\eqref{def:fluctuationDynamics1} and~\eqref{def:effectiveGenerator}, respectively.
Differentiating this operator with respect to time (in the strong-operator topology) yields (\cfr~\eqref{def:generator})
$$\adj{\mathcal{U}}_{\varrho,\ls L}(t)\big(\!-\nls i\,[\mathcal{N},\mathcal{L}_{\varrho,\ls L}(t)]+\ls \dot{\mathcal{G}}_{\varrho,\ls L}(t)\big)\ls \mathcal{U}_{\varrho,\ls L}(t),$$
because $\mathcal{G}_{\varrho,\ls L}(t)$ trivially commutes with $\mathcal{L}_{\varrho,\ls L}(t)$.
Here $\dot{\mathcal{G}}_{\varrho,\ls L}(t)$ denotes the densely defined operator associated with the (strong) time-derivative of $\mathcal{G}_{\varrho,\ls L}(t)$.
Indeed, $\mathcal{G}_{\varrho,\ls L}(t)$ is strongly continuous with respect to time, since $t\longmapsto \Psi^{\ls t}_{\!\varrho,\ls L}$ is continuously differentiable (and in particular, it is a Lipschitz continuous map).\newline
In principle, the domain of the commutator between two operators is the intersection between the domains of the two different compositions one can carry out with such operators; however, in this case, one has a significant cancellation in the difference, which allows defining an extension to a larger domain.
To avoid confusion, let $\mathrm{ad}_{\mathcal{N}}\big(\mathcal{L}_{\varrho,\ls L}(t)\mspace{-0.75mu}\big)$ denote such an extension of the commutator $[\mathcal{N}, \mathcal{L}_{\varrho,\ls L}(t)]$.
Its expression in the quadratic form representation is
\begin{equation}\label{def:commutatorNGenerator}
    \begin{split}
        \mathrm{ad}_{\mathcal{N}}\big(\mathcal{L}_{\varrho,\ls L}(t)\mspace{-0.75mu}\big)[\psi]\vcentcolon=&\:\frac{2\ls i}{\varrho}\,\Im\!\integrate[\Lambda^2_L]{V_L(\vec{x}\nls -\nls \vec{y})\,\Psi^{\ls t}_{\!\varrho,\ls L}(\vec{x})\Psi^{\ls t}_{\!\varrho,\ls L}(\vec{y})\,\scalar{a_{\vec{y}}a_{\vec{x}}\psi}{\psi};\!d\vec{x}d\vec{y}}\,+\\
        &+\frac{2\ls i}{\varrho}\,\Im\!\integrate[\Lambda^2_L]{V_L(\vec{x}\nls -\nls \vec{y})\,\Psi^{\ls t}_{\!\varrho,\ls L}(\vec{y})\,\scalar{a_{\vec{y}}a_{\vec{x}}\psi}{a_{\vec{x}}\psi};\!d\vec{x}d\vec{y}},\qquad \psi\in\dom{\mathcal{N}}.
    \end{split}
\end{equation}
By a density argument, for every $T\nls >\nls 0$ the above discussion yields the identity
\begin{equation}\label{eq:startGronwall}
    \frac{d}{d\ls t} \!\left(\mathbb{E}_{\,\mathcal{U}_{\varrho,\ls L}(t)\ls \xi}[\ls \mathcal{N}\ls ] + \mathcal{G}_{\varrho,\ls L}(t)[\ls \mathcal{U}_{\varrho,\ls L}(t)\ls \xi]\right)=-i\,\mathrm{ad}_{\mathcal{N}}\big(\mathcal{L}_{\varrho,\ls L}(t)\mspace{-0.75mu}\big)[\ls \mathcal{U}_{\varrho,\ls L}(t)\ls \xi]+\dot{\mathcal{G}}_{\varrho,\ls L}(t)[\ls \mathcal{U}_{\varrho,\ls L}(t)\ls \xi],
\end{equation}
for all $\xi\nls \in\nls \fdom{\mathcal{H}_{\varrho,\ls L}}\nls \subset\nls \dom{\mathcal{N}}$ and $t\nls \in\nls [0,T]$, where $\dot{\mathcal{G}}_{\varrho,\ls L}(t)[\ls \cdot\ls ]$ is defined in~\eqref{eq:derivativeEffectiveGenerator}.
Both sides of equation~\eqref{eq:startGronwall} are well-defined because $\mathcal{U}_{\varrho,\ls L}(t)$ leaves $\fdom{\mathcal{H}_{\varrho,\ls L}}$ invariant (see footnote~\ref{foo:invariantfdomH}, Section~\ref{sec:qcBEC}).

\medskip

\noindent We are now ready to control the expectation of the number operator in the fluctuation dynamics.
The first step is to find a proper upper bound for the r.h.s. of equation~\eqref{eq:startGronwall} in order to establish a Gr\"onwall inequality.
\begin{lemma}\label{th:numberPlusGeneratorGronwall}
    Let $t\longmapsto\Psi^{\ls t}_{\!\varrho,\ls L}\!\in C^{\mspace{0.75mu}1}\big(\mspace{0.75mu}[0,\infty),\mathfrak{A}^0(\Lambda_L)\mspace{-0.75mu}\big)\mspace{-0.75mu}\cap C^{\ls 0}\big(\mspace{0.75mu}[0,\infty), \mathfrak{A}^2(\Lambda_L)\mspace{-0.75mu}\big)$ be the unique solution to the Hartree equation~\eqref{eq:HartreePDE} fulfilling Assumptions~\ref{ass:tailCondition} and~\ref{ass:kineticTailCondition}.\newline
    Then, provided $S^{\ls t}_{\nls \varrho,\ls L}$ defined by~\eqref{def:FourierCoefficientsSumShortcut}, one has for every $0\nls <\nls T_{\nls\varrho,\ls L}\nls <\nls \big(2\ls \FT{V}_\infty(\vec{0})\ls (S^{\ls 0}_{\nls \varrho,\ls L})^2\big)^{-1}$ that for all $t\in[0,T_{\nls\varrho,\ls L}]$ and $\xi\nls \in\nls \fdom{\mathcal{H}_{\varrho,\ls L}}$ there exists $\omega_{\varrho,\ls L}\nls >\nls 0$ satisfying $\limsup\limits_{\varrho\to\infty}\limsup\limits_{L\to\infty}\,\omega_{\varrho,\ls L}<\infty$ such that for fixed $\varrho, L >0$
    $$\mathbb{E}_{\,\mathcal{U}_{\varrho,\ls L}(t)\ls \xi}[\ls \mathcal{N}\ls ] + \abs*{\mathcal{G}_{\varrho,\ls L}(t)[\ls \mathcal{U}_{\varrho,\ls L}(t)\ls \xi]}\leq e^{\omega_{\varrho,\ls L}\,t}\big(\,\mathbb{E}_{\,\xi}[\ls \mathcal{N}\ls ] + \abs*{\mathcal{G}_{\varrho,\ls L}(0)[\xi]}\nls \big)+\left(e^{\omega_{\varrho,\ls L}\,t}\nls -1\right)\nls\tfrac{\norm{\Psi_{\mspace{-3.75mu}\varrho,\mspace{0.75mu} L}}[2]^2\nls}{\varrho}\ls\norm{\xi}^2,$$
    where $\mathcal{G}_{\varrho,\ls L}(t)[\ls \cdot\ls ]$ is the Hermitian quadratic form associated with the operator defined in~\eqref{def:effectiveGenerator}.
        
    \begin{proof}
        Observe that one can adopt the same computations carried out in Proposition~\eqref{th:generatorTermsAPEstimates} to have control of~\eqref{def:commutatorNGenerator}, obtaining
        \begin{equation}
            \begin{split}
                \abs*{\mathrm{ad}_{\mathcal{N}}\big(\mathcal{L}_{\varrho,\ls L}(t)\mspace{-0.75mu}\big)[\psi]}\leq &\: \tfrac{2}{\varrho}\,\norm{V_L\!\ast\nls \abs{\Psi^{\ls t}_{\!\varrho,\ls L}}^2}[\infty]\,\norm{\mathcal{N}^{\frac{1}{2}}\psi}^2+2\ls\tfrac{\norm{\Psi_{\mspace{-3.75mu}\varrho,\mspace{0.75mu} L}}[2]\nls}{\varrho}\, \norm{V^{\ls 2}_L\nls \ast\nls \abs{\Psi^{\ls t}_{\!\varrho,\ls L}}^2}[\infty]^{\frac{1}{2}}\, \norm{\mathcal{N}^{\frac{1}{2}}\psi}\,\norm{\psi}\,+\\
                &+\tfrac{\sqrt{8}}{\varrho}\ls \norm{V_L\!\ast\nls \abs{\Psi^{\ls t}_{\!\varrho,\ls L}}^2}[\infty]^{\frac{1}{2}} \sqrt{\mathcal{V}_{\nls L}[\psi]\ls }\,\norm{\mathcal{N}^{\frac{1}{2}}\psi},\qquad \psi\in\dom{\mathcal{N}}.
            \end{split}
        \end{equation}
        Exploiting Young's inequality for convolutions and the fact that $\norm{V_L}[1]\!\leq\nls  \FT{V}_\infty(\vec{0})\nls =\nls \mathfrak{b}$, we get for $\psi\nls \in\nls \dom{\mathcal{N}}$
        \begin{align*}
                \abs*{\mathrm{ad}_{\mathcal{N}}\big(\mathcal{L}_{\varrho,\ls L}(t)\mspace{-0.75mu}\big)[\psi]}\leq &\: \tfrac{2\ls \mathfrak{b}}{\varrho}\,\norm{\Psi^{\ls t}_{\!\varrho,\ls L}}[\infty]^2\,\norm{\mathcal{N}^{\frac{1}{2}}\psi}^2+2\ls\tfrac{\norm{\Psi_{\mspace{-3.75mu}\varrho,\mspace{0.75mu} L}}[2]\nls}{\varrho}\,\norm{V_L}[2]\,\norm{\Psi^{\ls t}_{\!\varrho,\ls L}}[\infty]\, \norm{\mathcal{N}^{\frac{1}{2}}\psi}\,\norm{\psi}\,+\\
                &+\tfrac{\!\nls \sqrt{8\ls \mathfrak{b}\ls }\ls }{\varrho}\ls \norm{\Psi^{\ls t}_{\!\varrho,\ls L}}[\infty] \sqrt{\mathcal{V}_{\nls L}[\psi]\ls }\,\norm{\mathcal{N}^{\frac{1}{2}}\psi}.\\
                \leq &\left[2\nls \left(1+\tfrac{2}{1-\epsilon}\nls \right)\!\mathfrak{b}+\tfrac{4\ls \varsigma}{2\ls \varsigma\ls -\ls 1}\ls \norm{V_L}[2]^2\right]\!\tfrac{\norm{\Psi^{\ls t}_{\!\varrho,\ls L}}[\infty]^2}{\varrho}\,\mathbb{E}_{\ls \psi}[\ls \mathcal{N}\ls ]\ls +\\
                &+\tfrac{1-\varepsilon}{2\ls \varrho}\ls \mathcal{V}_{\nls L}[\psi]+\nls \left(\tfrac{1}{2}\nls -\nls \tfrac{1}{4\ls \varsigma}\right)\nls\tfrac{\norm{\Psi_{\mspace{-3.75mu}\varrho,\mspace{0.75mu} L}}[2]^2\nls}{\varrho}\,\norm{\psi}^2,\qquad\forall \varepsilon\in (0,1),\,\varsigma>\tfrac{1}{2}.
        \end{align*}
        Making use of~\eqref{eq:SumLinftyControlled}, we obtain
        \begin{equation*}
            \begin{split}
                \abs*{\mathrm{ad}_{\mathcal{N}}\big(\mathcal{L}_{\varrho,\ls L}(t)\mspace{-0.75mu}\big)[\psi]}\leq&\,\left[2\nls \left(1+\tfrac{2}{1-\epsilon}\nls \right)\!\mathfrak{b}+\tfrac{4\ls \varsigma}{2\ls \varsigma\ls -\ls 1}\ls \norm{V_L}[2]^2\right]\! \big(S^{\ls t}_{\nls \varrho,\ls L}\big)^2\,\mathbb{E}_{\ls \psi}[\ls \mathcal{N}\ls ]\ls +\\
                &+\tfrac{1-\varepsilon}{2}\ls \mathcal{H}_{\varrho,\ls L}[\psi]+\nls \left(\tfrac{1}{2}\nls -\nls \tfrac{1}{4\ls \varsigma}\right)\nls\tfrac{\norm{\Psi_{\mspace{-3.75mu}\varrho,\mspace{0.75mu} L}}[2]^2\nls}{\varrho}\,\norm{\psi}^2,
            \end{split}
        \end{equation*}
        where $S^{\ls t}_{\nls \varrho,\ls L}$ has been defined in~\eqref{def:FourierCoefficientsSumShortcut}.
        Here we can estimate the Hamiltonian in terms of $\mathcal{G}_{\varrho,\ls L}(t)[\ls \cdot\ls ]$, by means of Corollary~\ref{th:generatorAPEstimate}, so that for all $\varepsilon\in(0,1),$ and $\varsigma\nls >\nls \frac{1}{2}$
        \begin{equation*}
            \begin{split}
                \abs*{\mathrm{ad}_{\mathcal{N}}\big(\mathcal{L}_{\varrho,\ls L}(t)\mspace{-0.75mu}\big)[\psi]}\leq &\,\left[\nls \left(2+\tfrac{4}{1-\varepsilon}+\tfrac{3}{2}+\tfrac{1}{\varepsilon}\nls \right)\!\mathfrak{b}+\nls \left(\tfrac{4\ls \varsigma}{2\ls \varsigma\ls -\ls 1}+\tfrac{\varsigma\ls }{4}\nls \right)\!\norm{V_L}[2]^2\right]\! \big(S^{\ls t}_{\nls \varrho,\ls L}\big)^2\:\mathbb{E}_{\ls \psi}[\ls \mathcal{N}\ls ]\ls +\\
                &+\tfrac{1}{2}\,\mathcal{G}_{\varrho,\ls L}(t)[\psi]+\tfrac{\norm{\Psi_{\mspace{-3.75mu}\varrho,\mspace{0.75mu} L}}[2]^2\nls}{2\ls\varrho}\, \norm{\psi}^2,
            \end{split}
        \end{equation*}
        which implies that the following holds, by minimising with respect to $\varepsilon\in (0,1)$ and $\varsigma\nls >\nls \frac{1}{2}$
        \begin{equation}
            \abs*{\mathrm{ad}_{\mathcal{N}}\big(\mathcal{L}_{\varrho,\ls L}(t)\mspace{-0.75mu}\big)[\psi]}\leq\nls \left[\tfrac{25}{2}\ls \mathfrak{b}+\tfrac{25}{8}\,\norm{V_L}[2]^2\right]\! \big(S^{\ls t}_{\nls \varrho,\ls L}\big)^2\:\mathbb{E}_{\ls \psi}[\ls \mathcal{N}\ls ]+\tfrac{1}{2}\ls \mathcal{G}_{\varrho,\ls L}(t)[\psi]+\tfrac{\norm{\Psi_{\mspace{-3.75mu}\varrho,\mspace{0.75mu} L}}[2]^2\nls}{2\ls\varrho}\,\norm{\psi}^2.
        \end{equation}
        Applying Corollary~\ref{th:generatorDerivativeAPEstimates}, we follow the same procedure for $\dot{\mathcal{G}}_{\varrho,\ls L}(t)[\ls \cdot\ls ]$ (choosing $\varepsilon\nls =\nls \frac{1}{2}$ when using Corollary~\ref{th:generatorAPEstimate}), to obtain
        \begin{equation}
            \begin{split}
                \abs{\dot{\mathcal{G}}_{\varrho,\ls L}(t)[\psi]}\leq &\left[4\!\left(2\mathfrak{b}+\norm{V_L}[2]^2\right)\nls \mathfrak{b}^2\big(S^{\ls t}_{\nls \varrho,\ls L}\big)^6\!+4\mathfrak{b}^2\big(S^{\ls t}_{\nls \varrho,\ls L}\big)^4\!+\!\left(\tfrac{7\ls \mathfrak{b}}{2}+\norm{V_L}[2]^2\right)\!\big(S^{\ls t}_{\nls \varrho,\ls L}\big)^2+\right.\\
                &\left.+\,4\!\left(2\ls \mathfrak{b}+ \norm{V_L}[2]^2\right)\!\big(T^{\ls t}_{\nls \varrho,\ls L}\big)^2\!+6\ls \mathfrak{b}\,S^{\ls t}_{\nls \varrho,\ls L} T^{\ls t}_{\nls \varrho,\ls L}\right]\nls \mathbb{E}_{\ls \psi}[\ls \mathcal{N}\ls ]+\tfrac{1}{2}\ls \mathcal{G}_{\varrho,\ls L}(t)[\psi]+\tfrac{\norm{\Psi_{\mspace{-3.75mu}\varrho,\mspace{0.75mu} L}}[2]^2\nls}{2\ls\varrho}\,\norm{\psi}^2,
            \end{split}
        \end{equation}
        where $T^{\ls t}_{\nls \varrho,\ls L}$ has been introduced in~\eqref{def:FourierCoefficientsKineticShortcut}.
        Hence,
        \begin{align}
            \abs*{\mspace{-0.75mu}\left(\mspace{-0.75mu}\mathrm{ad}_{\mathcal{N}}\big(\mathcal{L}_{\varrho,\ls L}(t)\mspace{-0.75mu}\big)\nls +\dot{\mathcal{G}}_{\varrho,\ls L}(t)\!\right)\!\nls [\,\mathcal{U}_{\varrho,\ls L}(t)\ls \xi]}\nls \leq&\: \abs*{\mathrm{ad}_{\mathcal{N}}\big(\mathcal{L}_{\varrho,\ls L}(t)\mspace{-0.75mu}\big)[\,\mathcal{U}_{\varrho,\ls L}(t)\ls \xi]}+\abs{\dot{\mathcal{G}}_{\varrho,\ls L}(t)[\,\mathcal{U}_{\varrho,\ls L}(t)\ls \xi]}\nonumber\\
            \leq & \nls \left[4\!\left(2\ls \mathfrak{b}+\norm{V_L}[2]^2\right)\nls \mathfrak{b}^2\big(S^{\ls t}_{\nls \varrho,\ls L}\big)^6\!+4\mathfrak{b}^2\big(S^{\ls t}_{\nls \varrho,\ls L}\big)^4+\right.\nonumber\\
            &\;+\!\left(16\ls \mathfrak{b}+\tfrac{33}{8}\ls \norm{V_L}[2]^2\right)\!\big(S^{\ls t}_{\nls \varrho,\ls L}\big)^2+\nonumber\\
            &\;\left.+\,4\!\left(2\ls \mathfrak{b}+\norm{V_L}[2]^2\right)\!\big(T^{\ls t}_{\nls \varrho,\ls L}\big)^2\!+6\ls \mathfrak{b}\,S^{\ls t}_{\nls \varrho,\ls L} T^{\ls t}_{\nls \varrho,\ls L}\right]\nls \mathbb{E}_{\,\mathcal{U}_{\varrho,\ls L}(t)\ls \xi}[\ls \mathcal{N}\ls ]\ls +\\
            &\nls +\abs*{\mathcal{G}_{\varrho,\ls L}(t)[\ls \mathcal{U}_{\varrho,\ls L}(t)\ls \xi]}+\tfrac{\norm{\Psi_{\mspace{-3.75mu}\varrho,\mspace{0.75mu} L}}[2]^2\nls}{\varrho}\,  \norm{\xi}^2.\nonumber
        \end{align}
        Taking into account equations~(\ref{eq:SumNonlinearityControlInTime}, \ref{eq:SumKineticNonlinearityControlInTime}), one finds that the time-dependent function in the square brackets can be bounded from above by an increasing function.
        Hereafter, $t\longmapsto h_{\varrho,\ls L}(t)\leq h_{\varrho,\ls L}(T_{\nls\varrho,\ls L})$ denotes such an increasing function, with $0\nls \leq\nls  t \nls \leq\nls T_{\nls\varrho,\ls L}\nls <\nls \frac{1}{\vphantom{|^|}2\ls (S^{\ls 0}_{\nls \varrho,\ls L})^2\ls \mathfrak{b}}$.
        Thus,
        $$\abs*{\mspace{-0.75mu}\left(\mspace{-0.75mu}\mathrm{ad}_{\mathcal{N}}\big(\mathcal{L}_{\varrho,\ls L}(t)\mspace{-0.75mu}\big)\nls +\dot{\mathcal{G}}_{\varrho,\ls L}(t)\!\right)\!\nls [\,\mathcal{U}_{\varrho,\ls L}(t)\ls \xi]}\nls \leq \omega_{\varrho,\ls L}\!\left(\mathbb{E}_{\,\mathcal{U}_{\varrho,\ls L}(t)\ls \xi}[\ls \mathcal{N}\ls ]+\abs*{\mathcal{G}_{\varrho,\ls L}(t)[\ls \mathcal{U}_{\varrho,\ls L}(t)\ls \xi]}+\tfrac{\norm{\Psi_{\mspace{-3.75mu}\varrho,\mspace{0.75mu} L}}[2]^2\nls}{\varrho}\,  \norm{\xi}^2\nls \right)\!,$$
        where
        \begin{equation}
            \omega_{\varrho,\ls L}\vcentcolon=\max\{1,h_{\varrho,\ls L}(T_{\nls\varrho,\ls L})\}
        \end{equation}
        satisfies $\limsup\limits_{\varrho\to\infty}\limsup\limits_{L\to\infty} \,\omega_{\varrho,\ls L}<\infty$.
        This follows because the limit superior commutes with non-decreasing, continuous functions such as $\max\{1, \ls \cdot\ls \}$, and because
        $$\limsup\limits_{\varrho\to\infty}\limsup\limits_{L\to\infty} S^{\ls 0}_{\nls \varrho,\ls L}\! <\nls\infty,\qquad \limsup\limits_{\varrho\to\infty}\limsup\limits_{L\to\infty} T^{\ls 0}_{\nls \varrho,\ls L}\!<\nls \infty.$$
        Moreover, the $\Lp{2}$-norm of the potential is bounded in $L$, as pointed out in Remark~\ref{rmk:potentialL2}. 

        \smallskip

        \noindent Now, suppose $\xi\in\fdom{\mathcal{H}_{\varrho,\ls L}}$ is such that $\mathcal{G}_{\varrho,\ls L}(t)[\,\mathcal{U}_{\varrho,\ls L}(t)\ls \xi]=0$.
        In this case, we have
        $$\frac{d}{d\ls t} \,\mathbb{E}_{\,\mathcal{U}_{\varrho,\ls L}(t)\ls \xi}[\ls \mathcal{N}\ls ] = -i\,\mathrm{ad}_{\mathcal{N}}\big(\mathcal{L}_{\varrho,\ls L}(t)\mspace{-0.75mu}\big)[\ls \mathcal{U}_{\varrho,\ls L}(t)\ls \xi]\leq \omega_{\varrho,\ls L}\! \left(\mathbb{E}_{\,\mathcal{U}_{\varrho,\ls L}(t)\ls \xi}[\ls \mathcal{N}\ls ] + \tfrac{\norm{\Psi_{\mspace{-3.75mu}\varrho,\mspace{0.75mu} L}}[2]^2\nls}{\varrho}\,\norm{\xi}^2\right)$$
        and Gr\"onwall's lemma provides the result (one trivially has $0\nls \leq\nls  \abs{\mathcal{G}_{\varrho,\ls L}(0)[\xi]}$).\newline
        Finally, in case $\mathcal{G}_{\varrho,\ls L}(t)[\,\mathcal{U}_{\varrho,\ls L}(t)\ls \xi]\neq 0$, the following holds
        $$\frac{d}{d\ls t}\abs*{\mathcal{G}_{\varrho,\ls L}(t)[\,\mathcal{U}_{\varrho,\ls L}(t)\ls \xi]} \leq \abs*{\frac{d}{d\ls t}\, \mathcal{G}_{\varrho,\ls L}(t)[\,\mathcal{U}_{\varrho,\ls L}(t)\ls \xi]}=\abs*{\dot{\mathcal{G}}_{\varrho,\ls L}(t)[\,\mathcal{U}_{\varrho,\ls L}(t)\ls \xi]}\nls .$$
        Hence,
        \begin{align*}
            \frac{d}{d\ls t}\!\left(\mathbb{E}_{\,\mathcal{U}_{\varrho,\ls L}(t)\ls \xi}[\ls \mathcal{N}\ls ]+\abs*{\mathcal{G}_{\varrho,\ls L}(t)[\,\mathcal{U}_{\varrho,\ls L}(t)\ls \xi]}\right)\!&\leq \abs*{\mathrm{ad}_{\mathcal{N}}\big(\mathcal{L}_{\varrho,\ls L}(t)\mspace{-0.75mu}\big)[\,\mathcal{U}_{\varrho,\ls L}(t)\ls \xi]}+\abs*{\dot{\mathcal{G}}_{\varrho,\ls L}(t)[\,\mathcal{U}_{\varrho,\ls L}(t)\ls \xi]}\\
            &\leq \omega_{\varrho,\ls L}\!\left(\mathbb{E}_{\,\mathcal{U}_{\varrho,\ls L}(t)\ls \xi}[\ls \mathcal{N}\ls ]+\abs*{\mathcal{G}_{\varrho,\ls L}(t)[\ls \mathcal{U}_{\varrho,\ls L}(t)\ls \xi]}+\tfrac{\norm{\Psi_{\mspace{-3.75mu}\varrho,\mspace{0.75mu} L}}[2]^2\nls}{\varrho}\,  \norm{\xi}^2\nls \right)\!.
        \end{align*}
        Also in this case, Gr\"onwall's lemma yields the result.

    \end{proof}
\end{lemma}
Now, all the ingredients required to control the number of fluctuations are in place.
\begin{proof}[Proof of Lemma~\ref{th:fluctuationsNumber}]
    One has the trivial upper bound
    $$\mathbb{E}_{\,\mathcal{U}_{\varrho,\ls L}(t)\ls \xi_{\varrho,\ls L}}[\ls \mathcal{N}\ls ]\leq \mathbb{E}_{\,\mathcal{U}_{\varrho,\ls L}(t)\ls \xi_{\varrho,\ls L}}[\ls \mathcal{N}\ls ]+\abs*{\mathcal{G}_{\varrho,\ls L}(t)[\ls \mathcal{U}_{\varrho,\ls L}(t)\ls \xi_{\varrho,\ls L}]}\nls ,$$
    which is controlled by Lemma~\ref{th:numberPlusGeneratorGronwall} in terms of quantities evaluated for the initial quasi-vacuum state
    $$\mathbb{E}_{\,\mathcal{U}_{\varrho,\ls L}(t)\ls \xi_{\varrho,\ls L}}[\ls \mathcal{N}\ls ]\leq e^{\omega_{\varrho,\ls L}\,t}\big(\,\mathbb{E}_{\,\xi_{\varrho,\ls L}}[\ls \mathcal{N}\ls ] + \abs*{\mathcal{G}_{\varrho,\ls L}(0)[\xi_{\varrho,\ls L}]}\nls \big)+\left(e^{\omega_{\varrho,\ls L}\,t}\nls -1\right)\nls\tfrac{\norm{\Psi_{\mspace{-3.75mu}\varrho,\mspace{0.75mu} L}}[2]^2\nls}{\varrho},$$
    where $\omega_{\varrho,\ls L}$ fulfils 
    $$\limsup_{\varrho\to\infty}\limsup_{L\to\infty}\,\omega_{\varrho,\ls L}<\infty.$$
    By Corollary~\ref{th:generatorAPEstimate}, for all $\psi\nls \in\nls \fdom{\mathcal{H}_{\varrho,\ls L}}$ and $\varepsilon\nls >\nls 0$
    \begin{equation}
        \abs*{\mathcal{G}_{\varrho,\ls L}(0)[\psi]}\leq (1\nls +\nls \varepsilon)\ls \mathcal{H}_{\varrho,\ls L}[\psi]+\!\left[\left(3+\tfrac{2}{\varepsilon}\right)\nls \norm{V_L}[1]+\tfrac{1}{4}\ls \norm{V_L}[2]^2\right]\nls \tfrac{\norm{\Psi_{\mspace{-3.75mu}\varrho,\mspace{0.75mu} L}}[\infty]^2\!\nls }{\varrho}\;\mathbb{E}_{\,\psi}[\ls \mathcal{N}\ls ]+\tfrac{\norm{\Psi_{\mspace{-3.75mu}\varrho,\mspace{0.75mu} L}}[2]^2\nls}{\varrho}\ls\norm{\psi}^2.
    \end{equation}
    For the sake of simplicity, take $\varepsilon=1$.
    Thus, recalling $\norm{V_L}[1]\leq \FT{V}_\infty(\vec{0})=\mathfrak{b}$
    $$\mathbb{E}_{\,\mathcal{U}_{\varrho,\ls L}(t)\ls \xi_{\varrho,\ls L}}[\ls \mathcal{N}\ls ]\!\leq\nls e^{\omega_{\varrho,\ls L}\,t}\!\left[\ls 2\ls \mathcal{H}_{\varrho,\ls L}[\xi_{\varrho,\ls L}]+\nls \left(\! \left(5\ls \mathfrak{b}+\tfrac{1}{4}\ls \norm{V_L}[2]^2\right)\nls \tfrac{\norm{\Psi_{\mspace{-3.75mu}\varrho,\mspace{0.75mu} L}}[\infty]^2\nls}{\varrho}+1\right)\nls \mathbb{E}_{\,\xi_{\varrho,\ls L}}[\ls \mathcal{N}\ls ]\right]\nls +\nls \left(2\ls e^{\omega_{\varrho,\ls L}\,t}\! -1\right)\nls\tfrac{\norm{\Psi_{\mspace{-3.75mu}\varrho,\mspace{0.75mu} L}}[2]^2\nls}{\varrho}.$$
    Dividing both sides by $\varrho L^3$ shows that the r.h.s. vanishes in the iterated limit, owing to Proposition~\ref{th:energyOfVoid} (whose conditions are satisfied under Assumptions~\ref{ass:initialBEC},~\ref{ass:energyConsistence},~\ref{ass:tailCondition}, and~\ref{ass:kineticTailCondition}), the boundedness of $\norm{V_L}[2]$ (see Remark~\ref{rmk:potentialL2}), and Assumption~\ref{ass:tailCondition}, which ensures that
    $$\limsup_{\varrho\to\infty}\limsup_{L\to\infty} \tfrac{\norm{\Psi_{\mspace{-3.75mu}\varrho,\mspace{0.75mu} L}}[\infty]\!}{\sqrt{\varrho\ls }}\leq \limsup_{\varrho\to\infty}\limsup_{L\to\infty}S^{\ls 0}_{\varrho,\ls L} <\infty.$$

\end{proof}

\begin{proof}[Proof of Theorem~\ref{th:gamma1Convergence}]
    Consider the identity between kernels~\eqref{eq:gamma1MinusLimit}.
    Taking the Hilbert--Schmidt norms of the associated operators, applying the generalised triangle inequality for higher powers on the r.h.s.
    \begin{equation}
        \bigg\lvert\sum_{i\ls =1}^n a_i\bigg\rvert^p\!\!\leq n^{p-1}\sum_{i\ls =1}^n\, \abs{a_i}^p,\qquad \{a_i\}_{i\ls =1}^n\!\subset\C,\, p\geq 1,
    \end{equation}
    and recalling~\eqref{eq:NumberFromIntegral} yield
    $$\norm*{\gamma^{(1)}_{\varphi^{\ls t}_{\varrho,\ls L}}\!\nls -\tfrac{|\Psi^{\ls t}_{\!\varrho,\ls L}\rangle\langle \Psi^{\ls t}_{\!\varrho,\ls L}|}{\vphantom{|^{|}}\norm{\mathcal{N}^{1/2}\mathcal{W}(\Psi^{\ls 0}_{\!\varrho,\ls L})\ls \xi_{\varrho,\ls L}}^2\!\nls }\ls }[\mathrm{HS}]\!\leq\tfrac{\varrho L^3}{\norm{\vphantom{|^{|}}\mathcal{N}^{1/2}\mathcal{W}(\Psi^{\ls 0}_{\!\varrho,\ls L})\ls \xi_{\varrho,\ls L}}^2\!\nls }\,\sqrt{\tfrac{3}{(\varrho L^3)^2\!\nls }\nls \left(\mathbb{E}_{\,\mathcal{U}_{\varrho,\ls L}(t)\ls \xi_{\varrho,\ls L}}[\ls \mathcal{N}\ls ]\right)^{\!2}\!\nls +\tfrac{6}{\varrho L^3}\,\mathbb{E}_{\,\mathcal{U}_{\varrho,\ls L}(t)\ls \xi_{\varrho,\ls L}}[\ls \mathcal{N}\ls ]\ls }\ls .$$
    Since $\xi_{\varrho,\ls L}\!\in\nls \fdom{\mathcal{H}_{\varrho,\ls L}}$ is a quasi-vacuum state with respect to $\Psi^{\ls 0}_{\varrho,\ls L}$, the vector $\mathcal{W}(\Psi^{\ls 0}_{\!\varrho,\ls L})\ls \xi_{\varrho,\ls L}$ is a quasi-coherent state (with an expected number of particles close to $\varrho L^3$).
    Therefore, due to Lemma~\ref{th:fluctuationsNumber}
    \begin{equation}
        \lim_{\varrho\to\infty}\limsup_{L\to\infty} \norm*{\gamma^{(1)}_{\varphi^{\ls t}_{\varrho,\ls L}}\!\nls -\tfrac{|\Psi^{\ls t}_{\!\varrho,\ls L}\rangle\langle \Psi^{\ls t}_{\!\varrho,\ls L}|}{\vphantom{|^{|}}\norm{\mathcal{N}^{1/2}\mathcal{W}(\Psi^{\ls 0}_{\!\varrho,\ls L})\ls \xi_{\varrho,\ls L}}^2\!\nls }\ls }[\mathrm{HS}]\!=0.
    \end{equation}
    Finally, one has
    \begin{align*}
        \norm*{\gamma^{(1)}_{\varphi^{\ls t}_{\varrho,\ls L}}\!\nls -\tfrac{|\Psi^{\ls t}_{\!\varrho,\ls L}\rangle\langle \Psi^{\ls t}_{\!\varrho,\ls L}|}{\vphantom{|^|}\norm{\Psi^{\ls t}_{\!\nls\varrho,\mspace{0.75mu} L}}[2]^2}}[\mathrm{HS}]\!\leq &\, \norm*{\gamma^{(1)}_{\varphi^{\ls t}_{\varrho,\ls L}}\!\!\nls -\tfrac{|\Psi^{\ls t}_{\!\varrho,\ls L}\rangle\langle \Psi^{\ls t}_{\!\varrho,\ls L}|}{\vphantom{|^{|}}\norm{\mathcal{N}^{1/2}\mathcal{W}(\Psi^{\ls 0}_{\!\varrho,\ls L})\ls \xi_{\varrho,\ls L}}^2\!\nls }\ls }[\mathrm{HS}]\!\!+\, \abs*{\tfrac{\norm{\Psi^{\ls 0}_{\!\nls\varrho,\mspace{0.75mu} L}}[2]^2}{\vphantom{|^{|}}\norm{\mathcal{N}^{1/2}\mathcal{W}(\Psi^{\ls 0}_{\!\varrho,\ls L})\ls \xi_{\varrho,\ls L}}^2\!\nls }-1\ls}\\
        = & \, \norm*{\gamma^{(1)}_{\varphi^{\ls t}_{\varrho,\ls L}}\!\!\nls -\tfrac{|\Psi^{\ls t}_{\!\varrho,\ls L}\rangle\langle \Psi^{\ls t}_{\!\varrho,\ls L}|}{\vphantom{|^{|}}\norm{\mathcal{N}^{1/2}\mathcal{W}(\Psi_{\!\varrho,\ls L})\ls \xi_{\varrho,\ls L}}^2\!\nls }\ls }[\mathrm{HS}]\!\!+\, \abs*{\tfrac{\norm{\Psi_{\mspace{-3.75mu}\varrho,\mspace{0.75mu} L}}[2]^2}{\vphantom{|^{|}}\norm{\Psi_{\mspace{-3.75mu}\varrho,\mspace{0.75mu} L}}[2]^2\ls +\ls\scalar{\xi_{\varrho,\ls L}}{(\mathcal{N}+\ls\phi(\Psi_{\!\varrho,\ls L}))\ls \xi_{\varrho,\ls L}}}-1\ls}\\
        =&\, \norm*{\gamma^{(1)}_{\varphi^{\ls t}_{\varrho,\ls L}}\!\!\nls -\tfrac{|\Psi^{\ls t}_{\!\varrho,\ls L}\rangle\langle \Psi^{\ls t}_{\!\varrho,\ls L}|}{\vphantom{|^{|}}\norm{\mathcal{N}^{1/2}\mathcal{W}(\Psi_{\!\varrho,\ls L})\ls \xi_{\varrho,\ls L}}^2\!\nls }\ls }[\mathrm{HS}]\!\!+\, \tfrac{\abs*{\scalar{\xi_{\varrho,\ls L}}{(\mathcal{N}+\ls\phi(\Psi_{\!\varrho,\ls L}))\ls \xi_{\varrho,\ls L}}}}{\vphantom{|^{|}}\norm{\Psi_{\mspace{-3.75mu}\varrho,\mspace{0.75mu} L}}[2]^2\ls +\ls\scalar{\xi_{\varrho,\ls L}}{(\mathcal{N}+\ls\phi(\Psi_{\!\varrho,\ls L}))\ls \xi_{\varrho,\ls L}}}\\
        \leq &\,\norm*{\gamma^{(1)}_{\varphi^{\ls t}_{\varrho,\ls L}}\!\!\nls -\tfrac{|\Psi^{\ls t}_{\!\varrho,\ls L}\rangle\langle \Psi^{\ls t}_{\!\varrho,\ls L}|}{\vphantom{|^{|}}\norm{\mathcal{N}^{1/2}\mathcal{W}(\Psi_{\!\varrho,\ls L})\ls \xi_{\varrho,\ls L}}^2\!\nls }\ls }[\mathrm{HS}]\!\!+\, \tfrac{\norm{\mathcal{N}^{1/2}\xi_{\varrho,\ls L}}^2\ls +\ls 2\ls\norm{\Psi_{\mspace{-3.75mu}\varrho,\mspace{0.75mu} L}}[2]\norm{\mathcal{N}^{1/2}\xi_{\varrho,\ls L}}}{\vphantom{|^{|}}\norm{\Psi_{\mspace{-3.75mu}\varrho,\mspace{0.75mu} L}}[2]^2\ls +\ls\scalar{\xi_{\varrho,\ls L}}{(\mathcal{N}+\ls\phi(\Psi_{\!\varrho,\ls L}))\ls \xi_{\varrho,\ls L}}},
    \end{align*}
    which vanishes in the iterated limit.
    Specifically, the upper bound in the right-hand side asymptotically behaves as
    $$\sqrt{\tfrac{6}{\varrho L^3}\,\mathbb{E}_{\,\mathcal{U}_{\varrho,\ls L}(t)\ls \xi_{\varrho,\ls L}}[\ls \mathcal{N}\ls ]\ls }\ls +2\ls\sqrt{\tfrac{1}{\vphantom{|^|}\norm{\Psi_{\mspace{-3.75mu}\varrho,\mspace{0.75mu} L}}[2]^2\nls}\,\mathbb{E}_{\ls \xi_{\varrho,\ls L}}[\ls \mathcal{N}\ls ]\ls }\ls .$$
    This justifies the rate of convergence discussed in Remark~\ref{rmk:rateOfConvergence}
    
\end{proof}

\noindent\textbf{\large Declarations}

\medskip

\paragraph{Conflict of interest} No conflict of interest: the authors have no financial or non-financial interests to disclose.
\paragraph{Data Availability} Data and code sharing is not applicable to this article, as no new data were created or analysed and no code was used in this study.
\paragraph{Funding} Both authors have been financially supported by the European Research Council through the ERC Stg MaTCh, grant agreement no. 101117299.



\bibliographystyle{plainnat}
\bibliography{references}

\end{document}